\documentclass[12pt,a4paper]{report}
\usepackage[pdftex]{graphicx}
\usepackage[latin1]{inputenc}
\usepackage{amsthm}
\usepackage{amssymb}
\usepackage{amsmath,url}
\usepackage{slashed}
\usepackage{braket,dsfont}
\usepackage{euscript}
\usepackage{afterpage}
\newcommand\blankpage{%
    \null
    \thispagestyle{empty}%
    \addtocounter{page}{-1}%
    \newpage}

\theoremstyle{definition}

\newtheorem{thma}{Theorem}

\newtheorem{stmt}[thma]{Statement}
\newtheorem{theorem}{Theorem}

\newtheorem{definition}{Definition}

\def\d{\delta}

\def\ta{\~a}

\def\ccd{\c{c}}

\def\1{{\bf 1}}

\def\inbar{\vrule height1.5ex width.4ptdepth0pt}
\def\rlx{\relax\leavevmode}
\def\I{\leavevmode
\hbox{\small1\kern-3.8pt\normalsize1}}
\def\openone{\leavevmode
\hbox{\small1\kern-3.3pt\normalsize1}}
\def\Ione{\rlx{\rm 1\kern-2.7pt l}}
\font\cmss=cmss10 \font\cmsss=cmss10 at 7pt
\def\ZZ{\rlx\leavevmode \ifmmode\mathchoice
{\hbox{\cmssZ\kern-.4em Z}} {\hbox{\cmss Z\kern-.4em Z}}
{\lower.9pt\hbox{\cmsss Z\kern-.36em Z}}
{\lower1.2pt\hbox{\cmsssZ\kern-.36em Z}} \else{\cmss Z\kern-.4em
Z}\fi}
\def\Ik{\rlx{\rm I\kern-.18em k}}
\def\IC{\rlx\leavevmode
\ifmmode\mathchoice {\hbox{\kern.33em\inbar\kern-.3em{\rm C}}}
{\hbox{\kern.33em\inbar\kern-.3em{\rm C}}}
{\hbox{\kern.28em\sinbar\kern-.25em{\rm C}}}
{\hbox{\kern.25em\ssinbar\kern-.22em{\rm C}}}
\else{\hbox{\kern.3em\inbar\kern-.3em{\rm C}}}\fi}
\def\IP{\rlx{\rmI\kern-.18em P}}
\def\IR{\rlx{\rm I\kern-.18em R}}
\def\IN{\rlx{\rm I\kern-.20em N}}
\newcommand{\mb}\mbox

\newcommand{\ol}\overline
\newcommand{\ul}\underline
\newcommand{\ti}\tilde
\newcommand{\wt}\widetilde
\newcommand{\wh}\widehat
\newcommand{\bv}\breve
\newcommand{\dg}\dagger

\newcommand{\be}{\begin{equation}}
\newcommand{\ee}{\end{equation}}
\newcommand{\bl}{\begin{eqnarray}&}
\newcommand{\el}{&\end{eqnarray}}
\newcommand{\bq}{\begin{eqnarray}}
\newcommand{\eq}{\end{eqnarray}}

\begin{document}
\begin{titlepage}
  \begin{center}
    \LARGE{\textsc{ Universidade Federal Fluminense}} \\
           \textsc{Instituto de F\'isica} \\          
          
    \par\vfill
    \LARGE{{\bf \textit{Exploring new horizons of the Gribov problem in Yang-Mills theories}}}
    \par\vfill
   
    \par\vfill
    \small{\textsc{Ant\^onio Duarte Pereira Junior}} \\
          \small{\textsc{Niter\'oi, 2016}} \\          
          
  \end{center}
\end{titlepage}

\afterpage{\blankpage}

\clearpage
\thispagestyle{empty}
\begin{center}

ANT\^ONIO DUARTE PEREIRA JUNIOR

\vfill

EXPLORING NEW HORIZONS OF THE GRIBOV PROBLEM IN YANG-MILLS THEORIES

\vspace{3.0cm}

\begin{flushright}
\begin{minipage}{0.50\textwidth}

A thesis submitted to the Departamento\linebreak
de F\'isica - UFF in partial fulfillment\linebreak
of the requirements for the degree of\linebreak
Doctor in Sciences (Physics).

\end{minipage}
\end{flushright}

\vspace{3.0cm}

Advisor: Prof. Dr. Rodrigo Ferreira Sobreiro

\vfill

Niterói-RJ\\2016

\end{center}

\afterpage{\blankpage}

\clearpage
\begin{flushright}
\begin{minipage}{0.5\textwidth}

\vspace{15.0cm} 

\textit{To the reason of everything. To my cosmic love. To my state of love and trust. To Anna Gabriela.}

\end{minipage}
\end{flushright}

\afterpage{\blankpage}

\clearpage
\chapter*{Acknowledgements}
\addcontentsline{toc}{chapter}{Acknowledgements}

Scientific research is definitely a collaborative activity. But this is in the broad sense. Not only the scientific collaborators are fundamental to make progress in research, but also all support from different people outside the scientific community is crucial. In this small space, I try to express my genuine gratitude to scientific or not collaborators who helped me to start my career. It is a cliche, but now that I am writing this words, I realize how difficult is to describe the importance of all people that played some direct or indirect influence on my Ph.D activity. Nevertheless, I will try my best.

First, I would like to thank my advisor Prof. Rodrigo Sobreiro. His guidance during the last four years was crucial for me. He trusted me, gave me freedom and let me follow my personal interests. Since 2009 we have been discussing physics and is not even possible to discriminate what I have learned from him. Thank you very much for everything!

During my Ph.D I had the opportunity to spend one year and half at SISSA in Trieste. There I met two professors who taught me a lot of physics and scientific research in general: Prof. Loriano Bonora and Prof. Roberto Percacci. They received me, gave me all support and dedicated a lot of time to me. I am not able to express how generous they were and it is even more difficult to emphasize how inspired I am by them to continue my career. Grazie mille per tutti!

My most sincere thanks to my group mates in Brazil: Anderson Tomaz, Tiago Ribeiro and Guilherme Sadovski. First for their friendship, company and amazing coffees/lunches. Second for softening the hard times. Third for the invaluable discussions on physics and scientific collaboration. These were wonderful times! Also, in SISSA, I had the amazing company of a brazilian group mate: Bruno Lima de Souza. Thank you very much for making my time in Trieste even better, with lots of pizzas, japanese food, piadina, gelato and wonderful panini con porcina. I am grateful for his lessons on physics and Mathematica. Many many thanks for all of you being much more than group mates but marvelous friends! 

During this journey I met many people that strongly contributed to my scientific formation. My estimated friends and collaborators from UERJ: Pedro Braga, M\'arcio Capri, Diego Fiorentini, Diego Granado, Marcelo Guimar\~aes, Igor Justo, Bruno Mintz, Let\'icia Palhares and Silvio Sorella. They were fundamental to this thesis and every time I go to UERJ I learn more and more about non-perturbative quantum field theories. A special thanks goes to Marcelo for his amazing course on quantum field theories and to Silvio Sorella for his very clear and nice explanations and discussions on the Gribov problem. Part of this huge group but located a bit far from Maracan\~a is David Dudal. Thank you very much for infinitely many discussions about Yang-Mills theories and scientific career, for your patience, attention and support all the time. I am very grateful to Markus Huber for our discussions, conversations, e-mails and skypes. To Reinhard Alkofer for encouragement, very nice and deep comments on Yang-Mills theories and to highly interesting (and non-trivial) questions that made me improve my understanding on non-perturbative Yang-Mills theories. To Urko Reinosa for stimulating discussions and attention! Many thanks to Henrique Gomes and Flavio Mercati. They opened their Shape Dynamics doors to me. I am very grateful for discussions and several hints concerning scientific career. Also, their independent way of thinking is very stimulating and has been teaching me a lot in the last months. To Rodrigo Turcati, an estimated friend, for our discussions and coffees. In particular, many thanks for a careful reading of the manuscript.

Special thanks to my group mates at the Albert Einstein Institute, where I spent almost four months of my Ph.D. First, I express my gratitude to my supervisor Prof. Daniele Oriti for sharing his huge knowledge on quantum gravity and his very clear point of view of physics. To Joseph Bengeloun for discussions and for very nice and stimulating questions. My sincere acknowledgements to Alex Kegeles, Goffredo Chirco, Giovanni Tricella, Marco Finocchiaro, Isha Kotecha, Dine Ousmane-Samary and Cedrick Miranda for discussions and company during this period! 

I thank my professors at UFF: Luis Oxman for sharing his knowledge on quantum field theories and Yang-Mills as well as very nice conversations during lunch, Ernesto Galv\~ao for very nice courses since my undergraduate studies and for his support, Nivaldo Lemos for teaching me not only physics but a way of thinking of it, Marco Moriconi for his continuous support and discussions, Jorge S\'a Martins for his invaluable ``Lectures on Physics", Marcelo Sarandy for support and for sharing his ideas and knowledge, Caio Lewenkopf for courses, conversations and for being an extraordinary professional and Antonio Tavares da Costa for his support during my Ph.D studies.  

An impossible-to-express gratefulness to my family. Mom and Dad for giving me everything I needed to follow my dreams and for their effort to provide all the support irrespective of the occasion. They showed me the real meaning of dedication and obstinacy. Grandma and Grandpa (\textit{in memoriam}) for making the world a place free of problems to live. Granny was the first to take me to a center of research and Granpa the first to know (and support) my decision of studying physics. My parents in law for a wonderful time together and to give me support in many different situations. My uncle Antonio and aunt Lidia for support and care during my entire life. To my sisters Carol and Andresa and to my cousin Pedro for their love and support. To my in-laws Bruna, Liana, Artur, Biel, Bruno, Thiago, Cristiano, Carol and Sandra for their friendship and company. Last, but not least, to my nephew Lucas for being this wonderful kid! 

To the love of my life, Anna Gabriela, for giving me all the support and love a humankind can imagine. She encouraged me to follow all my desires, dreams and plans along these years. It is simply not possible to express my feelings here. Thank you for letting me understand the real meaning of love and happiness, for being my best friend, my best company and for all your care. I love you, minha Pequenininha. 

I am lucky to have so many extraordinary friends. To Bella and Duim for their company, frienship and love. To Bruno, Jimmy, RVS, Pedro and Rafael for being the brotherhood! To Fred, Victor, Ana Clara, D\'ebora, Tatyana, La\'is, St\'ephanie, J\'essica, Danilo, Laise, Allan, Luiza, L\'eo, Rosembergue, J\'ulia, Stefan, Tatiana (Bio),Brenno and Tikito for being part of my life. It is a pleasure to thank Mary Brand\~ao and Carlos Alexandre (\textit{in memoriam}) for giving me all the inspiration to follow my career, for their kindness and care.

Finally, I am grateful to CNPq, CAPES, SISSA and DAAD for financial support along the last years.

\pagestyle{empty}

\newpage
\begin{flushright}
\begin{minipage}{0.5\textwidth}

\vspace{15.0cm} 

\textit{``Physics is like sex: sure, it may give some practical results, but that's not why we do it."}\\
\sffamily \small Richard Feynman \rmfamily 

\end{minipage}
\end{flushright} 

\newpage 
\chapter*{Abstract}
\addcontentsline{toc}{chapter}{Abstract}

The understanding of the non-perturbative regime of Yang-Mills theories remains a challenging open problem in theoretical physics. Notably, a satisfactory description of the confinement of gluons (and quarks in full quantum chromodynamics) is not at our disposal so far. In this thesis, the Refined Gribov-Zwanziger framework, designed to provide a proper quantization of Yang-Mills theories by taking into account the existence of the so-called Gribov copies is explored. Successfully introduced in the Landau gauge, the Refined Gribov-Zwanziger set up does not extend easily to different gauges. The main reason is that a clear formulation of the analogue of the Gribov horizon in the Landau gauge is obstructed by technical difficulties when more sophisticated gauges are chosen. Moreover, the Refined Gribov-Zwanziger action breaks BRST symmetry explicitly, making the task of extracting gauge invariant results even more difficult. The main goal of the present thesis is precisely to provide a consistent framework to extend the Refined Gribov-Zwanziger action to gauges that are connected to Landau gauge via a gauge parameter. Our main result is the reformulation of the theory in the Landau gauge with appropriate variables such that a \textit{non-perturbative} BRST symmetry is constructed. This symmetry corresponds to a deformation of the standard BRST symmetry by taking into account non-perturbative effects. This opens a toolbox which allow us to explore what would be the Gribov horizon in different gauges as linear covariant and Curci-Ferrari gauges. Consistency with gauge independence of physical quantities as well as the computation of the gluon propagator in these gauges is provided. Remarkably, when lattice or functional methods results are available, we verify very good agreement with our analytical proposal giving support that it could provide some insights about the non-perturbative regime of Yang-Mills theories. A positivity violating gluon propagator in the infrared seems to be a general feature of the formalism. Gluons, therefore, cannot be interpreted as stable particles in the physical spectrum of the theory being, thus, confined. 

\clearpage
\chapter*{List of Publications}
\addcontentsline{toc}{chapter}{List of Publications}

Here I present my complete list of publications during my Ph.D. 

\begin{itemize}
	\item \textit{``More on the non-perturbative Gribov-Zwanziger quantization of linear covariant gauges,''} \\
  M.~A.~L.~Capri, D.~Dudal, D.~Fiorentini, M.~S.~Guimaraes, I.~F.~Justo, A.~D.~Pereira, B.~W.~Mintz, L.~F.~Palhares, R.~F.~Sobreiro and S.~P.~Sorella,\\
  Phys.\ Rev.\ D \textbf{93}, no. 6, 065019 (2016)


  \item \textit{``Non-perturbative treatment of the linear covariant gauges by taking into account the Gribov copies,''}\\
	M.~A.~L.~Capri, A.~D.~Pereira, R.~F.~Sobreiro and S.~P.~Sorella,\\  
  Eur.\ Phys.\ J.\ C \textbf{75}, no. 10, 479 (2015)

  \item \textit{``Exact nilpotent nonperturbative BRST symmetry for the Gribov-Zwanziger action in the linear covariant gauge,''}\\
	M.~A.~L.~Capri, D.~Dudal, D.~Fiorentini, M.~S.~Guimaraes, I.~F.~Justo, A.~D.~Pereira, B.~W.~Mintz, L.~F.~Palhares, R.~F.~Sobreiro and S.~P.~Sorella,\\  
  Phys.\ Rev.\ D \textbf{92}, no. 4, 045039 (2015)

	\item \textit{``Regularization of energy-momentum tensor correlators and parity-odd terms,''}\\
	L.~Bonora, A.~D.~Pereira and B.~L.~de Souza,\\  
  JHEP \textbf{1506}, 024 (2015)

	\item \textit{``Gribov ambiguities at the Landau-maximal Abelian interpolating gauge,''}\\
  A.~D.~Pereira, Jr. and R.~F.~Sobreiro,\\
  Eur.\ Phys.\ J.\ C \textbf{74}, no. 8, 2984 (2014)\\

  \item \textit{``On the elimination of infinitesimal Gribov ambiguities in non-Abelian gauge theories,''}	\\
 A.~D.~Pereira and R.~F.~Sobreiro,\\
 Eur.\ Phys.\ J.\ C \textbf{73}, 2584 (2013)\\
	
	\item \textit{``Dark gravity from a renormalizable gauge theory,''}\\
T.~S.~Assimos, A.~D.~Pereira, T.~R.~S.~Santos, R.~F.~Sobreiro, A.~A.~Tomaz and V.~J.~V.~Otoya,\\
    arXiv:1305.1468 [hep-th].
\end{itemize}

\noindent This thesis intends to cover in detail the first three ones plus further material to appear. 

After the submission of the thesis, the following papers were published:

\begin{itemize}
\item \textit{``Gauges and functional measures in quantum gravity I: Einstein theory,''}\\
  N.~Ohta, R.~Percacci and A.~D.~Pereira,\\  
  JHEP \textbf{1606}, 115 (2016)
	
\item \textit{``A local and BRST-invariant Yang-Mills theory within the Gribov horizon,''}\\
M.~A.~L.~Capri, D.~Dudal, D.~Fiorentini, M.~S.~Guimaraes, I.~F.~Justo, A.~D.~Pereira, B.~W.~Mintz, L.~F.~Palhares, R.~F.~Sobreiro and S.~P.~Sorella,\\  
  arXiv:1605.02610 [hep-th]

\item \textit{``Non-perturbative BRST quantization of Euclidean Yang-Mills theories in Curci-Ferrari gauges,''}\\
A.~D.~Pereira, R.~F.~Sobreiro and S.~P.~Sorella,\\
arXiv:1605.09747 [hep-th]	
\end{itemize}
\tableofcontents

\pagestyle{plain}

\chapter{Introduction}

Theoretical and experimental physicists face a very interesting moment of elementary particle physics: The Standard Model (SM) of particle physics (taking into account neutrinos masses), constructed about five decades ago seems to describe nature much better than we expected. Up to now, the most crucial tests it was submitted to were successfully overcame. The most recent and urgent test was the detection or not of the Higgs boson in the LHC. Its detection in 2012 \cite{Chatrchyan:2012ufa} was responsible for a great excitement among physicists and put the SM as one of the biggest intellectual achievements of humankind. Despite of aesthetic discussions concerning the beauty or not of the SM, it is undeniable it provides our best understanding of fundamental physics up to date. 

The current paradigm establishes we have four fundamental interactions in nature: The electromagnetic, strong, weak and gravitational. A consistent quantum description of the first three aforementioned interactions is provided by the SM. Gravity stays outside of this picture. Seemingly, the coexistence of gravity and quantum mechanics in a consistent framework requires a profound change in our current way of thinking of the other interactions. This challenging problem of providing a quantum theory of gravity is one of the biggest problems in theoretical physics and the lack of experimental data to guide us in a path instead of the other makes the problem even worse. 

We must comment, however, that is far from being accepted that the SM is the final word about the fundamental interactions. In particular, besides its apparent inconsistency with gravity, a prominent problem which, so far, has no satisfactory explanation within the SM is the existence of dark matter. Also, the strong CP problem and matter-antimatter asymmetry correspond to other examples of phenomena which are not currently accommodated in the SM. Those facts point toward a physics beyond the SM and this is strongly investigated in the current years. 

Nevertheless, despite of problems which are not (at least so far) inside the range of the SM power, there is a different class of problems which are those we believe the SM is able to describe, but due to technical difficulties or our ignorance about how to control all the scales believed to be described by the SM with a single mathematical tool, are still open. The focus of this thesis is precisely to get a better understanding on this sort of problem. 

\section{The gauge theory framework}

The SM model is built upon a class of quantum field theories known as \textit{non-Abelian gauge} or \textit{Yang-Mills theories}. This name is due to the seminal work by Yang and Mills \cite{Yang:1954ek} where these theories were introduced. Essentially, they generalize the $U(1)$ gauge invariance of electromagnetism to non-Abelian groups. The particular case considered by Yang and Mills was the group $SU(2)$. This group was supposed to represent the isotopic spin rotations and due to its non-Abelian nature, the analogue of photons in this model self-interacts. Also, due to the requirement of gauge invariance, these fields have to represent massless particles. At that time, this was a big obstacle for the Yang-Mills model of isotopic spin, since those massless particles should be easily observed and no such particle was detected. Nevertheless, the gauge invariance principle \cite{Salam:1961en} - the determination of the form of interactions due to the invariance under certain gauge symmetry - was theoretically powerful. An inspired work by Utiyama \cite{Utiyama:1956sy}, who considered a wider class of groups instead of just $SU(2)$, showed how such principle could coexist with electromagnetism, Yang-Mills theories and even general relativity. Also, the quantization of these theories was worked out by Feynman, Faddeev, Popov and De Witt in \cite{Feynman:1963ax,Faddeev:1967fc,DeWitt:1967uc}. 

However, it was the construction of the Glashow-Salam-Weinberg theory for the electroweak sector (the unification of electromagnetism and weak interactions) which brought Yang-Mills theories as the arena to formulate the elementary interactions. The realization that the strong interactions could be formulated through a Yang-Mills-type of theories was not easily clear. In fact, the construction of quantum chromodynamics (QCD) had very interesting turns which, for instance, led to the construction of string theory. The electroweak theory is a gauge theory with gauge group $SU(2)\times U(1)$ and QCD, a gauge theory for $SU(3)$. These sectors of the SM encompass very different physical mechanisms. In particular, the electroweak sector suffers a spontaneous symmetry breaking giving masses to the gauge fields, the gauge bosons. This mechanism is driven by the Higgs boson and after 2012, this picture is well grounded by experimental data. QCD, on the other hand, displays confinement, mass gap, chiral symmetry breaking and asymptotic freedom. In this thesis, we will particularly focus on the strong interaction sector. More precisely, we will disregard the existence of fermions along this thesis \textit{i.e.} we will focus on a \textit{pure} Yang-Mills theory. The reason is technical: The pure Yang-Mills theory already displays phenomena as confinement and mass gap. Being a simpler theory than full QCD, we believe it is useful to provide some insights for the more complicated theory with fermions. Nevertheless, although simpler, pure Yang-Mills theory is far from being simple. In the next section we introduce the referred problems in the context of pure Yang-Mills theories. 

\section{Pure Yang-Mills theories}

Pure Yang-Mills theories describe the dynamics of gauge bosons - which we will simply call ``gluons", since we are considering these theories in the context of the strong interactions. At the classical level, the action which dictates the dynamics of such particles is given by\footnote{We refer the reader to Ap.~\ref{appendixA} for our conventions.}

\begin{equation}
S_{\mathrm{YM}}=\frac{1}{4}\int d^dx\,F^{a}_{\mu\nu}F^{a}_{\mu\nu}\,.
\label{pym1}
\end{equation}  
The non-Abelian structure of the gauge group includes in \eqref{pym1} cubic and quartic interaction terms for the gluon field. These self-interaction terms drive a highly non-trivial dynamics for Yang-Mills theories already at the classical level. However, it is at the quantum level that intriguing phenomena take place. Even the quantization procedure itself already brings very subtle points (in fact, it is the topic studied in this entire thesis). Before pointing out these subtleties in the quantization procedure, we present some key results of the quantum theory.

\subsection{Asymptotic freedom}

One of the most remarkable features of quantum Yang-Mills theories is the so-called asymptotic freedom, \cite{Gross:1973id,Politzer:1973fx}. Working the explicit perturbative renormalization of these theories, we obtain at one-loop order, 

\begin{equation}
g^2(\mu)=\frac{1}{\frac{11N}{16\pi^2}\mathrm{ln}\frac{\mu^2}{\Lambda^2_{\mathrm{QCD}}}}\,,
\label{af1}
\end{equation}
where $\mu$ is an energy scale, $\Lambda^2_{\mathrm{QCD}}$ is a renormalization group invariant cut-off and $g$, the coupling constant of the theory. From expression \eqref{af1}, the running of the coupling is such that for high energies \textit{i.e.} $\mu>>\Lambda_{\mathrm{QCD}}$, $g$ goes to zero. In other words: In the UV (short distances), the coupling tends to zero. It means that the theory is UV complete and well-defined up to arbitrary high energy scales. This is precisely what is known as \textit{asymptotic freedom}. An immediate consequence of this fact is that for high energies, the coupling is small and thus, perturbation theory becomes an efficient tool to be applied, since the perturbative series is based on powers of $g$. For high energies, namely for short distances, gluons are weakly interacting and behave (almost) as free particles.

On the other hand, if we take smaller values for $\mu$, the value of $g$ increases. In particular, as $\mu\rightarrow\Lambda_{\mathrm{QCD}}$, $g\rightarrow\infty$. Sometimes this is referred as Landau pole and its existence is due to the breakdown of perturbation theory. To see this, we note that for

\begin{equation}
\mu^2=\Lambda^{2}_{\mathrm{QCD}}\mathrm{e}^{\frac{16\pi^2}{11N}}\,,
\label{af2}
\end{equation}
we have $g=1$. At this level, the perturbative expansion is not trustful since $g$ is not small. Hence, the expression \eqref{af1} is not meaningful at this scale. We see thus that the existence of the Landau pole is due to the fact that one enters the \textit{non-perturbative} regime where perturbation theory cannot be safely applied. Here, a challenging problem takes place: From the analytical point of view, the main tool at our disposal in standard quantum field theories is precisely perturbation theory. However, going towards the infrared scale of Yang-Mills, or, equivalently, large distances regime, perturbation theory does not apply and one needs a genuine \textit{non-perturbative} setting. Although many clever techniques were developed to attack non-perturbative phenomena, a systematic framework analogous to perturbation theory is not at our disposal up to date. The state-of-the-art is the gluing of complementary results from different non-perturbative approaches to try to build a consistent picture. We shall comment more about this later on. 

\subsection{Confinement}

Considering full QCD, a basic fact is observed: Quarks and gluons, the fundamental particles of the theory, are not observed in the spectrum of the theory as asymptotic states \cite{Alkofer:2000wg}. Even if we remove the quarks and stay just with gluons, this remains true. This phenomenon is known as \textit{confinement} of quark and gluons. We call QCD a \textit{confining} theory. Pure Yang-Mills are also \textit{confining}. Although widely accepted by our current understanding of elementary particle physics, an explanation of \textit{why} confinement exists still lacks. In fact, up to date we have several proposals to deal with this problem, but any of them give a full picture of the story. It is also accepted that QCD (or pure Yang-Mills) should be the correct framework to describe confinement. Our difficulty of providing a satisfactory understanding of this phenomenon should be directly associated with our ignorance of controlling the non-perturbative regime of QCD (or Yang-Mills theories). Flowing to the infrared, the theory becomes strongly interacting. Pictorially, as we separate quarks or gluons, the coupling which controls their interaction increases becoming so strong that we cannot separate them further. As such, quarks and gluons are not observed freely, but only in \textit{colorless} bound states, the hadrons. 

In this thesis, we will focus on one particular proposal to describe confinement, the so-called \textit{Refined Gribov-Zwanziger} framework. The essential feature of this set up is the finding that the standard quantization of Yang-Mills theories through the \textit{Faddeev-Popov} procedure is not completely satisfactory at the non-perturbative level. Taking into account an improvement of the Faddeev-Popov method and additional non-perturbative effects as the formation of non-trivial vacuum condensates bring features that are in agreement with the existence of confinement of gluons (and, possibly, quarks in recent proposals). Before pointing out the essential features of the Faddeev-Popov method which might be problematic in the infrared regime, we comment on another important non-perturbative effect.

\subsection{Dynamical mass generation} 

In four dimensions, the absence of a dimensionfull parameter in the classical pure Yang-Mills action could naively convince us that there is no room for the generation of mass parameters in the theory. On the other hand, it is known that the renormalization procedure intrinsically introduces a dimensionfull (arbitrary) scale $\mu$ and $\Lambda_{\mathrm{QCD}}$ in the game. We can see from eq.\eqref{af1} that the dimensionless coupling $g$ is written in terms of a ratio of dimensionfull parameters which arise from the renormalization procedure.   

Thus, having introduced those scales due to quantum effects, we might as well ask if physical mass parameters can emerge. For physical parameters, we mean those that respect the renormalization group equation. To make this statement more precise, we write the beta function associated with the coupling $g$ as a function of powers of $g^{2n}$,

\begin{equation}
\beta_{g^2} (g^2)=g^4\beta_0+g^6\beta_1+\ldots\,.
\label{dmg1}
\end{equation}
Assuming the existence of a \textit{physical} mass parameter $m$, we demand its invariance under the renormalization group equation

\begin{equation}
\mu\frac{d}{d\mu}m=\left(\mu\frac{\partial}{\partial \mu}+\beta_{g^2}(g^2)\frac{\partial}{\partial g^2}\right)m=0\,.
\label{dmg2}
\end{equation}
If we stick to leading order in perturbation theory \cite{Collins}, it is easy to check that 

\begin{equation}
m\,\propto\,\mu\,\mathrm{e}^{\frac{1}{\beta_0 g^2}}\,.
\label{dmg3}
\end{equation}
For an asymptotically free theory, $\beta_0<0$. Hence, at the perturbative regime, $m\rightarrow 0$, while as $g$ increases when flowing to the IR, $m$ does not vanish. We see thus a genuine non-perturbative nature of $m$ and as a final remark, due to the singularity at $g^2=0$ of expression \eqref{dmg3}, the perturbative expansion of $m$ might be problematic. 

The simple observation presented above shows that novel mass parameters might be dynamically generated at the quantum level. Moreover, this is a genuine non-perturbative feature, since perturbation theory is not able to generate the non-perturbative expression \eqref{dmg3}. A very nice concrete example of such phenomenon is given by the two-dimensional Gross-Neveu model \cite{Gross:1974jv}, where a non-trivial vacuum expectation value for $\langle\bar{\psi}\psi\rangle$ is dynamically generated. 

A similar issue might very well happen in pure Yang-Mills theories. In fact, the introduction of a mass parameter in the IR regime of Yang-Mills theories seems to be a very welcome feature. As usual, mass parameters provide a consistent infrared regularization. Therefore, it is logically acceptable to conceive that if Yang-Mills theories are supposed to describe the IR regime of the strong interactions and it is plagued by IR divergences, a mass parameter might be a consistent way of controlling such problems. 

Different phenomenological, theoretical and lattice results favor the existence of non-trivial vacuum condensates in Yang-Mills theories. In this thesis, they will play a crucial role and we shall return to this topic several times along the text. 

\section{Quantization of Yang-Mills theories}

A covariant quantization of Yang-Mills theories in $d$ Euclidean dimensions with gauge group $SU(N)$ can be formally written as

\begin{equation}
\EuScript{Z}_{\mathrm{YM}}=\int \left[\EuScript{D}A\right]\,\mathrm{e}^{-S_{\mathrm{YM}}}\,,
\label{qym1.0}
\end{equation}
where we use the notation $A\equiv\left\{A^{a}_{\mu}(x)\,,\,\,\,x\in\mathbb{R}^d\right\}$ with the formal definition of the measure as

\begin{equation}
\left[\EuScript{D}A\right]=\prod_{x}\prod_{\mu}\prod_{a}\mathrm{d}A^{a}_{\mu}(x)\,.
\label{qym2.0}
\end{equation}
This is a very formal definition and might not be properly well-defined. For a proper mathematical definition, a lattice construction of this measure is possible, see \cite{Wilson:1974sk}. 

To proceed with the standard perturbative quantization using \eqref{qym1.0}, we should define the Feynman rules for Yang-Mills theories. One of the building blocks is the gluon propagator which is obtained, by definition, out of the quadratic part of $S_{\mathrm{YM}}$. Up to quadratic order, the Yang-Mills action is 

\begin{equation}
S^{\mathrm{quad}}_{\mathrm{YM}}=\frac{1}{2}\int d^dx\,A^{a}_{\mu}\underbrace{\delta^{ab}\left(\partial_{\mu}\partial_{\nu}-\delta_{\mu\nu}\partial^2\right)}_{K^{ab}_{\mu\nu}}A^{b}_{\nu}\,.
\label{qym3.0}
\end{equation}
It is simple to note that action \eqref{qym3.0} is invariant under gauge transformations. For concreteness, we can write a gauge transformation (see eq.\eqref{a5}) and disregard terms containing the coupling $g$ to keep the same order as action \eqref{qym3.0},

\begin{equation}
A'^a_{\mu}=A^a_{\mu}-\partial_{\mu}\xi^a\,.
\label{qym4.0}
\end{equation}
For an arbitrary gauge parameter $\xi$, action \eqref{qym3.0} is left invariant by \eqref{qym4.0} due to the form of the kernel $K^{ab}_{\mu\nu}$. However, it is precisely the form of $K^{ab}_{\mu\nu}$ that ensures gauge invariance which gives rise to ``problems" to the very definition of the gluon propagator. The reason: Gauge invariance is ensured by the fact that $K^{ab}_{\mu\nu}$ develops zero-modes,

\begin{equation}
K^{ab}_{\mu\nu}\partial_{\nu}\xi^b=0\,.
\label{qym5.0}
\end{equation}
Hence, since the gluon propagator should be defined by $\left(K^{-1}\right)^{ab}_{\mu\nu}$, we face the problem that gauge invariance hinders the inversion of $K$ due to the presence of zero-modes. As is very well-known, this problem is cured by the gauge-fixing procedure, here called the Faddeev-Popov procedure, see \cite{Faddeev:1967fc}. In the next subsection we describe the method. As we shall see, this procedure relies on some assumptions which will play a key role in this thesis.

\subsection{Dealing with gauge symmetry - Part I}

The measure \eqref{qym2.0} is defined in the very wild space $\mathcal{A}$ of all gauge potential configurations $A=\left\{A^{a}_{\mu}(x)\,,\,\,\,x\in\mathbb{R}^d\right\}$. We formally write $\mathcal{A}\equiv\left\{A\right\}$. This measure enjoys invariance under the \textit{local} gauge transformations group $\mathcal{G}$, written as

\begin{equation}
\mathcal{G}=\left\{ U\right\}\,\,\,\,\mathrm{with}\,\,\,\,U=\left\{U(x)\,\in\,SU(N)\,;\,x\in\mathbb{R}^d\right\}\,.
\label{qym6.0}
\end{equation} 
Intuitively, is not difficult to accept that \eqref{qym2.0} is invariant under gauge transformations \eqref{a5}. The first term of the the gauge transformation corresponds to a ``rotation" on $A$ which preserves the measure, while the second term is just a harmless ``translation". Hence, by construction,

\begin{equation}
\left[\EuScript{D}A\right]=\left[\EuScript{D}A^{U}\right]\,,
\label{qym7.0}
\end{equation}

\begin{figure}[t]
	\centering
		\includegraphics[width=0.50\textwidth]{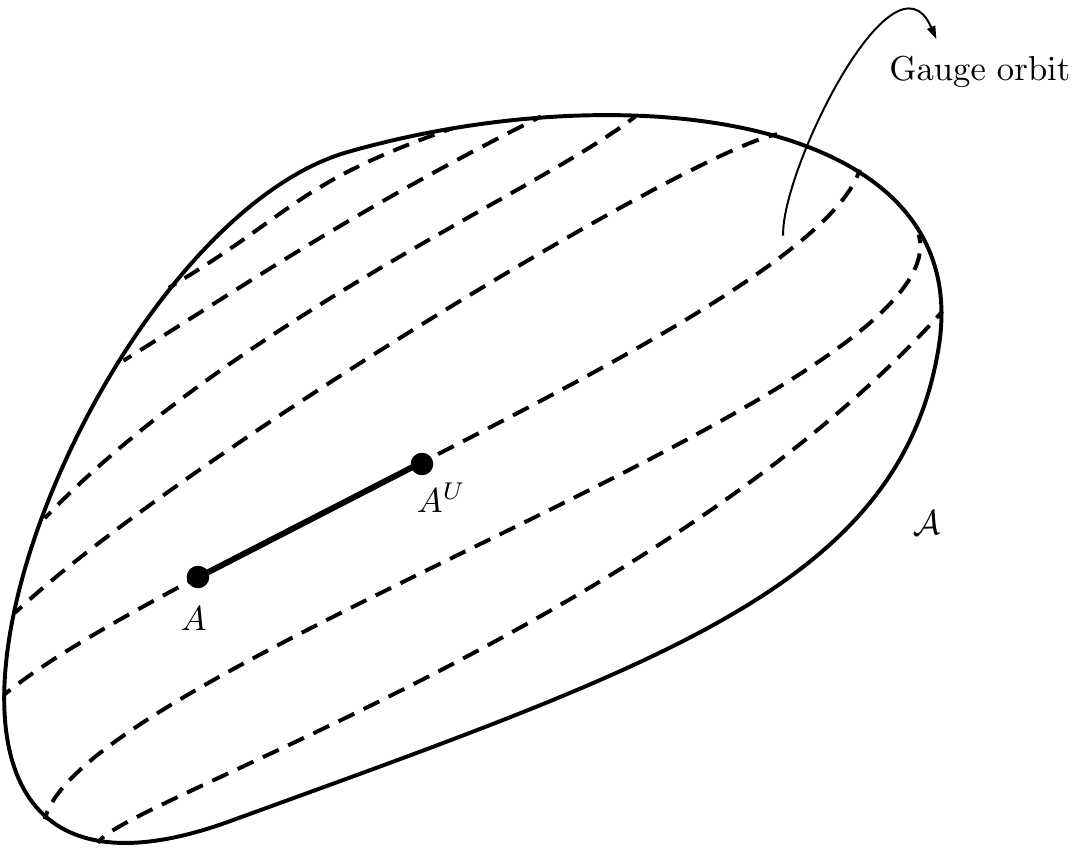}
	\caption{Space of gauge field configurations. This space is sliced by gauge orbits, namely, subset of gauge configurations which are related through gauge transformations.}
	\label{spaceA}
\end{figure}

\noindent which means the path integral measure is invariant under $\mathcal{G}$. An useful pictorial representation of the action of $\mathcal{G}$ on $\mathcal{A}$ is given by Fig.~\ref{spaceA}. Given a gauge field configuration $A\,\in\,\mathcal{A}$, we can pick an element of $\mathcal{G}$ and perform a gauge transformation parametrized by $U$. The resulting field is denoted as $A^U$. All gauge field configurations connected to $A$ via a gauge transformation lie on the dashed line represented in Fig.~\ref{spaceA} which $A$ and $A^U$ belong to. This dashed line is called a \textit{gauge orbit} and different dashed lines in Fig.~\ref{spaceA} represent different gauge orbits of gauge fields which are not related through the action of $\mathcal{G}$. The space of physically inequivalent configurations is denoted as $\mathcal{C}$ and is usually called \textit{moduli space}. It is written as the following quotient space,

\begin{equation}
\mathcal{C}=\frac{\mathcal{A}}{\mathcal{G}}\,.
\label{qym8.0}
\end{equation}
In principle, the path integral could be performed in the entire $\mathcal{A}$. This is achieved, in principle, in the lattice formulation of Yang-Mills theories. On the other hand, in the continuum, we argued just before this subsection that the construction of one of the building blocks of perturbation theory, the gluon propagator, is ill-defined due to gauge invariance. 

A possible strategy is to reduce our space of (functional) integration. Instead of integrating over the entire space $\mathcal{A}$, which contains physically equivalent as well as configurations which are not related through a gauge transformation, we pick one representative per gauge orbit. We then integrate over physically distinguishable gauge configurations. This idea is the basis of the gauge-fixing procedure. Although the idea is simple, we should be careful with the construction of such method. First, the resulting functional integration should be independent on the way we gauge fix, namely, on our choice of each representative per orbit. Second, our gauge-fixing choice should be such that for each gauge orbit we collect just one and only one representative. If each orbit contributes with a different number of representatives, their contribution will have different weights to the path integral. This will lead to inconsistencies with our previous requirement. A gauge-fixing which collects one representative per orbit is said to be \textit{ideal}.

In order to collect one representative per orbit, we introduce a \textit{gauge fixing} function(al) $F[A]$. The gauge-fixing is implemented by the equation

\begin{equation}
F[A]=0\,.
\label{qym9.0}
\end{equation}
We emphasize that $F[A]$ is an oversimplified notation. In fact, $F[A]$ denotes an application from $\mathcal{A}$ to $\mathbb{G}=\prod_x su(N)_x$, the local Lie algebra, and is in fact a set of functions given by

\begin{figure}[t]
	\centering
		\includegraphics[width=0.60\textwidth]{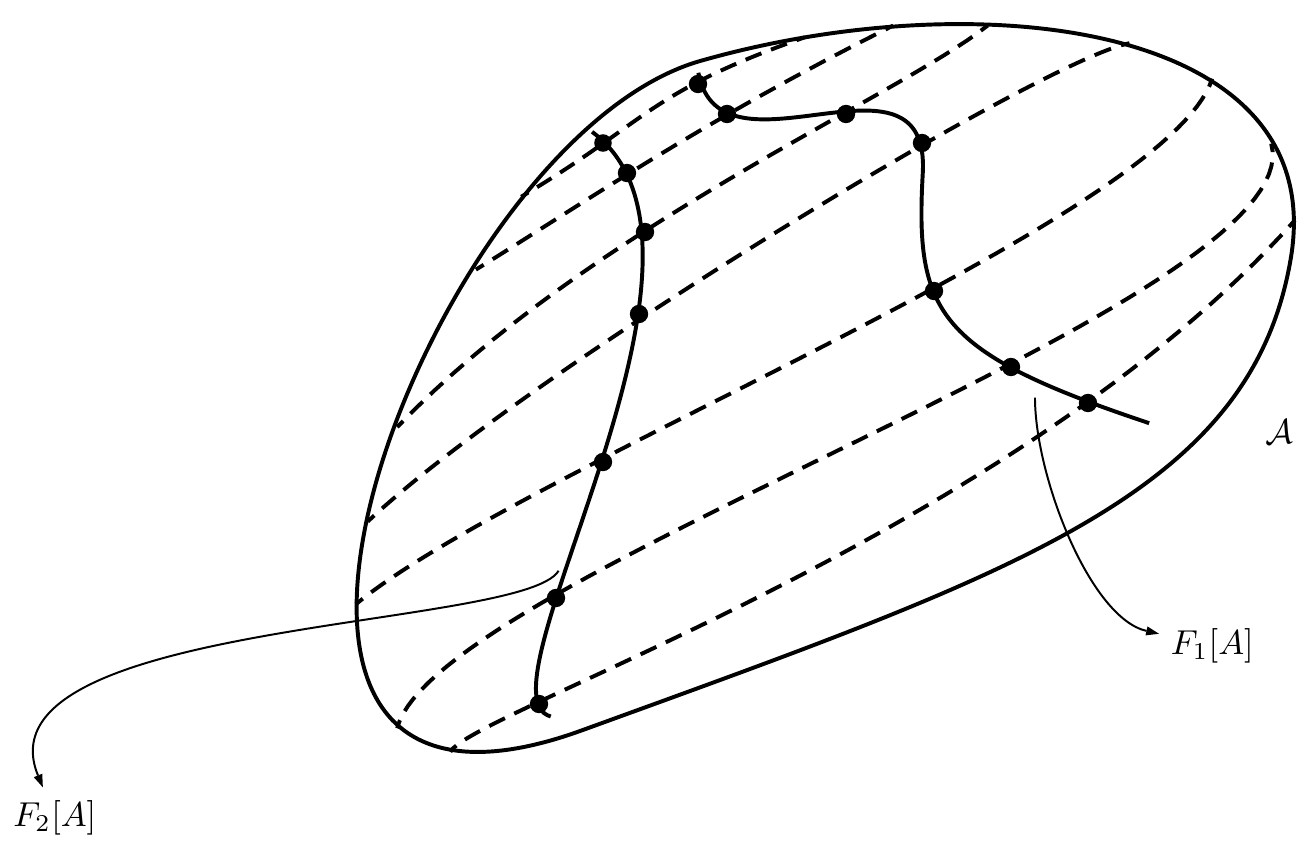}
		\caption{Two different choices of \textit{ideal} gauge-fixing functions $F_1[A]$ and $F_2[A]$. Each continuum section traced for each gauge-fixing function corresponds to a gauge slice.}
	\label{spaceA2}
\end{figure}

\begin{equation}
F[A]=\left\{F^a\left(A_{\mu}(x),\partial_{\mu}A_{\nu}(x),\ldots;x\right)\,;\,x\,\in\,\mathbb{R}^d\,;a=1,\ldots,N^2-1\right\}\,.
\label{qym10.0}
\end{equation}
By assumption, for each gauge orbit eq.\eqref{qym10.0} has a single solution. The collection of solutions of eq.\eqref{qym10.0} defines a \textit{gauge slice} or \textit{gauge section}, as represented in Fig.~\ref{spaceA2}. The gauge-fixing goal is to reduce the path integral measure, in a consistent fashion, to the gauge slice. As previously described, different choices $F_1[A]$ and $F_2[A]$ give rise to different gauge slices as shown im Fig.~\ref{spaceA2}, but the path integral result should not change for different gauge slicing. Following the representation of Fig.~\ref{spaceA2}, we give an example of a gauge-fixing which is not ideal in Fig.~\ref{spaceA3}. To construct the Faddeev-Popov procedure, we assume from now on that the gauge-fixing is ideal. We shall point out, however, that the entire motivation of this thesis is precisely the non-trivial fact that is not possible to choose an \textit{ideal} gauge-fixing condition (which is continuous in field space). In other words: All continuous gauge-fixing conditions are not ideal, \cite{Singer:1978dk}. This is the so-called Gribov problem. However, since we will devote all the rest of this thesis to this issue, we will ignore this for a while. 

\begin{figure}[t]
		\includegraphics[width=0.40\textwidth]{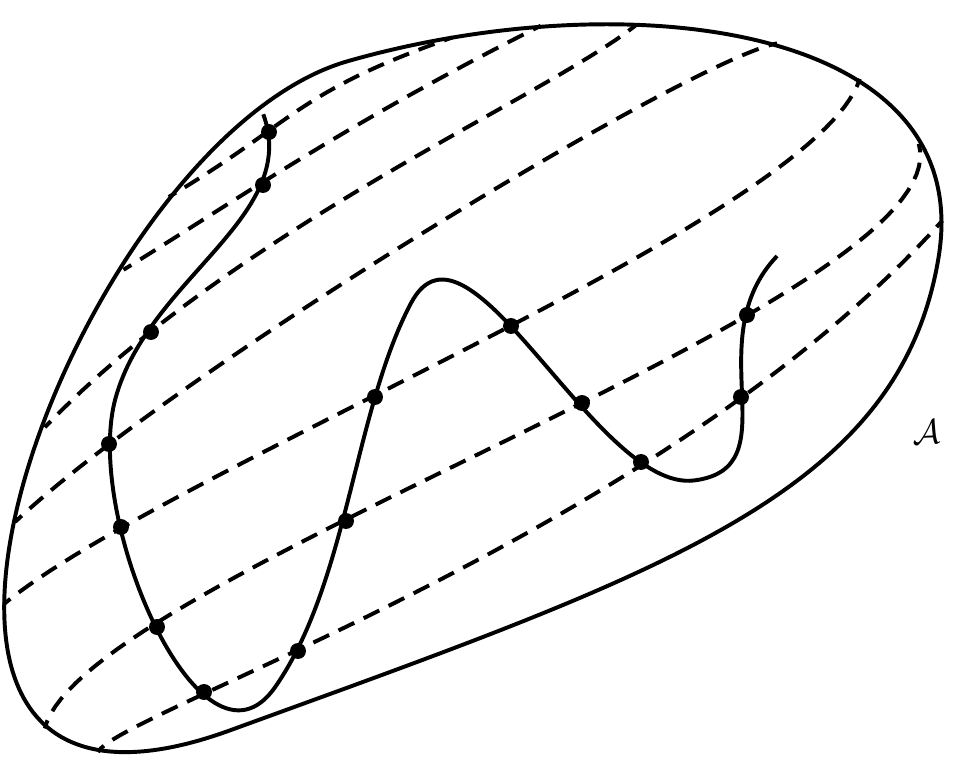}\,\,\,\,\,\,\,\,\,\,\,\,\,\,\,\,\,\,\,\,\,\,\,\,\,\,\,\,
		\includegraphics[width=0.40\textwidth]{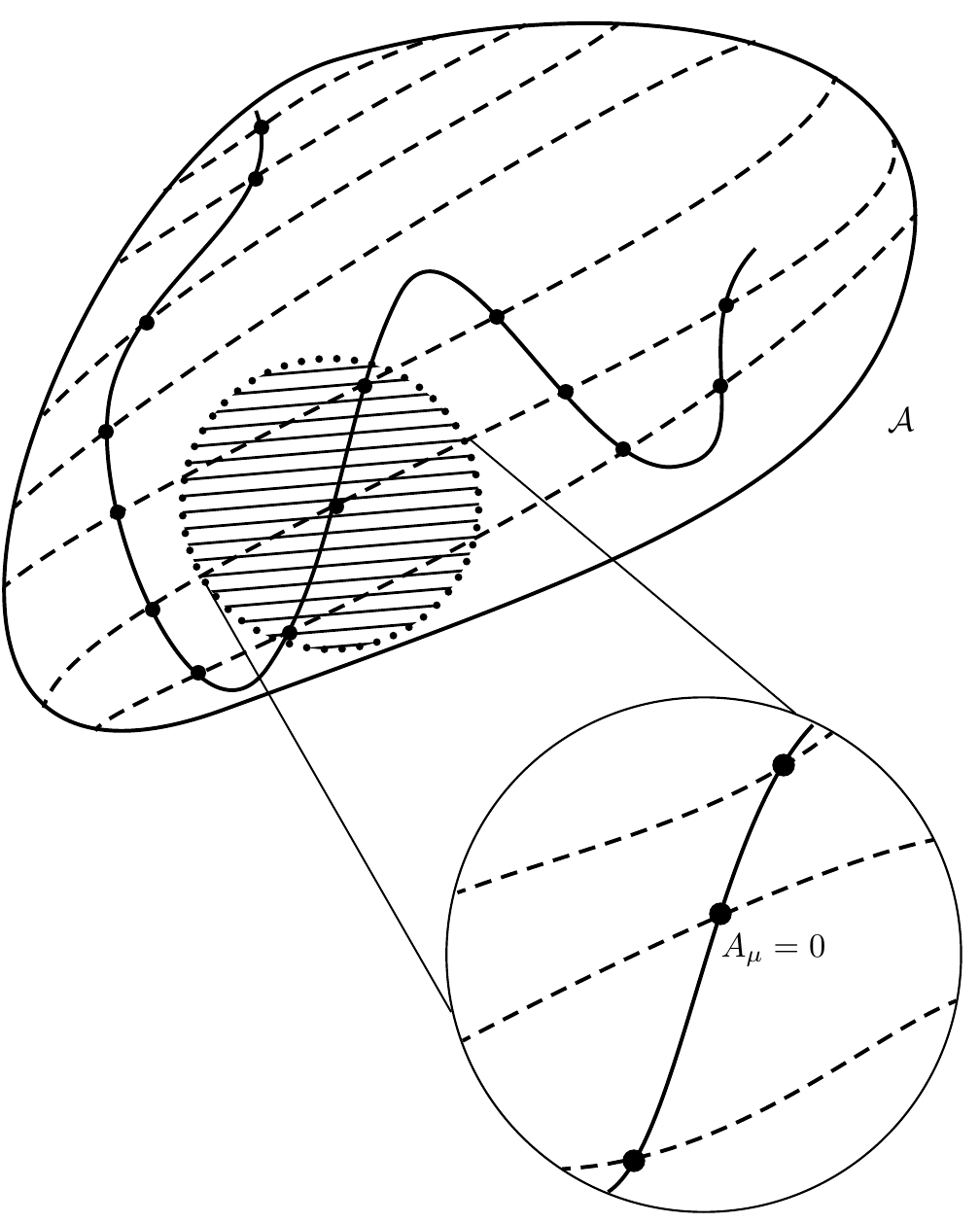}
	\caption{Left side: Gauge-fixing is not ideal. For some gauge orbits, the gauge slice intersects them more than once. Right side: Locally, it is always possible to find an ideal gauge-fixing. In particular, considering just small fluctuations around $A=0$, namely, perturbation theory, we can define a proper gauge-fixing.}
	\label{spaceA3}
\end{figure}

A final comment regarding the ``ideal" gauge choice is the following: Considering standard perturbation theory, we are concerned with small fluctuations around the classical vacuum $A=0$. Hence, given the point in $\mathcal{A}$ corresponding to $A=0$, we are just concerned with the vicinity of this point. Locally, an ideal gauge-fixing is always possible. As a consequence, at the perturbative regime, it is possible to choose a gauge-fixing function which is ideal. This is shown, pictorially, in Fig.~\ref{spaceA3}.

We want to build a path integral measure for a gauge slice corresponding to an ideal gauge fixing, namely, 

\begin{equation}
\left[\EuScript{D}A\right]\Big|_{F[A]=0}\,,
\label{qym11.0}
\end{equation}
where by assumption, if $F[A]=0$, then $F[A^{U}]\neq 0$ for all $U\,\in\,\mathcal{G}$. In the next subsection, we will present an ingenuous procedure to construct a measure \eqref{qym11.0}. 

\subsection{Dealing with gauge symmetry - Part II: Faddeev-Popov trick} \label{FPtrick}

In this subsection, the so-called Faddeev-Popov procedure \cite{Faddeev:1967fc} is reviewed. The main goal of this reminder is to emphasize where assumptions are made along the construction of the Faddeev-Popov method to prepare the reader for the introduction of the Gribov problem in the next chapter. Before working out the method in its full glory, we build a toy version of a result that will be used later on. 

Let us consider a real function $f(x)$ which has $n$ roots given by $\left\{x_1,\ldots,x_n\right\}$. Formally, we can write

\begin{equation}
\delta(f(x))=\sum^{n}_{i=1}\frac{\delta (x-x_i)}{|f'(x_i)|}\,,
\label{qym12.0}
\end{equation}
where $f'(x_i)$ denotes the derivative of $f(x)$ computed at $x_i$. We assume the derivative exists for all $x_i$ \textit{and} that it is different from zero. Then, we can integrate eq.\eqref{qym12.0} over $x$, 

\begin{equation}
\int dx\,\delta(f(x))=\sum^{n}_{i=1}\frac{1}{|f'(x_i)|}\,\,\,\,\,\Rightarrow\,\,\,\,\,\frac{1}{\sum^{n}_{i=1}\frac{1}{|f'(x_i)|}}\int dx\,\delta(f(x))=1\,.
\label{qym13.0}
\end{equation} 
Now, let us consider the particular case where $f(x)$ has \textit{just} one root $\tilde{x}$ and that the derivative $f'(\tilde{x})$ is \textit{positive}. Then, \eqref{qym13.0} reduces to

\begin{equation}
f'(\tilde{x})\int dx\,\delta(f(x))=1\,.
\label{qym14.0}
\end{equation} 
Eq.\eqref{qym14.0} is constructed upon two assumptions: \textit{(i)} The function $f(x)$ has only one root; \textit{(ii)} Its derivative computed at this root is positive. We call assumption \textit{(i)} as uniqueness and \textit{(ii)} as positivity. An useful analogy before turning to the gauge-fixing procedure is the following: If $f(x)$ plays the role of $F[A]$, then uniqueness plays the role of an \textit{ideal} gauge-fixing. The analogy with the positivity condition is not meaningful at the present moment, but we will return to this point soon. 

Returning to the Yang-Mills theory context, the functional generalization of \eqref{qym13.0} is\footnote{We already assume the gauge condition is ideal and thus, no analogous summation to the one present in eq.\eqref{qym13.0} appears.}

\begin{equation}
1=\int_{\mathcal{G}}\left[\EuScript{D}U\right]\Big|\mathrm{det}\left(\frac{\delta F[A^{\delta U}]}{\delta \xi}\right)\Big|\,\delta\left[F[A^U]\right]\,,
\label{qym15.0}
\end{equation}
\noindent where $\left[\EuScript{D}U\right]=\prod_x \mathrm{d}U(x)$, with $\mathrm{d}U$ the Haar measure of $SU(N)$. The functional derivative is taken with respect to $\xi$, the infinitesimal parameter associated with $\delta U=1-ig\xi$. Also, we emphasize that this functional derivative is computed at the root of $F[A^U]=0$ for a given gauge orbit. The original path integral for Yang-Mills theories is  written as

\begin{equation}
\EuScript{Z}_{\mathrm{YM}}=\int_{\mathcal{A}}\left[\EuScript{D}A\right]\mathrm{e}^{-S_{\mathrm{YM}}}\,,
\label{qym16.0}
\end{equation}
with gauge invariant measure and action $S_{\mathrm{YM}}$. Then, plugging the unity \eqref{qym15.0} in \eqref{qym16.0}, we obtain

\begin{equation}
\EuScript{Z}_{\mathrm{YM}}=\int_{\mathcal{A}}\left[\EuScript{D}A\right]\int_{\mathcal{G}}\left[\EuScript{D}U\right]\Big|\mathrm{det}\left(\frac{\delta F[A^{\delta U}]}{\delta \xi}\right)\Big|\,\delta\left[F[A^U]\right]\mathrm{e}^{-S_{\mathrm{YM}}}\,.
\label{qym17.0}
\end{equation}
Commuting the integral over $\mathcal{G}$ with the integral over $\mathcal{A}$ yields

\begin{equation}
\EuScript{Z}_{\mathrm{YM}}=\int_{\mathcal{G}}\left[\EuScript{D}U\right]\int_{\mathcal{A}}\left[\EuScript{D}A\right]\Big|\mathrm{det}\left(\frac{\delta F[A^{\delta U}]}{\delta \xi}\right)\Big|\,\delta\left[F[A^U]\right]\mathrm{e}^{-S_{\mathrm{YM}}(A)}\,.
\label{qym18.0}
\end{equation}
Making use of the gauge invariant measure, and the gauge invariant $S_{\mathrm{YM}}$, we can rewrite $\left[\EuScript{D}A\right]=\left[\EuScript{D}A^U\right]$ and $S_{\mathrm{YM}}(A)=S_{\mathrm{YM}}(A^U)$. Then,

\begin{equation}
\EuScript{Z}_{\mathrm{YM}}=\int_{\mathcal{G}}\left[\EuScript{D}U\right]\int_{\mathcal{A}}\left[\EuScript{D}A^U\right]\Big|\mathrm{det}\left(\frac{\delta F[A^{\delta U}]}{\delta \xi}\right)\Big|\,\delta\left[F[A^U]\right]\mathrm{e}^{-S_{\mathrm{YM}}(A^U)}\,.
\label{qym19.0}
\end{equation}
We see thus that integral over $\mathcal{A}$ performed over the dummy variable $A^U$ and the integral over $\mathcal{G}$ is completely factorized. This integral can be formally computed,

\begin{equation}
\int_{\mathcal{G}}\left[\EuScript{D}U\right]=\mathrm{Vol}(\mathcal{G})=\mathrm{Vol}(\prod_x SU(N)_x)\,.
\label{qym20.0}
\end{equation}
Of course, the volume of eq.\eqref{qym20.0} is infinity. However, this is just a prefactor of the path integral and is harmless for the computation of expectation values. Neglecting the normalization factor, we have

\begin{equation}
\EuScript{Z}_{\mathrm{YM}}=\int_{\mathcal{A}}\left[\EuScript{D}A\right]\Big|\mathrm{det}\left(\frac{\delta F[A^{\delta U}]}{\delta \xi}\right)\Big|\,\delta\left[F[A]\right]\mathrm{e}^{-S_{\mathrm{YM}}(A^U)}\,.
\label{qym21.0}
\end{equation}
Then, the Dirac delta functional $\delta\left[F[A]\right]$ projects the measure $\left[\EuScript{D}A\right]$ to the desired one, given by eq.\eqref{qym11.0}. The partition function \eqref{qym21.0} is the standard Faddeev-Popov one, presented in textbooks, \cite{pokorski,bailinlove,zinnjustin,weinberg2}. We can write \eqref{qym11.0} in local fashion by lifting the delta functional and the Faddeev-Popov determinant to the exponential. Here comes the \textit{positivity} assumption. In general, it is assumed that the Faddeev-Popov determinant is \textit{positive} and we can simply remove the absolute value of expression \eqref{qym21.0}. This is precisely the analogue of condition $(ii)$ aforementioned for the one-dimensional function $f(x)$. Hence, we write

\begin{equation}
\EuScript{Z}_{\mathrm{YM}}=\int_{\mathcal{A}}\left[\EuScript{D}A\right]\mathrm{det}\left(\frac{\delta F[A^{\delta U}]}{\delta \xi}\right)\,\delta\left[F[A]\right]\mathrm{e}^{-S_{\mathrm{YM}}(A^U)}\,.
\label{qym22.0}
\end{equation}
By introducing the Faddeev-Popov ghosts and the Nakanishi-Lautrup field, we modify Yang-Mills action by the introduction of the so-called ghost and gauge-fixing terms. 

To summarize, the construction of \eqref{qym22.0} relies on the assumptions:

\begin{itemize}

\item We are able to find a gauge condition $F[A]=0$ which selects \textit{one} representative per gauge orbit (uniqueness);

\item Out of $F[A]$, we construct the Faddeev-Popov operator computed at the root of $F[A]$. This operator is non-singular and even more, positive (positivity). 

\end{itemize}
As argued before, the first assumption is safe in perturbation theory. Through explicit examples, we can also show that the second assumption is well grounded in perturbation theory, at least for a large class of gauge conditions. Therefore, \eqref{qym22.0} is perfectly fine as long as perturbation theory is concerned. However, this procedure is not well-defined beyond perturbation theory. The reason is that the non-trivial topological structure of Yang-Mills theories forbids the construction of an \textit{ideal} gauge-fixing, thus spoiling the uniqueness assumption. This was explicitly verified by Gribov in \cite{Gribov:1977wm} and formalized by Singer in \cite{Singer:1978dk}. 

Since the Faddeev-Popov procedure breakdown beyond perturbation theory, we can ask ourselves if an improvement of such method could shed some light to the understanding of the IR physics of Yang-Mills theories. In this thesis, we will present some progress on this direction. By taking into account these non-trivial facts concerning a proper quantization of Yang-Mills theories, we will provide some evidence that non-perturbative physics can be reached in an analytical way. 

\section{Work plan}

As pointed out in the last section, in this thesis we investigate non-perturbative phenomena in pure Yang-Mills theories by constructing what would-be an optimized quantization scheme which takes into account problems in the standard Faddeev-Popov gauge-fixing procedure. The outline we follow is: In Ch.~\ref{ch.2} we discuss precisely the failure of gauge-fixing in Yang-Mills theories through particular examples of gauge conditions. We state the \textit{Gribov problem} and introduce some important concepts as the Gribov region. In Ch.~\ref{ch.3} we review the original attempt carried out by Gribov and later on by Zwanziger to deal with the presence of spurious configurations in the path integral domain of integration. At this stage we introduce the so-called Gribov-Zwanziger framework, a local and renormalizable way of dealing with the Gribov problem. Nevertheless, as originally worked out, this is restricted to Landau gauge. Also in this chapter we discuss a pivotal point of this thesis, namely, the BRST soft breaking of the Gribov-Zwanziger action. After that, in Ch.~\ref{RGZch}, we discuss the Refined Gribov-Zwanziger action. The refinement arises due to infrared instabilities of the Gribov-Zwanziger theory, which favors the formation of dimension-two condensates. As explicitly shown in this chapter, the Refined Gribov-Zwanziger scenario provides a consistent set up which leads to a gluon and ghost propagators with very good agreement with the most recent lattice data. Again, this is restricted to Landau gauge. Then, in Ch.~\ref{ch.5} we start to present our original results. Our aim is to extend the (Refined) Gribov-Zwanziger action to linear covariant gauges. Notably, in this class of gauges, several technical complications appear and the introduction of a gauge parameter plays a non-trivial role. Without BRST symmetry, the control of gauge dependence becomes a highly non-trivial task. In this chapter, we present a first attempt to handle the Gribov problem in linear covariant gauges, but we describe how this construction asks for an important conceptual change in the original Gribov-Zwanziger construction. This brings us to Ch.~\ref{nonpBRSTRGZ}, where we discuss in more detail the fate of BRST symmetry in the Refined Gribov-Zwanziger action. By a convenient change of variables we demonstrate in this chapter how to build a modified BRST transformation for the Refined Gribov-Zwanziger setting in the Landau gauge. These transformations enjoy nilpotency, a highly desired technical tool for the power of BRST. Also, the ``deformation" of such BRST transformations with respect to the standard one has an intrinsic non-perturbative nature. After this, in Ch.~\ref{LCGrevisited}, we come back to linear covariant gauges and construct a consistent framework with the non-perturbative BRST transformations. We show how this symmetry is powerful in order to control gauge dependence. Also, we propose a ``non-perturbative" BRST quantization scheme akin to the standard one. Also in this chapter, we discuss the fact that the refinement of the Gribov-Zwanziger action is not consistent when $d$, the spacetime dimension, is two. This fact brings different behaviors for the gluon propagator for different $d$. This result also holds in the particular case of Landau gauge and thus gives some hint that it might be more general. 

Albeit nilpotent, the non-perturbative BRST symmetry is not local, as well as the (Refined) Gribov-Zwanziger action in linear covariant gauges. In Ch.~\ref{locnonpBRST} we present a full local version of this construction and establishes a local quantum field theory which deals with the Gribov problem in linear covariant gauges and takes into account condensation of dimension-two operators in a non-perturbative BRST invariant way. We present the immediate consequences of the Ward identities of this action and gauge independence of correlation functions of physical operators as well as the non-renormalization of the longitudinal sector of the gluon propagator, a highly non-trivial fact in this context. 

In Ch.~\ref{CFGaugeCh} we extend our formalism to a class of one-parameter non-linear gauges, the Curci-Ferrari gauges. Due to the non-linearity of this gauge, novel dimension-two condensates must be taken into account and we provide our results for the gluon propagator. Again, we observe the same qualitative dependence on the behavior of this propagator with respect to spacetime dimension. So far, no lattice results are available for this class of gauges. We have thus the opportunity to test our formalism's power against future lattice data.

Finally we draw our conclusions and put forward some perspectives. A list of appendices collects conventions, techniques and derivations that are avoided along the text.

\chapter{The (in)convenient Gribov problem}\label{ch.2}

The Faddeev-Popov procedure \cite{Faddeev:1967fc,pokorski,bailinlove,zinnjustin,weinberg2}, although extremely efficient for perturbative computations, relies on two strong assumptions: First, the gauge condition is \textit{ideal}, which means it selects exactly one representative per gauge orbit. Second, the Faddeev-Popov operator associated with the \textit{ideal} gauge condition is positive (which ensures it does not develop zero-modes \textit{i.e.} it is invertible). Although we can argue these assumptions are well grounded \textit{a posteriori} by the nice agreement between perturbative computations and experimental measurements, it is not possible to guarantee the same happens when we start looking at non-perturbative scales. Actually, it is possible to prove, at least for some very useful gauges in perturbative analysis, that these assumptions do not hold at the non-perturbative level, \cite{Gribov:1977wm}. In other words, the Faddeev-Popov gauge fixing procedure is not able to remove all equivalent gauge field configurations at the non-perturbative level. This implies a residual over counting in this regime by the path integral and these spurious configurations are called \textit{Gribov copies}. The presence of Gribov copies in the quantization process is precisely the \textit{Gribov problem}, \cite{Gribov:1977wm}. In the following lines, we will make these words more precise and show in two particular gauges the manifestation of the Gribov problem. Although we intend to give an introduction to the problem for the benefit of the reader, we do not expose computational details. The reason for omitting these details is the existence of pedagogical and detailed reviews on the topic, see for instance \cite{Sobreiro:2005ec,Vandersickel:2011zc,Vandersickel:2012tz}.

\section{Landau gauge and Gribov copies} \label{section1.1}

An extremely popular gauge choice for continuum and lattice computations is the so-called \textit{Landau gauge}. This covariant gauge is defined by imposing the gauge field to be  transverse, namely

\begin{equation}
\partial_{\mu}A^{a}_{\mu}=0\,.
\label{1.1}
\end{equation}

\noindent In the same language used before, this is an ideal gauge choice \textit{if}, for a gauge field configuration that satisfies eq.(\ref{1.1}), an element of its gauge orbit, \textit{i.e.} all gauge field configurations $\tilde{A}^{a}_{\mu}$ which are obtained by a gauge transformation of $A^{a}_{\mu}$, will not satisfy condition (\ref{1.1}) \textit{i.e.}

\begin{equation}
\partial_{\mu}\tilde{A}^{a}_{\mu}\neq 0\,.
\label{1.2}
\end{equation}

\noindent To make our life simpler, we can test this assumption with infinitesimal gauge transformations at first place. A gauge field configuration $A'^{a}_{\mu}$ which is connected to $A^{a}_{\mu}$ via an infinitesimal gauge transformation\footnote{See the conventions in Appendix~\ref{appendixA}.} is given by

\begin{equation}
A'^{a}_{\mu}=A^{a}_{\mu}-D^{ab}_{\mu}\xi^{b}\,,
\label{1.3}
\end{equation}

\noindent with $\xi$ being the real infinitesimal gauge parameter associated with the gauge transformation. Therefore, we rewrite (\ref{1.2}) as

\begin{equation}
\partial_{\mu}A'^{a}_{\mu}=\partial_{\mu}A^{a}_{\mu}-\partial_{\mu}D^{ab}_{\mu}\xi^{b}=-\partial_{\mu}D^{ab}_{\mu}\xi^{b}\neq 0\,,
\label{1.4}
\end{equation}

\noindent whereby we used eq.(\ref{1.1}). The operator $-\partial_{\mu}D^{ab}_{\mu}$ with condition (\ref{1.1}) applied is nothing but the Faddeev-Popov operator in the Landau gauge. Hence, condition (\ref{1.2}) can be rephrased, at least at infinitesimal level, as the non appearance of zero-modes of the Faddeev-Popov operator (which, again, is an assumption in the standard Faddeev-Popov construction). Conversely, the equation 

\begin{equation}
\partial_{\mu}D^{ab}_{\mu}\xi^b=0\,.
\label{1.41}
\end{equation}

\noindent is known as \textit{copies equation} for obvious reasons. To check if the Faddeev-Popov operator indeed satisfies (\ref{1.4}), we can make a pure analysis of the spectrum of the Faddeev-Popov operator considering its eigenvalue equation,

\begin{equation}
\EuScript{M}^{ab}\chi^b\equiv -\partial_{\mu}D^{ab}_{\mu}\chi^{b} = \epsilon (A) \chi^a\,\,\, \Rightarrow\,\,\, (-\delta^{ab}\partial^{2}+\underbrace{gf^{abc}A^{c}_{\mu}\partial_{\mu}}_{(i)})\chi^b=\epsilon (A) \chi^a\,.
\label{1.5}
\end{equation}

\noindent Solving this eigenvalue problem in generality is a difficult task. Nevertheless, we just need to understand if zero-modes exist and a qualitative analysis is enough. A property which is very important to understand the spectrum of $\EuScript{M}$ in a qualitative way is that it is a Hermitian operator in Landau gauge\footnote{See the explicit proof in \cite{Sobreiro:2005ec}, for instance.}. First, we note that term $(i)$ characterizes the non-Abelian nature of the gauge group. Second, if we turn off the coupling $g$, eq.(\ref{1.5}) reduces to

\begin{equation}
-\partial^2\chi^a=\epsilon\chi^a\,.
\label{1.6}
\end{equation}

\noindent The operator $-\partial^2$ is a positive operator\footnote{We are working in Euclidean space to avoid issues of the validity of the Wick rotation at the non-perturbative regime.}. It implies the only\footnote{This is true under suitable boundary conditions. This also implies Abelian gauge theories do not develop non-trivial zero-modes.} zero-mode is given by $\chi=0$, which is a trivial solution. As a consequence, if we turn on $g$, for a small value of the product $gA^a_{\mu}$ we can see eq.(\ref{1.5}) as a perturbation of eq.(\ref{1.6}). Therefore the full operator $\EuScript{M}$ remains positive for sufficiently small $gA^a_{\mu}$. The regime where $gA^a_{\mu}$ is small is precisely where perturbation theory around the trivial vacuum $A=0$ is performed. It implies there are no zero-modes at the perturbative regime and, even more, the Faddeev-Popov operator is positive. These conditions match the assumptions on the Faddeev-Popov procedure and this is why we can state the Gribov problem plays no role at the perturbative level. But, increasing the value of $gA^a_{\mu}$, the first term of eq.(\ref{1.5}) will be comparable to (i) and a negative eigenvalue will emerge. This change of sign should be characterized by the presence of a zero-mode, \textit{i.e.} as long as we increase $gA^a_{\mu}$, we will hit a zero-mode of the Faddeev-Popov operator. As argued before, this phenomenon should occur as long as we walk away from the trivial vacuum $A=0$ and might be a genuine non-perturbative feature. Hence, the Faddeev-Popov procedure is well-defined in the ultraviolet (or perturbative) regime, while as the theory goes to the infrared (or the non-perturbative) region, these assumptions start to fail. 

\begin{figure}
	\centering
		\includegraphics[width=0.50\textwidth]{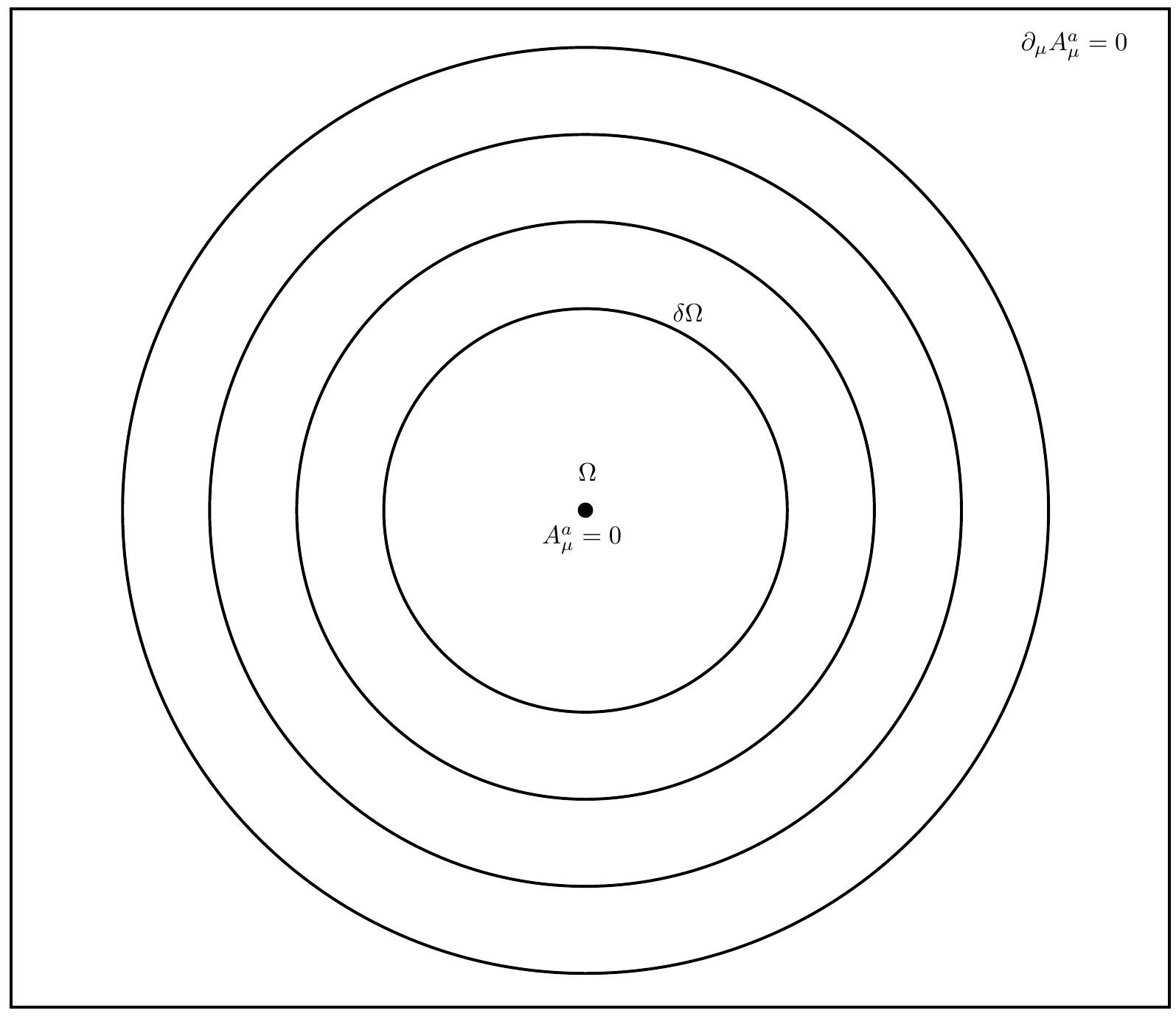}
	\caption{The gribov region $\Omega$ and its boundary $\delta\Omega$, the Gribov horizon.}
	\label{fig1}
\end{figure}

Pictorially, we can reproduce our qualitative analysis in a simple picture: In Fig.~\ref{fig1}, we consider the configuration space of transverse gauge fields \textit{i.e.} the fields that enjoy $\partial_{\mu}A^{a}_{\mu}=0$. At the center we have the trivial vacuum configuration $A=0$. The region enclosed by the boundary $\delta\Omega$ represents the place where the operator $\EuScript{M}$ is positive \textit{i.e.} we start walking from $A=0$, increasing the magnitude of $gA$, up to $\delta\Omega$ where the Faddeev-Popov operator reaches its first zero-modes. Beyond the boundary $\delta\Omega$ the Faddeev-Popov operator changes its sign and the other boundaries represent other places where zero-modes turn out to appear. The region $\Omega$, where the Faddeev-Popov operator is positive, is known as the \textit{first Gribov region} or simply \textit{Gribov region}. Its boundary $\delta\Omega$ is called \textit{Gribov horizon}. Actually, a series of very important properties concerning the Gribov region were proven by Dell'Antonio and Zwanziger in \cite{Dell'Antonio:1991xt}. In particular, the Gribov region enjoys the following properties:

\begin{itemize}
\item It is bounded in every direction;
\item It is convex;
\item All gauge orbits cross $\Omega$ at least once.
\end{itemize}

\noindent These very nice properties give support to the (partial) solution of the Gribov problem - proposed by Gribov himself in \cite{Gribov:1977wm} - where a restriction of the domain of the path integral to $\Omega$ is performed. We will discuss more details on this in the next section. The proofs of those properties of $\Omega$ are not presented here and the reader is referred to \cite{Vandersickel:2012tz,Dell'Antonio:1991xt}. Nevertheless, the essential tool to prove those properties is the definition of $\Omega$ by a minimization procedure of a functional. We introduce it here for completeness and also because this is an important device for gauge fixed lattice simulations, one of the sources we have to compare our results. Given the functional

\begin{equation}
\mathcal{A}_L=\frac{1}{2}\int d^dx~A^{a}_{\mu}A^{a}_{\mu}\,,
\label{1.7}
\end{equation}

\noindent we can impose the Landau gauge over a gauge field configuration $A^{a}_{\mu}$ by imposing the variation of $\mathcal{A}_L$ with respect to infinitesimal gauge transformations to be an extrema,

\begin{equation}
\delta\mathcal{A}_{\mathrm{L}}=\int d^dx~\xi^a\partial_{\mu}A^{a}_{\mu}=0\,\,\,\Rightarrow\,\,\,\partial_{\mu}A^{a}_{\mu}=0\,,
\label{1.8}
\end{equation}

\noindent with $\xi$ an arbitrary parameter. To define the Gribov region, we demand this extrema to be a minimum,

\begin{equation}
\delta^2\mathcal{A}_{\mathrm{L}}=-\int d^dx~\xi^a\partial_{\mu}D^{ab}_{\mu}\xi^b>0\,\,\,\Rightarrow\,\,\,-\partial_{\mu}D^{ab}_{\mu}>0\,.
\label{1.9}
\end{equation}

\noindent After this discussion, we establish the mathematical definition of the Gribov region which will play a very important role in this thesis. 

\begin{definition}\label{gribovregion}
The set of transverse gauge fields for which the Faddeev-Popov operator is positive is called the Gribov region $\Omega$ \textit{i.e.}
\begin{equation}
\Omega = \left\{A^{a}_{\mu}\,,\,\,\partial_{\mu}A^{a}_{\mu}=0\,\Big|-\partial_{\mu}D^{ab}_{\mu}>0\right\}\,.
\end{equation}
\end{definition} 

\noindent This definition deserves a set of comments: \textit{(a)} The Gribov region is free from \textit{infinitesimal} Gribov copies, namely, copies that correspond to zero-modes of the Faddeev-Popov operator; \textit{(b)} So far, we have not said anything about copies generated by finite gauge transformations. It implies $\Omega$ might not be free from all Gribov copies but of a class of copies; \textit{(c)} In fact, it was shown that $\Omega$ does contain copies, \cite{vanBaal:1991zw,Landim:2014fea}; \textit{(d)} In spite of our qualitative arguments to justify the presence of zero-modes of the Faddeev-Popov operator, it is possible to construct explicit examples of normalizable ones, \cite{Gribov:1977wm,Sobreiro:2005ec,Guimaraes:2011sf,Capri:2012ev}.

Since this region is not free from all Gribov copies, a natural question is how to define a region which is truly free from all spurious configurations. This region is known as \textit{Fundamental Modular Region (FMR)} (usually, is denoted as $\Lambda$) and as naturally expected is a subset of $\Omega$. This region is defined as the set of \textit{absolute} minimum of the functional (\ref{1.7}) where now we take into account variations with respect to finite gauge transformations. As $\Omega$, the FMR has many important properties which we list:

\begin{itemize}
\item It is bounded in every direction;
\item It is convex;
\item All gauge orbits cross $\Lambda$;
\item The boundary $\delta\Lambda$ shares points with $\delta\Omega$.

\end{itemize}

\noindent Again, we do not present the proofs, but we refer to \cite{Vandersickel:2012tz,vanBaal:1991zw}. At this level, we remind the reader that everything discussed up to now is about the Landau gauge choice. We emphasize that along the discussions and computations, the gauge condition (\ref{1.1}) was exhaustively employed. However, we should be able to work with different gauge choices (otherwise we are not dealing with a \textit{gauge} theory anymore). To show a second concrete example of the manifestation of the Gribov problem, we devote some words to the maximal Abelian gauge in the next section.

\section{Maximal Abelian gauge and Gribov copies}

The maximal Abelian gauge (MAG) is a very important gauge for non-perturbative studies in the context of the dual superconductivity model for confinement \cite{'tHooft:1981ht}. In particular, this gauge treats the diagonal and off-diagonal components of the gauge field in different footing which is very appropriate to the study of the so-called Abelian dominance \cite{Ezawa:1982bf}. Due to this decomposition, we first set our conventions and notation\footnote{These are compatible with those presented in Appendix~\ref{appendixA}, but we make them explicit for the benefit of the reader.}. The gauge field $A_{\mu}$ is decomposed into diagonal (Abelian) and off-diagonal (non-Abelian) parts as

\begin{equation}
A_{\mu} = A_{\mu}^{A}T^{A} = A_{\mu}^{a}T^{a} + A_{\mu}^{i}T^{i},
\label{1.10}
\end{equation}

\noindent with $T^a$ the off-diagonal sector of generators and $T^i$ the Abelian generators. The generators $T^{i}$ commute with each other and generate the Cartan subgroup of $SU(N)$. To avoid confusion, we have to keep in mind that capital indices $\left\{A, B, C,\ldots \right\}$ are related to the entire $SU(N)$ group, and so, they run in the set $\left\{1, \ldots, (N^{2} - 1)\right\}$. Small indices $\left\{a, b, c, \ldots h\right\}$ (from the beginning of the alphabet) represent the off-diagonal part of $SU(N)$ and they vary in the set $\left\{1, \ldots, N(N-1)\right\}$. Finally, small indices $\left\{i,j,k, \ldots\right\}$ (from the middle of the alphabet) describe the Abelian part of $SU(N)$ and they run in the range $\left\{1, \ldots, (N-1)\right\}$. From $SU(N)$ Lie algebra, we can write the following decomposed algebra

\begin{eqnarray}
\left[T^{a},T^{b}\right] &=&if^{abc}T^{c} + if^{abi}T^{i}, \nonumber \\
\left[T^{a},T^{i}\right] &=& -if^{abi}T^{b}, \nonumber \\
\left[T^{i},T^{j}\right] &=& 0.
\label{1.11}
\end{eqnarray}

\noindent The Jacobi identity (\ref{a0.1}) decomposes as

\begin{eqnarray}
f^{abi}f^{bcj} + f^{abj}f^{bic} &=& 0, \nonumber \\
f^{abc}f^{cdi} + f^{adc}f^{cib} + f^{aic}f^{cbd} &=& 0, \nonumber \\
f^{abc}f^{cde} + f^{abi}f^{ide} + f^{adc}f^{ceb} + f^{adi}f^{ieb} + f^{aec}f^{cbd} + f^{aei}f^{ibd} &=& 0.
\label{1.12}
\end{eqnarray}

Proceeding in this way, we can write the off-diagonal and diagonal components of an infinitesimal gauge transformation (\ref{a8}) with parameter $\alpha$ as

\begin{eqnarray}
\delta A_{\mu}^{a} &=& -(\EuScript{D}_{\mu}^{ab}\alpha^{b} + gf^{abc}A_{\mu}^{b}\alpha^{c} + gf^{abi}A_{\mu}^{b}\alpha^{i}), \nonumber \\
\delta A_{\mu}^{i} &=& - (\partial_{\mu}\alpha^{i} + gf^{abi}A_{\mu}^{a}\alpha^{b}),
\label{1.13}
\end{eqnarray}

\noindent where the covariant derivative $\EuScript{D}_{\mu}^{ab}$ is defined with respect to the Abelian component of the gauge field, \textit{i.e.}

\begin{equation}
\EuScript{D}_{\mu}^{ab} = \delta^{ab}\partial_{\mu} - gf^{abi}A_{\mu}^{i}.
\label{1.14}
\end{equation}

\noindent For completeness we show explicitly the decomposition of the Yang-Mils action, namely

\begin{equation}
S_{\mathrm{YM}} = \frac{1}{4}\int d^{d}x(F_{\mu \nu}^{a}F_{\mu \nu}^{a} + F_{\mu \nu}^{i}F_{\mu \nu}^{i}),
\label{1.15}
\end{equation}

\noindent with

\begin{eqnarray}
F_{\mu \nu}^{a} &=& \EuScript{D}_{\mu}^{ab}A_{\nu}^{b} - \EuScript{D}_{\nu}^{ab}A_{\mu}^{b} + gf^{abc}A_{\mu}^{b}A_{\nu}^{c}, \nonumber \\
F_{\mu \nu}^{i} &=& \partial_{\mu}A_{\nu}^{i} - \partial_{\nu}A_{\mu}^{i} + gf^{abi}A_{\mu}^{a}A_{\nu}^{b}.
\label{1.16}
\end{eqnarray}

Now, we introduce the gauge conditions that characterize the MAG. This gauge is obtained by fixing the non-Abelian sector in a Cartan subgroup covariant way,

\begin{equation}
\EuScript{D}_{\mu}^{ab}A_{\mu}^{b}=0\;.
\label{1.17}
\end{equation}

\noindent This condition does not fix the Abelian gauge symmetry, which is usually fixed by a Landau-like condition,

\begin{equation}
\partial_{\mu}A_{\mu}^i=0\;.
\label{1.18}
\end{equation}

\noindent With conditions (\ref{1.17}) and (\ref{1.18}) we can ask the same question we elaborated in the last section: Given a configuration $(A^a_{\mu},A^i_{\mu})$ that satisfies (\ref{1.17}) and (\ref{1.18}), is there a configuration $(\tilde{A}^a_{\mu},\tilde{A}^i_{\mu})$ that also satisfies the gauge condition \textit{and} is related to $(A^a_{\mu},A^i_{\mu})$ through a gauge transformation? Again, we restrict ourselves to infinitesimal gauge transformations and therefore apply eq.(\ref{1.13}) to (\ref{1.17}) and (\ref{1.18}). The result is

\begin{eqnarray}
\EuScript{M}^{ab}_{\mathrm{MAG}}\alpha^b&=&0\;,\nonumber\\
\partial_\mu\left(\partial_\mu \alpha^i+gf^{abi}A_\mu^a \alpha^b\right)&=&0\;,\label{1.19}
\end{eqnarray}

\noindent with

\begin{equation}
\EuScript{M}^{ab}_{\mathrm{MAG}} = \EuScript{D}_{\mu}^{ac}\EuScript{D}_{\mu}^{cb} + gf^{acd}A_{\mu}^{c}\EuScript{D}_{\mu}^{db} + g^{2}f^{aci}f^{bdi}A_{\mu}^{c}A_{\mu}^{d}.
\label{1.20}
\end{equation}

\noindent We recognize (\ref{1.19}) as the Gribov copies equations for the MAG. Although apparently we have two equations to characterize an infinitesimal Gribov copy in the MAG, we must notice the first equation of (\ref{1.19}) contains just the off-diagonal components of the gauge parameter. On the other hand, the second equation of (\ref{1.19}) contains both diagonal and off-diagonal parameters. It is easily rewritten as

\begin{equation}
\alpha^{i} = \frac{-gf^{abi}\partial_{\mu}(A^{a}_{\mu}\alpha^{b})}{\partial^{2}}\,,
\label{1.21}
\end{equation}

\noindent which implies that, once we solve the first equation of (\ref{1.19}), the diagonal components of the solution are not independent, but determined by eq.(\ref{1.21}). In this sense, the first equation of (\ref{1.19}) is truly the one which determines if a non-trivial solution is viable. Therefore, we take care just of the first equation and establish that infinitesimal Gribov copies exist in the MAG if and only if $\EuScript{M}_{\mathrm{MAG}}$ has non-trivial zero-modes. 

We have defined at the infinitesimal level the Gribov problem in the MAG. Remarkably, the operator $\EuScript{M}_{\mathrm{MAG}}$ is Hermitian and a similar analysis we carried out in Landau gauge can be employed here to characterize the spectrum of such operator. This analysis turns out to be very fruitful also in this case and is possible to define an analogue of the Gribov region $\Omega$. As before, we define a region $\Omega_{\mathrm{MAG}}$ which enjoys the following properties:

\begin{itemize}
\item It is unbounded in all Abelian (or diagonal) directions\,;
\item It is bounded in all non-Abelian (or off diagonal) directions\,;
\item It is convex\,.
\end{itemize}

\noindent Again, we do not expose the proofs of these properties but the reader can find all details in \cite{Capri:2005tj,Capri:2008vk,Capri:2010an,Capritese,arturotese}. Those properties can be pictorially represented by Fig.~\ref{fig2}, which is more precise in the case of $SU(2)$, where we have only one Abelian and two non-Abelian directions.

\begin{figure}
	\centering
		\includegraphics[width=0.50\textwidth]{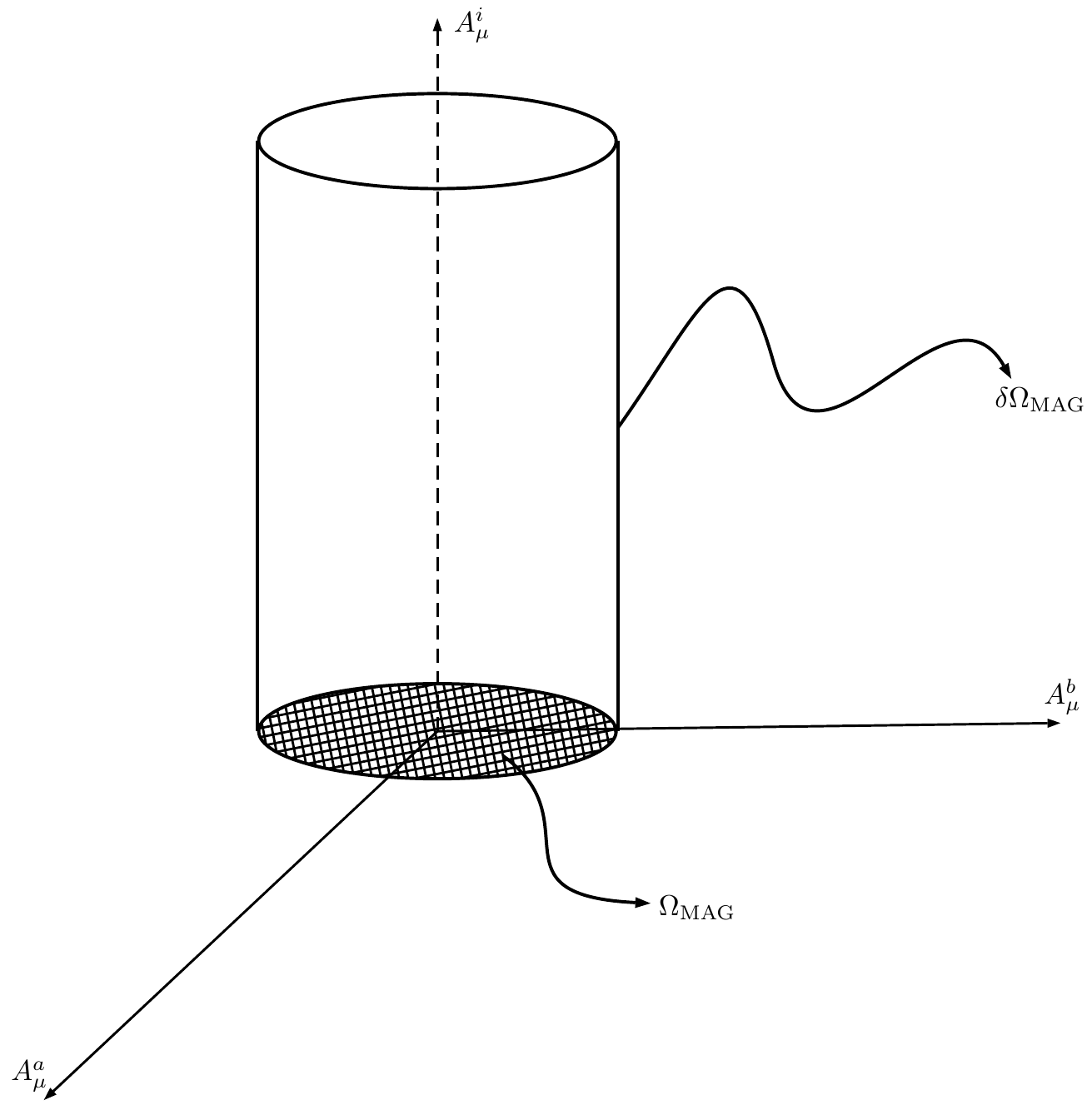}
	\caption{Gribov region $\Omega_{\mathrm{MAG}}$ in the MAG.}
	\label{fig2}
\end{figure}

This region is defined by the set of gauge field configurations which satisfy conditions (\ref{1.17}) and (\ref{1.18}) and render a positive Faddeev-Popov operator $\EuScript{M}_{\mathrm{MAG}}$. More precisely,

\begin{definition}
The Gribov region in the MAG is given by the set

\begin{equation}
\Omega_{\mathrm{MAG}}=\left\{(A^{a}_{\mu},A^{i}_{\mu})\,; D^{ab}_{\mu}A^{b}_{\mu}=0\,, \partial_{\mu}A^{i}_{\mu}=0\,\Big|\,\EuScript{M}^{ab}_{\mathrm{MAG}}>0\right\}\,.
\end{equation}
\end{definition}

\noindent Similarly to the Landau gauge, we can introduce a minimizing functional which taking the first variation and demanding the extrema condition gives the gauge condition and requiring this extrema is a minimum gives the positivity condition of the Faddeev-Popov operator $\EuScript{M}_{\mathrm{MAG}}$. The expression of this functional is 

\begin{equation}
\mathcal{A}_{\mathrm{MAG}}=\frac{1}{2}\int d^dx~A^{a}_{\mu}A^{a}_{\mu}\,.
\label{1.22}
\end{equation}
We note the minimizing functional depends just on the non-Abelian components of the gauge field and its first variation gives just condition (\ref{1.17}). We emphasize the same issue as in the Landau gauge: The Gribov region $\Omega_{\mathrm{MAG}}$ is not free of Gribov copies, but from infinitesimal ones. Again, a further reduction to a FMR is necessary to obtain a truly Gribov copies independent region. For explicit examples of normalizable zero-modes of $\EuScript{M}_{\mathrm{MAG}}$ we refer to \cite{Guimaraes:2011sf,Capri:2013vka}.

\section{General gauges and Gribov copies}

In the last two sections, we pointed out the existence of Gribov copies in two specific gauge choices. Also, we have argued about the existence of a region in configuration space which is free from infinitesimal Gribov copies and a region which truly free from all copies. However, to achieve such construction, we used particular properties of Landau and maximal Abelian gauges. In particular, it was crucial that these gauges have Hermitian Faddeev-Popov operator to perform the qualitative analysis we pursued and therefore to define the Gribov region. It indicates the definition of a ``Gribov region" strongly depends on the gauge condition we are working on. At this stage it is rather natural to ask if we could use our freedom to choose a smart gauge condition which does not suffer from the Gribov problem. First of all, from our previous analysis it is not difficult to realize that the Gribov problem is not a peculiar feature of these two gauge conditions. On the other hand, it is not simple to state that there is no gauge condition which is Gribov copies free just with these two examples. To accomplish such statement we should rely on the deep geometrical setting behind gauge theories, \cite{nakahara,bertlmann,frenkel}. This was achieved by Singer in \cite{Singer:1978dk} where he showed that, taking into account regularity conditions of the gauge field at infinity, there is no gauge fixing condition which is continuous in configuration space which is free from the Gribov problem. Clearly, we might be able to avoid some Singer's assumptions and construct a suitable gauge fixing which avoids copies \cite{Becchi:1998vv}. This task, however, has shown to be very hard because despite of the ``construction" of a copies free gauge-fixing, it should be still useful for performing concrete computations. 

It was advocated by Basseto \textit{et al.} in \cite{Bassetto:1983rq} that it is possible to find suitable boundary conditions for the gauge fields such that (space-like) planar gauges are free from copies. Also, Weinberg in \cite{weinberg2} refers to the axial gauge as copies free. Nevertheless, a common belief within these class of ``algebraic gauges" is the violation of Lorentz invariance, which could be an undesirable aspect for standard quantum field theories tools. More recently, a series of works by Quadri and Slavnov argued for the construction of gauges free from copies which also enjoy renormalizability and Lorentz invariance, \cite{Slavnov:2008xz,Quadri:2010vt,Quadri:2010xb}. Also, the construction of a gauge-fixing condition inspired in Laplacian center-gauges type implemented in the lattice was proposed in the continuum, \cite{Oxman:2015ira}. It is argued that this class of gauges should avoid the Gribov problem. 

Another valid point of view is the following: let us assume that Gribov copies are out there. Do they play a significant role in the quantization of Yang-Mills action? Following this reasoning, many authors discussed this issue, \cite{Sharpe:1984vi,Hirschfeld:1978yq,Fujikawa:1982ss,Fujikawa:1995fb,Fujikawa:1995gb}. However, we should emphasize an important aspect: So far, our arguments on the existence of Gribov copies, show they should appear at the \textit{non-perturbative} level. This fact tells us that within standard perturbation theory, every gauge choice seems to be a good choice in what concerns Gribov copies (disregarding possible intrinsic pathological gauges). Also, it suggests that the relevance or not of these copies should be tested in a full non-perturbative framework, a monumental achievement which is far beyond our current capabilities. Therefore, in this thesis, the point of view is that we have gauge fixed lattice simulations and Dyson-Schwinger equations that provide non-perturbative data for some quantities as correlation functions. Using standard perturbation theory, these correlation functions are not reproduced and we should improve somehow our methods. Our main point is that in a class of gauges which are implementable in lattice up to date, taking into account Gribov copies changes radically the infrared behavior of such Green's functions and a good agreement with lattice results emerges. Therefore, from this point of view, the copies play a very important role. Also, from a more formal and pragmatic point of view, removing copies from the path integral justifies the assumptions of the Faddeev-Popov method making it improved after all. In the next chapter we give a brief review of how copies can be (partially) eliminated from the path integral giving an optimized quantization of non-Abelian gauge theories. 

\chapter{Getting rid of gauge copies: the Gribov-Zwanziger action} \label{ch.3}

In Ch.~\ref{ch.2} we pointed out the existence of Gribov copies in the standard quantization of Yang-Mills theories. Furthermore, we showed a subclass of these copies, those generated by infinitesimal gauge transformations, are associated with zero-modes of the Faddeev-Popov operator in the Landau and maximal Abelian gauges. In fact, this statement is more general (as will be clear in the next chapters). Therefore, from a pure technical point of view the presence of zero-modes of the Faddeev-Popov operator makes the gauge fixing procedure ill-defined, see Subsect.~\ref{FPtrick}. In this sense it is desirable to remove these zero-modes/infinitesimal copies from the path integral domain in such a way the Faddeev-Popov procedure is well-grounded. From a more physical picture, we could state these configurations correspond to an over counting in the path integral quantization and therefore we should eliminate them to keep the integration over physically inequivalent configurations. In this picture, removing all Gribov copies (including those generated by a finite gauge transformations) is essential. Nevertheless we adopt the following point of view: The removal of infinitesimal Gribov copies makes the Faddeev-Popov procedure well defined and justifies the assumption concerning the \textit{positivity} or, at least, the non-vanishing of the Faddeev-Popov operator. Hence, this corresponds to an improvement with respect to the standard quantization procedure. To complete the elimination of all copies, we should also eliminate finite ones, but this is a much harder task and we stick with the simpler one, which is already highly non-trivial. 

A consistent way to eliminate infinitesimal Gribov copies was already proposed by Gribov himself in his seminal paper, \cite{Gribov:1977wm}. In his paper, however, he worked out just a ``semiclassical" solution (we will be more specific on the meaning of semiclassical in this context later on). Some years later, Zwanziger was able to generalize the elimination of Gribov copies up to all orders in perturbation theory, \cite{Zwanziger:1989mf}. In this chapter we will review the main features of these constructions. In the already cited reviews, the reader can find almost all details for the technical computational steps, \cite{Sobreiro:2005ec,Vandersickel:2011zc,Vandersickel:2012tz}. The takeaway message of this chapter is that the elimination of infinitesimal Gribov copies from the path integral domain can be implemented by a local and renormalizable action known as the \textit{Gribov-Zwanziger} or, simply, the GZ action. For simplicity, we restrict ourselves to the construction of the GZ action in the Landau gauge, but we could also work the construction for \textit{e.g.} the maximal Abelian gauge. 

\section{The no-pole condition}

In this section we provide a review of Gribov's proposal to eliminate (infinitesimal) copies, \cite{Gribov:1977wm}. The idea is simple once we know the existence of the Gribov region $\Omega$ (see Sect.~\ref{section1.1}). This region is defined in such a way that no infinitesimal Gribov copies live inside it. Therefore, Gribov's strategy was to restrict the functional integral domain precisely to $\Omega$. In this way, no zero-modes of the Faddeev-Popov operator are taken into account. Very important is the fact that this region, as discussed in Sect.~\ref{section1.1}, contains all physical configurations, since all gauge orbits cross it. In this way, the restriction to $\Omega$ is a true elimination of spurious configurations since all gauge fields have a representative inside $\Omega$. So, the partition function can be formally written as

\begin{equation}
\EuScript{Z}= \int_\Omega\left[\EuScript{D}A\right]\mathrm{e}^{-S_{\mathrm{YM}}} \equiv \int\left[\EuScript{D}A\right]\left[\EuScript{D}\bar{c}\right]\left[\EuScript{D}c\right]\EuScript{V}(\Omega)\delta(\partial_{\mu}A^{a}_{\mu})\mathrm{e}^{-S_{\mathrm{YM}}-\int d^4x\bar{c}^{a}\partial_{\mu}D^{ab}_{\mu}c^b}\,,
\label{2.1}
\end{equation}

\noindent with $\EuScript{V}(\Omega)$ the factor responsible to restrict the domain of integration to $\Omega$. Intuitively, the function $\EuScript{V}(\Omega)$ works as a step function which imposes a cut-off at the Gribov horizon $\partial\Omega$. To make this construction more concrete, we benefit from the fact that the Faddeev-Popov operator is intrinsically related with the ghosts two-point function. In fact,  performing the path integral over the ghosts fields (coupled to an external source) and writing the expression of the two-point function, we have

\begin{equation}
\langle \bar{c}^{a}(x)c^{b}(y)\rangle = \int \left[\EuScript{D}A\right]\EuScript{V}(\Omega)\;\delta(\partial_{\mu}A^{a}_{\mu})\mathrm{det}(\EuScript{M})\left[\EuScript{M}^{-1}\right]^{ab}(x,y)\mathrm{e}^{-S_{\mathrm{YM}}}\,.
\label{2.2}
\end{equation} 

\noindent We note this expression contains $\EuScript{M}^{-1}$ which is well-defined inside $\Omega$, but not in the entire configuration space. Also, we observe the relation of the inverse of $\EuScript{M}$ and the ghosts two-point function. The ideia to concretly implement the restriction to $\Omega$ formulated by Gribov thus follows: To restrict the path integral to the Gribov region we should demand the ghosts two-point function does not develop poles. This is the so-called \textit{no-pole condition}. To understand its content better, let us initiate the analysis with standard perturbation theory. At the tree-level, the ghosts two-point function goes as $1/p^2$ and the only pole is $p^2=0$. However, for any other value of $p$, this function is positive and we are ``safe", \textit{i.e.} inside $\Omega$. This approximation, however, is poor since it does not reveal non-trivial poles and we know they exist. So, if this was an exact result no copies would disturb the quantization. Going to one-loop order the two-point function is, in momentum space,

\begin{equation}
\langle \bar{c}^{a}(p)c^{b}(-p)\rangle = \delta^{ab}\mathcal{G}(p)\equiv \delta^{ab}\frac{1}{p^2}\frac{1}{\left(1-\frac{11g^2N}{48\pi^2}\mathrm{ln}\frac{\Lambda^2}{p^2}\right)^{\frac{9}{44}}}\,,
\label{2.3}
\end{equation}

\noindent with $\Lambda$ being an ultraviolet cut-off. The function $\mathcal{G}(p)$ has two singularities: $p^2=0$ and $p^2=\Lambda^2\exp\left(-\frac{1}{g^2}\frac{48\pi^2}{11N}\right)$. Now, within this approximation, we see a non-trivial pole. Reaching this pole means we are outside the Gribov region and this is precisely what we want to avoid. Therefore, the no-pole condition can be stated as the removal all poles of the ghost two-point function but $p^2=0$. The meaning of the singularity $p^2=0$ is the following: Since $1/p^2$ is a positive function and develops a pole just at $p^2=0$, it is associate with the Gribov horizon, where the first zero-modes appear. 

The characterization of the no-pole condition is performed by computing the ghosts two-point function considering the gauge field $A^{a}_{\mu}$ as a classical field. So, we consider the quantity\footnote{Note we have some ``abuse" of notation, since we used $\mathcal{G}$ for the form factor of the \textit{true} ghosts two-point function, where the functional integral over $A$ is performed.} 

\begin{equation}
\mathcal{G}(p)=\frac{1}{N^2-1}\int\left[\EuScript{D}\bar{c}\right]\left[\EuScript{D}c\right]\bar{c}^{a}(p)c^a(-p)\mathrm{e}^{-(S_{\mathrm{YM}}+\int d^4x~\partial_{\mu}D^{ab}_{\mu}c^{b})}\,.
\label{2.4}
\end{equation}

\noindent This quantity can be computed order by order in perturbation theory. Gribov's original computation was performed up to second order in perturbation theory. As explicitly shown in \cite{Sobreiro:2005ec,Vandersickel:2011zc,Vandersickel:2012tz}, eq.(\ref{2.4}) computed up to second order in $g$ reduces to 

\begin{figure}
	\centering
		\includegraphics[width=0.80\textwidth]{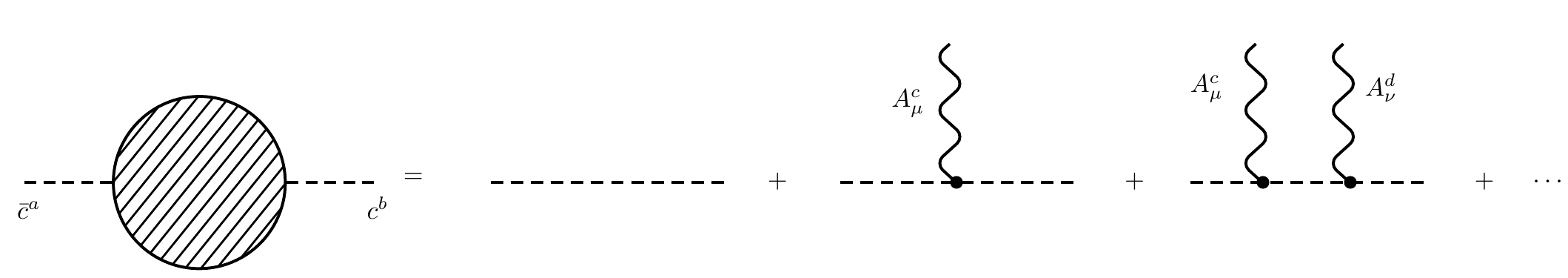}
	\caption{Ghost two-point function with insertions of gluons.}
	\label{fig:ghost2point}
\end{figure}

\begin{equation}
\mathcal{G}(p^2;A)=\frac{1}{p^2}+\frac{1}{V}\frac{1}{p^4}\frac{Ng^2}{N^2-1}\int \frac{d^{4}q}{(2\pi)^4}A^{a}_{\mu}(-q)A^{a}_{\nu}(q)\frac{(p-q)_{\mu}p_{\nu}}{(p-q)^2}\,,
\label{2.5}
\end{equation}

\noindent with $V$ the spacetime volume and the dependence on $A$ is made explicit to emphasize the functional integral over the gauge field was not performed. It is useful to rewrite eq.(\ref{2.5}) as 

\begin{equation}
\mathcal{G}(p^2,A)=\frac{1}{p^2}(1+\sigma(p,A))\,,
\label{2.6}
\end{equation}

\noindent with 

\begin{equation}
\sigma(p,A)=\frac{1}{V}\frac{1}{p^2}\frac{Ng^2}{N^2-1}\int\frac{d^4q}{(2\pi)^4}A^{a}_{\mu}(-q)A^{a}_{\nu}(q)\frac{(p-q)_{\mu}p_{\nu}}{(p-q)^{2}}\,.
\label{2.7}
\end{equation}

\noindent Is is possible to show $\sigma(p;A)$ decreases as $p^2$ increases, see \cite{Sobreiro:2005ec,Vandersickel:2011zc,Vandersickel:2012tz}. A nice trick is to write eq.(\ref{2.6}) as

\begin{equation}
\mathcal{G}(p,A)\approx\frac{1}{k^2}\frac{1}{1-\sigma(p,A)}\,,
\label{2.8}
\end{equation}

\noindent and due to the way $\sigma$ depends on $p^2$, it is enough to define the no-pole condition as 

\begin{equation}
\sigma(0,A)<1\,.
\label{2.9}
\end{equation}

\noindent To make eq.(\ref{2.9}) explicit we explore the transversality of the gauge field, namely, $p_\mu A^{a}_{\mu}(p)=0$ to write

\begin{equation}
A^{a}_{\mu}(-q)A^{a}_{\nu}(q)=\frac{1}{3}A^{a}_{\lambda}(-q)A^{a}_{\lambda}(q)\left(\delta_{\mu\nu}-\frac{q_{\mu}q_{\nu}}{q^2}\right)\equiv \frac{1}{3}A^{a}_{\lambda}(-q)A^{a}_{\lambda}(q)\mathcal{P}_{\mu\nu}(q)\,,
\label{2.10}
\end{equation}

\noindent and rewrite eq.(\ref{2.7}) as

\begin{eqnarray}
\sigma(p,A)&=&\frac{1}{3V}\frac{Ng^2}{N^2-1}\frac{p_{\mu}p_{\nu}}{p^2}\int\frac{d^4q}{(2\pi)^4}A^{a}_{\lambda}(-q)A^{a}_{\lambda}(q)\frac{1}{(p-q)^{2}}\mathcal{P}_{\mu\nu}(q)\nonumber \\
&=&\frac{1}{3V}\frac{Ng^2}{N^2-1}\left(\int\frac{d^4q}{(2\pi)^4}A^{a}_{\lambda}(-q)A^{a}_{\lambda}(q)\frac{1}{(p-q)^{2}}-\frac{p_{\mu}p_{\nu}}{p^2}\int\frac{d^4q}{(2\pi)^4}A^{a}_{\lambda}(-q)A^{a}_{\lambda}(q)\right.\nonumber\\
&\times&\left.\frac{1}{(p-q)^{2}}\frac{q_{\mu}q_{\nu}}{q^2}\right)\,.
\label{2.11}
\end{eqnarray}

\noindent We see that the limit $p^2\rightarrow 0$ must be carefully taken. As an intermediate step, 

\begin{eqnarray}
\sigma(0,A)&=&\frac{1}{3V}\frac{Ng^2}{N^2-1}\left(\int\frac{d^4q}{(2\pi)^4}A^{a}_{\lambda}(-q)A^{a}_{\lambda}(q)\frac{1}{q^{2}}-\lim_{p^2\rightarrow 0}\frac{p_{\mu}p_{\nu}}{p^2}\int\frac{d^4q}{(2\pi)^4}A^{a}_{\lambda}(-q)A^{a}_{\lambda}(q)\right.\nonumber\\
&\times&\left.\frac{1}{q^{2}}\frac{q_{\mu}q_{\nu}}{q^2}\right)\,,
\label{2.12}
\end{eqnarray}

\noindent which reduces to 

\begin{equation}
\sigma(0,A)=\frac{1}{4V}\frac{Ng^2}{N^2-1}\int\frac{d^4q}{(2\pi)^4}A^{a}_{\lambda}(-q)A^{a}_{\lambda}(q)\frac{1}{q^{2}}\,,
\label{2.13}
\end{equation}

\noindent by using 

\begin{equation}
\int d^4q\, f(q^2)\frac{q_{\mu}q_{\nu}}{q^2}=\frac{1}{d}\delta_{\mu\nu}\int d^4q\, f(q^2)\,.
\label{2.131}
\end{equation}

\noindent With this we can come back to eq.(\ref{2.9}) and define the explicit form of the factor $\EuScript{V}(\Omega)$, responsible to restrict the path integral domain to $\Omega$. Using eq.(\ref{2.9}),

\begin{equation}
\EuScript{V}(\Omega)=\theta(1-\sigma(0,A))\,,
\label{2.14}
\end{equation}

\noindent with $\theta(x)$ the standard Heaviside function. The path integral (\ref{2.1}) can be expressed as

\begin{equation}
\EuScript{Z}= \int\left[\EuScript{D}A\right]\left[\EuScript{D}\bar{c}\right]\left[\EuScript{D}c\right]\theta(1-\sigma(0,A))\delta(\partial_{\mu}A^{a}_{\mu})\mathrm{e}^{-S_{\mathrm{YM}}-\int d^4x\bar{c}^{a}\partial_{\mu}D^{ab}_{\mu}c^b}\,.
\label{2.15}
\end{equation}

The usual delta function in the partition function (\ref{2.15}) employs the Landau gauge condition. It is widely known we can lift this term to the exponential by means of the introduction of a Lagrange multiplier\footnote{See Ap.~\ref{appendixA}.} $b^a$. It is extremely desirable to lift the $\theta$ function as well and implement the restriction of the path integral domain as an effective modification of the gauge fixed action. To proceed in this direction, we make use of the integral representation of the $\theta$ function, namely

\begin{equation}
\theta(x)=\int^{+i\infty+\epsilon}_{-i\infty+\epsilon}\frac{d\beta}{2\pi i\beta}\mathrm{e}^{\beta x}\,.
\label{2.16}
\end{equation}

\noindent Using eq.(\ref{2.16}), we write the path integral (\ref{2.15}) as

\begin{equation}
\EuScript{Z}= \int^{+i\infty+\epsilon}_{-i\infty+\epsilon}\frac{d\beta}{2\pi i\beta}\int\left[\EuScript{D}A\right]\left[\EuScript{D}\bar{c}\right]\left[\EuScript{D}c\right]\mathrm{e}^{\beta (1-\sigma(0,A))}\delta(\partial_{\mu}A^{a}_{\mu})\mathrm{e}^{-S_{\mathrm{YM}}-\int d^4x\bar{c}^{a}\partial_{\mu}D^{ab}_{\mu}c^b}\,.
\label{2.17}
\end{equation}

\subsection{Gluon propagator and the Gribov parameter}\label{gpropgparam}

Expression (\ref{2.17}) is an explicit form for the partition function of $SU(N)$ Yang-Mills theory in the Landau gauge restricted to the Gribov region $\Omega$ (within Gribov's approximation, of course). A natural step is to compute the gluon propagator (one of the main actors of this thesis) with the modified partition function (\ref{2.17}). As usual, we retain the quadratic terms in $A^{a}_{\mu}$ in the gauge fixed action (and simply ignore the ghost sector which does not contribute to this computation),

\begin{equation}
\EuScript{Z}_q\left[J\right]=\int\frac{d\beta}{2\pi i\beta}\int\left[\EuScript{D}A\right]\mathrm{e}^{\beta(1-\sigma(0,A))}\mathrm{e}^{-\left(\int d^4x~(\partial_{\mu}A^{a}_{\nu}-\partial_{\nu}A^{a}_{\mu})^2+\int d^4x~\frac{1}{2\alpha}(\partial_{\mu}A^{a}_{\mu})^2+\int d^4x~A^a_{\mu}J^{a}_{\mu}\right)}\,,
\label{2.18}
\end{equation}

\noindent whereby $J^{a}_{\mu}$ is an external source introduced for the computation of correlation functions and $\alpha$ is a gauge parameter introduced to perform the functional integral over the Lagrange multiplier which enforces the gauge condition. At the end of the computation, we have to take the limit $\alpha\rightarrow 0$ to recover the Landau gauge. Performing standard computations, see \cite{Sobreiro:2005ec}, we have

\begin{equation}
\langle A^{a}_{\mu}(p)A^{b}_{\nu}(-p)\rangle = \int\frac{d\beta}{2\pi i\beta}\mathrm{e}^{\beta}\left(\mathrm{det}~K^{ab}_{\mu\nu}\right)^{-1/2}(K^{ab}_{\mu\nu})^{-1}(p)\,,
\label{2.19}
\end{equation}

\noindent with

\begin{equation}
K^{ab}_{\mu\nu}=\delta^{ab}\left(\beta\frac{1}{V}\frac{1}{2}\frac{Ng^2}{N^2-1}\frac{1}{p^2}\delta_{\mu\nu}+p^2\delta_{\mu\nu}+\left(\frac{1}{\alpha}-1\right)p_{\mu}p_{\nu}\right)\,.
\label{2.20}
\end{equation}

\noindent The determinant of the operator $K^{ab}_{\mu\nu}$ can be computed using standard techniques, but a pedagogical step by step guide can be found in the appendix of \cite{Vandersickel:2011zc,Vandersickel:2012tz}. We report the result here,

\begin{equation}
\left(\mathrm{det}~K^{ab}_{\mu\nu}\right)^{-1/2}=\mathrm{exp}\left[-\frac{3(N^2-1)}{2}V\int \frac{d^4q}{(2\pi)^4}~\mathrm{ln}\left(q^2+\frac{\beta Ng^2}{N^2-1}\frac{1}{2V}\frac{1}{q^2}\right)\right]\,.
\label{2.21}
\end{equation}

\noindent Plugging eq.(\ref{2.21}) into eq.(\ref{2.19}) we obtain

\begin{equation}
\langle A^{a}_{\mu}(p)A^{b}_{\nu}(-p)\rangle = \int\frac{d\beta}{2\pi i}\mathrm{e}^{f(\beta)}(K^{ab}_{\mu\nu})^{-1}(p)\,,
\label{2.22}
\end{equation}

\noindent with

\begin{equation}
f(\beta)=\beta-\mathrm{ln}\beta-\frac{3}{2}(N^2-1)V\int\frac{d^4q}{(2\pi)^4}\mathrm{ln}\left(q^2+\frac{\beta Ng^2}{N^2-1}\frac{1}{2V}\frac{1}{q^2}\right)\,.
\label{2.23}
\end{equation}

\noindent To perform the integral over $\beta$, we apply the steepest descent approximation method. Hence, we impose

\begin{equation}
f'(\beta_0)=0\,\,\Rightarrow\,\, 1=\frac{1}{\beta_0}+\frac{3}{4}Ng^2\int\frac{d^4q}{(2\pi)^4}\frac{1}{\left(q^4+\frac{\beta_0 Ng^2}{N^2-1}\frac{1}{2V}\right)}
\label{2.24}
\end{equation}

\noindent and define the so-called \textit{Gribov parameter} $\gamma$ by

\begin{equation}
\gamma^4=\frac{\beta_0}{4V(N^2-1)}\,.
\label{2.25}
\end{equation}

\noindent We note the Gribov parameter $\gamma$ has mass dimension. Since, formally $V\rightarrow \infty$, to have a finite value for $\gamma$ we should have $\beta_0 \propto V$. Therefore, the term $1/\beta_0$ can be neglected from eq.(\ref{2.24}). The result is

\begin{equation}
1=\frac{3}{4}Ng^2\int\frac{d^4q}{(2\pi)^4}\frac{1}{\left(q^4+2Ng^2\gamma^4\right)}
\label{2.26}
\end{equation} 

\noindent which is recognized as a \textit{gap equation} responsible to fix the Gribov (mass) parameter $\gamma$. Note the important fact that $\gamma$ acts as an infrared regulator for the integral (\ref{2.26}). Solving eq.(\ref{2.26}) gives

\begin{equation}
\int\frac{d^4q}{(2\pi)^4}\frac{1}{\left(q^4+2Ng^2\gamma^4\right)}=\frac{1}{(2\pi)^4}\int d\Omega_4\int^{\Lambda}_{0}dq\frac{q^3}{\left(q^4+2Ng^2\gamma^4\right)}\,,
\label{2.27}
\end{equation}

\noindent where $\int d\Omega_4=2\pi^2$ and $\Lambda$ is an ultraviolet cut-off. The remaining integral leads to 

\begin{eqnarray}
\int^{\Lambda}_{0}dq\frac{q^3}{\left(q^4+2Ng^2\gamma^4\right)}&=&\frac{1}{4}\left(\mathrm{ln}(\Lambda^4+2g^2N\gamma^4)-\mathrm{ln}(2g^2N\gamma^4)\right)\approx \frac{1}{4}\left(\mathrm{ln}(\Lambda^4)-\mathrm{ln}(2g^2N\gamma^4)\right)\nonumber \\
&=&\frac{1}{4}\mathrm{ln}\left(\frac{\Lambda^4}{2g^2N\gamma^4}\right)\,.
\label{2.28}
\end{eqnarray}

\noindent Plugging eq.(\ref{2.28}) into eq.(\ref{2.27}) and the result in eq.(\ref{2.26}), we obtain the following result for $\gamma^2$,

\begin{equation}
\gamma^2=\frac{\Lambda^2}{2g^2N}\mathrm{e}^{-\frac{64\pi^2}{3Ng^2}}\,.
\label{2.29}
\end{equation}

\noindent We will address more comments on the Gribov parameter soon, but let us return to the computation of the gluon propagator. The gluon propagator (\ref{2.22}) reduces to

\begin{equation}
\langle A^{a}_{\mu}(p)A^{b}_{\nu}(-p)\rangle = \delta^{ab}\frac{p^2}{p^4+2g^2N\gamma^4}\mathcal{P}_{\mu\nu}\equiv  \delta^{ab} \mathcal{D}(p)\mathcal{P}_{\mu\nu}\,,
\label{2.30}
\end{equation}

\noindent where we must emphasize we dropped the ${(2\pi i)}^{-1}\mathrm{e}^{f(\beta_0)}$ factor (absorbed in a normalization factor not written explicitly here) and took the $\alpha\rightarrow 0$ limit. Very often, it is referred to a propagator with the behavior of (\ref{2.30}) as of Gribov-type, \cite{Capri:2015mna}. We highlight the following properties of (\ref{2.30}) - see Fig.~\ref{fig3} for a qualitative plot of the Gribov and the usual perturbative propagator (the plots are not supposed to be numerically precise, but only illustrative):

\begin{itemize}
\item In the infrared regime, the gluon propagator form factor is suppressed by the presence of the Gribov parameter $\gamma$.

\item The gluon propagator form factor goes to zero at zero momentum. A very different behavior is observed in standard perturbation theory, where the form factor diverges at the origin.

\item Setting $\gamma\rightarrow 0$ or, equivalently, considering $p^2\rightarrow\infty$ we recover the standard $1/p^2$ perturbative form factor. 

\item The presence of the Gribov parameter generates two complex poles $p^2=\pm i\sqrt{2g^2N}\gamma^2$. This forbids us to interpret gluons as physical excitations and is interpreted as a manifestation of confinement.  
\end{itemize}

\begin{figure}[t]
	\centering
		\includegraphics[width=0.70\textwidth]{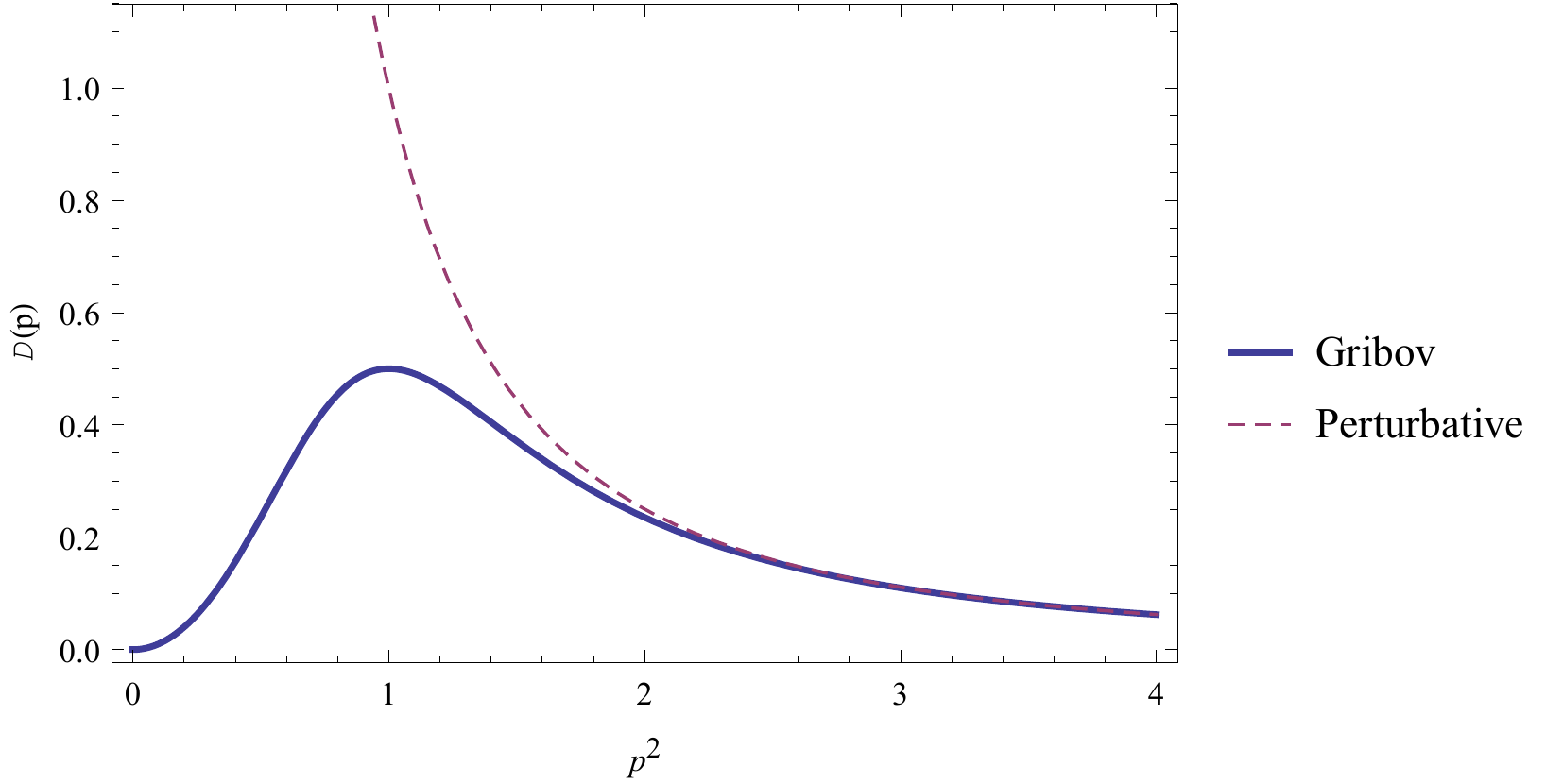}
	\caption{Qualitative comparison between the Gribov gluon propagator and the perturbative one.}
	\label{fig3}
\end{figure}

\noindent We see the restriction of the path integral domain to the Gribov region $\Omega$ affects substantially the gluon propagator. In particular (and as expected by previous discussions), the effects of taking into account Gribov copies are manifest in the infrared (non-perturbative region). The introduction of a boundary in the configuration space, namely, the Gribov horizon naturally generates a mass gap (the Gribov parameter). This massive parameter is not free but fixed in a self consistent way through the gap equation (\ref{2.26}). Also, since $\gamma$ is directly related to the restriction to $\Omega$, setting $\gamma\rightarrow 0$ should correspond to the standard Faddeev-Popov quantization. A propagator of Gribov-type is also known as a \textit{scaling} propagator.  We end this subsection with a remark concerning the gap equation (\ref{2.26}): Within Gribov's approximation, the gap equation must be regularized since the integral defining the equation is ultraviolet divergent. For a full consistent treatment using renormalization theory, we need a renormalizable action which implements the restriction to $\Omega$. This was achieved by Zwanziger in \cite{Zwanziger:1989mf} and will be discussed in the next section. Before, though, we briefly discuss the ghost propagators in Gribov's approximation. 

\subsection{Ghost propagator}

Now that we have the gluon propagator expression (\ref{2.30}) we can finish the computation of the one-loop ghost two-point function. From eq.(\ref{2.8}),

\begin{equation}
\langle \bar{c}^{a}(-p)c^b(p)\rangle_{\mathrm{1-loop}}=\delta^{ab}\frac{1}{p^2}\frac{1}{1-\langle \sigma(p,A) \rangle}\,,
\label{2.31}
\end{equation}

\noindent with

\begin{eqnarray}
\langle \sigma(p,A) \rangle &=& Ng^2\frac{p_{\mu}p_{\nu}}{p^2}\int\frac{d^4q}{(2\pi)^4}\langle A^{a}_{\mu}(-q)A^{a}_{\nu}(q)\rangle\frac{1}{(p-q)^{2}}\nonumber \\
&=&  Ng^2\frac{p_{\mu}p_{\nu}}{p^2}\int\frac{d^4q}{(2\pi)^4}\frac{q^2}{q^4+2g^2N\gamma^4}\frac{1}{(p-q)^{2}}\left(\delta_{\mu\nu}-\frac{q_{\mu}q_{\nu}}{q^2}\right)\,.
\label{2.32}
\end{eqnarray}

\noindent Inhere, we will restrict ourselves to the infrared behavior of the correlation function \textit{i.e.} around $p^2\approx 0$. To proceed, we will use a trick which essentially consists in writing ``one" in a fancy way. Let us begin with the following relation,

\begin{equation}
\int \frac{d^4q}{(2\pi)^4}\frac{1}{q^4+2g^2N\gamma^4}\left(\delta_{\mu\nu}-\frac{q_{\mu}q_{\nu}}{q^2}\right)=\frac{3}{4}\delta_{\mu\nu}\int \frac{d^4q}{(2\pi)^4}\frac{1}{q^4+2g^2N\gamma^4}
\label{2.33}
\end{equation}

\noindent Invoking the gap equation (\ref{2.26}), we can rewrite eq.(\ref{2.33}) as

\begin{equation}
\frac{3}{4}\delta_{\mu\nu}\int \frac{d^4q}{(2\pi)^4}\frac{1}{q^4+2g^2N\gamma^4}=\frac{1}{Ng^2}\delta_{\mu\nu}\,,
\label{2.34}
\end{equation}

\noindent and contracting eq.(\ref{2.34}) with $p_{\mu}p_{\nu}/p^2$ and using eq.(\ref{2.33}) again, we obtain

\begin{equation}
\frac{p_{\mu}p_{\nu}}{p^2}\int \frac{d^4q}{(2\pi)^4}\frac{1}{q^4+2g^2N\gamma^4}\left(\delta_{\mu\nu}-\frac{q_{\mu}q_{\nu}}{q^2}\right)=\frac{1}{Ng^2}\,.
\label{2.35}
\end{equation}

\noindent This implies, 

\begin{equation}
Ng^2\frac{p_{\mu}p_{\nu}}{p^2}\int \frac{d^4q}{(2\pi)^4}\frac{1}{q^4+2g^2N\gamma^4}\left(\delta_{\mu\nu}-\frac{q_{\mu}q_{\nu}}{q^2}\right)=1\,.
\label{2.36}
\end{equation}

\noindent Now that we have this weird (but convenient) way of expressing the unity, we can write

\begin{eqnarray}
1-\langle \sigma(p,A) \rangle &=& Ng^2\frac{p_{\mu}p_{\nu}}{p^2}\int\frac{d^4q}{(2\pi)^4}\frac{1}{q^4+2g^2N\gamma^4}\underbrace{\left(1-\frac{q^2}{(p-q)^{2}}\right)}_{(\ast)}\left(\delta_{\mu\nu}-\frac{q_{\mu}q_{\nu}}{q^2}\right)\nonumber\\
&\equiv& Ng^2\frac{p_{\mu}p_{\nu}}{p^2}\mathcal{C}_{\mu\nu}(p)\,,
\label{2.37}
\end{eqnarray}

\noindent from which we immediately obtain

\begin{equation}
\mathcal{C}_{\mu\nu}(0)=0\,.
\label{2.38}
\end{equation}

\noindent To obtain a more complete information about the limit $p\rightarrow 0$, we take advantage from the fact that $(\ast)$ can be expressed as

\begin{equation}
(\ast)=\left(1-\frac{q^2}{(p-q)^{2}}\right)=\frac{p^2-2p\cdot q}{q^2\left(\frac{p^2}{q^2}-2\frac{p\cdot q}{q^2}+1\right)}\approx \frac{p^2-2\overbrace{p\cdot q}^{(\ast\ast)}}{q^2}\,,
\label{2.39}
\end{equation}

\noindent whereby we retained terms up to $p^2$. We note the term $(\ast\ast)$ forms an odd function on $q$ to be integrated within a symmetric interval. Therefore $(\ast\ast)$ gives an automatic vanishing term and we can rewrite $\mathcal{C}_{\mu\nu}(p)$ as

\begin{equation}
\lim_{p\rightarrow 0}\mathcal{C}_{\mu\nu}(p)=p^2\int\frac{d^4q}{(2\pi)^4}\frac{1}{q^2}\frac{1}{q^4+2g^2N\gamma^4}\left(\delta_{\mu\nu}-\frac{q_{\mu}q_{\nu}}{q^2}\right)=\frac{3p^2}{4}\delta_{\mu\nu}\int\frac{d^4q}{(2\pi)^4}\frac{1}{q^2}\frac{1}{q^4+2g^2N\gamma^4}\,.
\label{2.40}
\end{equation}

\noindent The integral (is UV finite) can be easily performed,

\begin{eqnarray}
\lim_{p\rightarrow 0}\mathcal{C}_{\mu\nu}(p)&=&\frac{3p^2}{4}\delta_{\mu\nu}\int\frac{d^4q}{(2\pi)^4}\frac{1}{q^2}\frac{1}{q^4+2g^2N\gamma^4}=\frac{3p^2}{4}\delta_{\mu\nu}\int d\Omega_{4}\int^{\infty}_{0} \frac{dq}{(2\pi)^4}\frac{q}{q^4+2g^2N\gamma^4}\nonumber\\
&=& \frac{3p^2}{128\pi}\frac{1}{\sqrt{2g^2N}\gamma^2}\delta_{\mu\nu}\,.
\label{2.41}
\end{eqnarray}

\noindent With eq.(\ref{2.41}) the infrared behavior of the ghost two-point function is given by

\begin{equation}
\lim_{p\rightarrow 0}\langle \bar{c}^{a}(-p)c^b(p)\rangle_{\mathrm{1-loop}}=\delta^{ab}\frac{128\pi\sqrt{2g^2N}\gamma^2}{3}\frac{1}{p^4}\,.
\label{2.42}
\end{equation}

\noindent We see from (\ref{2.42}) that the ghost propagator is enhanced \textit{i.e.} more singular near $p=0$ than the usual $1/p^2$ one. At this stage we just present expressions (\ref{2.30}) and (\ref{2.42}) in $d=4$. It is possible to show, however, these results are also valid for $d=2,3$ in the context of the Gribov modification. A more detailed discussion on the relation between propagators and spacetime dimensions will be presented in Ch.~\ref{LCGrevisited}.
 
\section{Zwanziger's Horizon function}

The solution proposed by Gribov and presented in the last section, although self-consistent, had the limitation of implementing the no-pole condition just at leading order. An all order implementation would be desirable to understand in more details the effects of the restriction to $\Omega$. The first effort in the direction of restricting the path integral domain to the Gribov region $\Omega$ to all orders in perturbation theory was done by Zwanziger in \cite{Zwanziger:1989mf}. In his seminal paper, Zwanziger implemented the restriction to $\Omega$ using a different strategy. Instead of dealing with the ghost propagator he managed to study directly the Faddeev-Popov operator spectrum,

\begin{equation}
\EuScript{M}^{ab}\chi^b=-\partial_{\mu}D^{ab}_{\mu}\chi^b=\epsilon(A)\chi^a\,\,\, \mathrm{with}\,\,\, \partial_{\mu}A^{a}_{\mu}=0\,,
\label{2.43}
\end{equation}

\noindent and defining properly the Gribov region $\Omega$ by the condition

\begin{equation}
\epsilon_{\mathrm{min}}(A)\geq 0\,,
\label{2.44}
\end{equation}

\noindent \textit{i.e.} the $A$-dependent minimum eigenvalue of the Faddeev-Popov operator should be non-negative. This defines precisely the region where the operator $\EuScript{M}$ is positive \textit{i.e.} the Gribov region $\Omega$. With this, he computed the trace of the Faddeev-Popov operator and found the following expression\footnote{This computation is lengthy and we refer to \cite{Vandersickel:2011zc,Vandersickel:2012tz,Zwanziger:1989mf} for details.},

\begin{equation}
\mathrm{Tr}\;\EuScript{M}(A) = Vd(N^2-1)-H_L(A)\,,
\label{2.45}
\end{equation} 

\noindent with $d$ the spacetime dimension and $V$ its volume. The function $H_L(A)$ will play a prominent role in this thesis and is the so-called \textit{horizon function}. Explicitly, 

\begin{equation}
H_L(A)=g^2\int d^dxd^dy~f^{abc}A^{b}_{\mu}(x)\left[\EuScript{M}^{-1}(A)\right]^{ad}(x,y)f^{dec}A^{e}_{\mu}(y)\,.
\label{2.46}
\end{equation} 

\noindent Zwanziger argued (see \cite{Vandersickel:2011zc,Vandersickel:2012tz,Zwanziger:1989mf} for details) that condition (\ref{2.44}) is well implemented by demanding the non-negativity of (\ref{2.45}). This equivalence should hold at the thermodynamic or infinity spacetime volume limit. Also, under considerations of the thermodynamic limit and the implications of this for the equivalence between canonical and microcanonical ensembles, Zwanziger implemented condition (\ref{2.44}) in the path integral. The result is

\begin{equation}
\EuScript{Z}=\int\left[\EuScript{D}A\right]\left[\EuScript{D}\bar{c}\right]\left[\EuScript{D}c\right]\left[\EuScript{D}\bar{b}\right]\mathrm{e}^{-S_{\mathrm{GZ}}}\,,
\label{2.47}
\end{equation}

\noindent with

\begin{equation}
S^{L}_{\mathrm{GZ}}=S_{\mathrm{YM}}+\int d^dx\left(b^a\partial_{\mu}A^{a}_{\mu}+\bar{c}^{a}\partial_{\mu}D^{ab}_{\mu}c^{b}\right)+\gamma^{4}H_L(A)-\gamma^{4}Vd(N^2-1)\,,
\label{2.48}
\end{equation}

\noindent where $\gamma$ is a mass parameter (the same Gribov parameter we introduced in Subsect.~\ref{gpropgparam}) which is not free but fixed through the so-called \textit{horizon condition}

\begin{equation}
\langle H_{L}(A)\rangle=Vd(N^2-1)\,,
\label{2.49}
\end{equation}

\noindent whereby expectation values $\langle\ldots\rangle$ are taken with respect to the path integral with modified measure (\ref{2.47}) - this is the reason why although $\gamma$ is not apparent in eq.(\ref{2.49}) it will enter the expectation value computation. The action defined by (\ref{2.48}) is the so-called \textit{Gribov-Zwanziger} action (which we shall frequently refer to as GZ action). Two important remarks about the GZ action can be immediately done: $(i)$ This action effectively implements the restriction of the path integral domain to the Gribov region $\Omega$ and therefore removes infinitesimal Gribov copies from the functional integral; $(ii)$ the horizon function contain the inverse of the Faddeev-Popov operator, which is a well-defined object since the restriction to $\Omega$ ensures $\EuScript{M}$ is positive. Due to the form of the horizon function, the GZ action is clearly non-local, an inconvenient feature for the application of standard quantum field theories techniques. Remarkably it is possible to reformulate the GZ action in local fashion by the introduction of a suitable set of auxiliary fields. This procedure will be described in the next subsection. To close this discussion, we emphasize a highly non-trivial feature: Although Gribov and Zwanziger pursued different paths to construction a partition function that takes into account infinitesimal Gribov copies, it was shown that working Gribov's procedure to all orders in perturbation theory leads to the same result Zwanziger's found \cite{Gomez:2009tj,Capri:2012wx}. This is a non-trivial check and also very reassuring. 

\subsection{Localization of the Gribov-Zwanziger action}\label{localizationGZaction}

The GZ action can be cast in a local form by the introduction of a suitable set of auxiliary fields. Essentially, we want to localize the horizon function term and for this, we write the following identity,

\begin{eqnarray}
\mathrm{e}^{-\gamma^{4}H_L(A)}&=&\exp{\left[\int d^dxd^dy~\underbrace{\gamma^2gf^{abc}A^{b}_{\mu}(x)}_{\bar{J}^{ac}_{\mu}}\left[-\EuScript{M}^{-1}(A)\right]^{ad}(x,y)\underbrace{\gamma^2gf^{dec}A^{e}_{\mu}(y)}_{J^{dc}_{\mu}}\right]}\nonumber\\
&=&(\mathrm{det}\;(-\EuScript{M}))^{d(N^2-1)}\int \left[\EuScript{D}\bar{\varphi}\right]\left[\EuScript{D}\varphi\right]\exp \left[\int d^dx\int d^dy\left(\bar{\varphi}^{ac}_{\mu}(x)\EuScript{M}^{ab}(x,y)\varphi^{bc}_{\mu}(y)\right)\right.\nonumber\\
&+&\left.\int d^dx\left(\bar{J}^{ab}_{\mu}(x)\varphi^{ab}_{\mu}(x)+\bar{\varphi}^{ab}_{\mu}(x)J^{ab}_{\mu}(x)\right)\right]\,,
\label{2.50}
\end{eqnarray}

\noindent with $(\bar{\varphi},\varphi)^{ab}_{\mu}$ a pair of bosonic fields. Finally, we can also lift to an exponential the term $(\mathrm{det}\;(-\EuScript{M}))^{d(N^2-1)}$ by the introduction of a pair of anti-commuting fields $(\bar{\omega},\omega)^{ab}_{\mu}$,

\begin{equation}
(\mathrm{det}\;(-\EuScript{M}))^{d(N^2-1)}=\int \left[\EuScript{D}\bar{\omega}\right]\left[\EuScript{D}\omega\right]\exp\left[-\int d^dx\int d^dy~\bar{\omega}^{ac}_{\mu}(x)\EuScript{M}^{ab}(x,y)\omega^{bc}_{\mu}(y)\right]\,.
\label{2.51}
\end{equation}

\noindent This implies the horizon function term can be rewritten as 

\begin{eqnarray}
\mathrm{e}^{-\gamma^{4}H_L(A)}&=& \int \left[\EuScript{D}\bar{\omega}\right]\left[\EuScript{D}\omega\right]\left[\EuScript{D}\bar{\varphi}\right]\left[\EuScript{D}\varphi\right]\exp\left[\int d^dx\int d^dy\left(\bar{\varphi}^{ac}_{\mu}(x)\EuScript{M}^{ab}(x,y)\varphi^{bc}_{\mu}(y)\right.\right.\nonumber\\
&-&\left.\left.\bar{\omega}^{ac}_{\mu}(x)\EuScript{M}^{ab}(x,y)\omega^{bc}_{\mu}(y)\right)-\gamma^2g\int d^dx~f^{bac}A^{b}_{\mu}(x)(\varphi+\bar{\varphi})^{ac}_{\mu}(x)\right]\,,
\label{2.52}
\end{eqnarray}

\begin{table}[t]
\centering
\begin{tabular}{|c|c|c|c|c|}
\hline
Fields & $\overline{\varphi}$ & $\varphi$ & $\overline{\omega}$ & $\omega$ \\ \hline
Dimension & 1 & 1 & 1 & 1\\
Ghost number & 0 & 0 & $-1$ & 1 \\ \hline
\end{tabular}
\caption{Quantum numbers of the auxiliary fields.}
\label{table2}
\end{table}

\noindent and the partition function for the GZ action is expressed as

\begin{equation}
\EuScript{Z}=\int \underbrace{\left[\EuScript{D}A\right]\left[\EuScript{D}b\right]\left[\EuScript{D}\bar{c}\right]\left[\EuScript{D}c\right]\left[\EuScript{D}\bar{\omega}\right]\left[\EuScript{D}\omega\right]\left[\EuScript{D}\bar{\varphi}\right]\left[\EuScript{D}\varphi\right]}_{\left[\EuScript{D}\mu_{\mathrm{GZ}}\right]}\mathrm{e}^{-S^{L}_{\mathrm{GZ}}}\,,
\label{2.53}
\end{equation}

\noindent where we define $\left[\EuScript{D}\mu_{\mathrm{GZ}}\right]$ for convenience and

\begin{eqnarray}
S^{L}_{\mathrm{GZ}} &=& S_{\mathrm{YM}}+\int d^dx\left(b^a\partial_{\mu}A^{a}_{\mu}+\bar{c}^{a}\partial_{\mu}D^{ab}_{\mu}c^b\right)-\int d^dx\left(\bar{\varphi}^{ac}_{\mu}\EuScript{M}^{ab}\varphi^{bc}_{\mu}-\bar{\omega}^{ac}_{\mu}\EuScript{M}^{ab}\omega^{bc}_{\mu} \right)\nonumber\\
&+&\gamma^2\int d^dx~gf^{abc}A^{a}_{\mu}(\varphi+\bar{\varphi})^{bc}_{\mu}-\int d^dx~d\gamma^4(N^2-1)\,.
\label{2.54}
\end{eqnarray}  

\noindent is the \textit{local} Gribov-Zwanziger action. We will call simply Gribov-Zwanziger (GZ) action either (\ref{2.48}) or (\ref{2.54}). The GZ action, besides written in a local form, is renormalizable at all orders in perturbation theory (actually, we have to perform a shift on the $\omega$ field which will be discussed in the next subsection) \cite{Vandersickel:2011zc,Vandersickel:2012tz,Zwanziger:1989mf,Dudal:2008sp}. Therefore, in the Landau gauge we have a local and renormalizable action which takes into account the restriction to the Gribov region $\Omega$ \textit{i.e.} eliminates at least all infinitesimal Gribov copies. In this local fashion, the horizon condition which fixes the $\gamma$-parameter is written as

\begin{equation}
\frac{\partial\mathcal{E}_0}{\partial\gamma^2}=0\,\,\,\Rightarrow\,\,\,-\langle gf^{abc}A^{a}_{\mu}(\varphi+\bar{\varphi})^{bc}_{\mu}\rangle+2\gamma^2d(N^2-1)=0\,,
\label{2.55}
\end{equation}

\noindent where

\begin{equation}
\mathrm{e}^{-V\mathcal{E}_0}=\int \left[\EuScript{D}\mu_{\mathrm{GZ}}\right]\mathrm{e}^{-S^{L}_{\mathrm{GZ}}}\,.
\label{2.56}
\end{equation}

\noindent We refer to the fact that working out condition (\ref{2.55}) to leading order, we end up with the gap equation (\ref{2.26}). This is precisely the very first evidence of the equivalence between Gribov's no-pole condition and Zwanziger's horizon condition. 

\subsection{The fate of BRST symmetry}\label{brstbreakingGZ}

In this subsection we introduce one of the main features we will further explore in this thesis: The Gribov-Zwanziger action (\ref{2.54}) \textit{breaks} the BRST symmetry - see \cite{Becchi:1974xu,Becchi:1974md,Becchi:1975nq,Tyutin:1975qk,Barnich:2000zw}- explicitly but in a \textit{soft} way. To explain this fact in a clear way, we need to clarify a few aspects before. First of all, action (\ref{2.54}) within the path integral admits the following non-local field redefinition with trivial Jacobian,

\begin{equation}
\omega^{ab}_{\mu}\,\,\,\longrightarrow\,\,\,\omega^{ab}_{\mu}+gf^{dlm}\int d^dy\left[\EuScript{M}^{-1}\right]^{ad}(x,y)\partial_{\nu}\left(\varphi^{mb}_{\mu}D^{le}_{\nu}c^e\right)\,.
\label{2.57}
\end{equation}

\noindent The resulting action, after field redefinition (\ref{2.57}) and taking into account the triviality of the Jacobian in the path integral measure is written as 

\begin{eqnarray}
S^{L}_{\mathrm{GZ}} &=& S_{\mathrm{YM}}+\int d^dx\left(b^a\partial_{\mu}A^{a}_{\mu}+\bar{c}^{a}\partial_{\mu}D^{ab}_{\mu}c^b\right)-\int d^dx\left(\bar{\varphi}^{ac}_{\mu}\EuScript{M}^{ab}\varphi^{bc}_{\mu}-\bar{\omega}^{ac}_{\mu}\EuScript{M}^{ab}\omega^{bc}_{\mu}\right.\nonumber\\
&+&\left.gf^{adl}\bar{\omega}^{ac}_{\mu}\partial_{\nu}\left(\varphi^{lc}_{\mu}D^{de}_{\nu}c^e\right) \right)
+\gamma^2\int d^dx~gf^{abc}A^{a}_{\mu}(\varphi+\bar{\varphi})^{bc}_{\mu}\nonumber\\
&-&\int d^dx~d\gamma^4(N^2-1)\,.
\label{2.58}
\end{eqnarray}

\noindent This shift on the $\omega$ field is relevant for the BRST discussion. For the standard Faddeev-Popov fields $(A,\bar{c},c,b)$ the BRST transformations are given by eq.(\ref{a11}) while the auxiliary localizing fields $(\varphi,\bar{\varphi},\omega,\bar{\omega})$ transforms as BRST doublets in such a way they never enter the non-trivial part of the cohomology of $s$, the BRST operator, \cite{Piguet:1995er}. So, the complete set of BRST transformations is

\begin{align}
sA^{a}_{\mu}&=-D^{ab}_{\mu}c^b\,,     &&sc^a=\frac{g}{2}f^{abc}c^bc^c\,, \nonumber\\
s\bar{c}^a&=b^{a}\,,     &&sb^{a}=0\,, \nonumber\\
s\varphi^{ab}_{\mu}&=\omega^{ab}_{\mu}\,,   &&s\omega^{ab}_{\mu}=0\,, \nonumber\\
s\bar{\omega}^{ab}_{\mu}&=\bar{\varphi}^{ab}_{\mu}\,,         &&s\bar{\varphi}^{ab}_{\mu}=0\,,
\label{2.59}
\end{align}

\noindent with $s^2=0$. The GZ action (\ref{2.58}) can be rewritten as

\begin{eqnarray}
S^{L}_{\mathrm{GZ}} &=& S_{\mathrm{YM}}+s\int d^dx~\bar{c}^{a}\partial_{\mu}A^{a}_{\mu}-s\int d^dx~\bar{\omega}^{ac}_{\mu}\EuScript{M}^{ab}\varphi^{bc}_{\mu}\nonumber\\
&+&\gamma^2\int d^dx~gf^{abc}A^{a}_{\mu}(\varphi+\bar{\varphi})^{bc}_{\mu}-\int d^dx~d\gamma^4(N^2-1)\,,
\label{2.60}
\end{eqnarray}

\noindent and we see explicitly that the shift on $\omega$ allows us to write the auxiliary fields sector in a BRST exact way due to the presence of the term 

\begin{equation}
\int d^dx~gf^{adl}\bar{\omega}^{ac}_{\mu}\partial_{\nu}\left(\varphi^{lc}_{\mu}D^{de}_{\nu}c^e\right)\,.
\label{2.61}
\end{equation}

\noindent A comment is relevant here: Written in the form (\ref{2.54}), is clear that as long as we take $\gamma=0$, we recover the Faddeev-Popov action, since the integration over the auxiliary fields sector is just an insertion of the unity. On the other hand, written as (\ref{2.60}) the limit $\gamma=0$ leads to an addition of a BRST exact term which does not affect physical gauge invariant observables. Therefore, independently of the form we write the GZ action (\ref{2.54}) or (\ref{2.60}), the Faddeev-Popov action and its physical content are recovered as long as the parameter $\gamma$ which implements the restriction to the Gribov region is set to zero. Of course, this is a very important check of consistency.

Given the nilpotency of the BRST operator $s$, it is very easy from eq.(\ref{2.60}) to obtain the following expression

\begin{equation}
sS^{L}_{\mathrm{GZ}}=\Delta_{\gamma^2}=\gamma^2gf^{abc}\int d^dx\left(A^{a}_{\mu}\omega^{bc}_{\mu}-D^{ad}_{\mu}c^{d}(\varphi+\bar{\varphi})^{bc}_{\mu}\right)\,.
\label{2.62}
\end{equation}

\noindent Eq.(\ref{2.62}) shows an outstanding feature of the GZ action: It breaks the BRST symmetry explicitly but the presence of the Gribov parameter $\gamma$ makes the breaking \textit{soft}. From the explicit solution of the gap equation which fixes $\gamma$, given by eq.(\ref{2.29}) we see that as long as we go to the UV regime $\gamma\rightarrow 0$ and eq.(\ref{2.62}) reduces to 

\begin{equation}
sS^{L}_{\mathrm{GZ}}=sS^{L}_{\mathrm{FP}}=0\,.
\label{2.63}
\end{equation}

\noindent Therefore, in the UV sector the BRST breaking term vanishes and we recover all the known features of the standard Faddeev-Popov action. On the other hand, in the IR, the breaking cannot be neglected and within the GZ framework, this breaking is manifest. Again, this is another consequence of the fact that Gribov copies play a relevant role at the infrared regime. The role of this soft breaking is subject of investigation up to date and a full understanding of this feature still lacks. Essentially, it is precisely this breaking which will motivate us to a reformulation of the Gribov-Zwanziger scenario in Ch.~\ref{nonpBRSTRGZ}. We mention this is not a particular feature of Landau gauge, but also in the construction of the Gribov-Zwanziger action for the MAG \cite{Capri:2008vk,Capri:2010an}.  

It is not difficult to understand the breaking of BRST symmetry at the qualitative level: The infinitesimal gauge transformation (\ref{a8}) is formally the same as the BRST transformation of the gauge field (\ref{2.59}), just exchanging the infinitesimal gauge parameter by the anti-commuting ghost field $c$. Heuristically, we can identify infinitesimal gauge transformations with BRST transformations for the gauge field (a more detailed discussion on this feature was explored on \cite{Pereira:2013aza,Pereira:2014apa}). As discussed in Ch.~\ref{ch.2}, the Gribov region is free from infinitesimal gauge copies. So, if we choose a configuration which lies in $\Omega$ and perform an infinitesimal gauge/BRST transformation, the resulting configuration necessarily is located outside $\Omega$ and since we are cutting off these configurations, the breaking of BRST symmetry seems to be unavoidable. We underline this is a heuristic argument and should not be taken as an ultimate reasoning.

\section{$\gamma$ is a physical parameter}\label{gammaphys}

For completeness, we expose here an important result explicitly presented in \cite{Dudal:2008sp}. The Gribov parameter is not akin to a gauge parameter, but a truly physical parameter of the theory and therefore can enter physical quantities like gauge invariant correlation functions. The elegant algebraic proof of \cite{Dudal:2008sp} goes as follows: Taking the derivative of (\ref{2.60}) and acting with the BRST operator $s$, we obtain

\begin{equation}
s\frac{\partial S^{L}_{\mathrm{GZ}}}{\partial\gamma^2}=\frac{\Delta_{\gamma^2}}{\gamma^2}=gf^{abc}\int d^dx\left(A^{a}_{\mu}\omega^{bc}_{\mu}-(D^{ad}_{\mu}c^{d})(\varphi+\bar{\varphi})^{bc}_{\mu}\right)\,,
\label{2.64}
\end{equation}

\noindent and since $s^2=0$, the only possibility we have is

\begin{equation}
\frac{\partial S^{L}_{\mathrm{GZ}}}{\partial\gamma^2}\neq s(\mathrm{something})\,,
\label{2.65}
\end{equation}

\noindent \textit{i.e.} is not a BRST exact quantity. Since the $\gamma$-dependent part depends on the auxiliary fields $(\varphi,\bar{\varphi},\omega,\bar{\omega})$ which form a BRST quartet, it cannot belong to the non-trivial cohomology of $s$. So, the only way to have a BRST invariant $\gamma$-dependent terms is to have an exact BRST term (which is not the case (\ref{2.64})). However, let us assume the contrary, namely,

\begin{equation}
S_{\gamma}=s\Sigma_{\gamma}\,,
\label{2.66}
\end{equation}

\noindent with $S_{\gamma}$ the $\gamma$-dependent part of the action. Immediately follows 

\begin{equation}
\frac{\partial S_{\gamma}}{\partial\gamma^2}=s\frac{\partial\Sigma_{\gamma}}{\partial\gamma^2}\,.
\label{2.67}
\end{equation}

\noindent From (\ref{2.67}) we can show $\gamma$ is an unphysical parameter from a direct computation: Consider a gauge invariant quantity $\mathcal{O}$. We can write

\begin{equation}
\frac{\partial}{\partial\gamma^2}\langle \mathcal{O} \rangle = \int \left[\EuScript{D}\mu_{\mathrm{GZ}}\right]\mathcal{O}(s\Sigma_{\gamma})\mathrm{e}^{-S^{L}_{\mathrm{GZ}}}=\int \left[\EuScript{D}\mu_{\mathrm{GZ}}\right]s(\mathcal{O}\Sigma_{\gamma})\mathrm{e}^{-S^{L}_{\mathrm{GZ}}}=\langle s(\mathcal{O}\Sigma_{\gamma}) \rangle = 0\,,
\label{2.68}
\end{equation}

\noindent which leads a complete independence of $\mathcal{O}$ from $\gamma$. Therefore we conclude that the BRST breaking is the mechanism which ensures the physical character of $\gamma$ because (\ref{2.66}) does not hold. The BRST soft breaking is a very important mechanism to implement a physical self-consistent non-perturbative mass parameter in a local and renormalizable way \cite{Dudal:2008sp}.

\section{Alternative solutions: a brief comment}

In this chapter, we have presented an overview of the original approach developed by Gribov and Zwanziger to circumvent the existence of Gribov copies and provide a consistence quantization of Yang-Mills theories. Although implemented in different ways, the essence of their methods lies on the existence of the Gribov region with all of nice properties it enjoys. Nevertheless, it is possible to at least conceive a different solution of the Gribov problem in a manageable way. Recently, two different ``alternative" methods to deal with gauge copies were proposed. 

In \cite{Pereira:2013aza}, the similarity between infinitesimal gauge and BRST transformations was explored. The infinitesimal copies equation is defined by performing an infinitesimal gauge transformation over the gauge-fixing condition and enforcing the gauge condition again. If such equation has solutions, them we automatically have gauge copies. In \cite{Pereira:2013aza}, the copies equation was derived by taking the BRST transformation of the gauge condition instead. Since these transformations are formally the same, this is simply an obvious fact. However, in order to avoid copies, what is required is that the resulting equation does not have any solutions. In other words: We spoil the copies equation and implement this property as a constraint in the theory. Formally, given a gauge condition $F^a=0$, we demand $sF^a\neq 0$ for $F^a=0$. Introducing this constraint in the gauge fixed Yang-Mills action automatically implies a breaking of the BRST symmetry. It turns out that, if we require that the UV regime of Yang-Mills theories is not affected by this constraint, we recover the Gribov-Zwanziger when $F^a$ is the Landau gauge condition. It was also shown that this is also true for the MAG. 

Since this method does not rely on any geometric property of a region akin to the Gribov region, there is no \textit{a priori} requirement of dealing with a Hermitian Faddeev-Popov operator. Hence, the method could be applied to gauges with a non-Hermitian Faddeev-Popov operator. This was exploited in \cite{Pereira:2014apa}, where this method was employed to deal with the Gribov problem in an interpolating gauge among Landau and maximal Abelian gauges, which has a non-Hermitian Faddeev-Popov operator. 

The method, however, is not completely unambiguous and further investigations are necessary to provide a closed framework. In particular, there is an ambiguity on the definition of the gap equation that fixes the Gribov-parameter that is not resolved so far. Also, the prescription of the method is such that the gauge condition $F^a$, albeit arbitrary, should be a function of the gauge field $A$ only. 

In \cite{Serreau:2012cg,Serreau:2013ila,Serreau:2015yna}, a different strategy was adopted. Instead of removing the copies from the quantization, what is proposed is an averaging over them. Hence, a proper weight is assigned to the copies and it is possible to ``control" their contribution to the path integral. This method seems to have the remarkable property of taking care of finite gauge copies also. Again, the breaking of BRST symmetry is observed. Also, a nice agreement with the standard Gribov-Zwanziger framework is obtained.

\chapter{Refinement of the Gribov-Zwanziger action}\label{RGZch}

Until 2007 gauge fixed lattice simulations in the Landau gauge were pointing towards an IR suppressed, positivity violating, vanishing at zero momentum gluon propagator and an enhanced ghost propagator in the IR \cite{Sternbeck:2004xr,Cucchieri:2004mf} for $d=2,3,4$. This \textit{scaling} behavior is precisely the one predicted by the GZ framework as reported in Ch.~\ref{ch.3}. Also, solutions for Dyson-Schwinger equations \cite{Alkofer:2000wg,vonSmekal:1997ohs,Lerche:2002ep,Alkofer:2003jj,Huber:2010ne} and results from the functional renormalization group approach \cite{Pawlowski:2003hq} were in agreement with this scenario.

However, more accurate lattice simulations performed in larger volumes changed this picture: The gluon propagator remained suppressed in the infrared, positivity violating but attaining a \textit{non-vanishing} value at zero momentum and the ghost propagator was \textit{no longer} enhanced in the IR for $d=3,4$ while for $d=2$ the scaling behavior was maintained, \cite{Cucchieri:2007md,Bogolubsky:2007ud,Cucchieri:2008fc,Cucchieri:2008yp,Cucchieri:2008mv,Bogolubsky:2009dc,Bogolubsky:2009qb,Cucchieri:2009zt,Maas:2011se}. This behavior observed by the most recent lattice simulations is known as \textit{decoupling} (massive) type. Together with the new lattice results, different groups working on functional approaches were able to find the so-called decoupling propagator \cite{Aguilar:2008xm,Boucaud:2008ky,Fischer:2008uz}. Therefore, in the context of this thesis, a natural question seems to be how to reconcile the GZ scenario with the decoupling behavior. Although naively incompatible, it was shown in \cite{Dudal:2008sp,Dudal:2007cw} that taking into account further non-perturbative effects besides the elimination of infinitesimal Gribov copies, it is possible to construct a framework which is based on a local and renormalizable action which reproduces the most recent lattice results in a good qualitative agreement. This framework is the so-called \textit{Refined Gribov-Zwanziger} (RGZ) scenario. These further non-perturbative effects comes from the formation of condensates in the theory, an eminent non-perturbative feature. Last decade witnessed a grown interest in the introduction of dimension 2 condensates in standard Yang-Mills theories due to their power to lower the effective potential value with respect to the without-condensate one, \cite{Knecht:2001cc,Verschelde:2001ia,Kondo:2001nq,Dudal:2002aj,Dudal:2002xe,Dudal:2002pq,Dudal:2003vv,Dudal:2003gu,Dudal:2003pw,Dudal:2003dp,Kondo:2003uq,Dudal:2003pe,Dudal:2003tc,Dudal:2003np,Dudal:2003by,Sobreiro:2004us,Dudal:2004rx,Sobreiro:2004yj,Dudal:2005na,Sorella:2005zh,Sorella:2006ax,Capri:2005vw,Lemes:2006aw}. To motivate even more the introduction of such composite operators in the game, we provide a straightforward computation within the GZ framework which explicitly shows the existence of particular examples of dimension two condensates.

\section{Infrared instabilities of the GZ action}

Let us consider the introduction of the following dimension two operators

\begin{equation}
 A^{a}_{\mu}(x)A^{a}_{\mu}(x)\,\,\,\,\mathrm{and}\,\,\,\, \bar{\varphi}^{ab}_{\mu}(x)\varphi^{ab}_{\mu}(x)-\bar{\omega}^{ab}_{\mu}(x)\omega^{ab}_{\mu}(x)\,,
\label{3.1}
\end{equation}

\noindent in the GZ path integral coupled to constant sources $m$ and $J$. So,

\begin{equation}
\mathrm{e}^{-V\mathcal{E}(m,J)}=\int\left[\EuScript{D}\mu_{\mathrm{GZ}}\right]\mathrm{e}^{-S^{L}_{\mathrm{GZ}}-m\int d^dx A^{a}_{\mu}A^{a}_{\mu}+J\int d^dx (\bar{\varphi}^{ab}_{\mu}\varphi^{ab}_{\mu}-\bar{\omega}^{ab}_{\mu}\omega^{ab}_{\mu})}\,.
\label{3.2}
\end{equation}

\noindent From eq.(\ref{3.2}) it is immediate,

\begin{eqnarray}
\langle \bar{\varphi}^{ac}_{\mu}\varphi^{ac}_{\mu}-\bar{\omega}^{ac}_{\mu}\omega^{ac}_{\mu}\rangle &=& - \frac{\partial {\cal E}(m,J)}{\partial J}\Big|_{J=m=0}\,, \nonumber  \\
\langle A^{a}_{\mu}A^{a}_{\mu}\rangle &=& \frac{\partial {\cal E}(m,J)}{\partial m}\Big|_{J=m=0}\,.
\label{3.3}
\end{eqnarray}

\noindent and we can easily compute explicitly eq.(\ref{3.3}) at one-loop order. The result is 

\begin{equation}
{\cal E}(m,J)=\frac{(d-1)(N^2-1)}{2}\int \frac{d^dk}{(2\pi)^d}~\mathrm{ln}\left(k^2+\frac{2\gamma^4g^2N}{k^2+J}+2m\right)-d\gamma^4(N^2-1)\,. 
\label{3.4}
\end{equation}

\noindent Eqs.(\ref{3.3}) and (\ref{3.4}) give thus

\begin{equation}
\langle \bar{\varphi}^{ac}_{\mu}\varphi^{ac}_{\mu}-\bar{\omega}^{ac}_{\mu}\omega^{ac}_{\mu}\rangle = \gamma^4g^2N(N^2-1)(d-1)\int \frac{d^dk}{(2\pi)^d}\frac{1}{k^2}\frac{1}{(k^4+2g^2\gamma^4N)}
\label{3.5}
\end{equation}

\noindent and

\begin{equation}
\langle A^{a}_{\mu}A^{a}_{\mu}\rangle = -\gamma^4(N^2-1)(d-1)\int\frac{d^dk}{(2\pi)^d}\frac{1}{k^2}\frac{2g^2N}{(k^4+2g^2\gamma^4N)}\,.
\label{3.6}
\end{equation}

\noindent In $d=4,3$ these integrals are perfectly convergent and the the correlation functions (\ref{3.5}) and (\ref{3.6}) are non-vanishing (even imposing the gap equation) as long as $\gamma\neq 0$. The case of $d=2$ will not be discussed in the moment, but we will return to it in Ch.~\ref{nonpBRSTRGZ}. Albeit simple, this computation tells us once we introduce the non-perturbative $\gamma$ parameter we have the dynamical formation of dimension 2 condensates which can be probed already at one-loop order. This, of course, is a consequence of the fact that the perturbative series is supplied with the non-pertubative information carried by $\gamma$ and therefore we can capture non-pertubative effects - like condensates formation - from perturbation theory. The non-trivial Gribov background imposed by the horizon gives its own contribution to the formation of condensates. The presence of these condensates reveals the GZ action is plagued by IR instabilities which we can take into account from the beginning. The result is precisely the \textit{Refined Gribov-Zwanziger} action. Although we will not enter in this discussion, the condensates (\ref{3.5}) and (\ref{3.6}) are not the only ones formed. However, they are enough to capture all the relevant features of the addition of dimension 2 condensates in the GZ action and therefore we restrict ourselves to them. For more details on this issue, we refer to \cite{Gracey:2010cg,Dudal:2011gd}.

\section{The \textit{Refined} Gribov-Zwanziger action}

In order to present the framework and the important outcomes of the RGZ action, we begin by \textit{defining} it as\footnote{We omit the vacuum terms for the present purposes.}

\begin{equation}
S^{L}_{\mathrm{RGZ}} = S^{L}_{\mathrm{GZ}} + \frac{{m}^2}{2}\int d^dx~A^{a}_{\mu}A^{a}_{\mu}-{M}^2\int d^dx~(\bar{\varphi}^{ab}_{\mu}\varphi^{ab}_{\mu}-\bar{\omega}^{ab}_{\mu}\omega^{ab}_{\mu})\,,
\label{3.7}
\end{equation}

\noindent with $m$ and $M$ mass parameters. For organizational purposes, we rewrite action (\ref{3.7}) in the following way

\begin{equation}
S^{L}_{\mathrm{RGZ}} = S^{L}_{\mathrm{GZ}} + S_{A^{2}}+S_{\bar{\varphi}\varphi}\,,
\label{3.8}
\end{equation}

\noindent where the notation is established in an obvious way. The construction of (\ref{3.8}) can be established step by step by the introduction of one condensate and then the other, which allows the identification of the effects of each condensate separately. This is how it was originally proceeded in \cite{Dudal:2008sp}. We follow this construction but the very technical details are indicated in the appropriated literature. 

\subsection{The construction of $S^{L}_{\mathrm{GZ}} + S_{A^{2}}$} \label{gzplusa2}

The introduction of $S_{A^{2}}$ to the GZ action was studied in great detail in \cite{Sobreiro:2004us,Dudal:2005na} (see also \cite{Vandersickel:2011zc,Dudal:2008sp,Sobreiro:2007tv}). In this section, we review the most important features of this construction. Before discussing the introduction of the local composite operator $A^2_{\mu}$, we introduce a relevant language for the GZ action. The GZ action can be written as\footnote{Again, vacuum terms are disregarded.}

\begin{equation}
S^{L}_{\mathrm{GZ}} = S_0+\gamma^2\int d^dx~gf^{abc}A^{a}_{\mu}(\varphi+\bar{\varphi})^{bc}_{\mu}\,,
\label{3.9}
\end{equation}

\noindent where

\begin{eqnarray}
S_0&=&S_{\mathrm{YM}}+\int d^dx\left(b^a\partial_{\mu}A^{a}_{\mu}+\bar{c}^{a}\partial_{\mu}D^{ab}_{\mu}c^b\right)-\int d^dx\left(\bar{\varphi}^{ac}_{\mu}\EuScript{M}^{ab}\varphi^{bc}_{\mu}\right.\nonumber\\
&-&\left.\bar{\omega}^{ac}_{\mu}\EuScript{M}^{ab}\omega^{bc}_{\mu}\right)+\int d^dx~gf^{adl}\bar{\omega}^{ac}_{\mu}\partial_{\nu}\left(\varphi^{lc}_{\mu}D^{de}_{\nu}c^e\right)\,. 
\label{3.10}
\end{eqnarray}

\noindent We note from eq.(\ref{3.10}) that the auxiliary field sector enjoys a $U(d(N^2-1))$ symmetry due to the pattern of contraction of the second pair of indices $(a,\mu)$ of each field. It means we can employ a multi-index notation, namely, $i=(a,\mu)$ and rewrite eq.(\ref{3.10}) as 

\begin{eqnarray}
S_0&=&S_{\mathrm{YM}}+\int d^dx\left(b^a\partial_{\mu}A^{a}_{\mu}+\bar{c}^{a}\partial_{\mu}D^{ab}_{\mu}c^b\right)-\int d^dx\left(\bar{\varphi}^{a}_i\EuScript{M}^{ab}\varphi^{b}_i\right.\nonumber\\
&-&\left.\bar{\omega}^{a}_i\EuScript{M}^{ab}\omega^{b}_i\right)+\int d^dx~gf^{adl}\bar{\omega}^{a}_i\partial_{\nu}\left(\varphi^{l}_iD^{de}_{\nu}c^e\right)\,. 
\label{3.11}
\end{eqnarray}

\noindent As explicitly shown in eq.(\ref{2.60}), $S_0$ is invariant under BRST transformations. On the other hand, the second term of (\ref{3.9}) can be expressed as

\begin{equation}
\gamma^2\int d^dx~gf^{abc}A^{a}_{\mu}(\varphi+\bar{\varphi})^{bc}_{\mu}=\gamma^2\int d^dx\left(D^{cb}_{\mu}\varphi^{bc}_{\mu}+D^{cb}_{\mu}\bar{\varphi}^{bc}_{\mu}\right)\,,
\label{3.12}
\end{equation}

\noindent since they differ just by total derivatives. We can treat $gfA\varphi$ and $gfA\bar{\varphi}$ as composite operators that are introduced in the action $S_0$ via the introduction of local external sources $M^{ab}_{\mu\nu}$ and $V^{ab}_{\mu\nu}$ as

\begin{equation}
\int d^dx\left(M^{ai}_{\mu}D^{ab}_{\mu}\varphi^{bi}+V^{ai}_{\mu}D^{ab}_{\mu}\bar{\varphi}^{bi}\right)\,.
\label{3.13}
\end{equation}

\noindent and expression (\ref{3.12}) is obtained by setting the sources to the ``physical values"

\begin{equation}
M^{ab}_{\mu\nu}\big|_{\mathrm{phys}}=V^{ab}_{\mu\nu}\big|_{\mathrm{phys}}=\gamma^2\delta_{\mu\nu}\delta^{ab}\,.
\label{3.14}
\end{equation}

\noindent Following the formalism described in Ap.~\ref{LCO}, we have to introduce sources in such a way to form BRST doublets, namely

\begin{eqnarray}
sU^{ai}_{\mu}&=&M^{ai}_{\mu}\,,\,\,\,\,\,\,\,\,\,\,\,\, sM^{ai}_{\mu}=0\nonumber\\
sV^{ai}_{\mu}&=&N^{ai}_{\mu}\,,\,\,\,\,\,\,\,\,\,\,\,\, sN^{ai}_{\mu}=0\,.
\label{3.142}
\end{eqnarray}

\noindent The resulting BRST invariant action containing sources and composite operators can be expressed as

\begin{eqnarray}
S_{\mathrm{s}}&=&s\int d^dx\left(U^{ai}_{\mu}D^{ab}_{\mu}\varphi^{b}_{i}+V^{ai}_{\mu}D^{ab}_{\mu}\bar{\omega}^{b}_i-U^{ai}_{\mu}V^{ai}_{\mu}\right)\nonumber\\
&=&\int d^dx\left(M^{ai}_{\mu}D^{ab}_{\mu}\varphi^{b}_{i}-U^{ai}_{\mu}D^{ab}_{\mu}\omega^{b}_i-gf^{abc}U^{ai}_{\mu}(D^{cd}_{\mu}c^d)\varphi^{b}_i+N^{ai}_{\mu}D^{ab}_{\mu}\bar{\omega}^{b}_i\right.\nonumber\\
&+&\left.V^{ai}_{\mu}D^{ab}_{\mu}\bar{\varphi}^{b}_i+gf^{abc}V^{ai}_{\mu}(D^{cd}_{\mu}c^d)\bar{\omega}^{b}_i-M^{ai}_{\mu}V^{ai}_{\mu}+U^{ai}_{\mu}N^{ai}_{\mu}\right)\,.
\label{3.143}
\end{eqnarray}

\noindent After the introduction of the the complete set of sources, the physical limit which reconstruct the GZ action is then

\begin{equation}
M^{ab}_{\mu\nu}\big|_{\mathrm{phys}}=V^{ab}_{\mu\nu}\big|_{\mathrm{phys}}=\gamma^2\delta_{\mu\nu}\delta^{ab}\,\,\,\,\mathrm{and}\,\,\,\,U^{ab}_{\mu\nu}\big|_{\mathrm{phys}}=N^{ab}_{\mu\nu}\big|_{\mathrm{phys}}=0\,.
\label{3.144}
\end{equation}

\noindent The action $S_0$ with the inclusion of the composite operators (\ref{3.13}) is renormalizable to all order in perturbation theory as we already mentioned \cite{Zwanziger:1989mf,Dudal:2008sp,Zwanziger:1992qr}. However, what is highly non-trivial and a very interesting property of the GZ action is that the renormalizability is preserved if we add the local composite operator $A^{2}_{\mu}$. For a detailed exposition of the proof within the algebraic renormalization framework \cite{Piguet:1995er} we refer to \cite{Dudal:2005na}. Although we do not present the proof here, we give some important steps towards the proof which are very important in a wild range of applications of introduction of local composite operators in the GZ action. Our aim now is to\footnote{We refer to Ap.~\ref{LCO} for details on the local composite operator technique.} introduce besides the composite operators $gfA\varphi$ and $gfA\bar{\varphi}$, the LCO $\mathcal{O}_{A^2}=A^{a}_{\mu}A^{a}_{\mu}$ to $S_0$. This operator is introduced in the following way

\begin{eqnarray}
S'_{A^{2}}&=&s\int d^dx\left(\frac{1}{2}\eta A^{a}_{\mu}A^{a}_{\mu}-\frac{1}{2}\zeta\tau\eta\right)\nonumber\\
&=&\int d^dx\left(\frac{1}{2}\tau A^{a}_{\mu}A^{a}_{\mu}+\eta A^{a}_{\mu}\partial_{\mu}c^{a}-\frac{1}{2}\zeta\tau^2\right)\,,
\label{3.145}
\end{eqnarray}

\noindent with

\begin{equation}
s\eta=\tau\,\,\,\,\, \mathrm{and}\,\,\,\,\, s\tau=0\,,
\label{3.146}
\end{equation}

\noindent local sources. As explained in Ap.~\ref{LCO}, the introduction of quadratic terms in the sources (see (\ref{3.143}) and (\ref{3.145})) is allowed by power counting and are important to absorb novel divergences. Remarkably, although a parameter $\zeta$ is introduced for the quadratic terms of the sources associated with the $\mathcal{O}_{A^2}$ LCO, no extra parameter is needed for the sources associated with the $gfA\varphi$ and $gfA\bar{\varphi}$. This is a non-trivial feature of the GZ action and should be verified through the complete renormalizability analysis \cite{Sobreiro:2004us,Dudal:2005na}. The action $S$ defined by

\begin{equation}
S=S_s+S'_{A^2}\,,
\label{3.147}
\end{equation}

\noindent enjoys BRST invariance due to the presence of the external sources. Therefore we have a more general action which contains the GZ action with a LCO $\mathcal{O}_{A^2}$ as a particular case. We note that at the physical level (\ref{3.144}), part of the action $S$ reduces to

\begin{equation}
S_{\gamma}=\int d^dx\left(\gamma^2gf^{abc}A^{a}_{\mu}\varphi^{bc}_{\mu}+\gamma^2gf^{abc}A^{a}_{\mu}\bar{\varphi}^{bc}_{\mu}-d\gamma^4(N^2-1)\right)\,,
\label{3.148}
\end{equation}

\noindent where the quadratic term on the sources generates the vacuum term $d\gamma^4(N^2-1)$. Our aim is to study the condensation of $\mathcal{O}_{A^2}$, therefore at the physical level $\eta=0$. Hence at the physical level (\ref{3.147}) reduces to

\begin{equation}
S=S_0+S_{\gamma}+\int d^dx\left(\frac{\tau}{2}A^{a}_{\mu}A^{a}_{\mu}-\frac{\zeta}{2}\tau^2\right)\,.
\label{3.149}
\end{equation}

\noindent Employing the procedure described in Sect.~\ref{HSFields} a Hubbard-Stratonovich field $\sigma$ associated with $\mathcal{O}_{A^2}$ is introduced via the unity

\begin{equation}
1=\EuScript{N}\int \left[\EuScript{D}\sigma\right]\exp\left[-\frac{1}{2\zeta}\int d^dx\left(\frac{\sigma}{g}+\frac{1}{2}A^{a}_{\mu}A^{a}_{\mu}-\zeta\tau\right)^2\right]\,.
\label{3.1491}
\end{equation}

\noindent This procedure removes the $\tau^2$ term and the resulting action is 

\begin{equation}
S=S_0+S_\gamma+S_\sigma-\int d^dx~\tau\frac{\sigma}{g}\,,
\label{3.1492}
\end{equation}

\noindent with
\begin{equation}
S_\sigma=\frac{\sigma^2}{2g^2\zeta}+\frac{\sigma}{2g\zeta}A^{a}_{\mu}A^{a}_{\mu}+\frac{1}{8\zeta}(A^{a}_{\mu}A^{a}_{\mu})^2\,.
\label{3.1493}
\end{equation}

\noindent As dicussed in Sect.~\ref{HSFields},

\begin{equation}
\langle \mathcal{O}_{A^2}\rangle=-\frac{1}{g}\langle\sigma\rangle\,,
\label{3.1494}
\end{equation}

\noindent and a nonvanishing value for $\langle\sigma\rangle$ implies a non-trivial value for the vacuum expectation value for the LCO $\mathcal{O}_{A^2}$. To see if a non-trivial value of $\langle\sigma\rangle$ is energetically favored, we must solve the following equation

\begin{equation}
\frac{\partial\Gamma}{\partial \sigma}=0\,,
\label{3.1495}
\end{equation}

\noindent for constant field configurations $\sigma$ with $\Gamma$ the quantum action defined in Ap.~\ref{EA}. On the other hand the Gribov parameter $\gamma$ is fixed by the horizon condition,

\begin{equation}
\frac{\partial\Gamma}{\partial\gamma^2}=0\,.
\label{3.1496}
\end{equation}

\noindent This equation admits the solution $\gamma=0$ and we should stress, however, this is a consequence of the localization procedure through the auxiliary fields and should be discarded. This solution corresponds to not restrict the path integral to $\Omega$. The discussion of the solutions for the gap equations (\ref{3.1495}) and (\ref{3.1496}) are carefully analyzed at one-loop order in \cite{Dudal:2005na,Sobreiro:2007tv}. The results obtained in \cite{Dudal:2005na} seems to require a higher order computation to obtain a definite picture for the vacuum energy. Taking into account the inclusion of the LCO $\mathcal{O}_{A^2}$ to the GZ action we have the following effective action which takes into account the restriction of the path integral domain to $\Omega$ and the non-perturbative effects of the condensation of the operator $\mathcal{O}_{A^2}$,
 
\begin{eqnarray}
S_0+S_{A^2}&=& S_{\mathrm{YM}}+\int d^dx\left(b^a\partial_{\mu}A^{a}_{\mu}+\bar{c}^{a}\partial_{\mu}D^{ab}_{\mu}c^b\right)-\int d^dx\left(\bar{\varphi}^{ac}_{\mu}\EuScript{M}^{ab}\varphi^{bc}_{\mu}\right.\nonumber\\
&-&\left.\bar{\omega}^{ac}_{\mu}\EuScript{M}^{ab}\omega^{bc}_{\mu}\right)+\int d^dx~gf^{adl}\bar{\omega}^{ac}_{\mu}\partial_{\nu}\left(\varphi^{lc}_{\mu}D^{de}_{\nu}c^e\right)+\gamma^2\int d^dx~gf^{abc}A^{a}_{\mu}(\varphi+\bar{\varphi})^{bc}_{\mu}\nonumber\\
&+&\frac{m^2}{2}\int d^dx~A^{a}_{\mu}A^{a}_{\mu}\,.
\label{3.15}
\end{eqnarray}

\noindent From eq.(\ref{3.15}) the expression for the tree-level gluon propagator is

\begin{equation}
\langle A^{a}_{\mu}(p)A^{b}_{\nu}(-p)\rangle_{S_0+S_{A^2}}=\delta^{ab}\frac{p^2}{p^2(p^2+m^2)+2g^2\gamma^4N}\mathcal{P}_{\mu\nu}\,,
\label{3.16}
\end{equation}

\noindent while for the ghost propagator at one-loop order,

\begin{equation}
\lim_{p\rightarrow0}\frac{\delta^{ab}}{N^2-1}\langle c^a(p)c^b(-p)\rangle_{S_0+S_{A^2}} \approx \frac{4}{3Ng^2\mathcal{J}p^4}
\label{3.17}
\end{equation}

\noindent with

\begin{equation}
\mathcal{J}=\int \frac{d^dk}{(2\pi)^d}\frac{1}{k^2(k^4+m^2k^2+2g^2N\gamma^4)}\,,
\label{3.18}
\end{equation}

\noindent which is real and finite for $d=3,4$. We see from eq.(\ref{3.16}) and (\ref{3.17}) that the introduction of $\mathcal{O}_{A^2}$ to the Gribov-Zwanziger action does not alter qualitatively the propagators obtained from the standard GZ action (see (\ref{2.30}) and (\ref{2.42})). Therefore, the gluon propagator is still supressed at the IR vanishing at zero momentum and positivity violating \cite{Dudal:2005na}. Also the ghost propagator at one-loop order maintains its enhancement at $p\approx 0$. This implies that taking into account the condensation of the $\mathcal{O}_{A^{2}}$ operator is not enough to capture the \textit{decoupling} behavior reported by recent lattice simulations.

\subsection{The birth of $S^{L}_{\mathrm{RGZ}}$}

As discussed in Subsect.~\ref{gzplusa2}, the inclusion of dynamical effects associated with the LCO $\mathcal{O}_{A^2}$ is not enough to generate a decoupling gluon propagator. However, we should remind the straightforward computation (\ref{3.5}), where a non-trivial expectation value for the LCO

\begin{equation}
\mathcal{O}_{\varphi\omega}=\bar{\varphi}^{ab}_{\mu}(x)\varphi^{ab}_{\mu}(x)-\bar{\omega}^{ab}_{\mu}(x)\omega^{ab}_{\mu}(x)\,,
\label{3.19}
\end{equation}

\noindent was reported. Although these auxiliary fields were introduced with the goal of localizing the horizon function, we cannot simply forget they will develop their own quantum dynamics and eventually condensate. Therefore, introducing the operator $\mathcal{O}_{\varphi\omega}$ under the LCO formalism and studying its role for the effective potential is a very important task. In this case, we have to introduce the following term

\begin{equation}
S'_{\varphi\omega}=\int d^dx\left(s\left(-J\bar{\omega}^{a}_{i}\varphi^{a}_{i}\right)+\rho J\tau\right)=\int d^dx\left(-J\left(\bar{\varphi}^{a}_{i}\varphi^{a}_{i}-\bar{\omega}^{a}_{i}\omega^{a}_{i}\right)+\rho J\tau\right)\,,
\label{3.20}
\end{equation}

\noindent with $\rho$ a parameter and $J$ an external source invariant under BRST transformations,

\begin{equation}
sJ=0\,.
\label{3.21}
\end{equation}

\noindent There are several details on the introduction of $\mathcal{O}_{\varphi\omega}$ that will be omitted in this thesis and we refer to \cite{Vandersickel:2011zc,Dudal:2008sp,Dudal:2007cw}. However let us emphasize some important features:

\begin{itemize}
\item The introduction of $\mathcal{O}_{\varphi\omega}$ does not spoil the renormalizability of the $S^{L}_{\mathrm{GZ}} + S_{A^{2}}$ action.

\item In the LCO formalism we should introduce a term like $\beta J^2$ to remove divergences proportional to $J^2$ which are typical of this procedure. However, for this particular LCO $\mathcal{O}_{\varphi\omega}$ these divergences do not show up in the correlation functions and we can simply ignore this term. 

\item The LCO $\mathcal{O}_{\varphi\omega}$ is a BRST exact object. Hence, in a BRST invariant theory it is trivial that $\langle \mathcal{O}_{\varphi\omega} \rangle=0$. However, we should remind ourselves the restriction of the path integral to the Gribov region $\Omega$ breaks BRST softly in such a way that

\begin{equation}
\langle \mathcal{O}_{\varphi\omega} \rangle \neq 0
\label{3.22}
\end{equation}

is allowed. 
\end{itemize}

As done before, we introduce a term to $S^{L}_{\mathrm{GZ}} + S_{A^{2}}$ to take into account the condensation of $\mathcal{O}_{\varphi\omega}$ from the beginning, giving rise to the \textit{Refined Gribov-Zwanziger action},

\begin{equation}
S^{L}_{\mathrm{RGZ}} = S^{L}_{\mathrm{GZ}} + \frac{{m}^2}{2}\int d^dx~A^{a}_{\mu}A^{a}_{\mu}-{M}^2\int d^dx~(\bar{\varphi}^{ab}_{\mu}\varphi^{ab}_{\mu}-\bar{\omega}^{ab}_{\mu}\omega^{ab}_{\mu})\,.
\label{3.23}
\end{equation}

\noindent It is clear from the quadratic coupling $A(\bar{\varphi}+\varphi)$ that the introduction of a mass like term $\bar{\varphi}\varphi$ will affect the gluon propagator. Clearly, if the gluon propagator is modified then the resulting ghost propagator should change due to the gluon-ghost vertex. This poses a fundamental question: Does the introduction of these LCO's in the action keeps the theory inside the Gribov region $\Omega$ since the ghost propagator is affected ? The answer is no and to fix this problem, a vacuum term should be added to the action. We refer to \cite{Dudal:2008sp} for details on this and we will return to this point soon. The essence of this modification relies on the fact that a vacuum term will contribute to the gap equation which is directly used in the computation of the ghost two-point function. Before going through this point, from eq.(\ref{3.23}) it is easy to compute the gluon propagator (at tree-level) since the novel vacuum term does not affect it. The expression is

\begin{eqnarray}
\langle A^{a}_{\mu}(p)A^{b}_{\nu}(-p)\rangle_{\mathrm{RGZ}} &=& \delta^{ab}\frac{p^2+M^2}{(p^2+m^2)(p^2+M^2)+2g^2N\gamma^4}\left(\delta_{\mu\nu}-\frac{p_{\mu}p_{\nu}}{p^2}\right)\nonumber\\
&\equiv& \delta^{ab}\mathcal{D}(p^2)\left(\delta_{\mu\nu}-\frac{p_{\mu}p_{\nu}}{p^2}\right)\,.
\label{3.24}
\end{eqnarray}

\noindent Remarkably expression (\ref{3.24}) displays infrared suppression and $\mathcal{D}(0)$ attains a non-vanishing value a fact, as mentioned before, reported by the most recent lattice data. A qualitative representation of (\ref{3.24}) is displayed on Fig.~\ref{fig:rgzprop}.

\begin{figure}
	\centering
		\includegraphics[width=.70\textwidth]{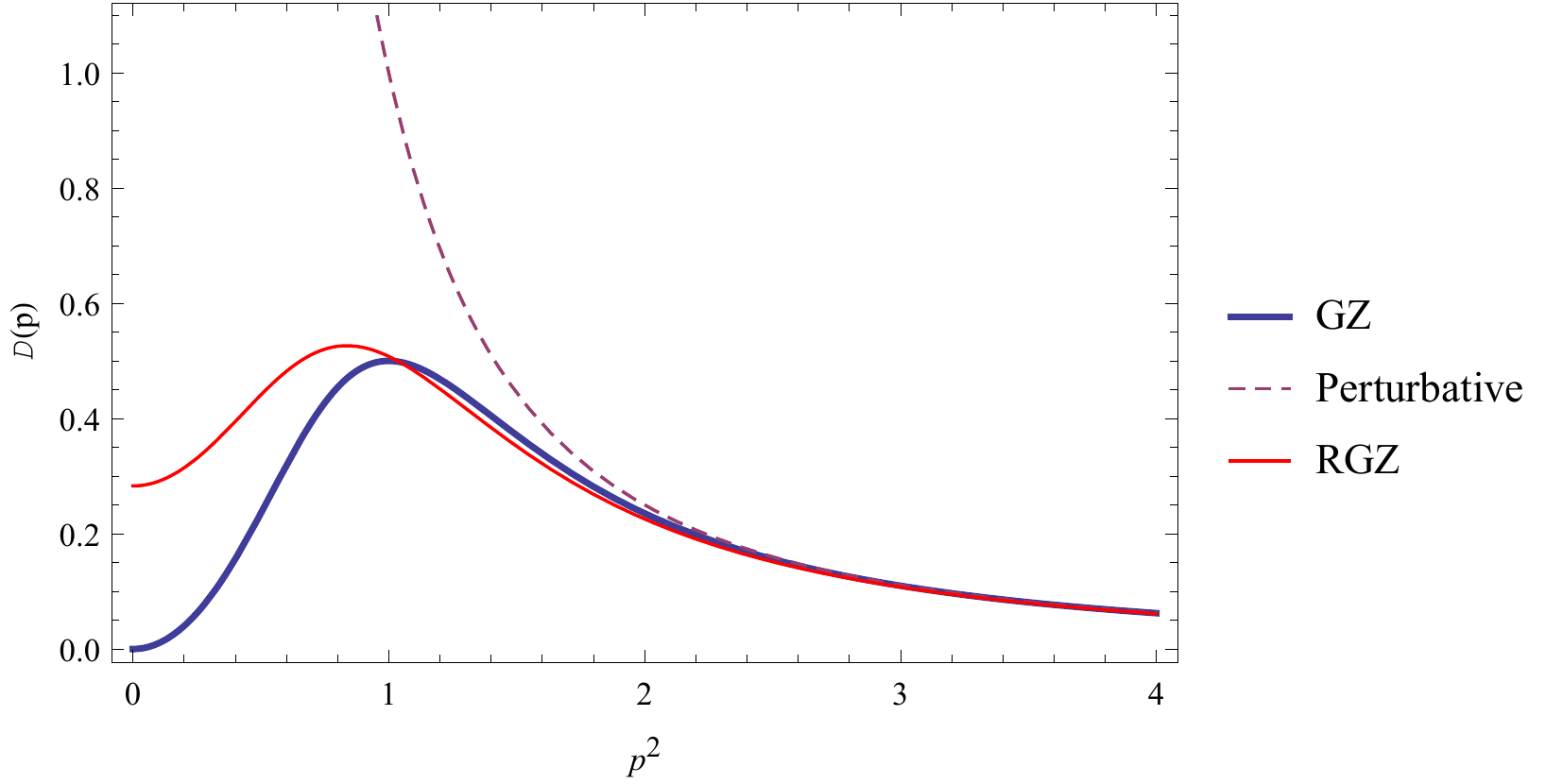}
	\caption{Qualitative representation of the gluon form factor in the GZ, RGZ and perturbative actions.}
	\label{fig:rgzprop}
\end{figure}

Therefore the RGZ framework enables an agreement with lattice simulations at least qualitatively and a dynamical justification for the finiteness of $\mathcal{D}(0)$. We must emphasize the fact that even for $m^2=0$ the lattice behavior is still qualitatively reproduced by eq.(\ref{3.24}). Therefore, the non-trivial finite value for $\mathcal{D}(0)$ comes from the condensation of the localizing fields $\mathcal{O}_{\varphi\omega}$. Now, let us return to the issue of whether the refinement of the Gribov-Zwanziger action is consistent with the horizon (or no-pole) condition without the addition of novel vacuum terms. To begin with, we add a vacuum term parametrized in the following way 

\begin{equation}
S_{\upsilon}=2\frac{d(N^2-1)}{\sqrt{2g^2N}}\int d^dx~\upsilon\gamma^2M^2\,,
\label{3.25}
\end{equation}

\noindent where the specific choice of the coefficient is for computational convenience and $\upsilon$ is a new parameter to be fixed. Now, let us proceed to the computation of the ghost propagator. We use eq.(\ref{2.31}) and the first line of eq.(\ref{2.32}) to write\footnote{For simplicity we set $m^2=0$ since it does not affect the qualitative behavior of the gluon propagator.}

\begin{eqnarray}
\langle\sigma(p^2)\rangle&=&\frac{N}{N^2-1}\frac{g^2}{p^2}\int \frac{d^dq}{(2\pi)^d}\frac{(p-q)_{\mu}p_{\nu}}{(p-q)^2}\langle A^{a}_{\mu}(-q)A^{a}_{\nu}(q)\rangle\nonumber\\
&=&Ng^2\frac{p_\mu p_\nu}{p^2}\int \frac{d^dq}{(2\pi)^d}\frac{1}{(p-q)^2}\frac{q^2+M^2}{q^2(q^2+M^2)+2g^2N\gamma^4}\left(\delta_{\mu\nu}-\frac{q_{\mu}q_{\nu}}{q^2}\right)\nonumber\\
\label{3.26}
\end{eqnarray}

\noindent where we have used expression (\ref{3.24}) with $m^2=0$ and the transversality of the gluon propagator. We expand expression (\ref{3.26}) around $p^2\approx 0$ and define $\lambda^4=2g^2N\gamma^4$, we obtain

\begin{equation}
\langle\sigma(p^2\approx 0)\rangle=Ng^2\frac{p_\mu p_\nu}{p^2}\int \frac{d^dq}{(2\pi)^d}\frac{1}{q^2}\frac{q^2+M^2}{q^2(q^2+M^2)+\lambda^4}\left(\delta_{\mu\nu}-\frac{q_{\mu}q_{\nu}}{q^2}\right)\,.
\label{3.27}
\end{equation}

\noindent Making use of 

\begin{equation}
\int d^dq~f(q^2)q_\mu q_\nu=\frac{1}{d}\delta_{\mu\nu}\int d^dq~f(q^2)q^2\,,
\label{3.28}
\end{equation}

\noindent we end up with

\begin{equation}
\langle\sigma(p^2\approx 0)\rangle=Ng^2\left(\frac{d-1}{d}\right)\int \frac{d^dq}{(2\pi)^d}\frac{1}{q^2}\frac{q^2+M^2}{q^2(q^2+M^2)+\lambda^4}\,.
\label{3.29}
\end{equation}

\noindent In the computation of the ghost propagator for the GZ action we make use of the gap equation. To derive the gap equation in the RGZ scenario, we use the one-loop vacuum energy

\begin{equation}
\mathcal{E}^{(1)}_v=\frac{(d-1)(N^2-1)}{2}\int \frac{d^dq}{(2\pi)^d}~\mathrm{ln}\left(q^2+\frac{\lambda^4}{q^2+M^2}\right)-d\gamma^4(N^2-1)+2\frac{d(N^2-1)}{\sqrt{2g^2N}}\upsilon\gamma^2 M^2\,,
\label{3.30}
\end{equation}

\noindent and we notice the presence of $\upsilon$ in the expression. It is convenient to deal with the following normalization,

\begin{equation}
\tilde{\mathcal{E}}^{(1)}_v\equiv\frac{\mathcal{E}^{(1)}_v}{N^2-1}\frac{2g^2N}{d}=-\lambda^4+2\lambda^2 M^2\upsilon+g^2N\frac{d-1}{d}\int \frac{d^dq}{(2\pi)^d}~\mathrm{ln}\left(q^2+\frac{\lambda^4}{q^2+M^2}\right)\,.
\label{3.31}
\end{equation} 

\noindent The one-loop gap equation is obtained through the condition

\begin{equation}
\frac{\partial\tilde{\mathcal{E}}^{(1)}_v}{\partial\lambda^2}=0\,,
\label{3.32}
\end{equation}

\noindent which is equivalent to

\begin{equation}
-1+\frac{M^2}{\lambda^2}\upsilon+g^2N\frac{d-1}{d}\int \frac{d^dq}{(2\pi)^d}\frac{1}{q^2(q^2+M^2)+\lambda^4}=0\,.
\label{3.33}
\end{equation}

\noindent We emphasize that in the derivation of eq.(\ref{3.33}) we tacitly removed the $\lambda=0$ solution which emerges as an artificial solution from the localization procedure and corresponds to not restrict the path integral domain to the Gribov region $\Omega$. The eq.(\ref{3.33}) is used to fix the Gribov parameter $\gamma$ (or, equivalently, $\lambda$). However, it contains two extra parameters which are in principle free: $\upsilon$ and $M$. The parameter $M$ can be determined dynamically as $m$ in the previous section. This is explored in details in \cite{Dudal:2008sp} and we shall devote few words on this later on. To fix the parameter $\upsilon$, we introduce an important boundary condition. It is the following: The function $\sigma(p^2)$ in the standard GZ action satisfies $\sigma(0)=1$. Therefore, when we introduce a mass $M$, we expect the following property

\begin{equation}
\lim_{M^2\rightarrow 0}\langle\sigma(0,M^2)\rangle=1\,.
\label{3.34}
\end{equation}

\noindent A reasonable assumption which is physically expected is that the function $\sigma$ is stationary around $M^2=0$, namely, for infinitesimal values of $M^2$ we do not expect variations of $\langle\sigma(0,M^2)\rangle$ of first order. Hence, 

\begin{equation}
\frac{\partial\langle\sigma (0,M^2)\rangle}{\partial M^2}\Big|_{M^2=0}=0\,.
\label{3.35}
\end{equation}

\noindent Eq.(\ref{3.35}) fixes the parameter $\upsilon$. To see this, we rewrite eq.(\ref{3.29}) as

\begin{eqnarray}
\langle\sigma(0)\rangle&=&Ng^2\frac{d-1}{d}\left(\int \frac{d^dq}{(2\pi)^d}\frac{1}{q^4+q^2M^2+\lambda^4(M^2)}\right.\nonumber\\
&+&\left.M^2\int \frac{d^dq}{(2\pi)^d}\frac{1}{q^2}\frac{1}{q^4+q^2M^2+\lambda^4(M^2)}\right)\,,
\label{3.36}
\end{eqnarray}

\noindent and imposing condition (\ref{3.35}) we immediately get

\begin{equation}
\int\frac{d^dq}{(2\pi)^d}\frac{q^2}{(q^4+\lambda^4(0))^2}=\int\frac{d^dq}{(2\pi)^d}\frac{1}{q^2}\frac{1}{(q^4+\lambda^4(0))}\,.
\label{3.37}
\end{equation}

\noindent Now, we take the derivative with respect to $M^2$ of eq.(\ref{3.33}) and impose $M^2=0$. The resulting expression is

\begin{equation}
\frac{\upsilon}{\lambda^2(0)}-g^2N\frac{d-1}{d}\int \frac{d^dq}{(2\pi)^d}\frac{q^2}{(q^4+\lambda^4(0))^2}=0\,,
\label{3.38}
\end{equation}

\noindent and making use of (\ref{3.37}), we end up with

\begin{equation}
\upsilon=\lambda^2(0)g^2N\frac{3}{4}\int \frac{d^4q}{(2\pi)^4}\frac{1}{q^2}\frac{1}{q^4+\lambda^{4}(0)}=\frac{3g^2N}{128\pi}\,,
\label{3.39}
\end{equation}

\noindent where we set $d=4$ for concreteness. In this way, we fixed the parameter $\upsilon$. Substituting eq.(\ref{3.33}) in eq.(\ref{3.36}), we obtain

\begin{equation}
\langle\sigma(0)\rangle=1-\frac{M^2}{\lambda^2}\upsilon+M^2Ng^2\frac{d-1}{d}\int \frac{d^dq}{(2\pi)^d}\frac{1}{q^2}\frac{1}{q^4+q^2M^2+\lambda^4(M^2)}\,,
\label{3.40}
\end{equation}

\noindent with $\upsilon$ fixed by eq.(\ref{3.39}). It is very clear from eq.(\ref{3.40}) that for $M^2=0$ we recover the standard GZ result, namely $\langle\sigma(0)\rangle=1$. Also, if we didn't include the novel vacuum term $\upsilon$, it is clear that  $\langle\sigma(0)\rangle$ would violate the no-pole condition \textit{i.e.} $\langle\sigma(0)\rangle>1$. In this case, the theory is pushed outside the Gribov horizon and this contradicts the initial hypothesis of the GZ construction. Therefore, the presence of an extra vacuum term is essential to balance the right-hand side of eq.(\ref{3.40}) in such a way that $\langle\sigma(0)\rangle<1$. Of course, in the form (\ref{3.40}) it is not possible to guarantee $\langle\sigma(0)\rangle<1$ in an obvious way, but only to see this is a possibility. To see it concretely, we plug eq.(\ref{3.39}) in (\ref{3.40}) and after few manipulations we end up with

\begin{equation}
\langle\sigma(p^2\approx 0)\rangle=1-\frac{3x^2}{4}g^2N\int \frac{d^4q}{(2\pi)^4}\frac{1}{(1+q^4)(q^4+xq^2+1)}\,,
\label{3.41}
\end{equation}

\noindent with\footnote{The mass parameter $M^2$ must be non-negative to avoid tachyonic modes in the $\bar{\omega}\omega$ sector.} $x=M^2/\lambda^2(M^2)\geq 0$. From eq.(\ref{3.41}), we see there is no dependence on the momentum $p^2$ (we remember this is not the case in eq.(\ref{2.41}) for the standard GZ construction). The immediate consequence is that the ghost propagator at one-loop order is not enhanced differently from the GZ framework and this is in agreement with the most recent lattice at the qualitative level. 

Before going ahead with further features of the RGZ scenario, we believe it is worth to summarize what we have so far for the benefit of the reader: The GZ action was constructed to take into account the presence of (infinitesimal) Gribov copies in the quantization of Yang-Mills theories. This action is originally written in a non-local fashion due to the presence of the horizon function and contains a mass parameter, the Gribov parameter, which is fixed by a gap equation. On the other hand, auxiliary fields can be introduced in order to localize the horizon function. The resulting GZ action is local and renormalizable to all orders in perturbation theory. Nevertheless, in this chapter we have shown that even at the perturbative level, we can probe IR instabilities of this action. They are associated with the formation of dimension two condensates, which as shown at one-loop order are proportional to the Gribov parameter, see eq.(\ref{3.5}) and eq.(\ref{3.6}). Therefore, two mass terms were introduced in the original GZ action to take into account those instabilities. The first mass term $m^2$ associated with the LCO $\mathcal{O}_{A^2}$ does not change the qualitative behavior of the gluon and ghost propagators derived from the GZ action: An infrared suppressed positivity violating gluon propagator in the IR which attains a vanishing value at zero momentum and an enhanced ghost propagator in the deep IR. However, the mass $M^2$ associated with the LCO $\mathcal{O}_{\varphi\omega}$ modifies the propagators as just showed and are in very good qualitative agreement with lattice results. The introduction of such condensates in the GZ action is the so-called \textit{refinement}. As also discussed, the refinement obliged us to introduce a novel vacuum term to keep the theory inside the Gribov region $\Omega$. In summary, the RGZ action contains the following parameters apart from the ordinary Yang-Mills coupling $g$: The Gribov parameter $\gamma$, the mass $m^2$, the mass $M^2$ and the parameter $\upsilon$. The Gribov parameter is fixed by the horizon condition while the new parameter $\upsilon$ is fixed by the boundary condition (\ref{3.35}) and eq.(\ref{3.38}). The mass parameters associated with the condensates should be dynamically fixed and this is a very important (and not so easy) task. As described in Ap.~\ref{LCO} with particular application in Subsect.~\ref{gzplusa2}, we can construct a consistent effective action taking into account the LCO's $\mathcal{O}_{A^2}$ and $\mathcal{O}_{\varphi\omega}$. The non-trivial values of $\langle \mathcal{O}_{A^2}\rangle$ and $\langle\mathcal{O}_{\varphi\omega}\rangle$ are obtained by demanding the minimization of the effective action with respect to the parameters $m^2$ and $M^2$. In summary, for each mass parameter introduced in the RGZ we have a gap equation, namely

\begin{equation}
\frac{\partial\Gamma}{\partial\gamma^2}=0\,,\,\,\,\,\,\,\,\,\frac{\partial\Gamma}{\partial m^2}=0\,,\,\,\,\,\,\,\,\,\frac{\partial\Gamma}{\partial M^2}=0\,.
\label{3.42}
\end{equation}

\noindent We should emphasize two points: \textit{(i)} The gap equations (\ref{3.42}) fix the the mass parameters in a dynamical way. Therefore, they are not free parameters and this implies no free extra parameters are introduced in the theory; \textit{(ii)} Solving the gap equations (\ref{3.42}) might not be easy (and in fact is not), but they guarantee we have a consistent way to fix those parameters. Unfortunately, solving these equations at leading order is not enough in general, making this task quite involved. Therefore, playing with fits with lattice data or Schwinger-Dyson results is an efficient way to fix those parameters in practice. Since this sort of analysis will not be extended in the second part of this thesis, we simply refer to \cite{Vandersickel:2011zc,Dudal:2008sp,Dudal:2005na,Dudal:2011gd} for details on the solutions of (\ref{3.42}) and let the take away message that we have a consistent way of fixing the mass parameters and no free parameters are introduced. 

\section{Features of the RGZ action}

The RGZ action defined by eq.(\ref{3.7}) displays many interesting features. In this section we point out a subset of them without full details, but giving the appropriate references. We emphasize, however, the study of different features of the RGZ action is an active research topic.  

\subsection{Gluon propagator: positivity violation} \label{positivityviol}

The tree-level gluon propagator obtained from the RGZ action, given by eq.(\ref{3.24}), as already described previously, attains a finite form factor at zero momentum. Also, a crucial property which is also reported from lattice simulations, is the positivity violation. To show this, we write the gluon propagator form factor in its K\"all\'en-Lehmann spectral representation as

\begin{equation}
\mathcal{D}(p^2)=\int^{\infty}_{0}dM^2\frac{\rho(M^2)}{p^2+M^2}\,,
\label{3.43}
\end{equation}

\noindent where $\rho(M^2)$ is the spectral density function. To associate the gluon field with a stable particle, the spectral density function \textit{must} be positive. If $\rho(M^2)<0$ for a given $M^2$, we say $\mathcal{D}(p^2)$ is positivity violating and we cannot associate the gluon propagator with a physical particle in the spectrum of the theory. This is advocated as a signal of confinement, see \cite{Cucchieri:2004mf,vonSmekal:1997ohs}. A practical procedure to check if $\rho(M^2)$ is not positive for all values of $M^2$ is to use the temporal correlator $\mathcal{C}(t)$ given by

\begin{eqnarray}
\mathcal{C}(t)&=&\int^{\infty}_{0}dM~\rho(M^2)\,\mathrm{e}^{-Mt}=\frac{1}{2\pi}\int^{\infty}_{-\infty}dp~\mathcal{D}(p^2)\,\mathrm{e}^{-ipt}\nonumber\\
&=&\frac{1}{2\pi}\int^{\infty}_{-\infty}dp~\frac{p^2+M^2}{(p^2+m^2)(p^2+M^2)+2g^2N\gamma^4}\,\mathrm{e}^{-ipt}\,.
\label{3.44}
\end{eqnarray}

\noindent The temporal correlator is useful to test if the spectral density is not always positive \textit{i.e.} if $\rho(M^2)$ is positive, then $\mathcal{C}(t)$ is necessarily positive. The converse is not necessarily true. However, if $\mathcal{C}(t)$ is negative for a given value of $t$, then $\rho(M^2)$ cannot be positive for all $M^2$. It was proved in \cite{Vandersickel:2011zc,Vandersickel:2012tz,Zwanziger:1989mf,Dudal:2005na} that as long as the Gribov parameter $\gamma$ is different from zero, the gluon propagator is positivity violating. The case with $M^2=0$ obviously violates positivity since $\mathcal{D}(0)=0$ and, therefore, $\rho(M^2)$ cannot be positive for all $M$. For $\mathcal{D}(0)\neq 0$, this conclusion is not immediate, but it is possible to show the positivity violation also, \cite{Dudal:2008sp}. The conclusion we can reach with this analysis is that irrespective of the inclusion of further non-perturbative effects as condensates, the restriction to the Gribov region $\Omega$ implies a positivity violating gluon propagator. This is a non-trivial feature, but we should emphasize at this point that this analysis was carried out at tree-level. In principle, we should prove that higher order corrections preserve the positivity violating behavior. 

\subsection{Strong coupling constant}

It is possible to provide a renormalization group invariant definition of an effective strong coupling constant $g^2_R(p^2)$ using the gluon and ghost form factors \cite{Alkofer:2000wg,vonSmekal:1997ohs,Bloch:2003sk},

\begin{equation}
g^2_R(p^2)=g^2(\bar{\mu})\tilde{\mathcal{D}}(p^2,\bar{\mu}^2)\tilde{\mathcal{G}}^2(p^2,\bar{\mu}^2)\,,
\label{3.45}
\end{equation}

\noindent where 

\begin{equation}
\tilde{\mathcal{D}}(p^2)=p^2\mathcal{D}(p^2)\,,\,\,\,\,\,\mathrm{and}\,\,\,\,\,\,\tilde{\mathcal{G}}(p^2)=p^2\mathcal{G}(p^2)\,,
\label{3.46}
\end{equation}

\noindent are the gluon and ghost form factors\footnote{Note that we use ``form factor" both for $\mathcal{D}$ or $p^2\mathcal{D}$. However, we emphasize which expression we are using to avoid confusion.} respectively. Before the paradigm change in 2007, where lattice simulations reported a finite value for $\mathcal{D}(0)$ and a non enhanced ghost propagator, many evidences from Schwinger-Dyson results, lattice simulations and the GZ approach pointed towards the existence of a non-vanishing infrared fixed point for $g^{2}_R(p^2)$. In particular, these studies proposed a power law behavior for $\tilde{\mathcal{D}}(p^2)$ and $\tilde{\mathcal{G}}(p^2)$ in the IR, namely

\begin{equation}
\lim_{p\rightarrow 0}\tilde{\mathcal{D}}(p^2)\propto (p^2)^\theta\,\,\,\,\,\,\,\,\,\,\,\,\mathrm{and}\,\,\,\,\,\,\,\,\,\,\,\,\lim_{p\rightarrow 0}\tilde{\mathcal{G}}(p^2)\propto (p^2)^\omega\,,
\label{3.47}
\end{equation}

\noindent with the constraint

\begin{equation}
\theta+2\omega=0\,,
\label{3.48}
\end{equation}

\noindent which implies a non-vanishing IR fixed point. If we set $M^2=0$ in eq.(\ref{3.24}), conditions (\ref{3.47}) and (\ref{3.48}) are satisfied with $\theta=2$ and $\omega=-1$. We remember the reader that for $M^2=0$, the gluon and ghost propagator are of scaling type. On the other hand, in the RGZ scenario, namely, $M^2\neq 0$ we have a massive/decoupling gluon propagator and a ghost propagator which is not enhanced anymore. In this situation, 

\begin{equation}
\lim_{p^2\rightarrow 0}\tilde{\mathcal{D}}(p^2) \propto p^2\,\,\,\,\,\,\,\,\,\,\mathrm{and}\,\,\,\,\,\,\,\,\,\,\lim_{p^2\rightarrow 0}\tilde{\mathcal{G}}(p^2) \propto (p^2)^0\,.
\label{3.49}
\end{equation}

\noindent The relation (\ref{3.48}) is violated by (\ref{3.49}) and the strong coupling constant vanishes as long as $p^2=0$. This behavior observed in the RGZ framework is in agreement with large volume lattice simulations \cite{Cucchieri:2007md,Cucchieri:2008fc,Cucchieri:2007rg,Cucchieri:2008qm}.

\subsection{BRST soft breaking}

In Subsect.~\ref{brstbreakingGZ}, we have pointed out the influence of the restriction of the path integral domain to the Gribov region $\Omega$. The resulting effect is that the BRST symmetry is softly broken by a term proportional to the Gribov parameter $\gamma^2$. In the RGZ scenario, the addition of dimension two operators does not change the situation. In particular, the operator $\mathcal{O}_{\varphi\omega}$ is BRST invariant and $\mathcal{O}_{A^2}$ is BRST invariant on-shell. On the other hand, this breaking is a mechanism responsible for ensuring the Gribov parameter $\gamma$ is physical as we showed in Sect.~\ref{gammaphys}. Very recently, some evidences from lattice simulations that BRST symmetry is indeed broken in the IR were provided in \cite{Cucchieri:2014via}.

\subsection{What about unitarity?}

It is well-known that BRST symmetry plays a very important role in the proof of unitarity at the perturbative regime of Yang-Mills theories, \cite{Kugo:1979gm}. However, there is some ``folk theorem" that is frequently evoked in the literature which says that if we spoil BRST symmetry we automatically destroy unitarity (and sometimes, also the opposite). This is a naive conclusion and, most important, general arguments based on asymptotic states of elementary particles of the theory are frequently used. However, in the case of pure Yang-Mills theories at the non-perturbative level, we are dealing with a confining gauge theory and therefore defining asymptotic states for its elementary excitations is not even possible. To see how subtle is the unitarity issue, we have to guarantee that the $S$-matrix associated with physical excitations of the spectrum of the theory must be unitary. But as just discussed in Subsect.~\ref{positivityviol}, gluons cannot be associated with physical excitations in the spectrum of the theory in the non-perturbative sector and the issue of unitarity of the $S$-matrix is meaningless for them. Certainly, this is signal that in the non-perturbative regime we do not have to expect to treat the unitarity issue in the same footing as in the perturbative level. A very important comment is that having a nilpotent BRST symmetry gives us a powerful tool to identify a subspace of renormalizable ``quantum" operators which have a correspondence with classical physical (gauge invariant operators) operators. This identification, however, is not the full story. We still have to prove this subspace is \textit{physical} by ensuring positivity. Proving this is independent from the existence of a nilpotent BRST symmetry which was used just to \textit{identify} such candidate space. An interesting progress in this issue in the context of the GZ action which breaks the BRST symmetry was addressed in \cite{Dudal:2012sb}, where the identification of such subspace of renormalizable quantum operators was done. Again, what remains is the proof of the physicality of such subspace, but this task is as difficult as if we considered standard BRST invariant Yang-Mills theories. Also an interesting counterexample of a theory that enjoys a nilpotent BRST symmetry but is not unitary was worked out in \cite{Dudal:2007ch}.

\subsection{(Recent) Applications}

The RGZ framework provides a local and renormalizable way of implementing the restriction of the path integral domain to $\Omega$ and taking into account the formation of condensates. In recent years, the RGZ action was used for various applications and so far, very nice results were obtained. It is far beyond the scope of this thesis to trace all the details of these applications, but we will mention at least a set of interesting results that might be faced as a further encouragement to a deep understanding of the effects of Gribov copies and dimension two condensates to non-perturbative Yang-Mills theories. As a (sub)set of applications we have,

\begin{itemize}

\item \textit{Computation of glueball spectra:} In \cite{Vandersickel:2011zc,Dudal:2009zh,Dudal:2010cd,Dudal:2013wja} the RGZ framework was employed to compute masses of glueballs. Remarkably, the results obtained compare well with lattice results even within a simple approximation used. 

\item \textit{Casimir energy:} In \cite{Canfora:2013zna}, the computation of the Casimir energy for the MIT bag model was perfomed using the RGZ gluon propagator. This provided a ``correct" sign for the Casimir energy, namely a negative one, a non-trivial result since the computation using the perturbative propagator indicates a positive energy \cite{Boyer:1968uf}.

\item \textit{Non-perturbative matter propagators:} Very recently, the authors of \cite{Capri:2014bsa} proposed an extension of the BRST soft breaking to the matter sector to be included in the pure Yang-Mills action. In particular, they considered scalar matter in the adjoint representation of the gauge group and quarks. Besides the soft BRST breaking effect, the effects of condensation of LCO's is also considered and the resulting propagators of the matter fields studied are in good qualitative agreement with lattice results, see \cite{Capri:2014bsa} and references therein. 

\item \textit{Finite temperature studies:} In \cite{Canfora:2013kma}, a study of the GZ gap equation at finite temperature was done. More recently, \cite{Canfora:2015yia}, a study concerning the Polyakov loop was also performed, again in the GZ action, but the first extensions to the RGZ framework were done. This topic might be of major interest in the following years and dealing with the RGZ framework at finite temperature, albeit challenging, is potentially a very rich arena to provide other physical quantities to be compared with lattice and functional approaches. 

\end{itemize}

To conclude this chapter, we remark that we have at our disposal a powerful analytical framework that takes into account non-perturbative features and provides a list of results which is compatible with recent lattice and functional methods data. We should emphasize, however, a strong limitation of this framework: Everything said so far was particularly valid in the Landau gauge. With the very exceptions of the maximal Abelian and Coulomb\footnote{Although a naive construction of the RGZ action for Coulomb gauge is possible, we remind that very few is known about the renormalizability properties of this gauge. Hence, one might argue that in the Coulomb gauge we do not have a complete satisfactory formulation so far.} gauges, there is no construction of the RGZ action for other gauges from first principles. To be fair enough, the construction of a RGZ action for Coulomb gauge is also limited due to our ignorance about the renormalizability properties of this gauge. Hence, we could honestly state that we are able to construct a local and renormalizable RGZ action for the Landau and maximal Abelian gauges from first principles. These are special particular cases due to the fact that they have a Hermitian Faddeev-Popov operator as pointed out in Ch.~\ref{ch.2}. It is precisely this property which ensures a geometric construction of a Gribov region which is suitable to restrict the path integral domain. However, the Faddeev-Popov operator is not Hermitian in general, a fact that jeopardizes a general construction of the GZ action from first (geometric) principles. Since we are dealing with a gauge theory we expect that physical quantities should be independent from the gauge choice and a proof of principle would be to compute a physical quantity in one gauge and compare with the same quantity in a different gauge. We know this statement is under control in the perturbative regime due to the BRST symmetry, which is crucial to prove gauge parameter independence of physical quantities. In the (R)GZ framework we have neither BRST symmetry nor the formulation of the action in a general gauge and this issue can pose a very natural question: Is the (R)GZ scenario consistent with gauge independence of physical quantities? In \cite{Lavrov:2011wb,Lavrov:2012gb,Lavrov:2015pka} it was advocated that the BRST soft breaking leads to inconsistency with gauge parameter independence of observables and to consistent extend the GZ action to different gauges we should follow a general receipt \cite{Lavrov:2013boa,Reshetnyak:2013bga,Moshin:2014xka,Moshin:2015gsa}. The proposal is very general: It takes for granted the (R)GZ formulation for Landau gauge and under the use of finite field dependent BRST transformations \cite{Joglekar:1994tq,Lavrov:2013rla,Upadhyay:2015lha} they generate a general expression for a would be (R)GZ action in an arbitrary gauge. Although very general, it is not clear, for instance, how to obtain the (R)GZ action for the maximal Abelian gauge in this framework such that it matches the one proposed from first principles. Also, it is not clear to us how to cast the formalism proposed in a local form. Clearly, these questions should be clarified and possibly will be addressed in the near future. 
	In any case, the second part of this thesis is devoted to the task of extending the RGZ action to other gauges from first principles and see how this extensions affect the way we should formulate a local action which takes into account gauge copies and LCO's condensation. As we shall present, a deep conceptual change in the way we formulate the RGZ action in the Landau gauge will enable us to construct a BRST symmetry for this action. With this, a BRST quantization is proposed and all consistency regarding gauge independence of observables is automatically guaranteed. However, we will present this (re)construction of the RGZ action in a step by step way. First we address the problem of the construction of the RGZ action is a general linear covariant gauge. It is precisely this problem that will suggest a complete reformulation of the RGZ framework in the Landau gauge and an introduction of a new (non-perturbative) BRST symmetry. The rest will be the very first consequences of this reformulation and the formal developments of the new BRST symmetry. 

\chapter{A first step towards linear covariant gauges} \label{ch.5}

In this chapter we give a tentative construction of the Gribov-Zwanziger action (and its refinement) in general linear covariant gauges (LCG). As we discussed in the previous chapters, given a gauge condition, an infinitesimal Gribov copy is associated with a zero-mode of the Faddeev-Popov operator. The basis of the solution proposed by Gribov and Zwanziger relies on the hermiticity of such operator. This is the case for the Landau gauge (a particular case of linear covariant gauge) and also for the maximal Abelian gauge. Linear covariant gauges are defined by 

\begin{equation}
\partial_{\mu}A^{a}_{\mu}=\alpha b^a\,,
\label{6.1}
\end{equation}

\noindent with $\alpha \geq 0$. The choice $\alpha = 0$ corresponds to the Landau gauge and, consequently, we have a Hermitian Faddeev-Popov operator. In general, however, for $\alpha >0$ the Faddeev-Popov operator is not Hermitian. To see this, we recognize the Faddeev-Popov operator in linear covariant gauge,

\begin{equation}
\EuScript{M}^{ab}_{\mathrm{LCG}}=-\partial_{\mu}D^{ab}_{\mu}=-\delta^{ab}\partial^2+gf^{abc}A^{c}_{\mu}\partial_{\mu}+\underbrace{gf^{abc}(\partial_{\mu}A^{c}_{\mu})}_{(\ast)}\,,
\label{6.2}
\end{equation}

\noindent where we notice that $(\ast)$ is zero in the Landau gauge, but generically non-vanishing for arbitrary $\alpha$. This fact jeopardizes the hermiticity of $\EuScript{M}^{ab}_{\mathrm{LCG}}$. Let us consider two test functions $\psi^a$ and $\chi^a$,

\begin{eqnarray}
(\EuScript{M}^{ab}_{\mathrm{LCG}}\psi^b,\chi^a)&=&\int d^dx \left(\EuScript{M}^{ab}_{\mathrm{LCG}}\psi^{\dagger b}\right)\chi^b=-\int d^dx\left(\partial_{\mu}(\delta^{ab}\partial_{\mu}-gf^{abc}A^{c}_{\mu})\psi^{\dagger b}\right)\chi^a\nonumber \\
&=&-\int d^dx\,(\delta^{ab}\partial^2\psi^{\dagger b})\chi^a+\int d^dx\,gf^{abc}(\partial_{\mu}A^{c}_{\mu})\psi^{\dagger b}\chi^a+\int d^dx\,gf^{abc}A^{c}_{\mu}(\partial_{\mu}\psi^{\dagger b})\chi^a\nonumber\\
&=&-\int d^dx\,\psi^{\dagger a}\delta^{ab}\partial^2\chi^b-\int d^dx\,\psi^{\dagger a}gf^{abc}(\partial_{\mu}A^{c}_{\mu})\chi^b+\int d^dx\,\psi^{\dagger a}gf^{abc}(\partial_{\mu}A^{c}_{\mu})\chi^b\nonumber\\
&+&\int d^dx\,\psi^{\dagger b}gf^{abc}A^{c}_{\mu}(\partial_{\mu}\chi^a)= (\psi^a,\EuScript{M}^{ab}_{\mathrm{LCG}}\chi^b)-\int d^dx\,\psi^{\dagger a}gf^{abc}(\partial_{\mu}A^{c}_{\mu})\chi^b\nonumber\\
&\neq& (\psi^a,\EuScript{M}^{ab}_{\mathrm{LCG}}\chi^b)\,,
\label{6.3}
\end{eqnarray}

\noindent where is clear that $\EuScript{M}^{ab}_{\mathrm{LCG}}$ is Hermitian only if $\partial_{\mu}A^{a}_{\mu}=0$. From this explicit proof, we start to identify the difficulty to extend what was done in the Landau gauge for a general linear covariant gauge. To begin with, we are not even able to define a region where $\EuScript{M}^{ab}_{\mathrm{LCG}}$ is positive, since it is not Hermitian. This obstacle hinders a trivial construction of a region akin to the Gribov region $\Omega$ in the Landau gauge in linear covariant gauges. On the other hand, an idea put forwarded in \cite{Sobreiro:2005vn} is to use an ``auxiliary" operator which is Hermitian and work an analogous construction of the Gribov-Zwanziger solution. This Hermitian operator is such that avoiding its zero-modes we automatically avoid zero-modes of the Faddeev-Popov operator $\EuScript{M}^{ab}_{\mathrm{LCG}}$. In the next subsection we will present this construction and out of this, we will propose a RGZ action for linear covariant gauges. As we will point out in the end of this chapter, this construction shows an inconsistency with the standard RGZ action in the Landau gauge. This inconsistency will motivate us to reformulate of the RGZ action, and this is presented in the next chapter.

\section{Gribov region in LCG: A proposal} \label{gribovreg}

Performing a gauge transformation over (\ref{6.1}), we define the Gribov copies equation in LCG, 

\begin{equation}
\EuScript{M}^{ab}_{\mathrm{LCG}}\xi^b \equiv  -\partial_{\mu}D^{ab}_{\mu}\xi^b=0\,,
\label{6.4}
\end{equation}

\noindent with $\xi^a$ being the infinitesimal parameter of the gauge transformation.

\noindent An attempt to define a region free from infinitesimal copies in linear covariant gauges was done in \cite{Sobreiro:2005vn}. In this work, a region free from infinitesimal copies was identified, under the assumption that the gauge parameter $\alpha$ is infinitesimal, {\it i.e.} $\alpha \ll 1$. The region introduced in \cite{Sobreiro:2005vn} is defined by demanding that  the transverse component of the gauge field, $A^{Ta}_\mu= \left(\delta_{\mu\nu}-\partial_{\mu}\partial_{\nu}/\partial^2\right)A^{a}_{\nu}$,  belongs to the Gribov region $\Omega$ of Landau gauge. Here, we review and extend this result for a finite value of $\alpha$ presented in \cite{Capri:2015pja}. The extension relies on  the following theorem, 

\begin{theorem}
If the transverse component of the gauge field\footnote{Remind the definition of $\Omega$, Def.~\ref{gribovregion}.} $A^{Ta}_{\mu}=\left(\delta_{\mu\nu}-\partial_{\mu}\partial_{\nu}/\partial^2\right)A^{a}_{\nu}$ $\in$ $\Omega$, then the equation $\EuScript{M}^{ab}_{\mathrm{LCG}}\xi^b=0$ has only the trivial solution $\xi^b=0$.
\label{thma}
\end{theorem}

\begin{proof}
Since $A^{aT}_\mu$ $\in$ $\Omega$ by assumption, the  operator

\begin{equation}
\EuScript{M}^{Tab}\equiv -\delta^{ab}\partial^2+gf^{abc}A^{Tc}_{\mu}\partial_{\mu}\,,
\label{proof1}
\end{equation}

\noindent is positive definite and, therefore, is invertible. As a consequence, eq.(\ref{6.4}) can be rewritten as

\begin{eqnarray}
\xi^a(x,\alpha) &=& -g\left[(\EuScript{M}^T)^{-1}\right]^{ad}f^{dbc}\partial_{\mu}(A^{Lc}_{\mu}\xi^{b}) \nonumber \\
&=& -g\alpha\left[(\EuScript{M}^T)^{-1}\right]^{ad}f^{dbc}\partial_{\mu}\left(\left(\frac{\partial_{\mu}}{\partial^2}b^{c}\right)\xi^{b}\right)\,,
\label{proof2}
\end{eqnarray}

\noindent where the gauge condition (\ref{6.1}) was used. We consider here zero modes  $\xi(x,\alpha)$ which are smooth functions of the gauge parameter $\alpha$. This requirement is motivated by the physical consideration that the quantity $\xi(x, \alpha)$ corresponds to the parameter of an infinitesimal gauge transformation. On physical grounds, we expect thus a regular behaviour of $\xi(x,\alpha)$ as function of $\alpha$, \textit{i.e.} infinitesimal modifications on the value of $\alpha$ should not produce a drastic singular behaviour of $\xi$, a feature also supported by the important fact that an acceptable zero mode has to be a square-integrable function, {\it i.e.} $\int d^4x \;\xi^a \xi^a <\infty$. Also, the $\alpha$-dependence should be such that in the limit $\alpha \rightarrow 0$ we recover the zero-modes of the Landau gauge.  Therefore, we require smoothness of $\xi(x,\alpha)$ with respect to $\alpha$.  Thus, for a certain radius of convergence $\EuScript{R}$, we can write the zero-mode $\xi(x,\alpha)$ as a Taylor expansion in $\alpha$,

\begin{equation}
\xi^a(x,\alpha)=\sum^{\infty}_{n=0}\alpha^n\xi^a_{n}(x)\,.
\label{proof3}
\end{equation}

\noindent For such radius of convergence, we can plug eq.(\ref{proof3}) into eq.(\ref{proof2}), which gives

\begin{equation}
\sum^{\infty}_{n=0}\alpha^n\xi^a_{n}(x)=-\sum^{\infty}_{n=0}g\alpha^{n+1}\left[(\EuScript{M}^T)^{-1}\right]^{ad}f^{dbc}\partial_{\mu}\left(\left(\frac{\partial_{\mu}}{\partial^2}b^{c}\right)\xi^{b}_{n}\right)\equiv \sum^{\infty}_{n=0}\alpha^{n+1}\phi^a_{n}\,.
\label{proof4}
\end{equation}

\noindent Since $\alpha$ is arbitrary, eq.(\ref{proof4}) should hold order by order in $\alpha$, which implies

\begin{equation}
\xi^a_0=0\,\,\,\, \Rightarrow \,\,\,\, \phi^a_0 = -\alpha\left[(\EuScript{M}^T)^{-1}\right]^{ad}f^{dbc}\partial_{\mu}\left(\left(\frac{\partial_{\mu}}{\partial^2}b^{c}\right)\xi^{b}_{0}\right)=0\,,
\label{proof5}
\end{equation}

\noindent at zeroth order. Therefore, 

\begin{equation}
\xi^a_1 = \alpha\phi^a_0=0\,,
\label{proof5a}
\end{equation}

\noindent and by recursion,

\begin{equation}
\xi^a_n = \alpha\phi^a_{n-1}=0,\,\,\forall n\,.
\label{proof6}
\end{equation}

\noindent Hence, the zero-mode $\xi(x,\alpha)$ must be identically zero within $\EuScript{R}$. Due to the requirement of smoothness, {\it i.e.}  of differentiability and  continuity of $\xi$, the zero-mode must vanish everywhere.        $\blacksquare$
\end{proof}

\noindent A comment is in order here. As emphasized before, Theorem~\ref{thma} holds for zero-modes which have a Taylor expansion in powers of $\alpha$. Although this does not seem to be a very strong requirement, since we expect smooth functions of $\alpha$, \textit{i.e.} small perturbations on $\alpha$ should not result on abrupt changes on $\xi$, one could think about the possibility to have zero-modes which might eventually display a pathological behavior, {\it i.e.} which could be singular for some values of the gauge parameter $\alpha$. For this reason, we shall refer to the zero-modes that admit a Taylor expansion as \textit{regular} zero-modes. Motivated by the previous theorem, we introduce the following ``Gribov region" ${\Omega}_{\mathrm{LCG}}$ in linear covariant gauges:  

\begin{definition}
The Gribov region ${\Omega}_{\mathrm{LCG}}$ in linear covariant gauges is given by

\begin{equation}
{\Omega}_{\mathrm{LCG}} = \left\{A^{a}_{\mu},\,\, \partial_{\mu}A^{a}_{\mu}-\alpha b^a =0,\,\,\EuScript{M}^{Tab} > 0 \right\}\,,    \; 
\label{defiLCG}
\end{equation}

\label{defi}
\end{definition}
\noindent where the operator $\EuScript{M}^{Tab}$ is given in eq.\eqref{proof1}

\noindent From the previous Theorem, it follows that the  region ${\Omega}_{\mathrm{LCG}}$ is free from  infinitesimal Gribov copies which are regular, {\it i.e.} smooth functions of the gauge parameter $\alpha$.   

\section{A natural candidate for the GZ action in LCG} \label{horizonLCG}

Definition~\ref{defi} provides a consistent candidate for the Gribov region in LCG. Following the strategy employed  by Gribov in \cite{Gribov:1977wm} and generalized by Zwanziger \cite{Zwanziger:1989mf}, we should restrict the path integral to the region ${\Omega}_{\mathrm{LCG}}$, \textit{i.e.}

\begin{equation}
\EuScript{Z} = \int_{{\Omega}_{\mathrm{LCG}}}  \left[\EuScript{D}\mathbf{\Phi}\right] \mathrm{e}^{-(S_{\mathrm{YM}}+S_\mathrm{gf})}\,,
\label{hor1}
\end{equation}

\noindent where $\mathbf{\Phi}$ represents all fields of the theory $(A,\bar{c},c,b)$. From equation \eqref{defiLCG}, one immediately sees that the region ${\Omega}_{\mathrm{LCG}}$ is defined by the positivity of the operator $\EuScript{M}^T$ which contains only the transverse component of the gauge field, eq.\eqref{proof1}. In other words, $\EuScript{M}^T$ is nothing but the Faddeev-Popov operator of the Landau gauge.  As a consequence, the whole procedure performed by Gribov \cite{Gribov:1977wm} and Zwanziger \cite{Zwanziger:1989mf} in the case of the Landau gauge, can be repeated here, although one has to keep in mind that the restriction of the domain of integration to the region ${\Omega}_{\mathrm{LCG}}$ affects only the transverse component of the gauge field, while the longitudinal sector remains unmodified. Therefore, following \cite{Gribov:1977wm,Zwanziger:1989mf}, for the restriction to the region ${\Omega}_{\mathrm{LCG}}$ we write 

\begin{equation}
\int_{{\Omega}_{\mathrm{LCG}}}  \left[\EuScript{D}\mathbf{\Phi}\right] \mathrm{e}^{-(S_{\mathrm{YM}}+S_\mathrm{gf})} = \int  \left[\EuScript{D}\mathbf{\Phi}\right] \mathrm{e}^{-\tilde{S}_{\mathrm{GZ}}} \;, \label{new3}   \; 
\end{equation}

\noindent where the action  $\tilde{S}_{\mathrm{GZ}}$ is given by 

\begin{equation}
\tilde{S}_{\mathrm{GZ}} = S_{\mathrm{YM}} + S_{\mathrm{gf}} + \gamma^4H(A^T) - dV\gamma^4(N^2-1)\,.
\label{hor3}
\end{equation}

The quantity $H(A^T)$ is the non-local horizon function which depends only on the transverse component of the gauge field $A^T$, namely 

\begin{equation}
H(A^T) =  g^2 \int d^dx~f^{adc}A^{Tc}_{\mu}[\left(\EuScript{M}^T\right)^{-1}]^{ab}f^{bde}A^{Te}_{\mu}\,.
\label{hor4}
\end{equation}

\noindent The parameter $\gamma$ in eq.\eqref{hor3} is the Gribov parameter. As in the case of the Landau gauge, it is  determined in a  self-consistent  way by the gap equation

\begin{equation}
\langle H(A^T) \rangle = dV(N^2-1)\,,
\label{gapequation}
\end{equation}

\noindent where the vacuum expectation value $\langle H(A^T) \rangle$ has to be evaluated  now with the measure defined in eq.\eqref{new3}.

\noindent The effective action (\ref{hor3}) implements the restriction to the region ${\Omega}_{\mathrm{LCG}}$. Here, an important feature has to be pointed out. Formally, the horizon function (\ref{hor4}) is the same as in the Landau gauge. However, in this case, although the longitudinal component of the gauge field does not enter the horizon function, it appears explicitly in the action $\tilde{S}_{\mathrm{GZ}}$. As we shall see, this property will give rise to several differences with respect to the Landau gauge. Another important point to be emphasized concerns  the vacuum term $dV\gamma^4(N^2-1)$ in expression \eqref{hor3}. This term is related to the spectrum of the operator $\EuScript{M}^{T}$ \cite{Vandersickel:2012tz} which does not depend on the gauge parameter $\alpha$. Therefore, at least at the level of the construction of the effective action $\tilde{S}_{\mathrm{GZ}}$, eq.(\ref{hor3}), the vacuum term is independent from $\alpha$. Of course, we should check out if quantum corrections might eventually introduce some $\alpha$-dependence in the vacuum term. This would require a lengthy analysis of the renormalizability properties of $\tilde{S}_{\mathrm{GZ}}$.

\noindent From now on, we will refer to the action (\ref{hor3}) as the Gribov-Zwanziger action in LCG. As it happens in the Landau gauge, this action is non-local, due to the non-locality of the horizon function. However, as we shall see in the next section, it is possible to localize this action by the introduction of a suitable set of auxiliary fields. Here,  differences with respect to the Landau gauge will show up, due to the unavoidable presence of the longitudinal component $A^{aL}_\mu$ of the gauge field. 

\section{Localization of the GZ action in LCG} \label{localization}

In order to have a suitable framework to apply the usual tools of quantum field theory, we have to express the action (\ref{hor3}) in local form. In the case of linear covariant gauges, the localization is not as direct as in Landau gauge (see Subsect.~\ref{localizationGZaction}). The difficulty relies on  the fact that the horizon function (\ref{hor4}) has two kinds of non-localities. The first one is the same as in Landau gauge, \textit{i.e.} the presence of the inverse of the Faddeev-Popov operator (or, in this case, an auxiliary operator)  $\EuScript{M}^T$. The other one follows from the fact that the decomposition of the  gauge field into transverse and longitudinal components is non-local, see eq.(\ref{a9.1}). Therefore, if we apply the same procedure used in the Landau gauge, the localization of the horizon function would give rise to a term of the type

\begin{equation}
\int d^dx~g\gamma^2f^{abc}A^{Tc}_{\mu}(\varphi+\bar{\varphi})^{ab}_{\mu}\,,
\label{hor5}
\end{equation}

\noindent which is still a non-local term, due to the presence of the transverse component  $A^{aT}_\mu$. However, it is possible to localize the action (\ref{hor3}) using an additional step. First, let us write the transverse component  of the gauge field as

\begin{equation}
A^{Ta}_{\mu} = A^{a}_{\mu} - h^{a}_{\mu}\,,
\label{hor6}
\end{equation}

\noindent where the  field $h^a_\mu$ will be identified with the longitudinal component, {\it i.e.} we shall impose that 

\begin{equation}
h^{a}_{\mu} = \frac{\partial_{\mu}\partial_{\nu}}{\partial^2}A^{a}_{\nu}\,.
\label{hor7}
\end{equation}

\noindent In other words, we introduce an extra field $h^a_\mu$ and state that the transverse part of the gauge field can be written in a local way using eq.(\ref{hor6}). Clearly, we must impose a constraint to ensure that, on-shell, eq.(\ref{hor6}) is equivalent to the usual decomposition eq.(\ref{a9.1}). Before introducing this constraint, we rewrite the horizon function in terms of $h^a_\mu$. As a matter of notation, we will denote $\EuScript{M}^{T}$ as $\EuScript{M}(A-h)$, when the transverse gauge field is expressed in terms of $h^a_\mu$. The horizon function \eqref{hor4} is now written as  

\begin{equation}
H(A,h) = g^2 \int d^dx~f^{adc}(A^{c}_{\mu}-h^{c}_{\mu})[\left(\EuScript{M}(A-h)\right)^{-1}]^{ab}f^{bde}(A^{e}_{\mu}-h^{e}_{\mu})\,.
\label{hor8}
\end{equation} 

\noindent The constraint\footnote{In \cite{Capri:2015pja}, a second constraint was introduced. At the classical level, this second constraint is not relevant and it turns out that it is not convenient at the quantum level. All the necessary properties are captured by (\ref{hor7}). Even the present localization procedure is not the most convenient one. The reason has to do with renormalizability, which is discussed in details in \cite{Capriren,Terin}. However, for our current purposes, this localization is enough.} given by eq.(\ref{hor7}) is imposed by the introduction of a Lagrange multiplier $\lambda^{a}_{\mu}$, \textit{i.e.}, by the introduction of the term

\begin{equation}
S_{\lambda}=\int d^dx~\lambda^{a}_{\mu}(\partial^2h^{a}_{\mu}-\partial_{\mu}\partial_{\nu}A^{a}_{\nu})\,.
\label{hor9}
\end{equation}

Therefore, the introduction of the extra field $h^a_\mu$ in eq.(\ref{hor8}) by means of the constraints (\ref{hor7}) implemented by the terms (\ref{hor9}) provides an action $S'_{\mathrm{GZ}}$ 

\begin{equation}
S'_{\mathrm{GZ}}=S_{\mathrm{YM}}+S_{\mathrm{gf}}+ S_{\lambda}+\gamma^4H(A,h) - dV\gamma^4(N^2-1)
\label{hor12}
\end{equation}

\noindent which is on-shell equivalent to the the non-local Gribov-Zwanziger action (\ref{hor3}). The introduction of the fields $h^{a}_{\mu}$and $\lambda^{a}_{\mu}$ has to be done through a BRST doublet \cite{Piguet:1995er} to avoid the appearance of such fields in the non-trivial part of the cohomology of the BRST operator $s$, a property which will be important for the renormalizability analysis. Therefore, we introduce the BRST doublets $(h^{a}_{\mu},\xi^{a}_{\mu})$ and $(\bar{\lambda}^{a}_{\mu},\lambda^{a}_{\mu})$, \textit{i.e.}

\begin{alignat}{2}
  sh^{a}_{\mu} &= \xi^{a}_{\mu}\,,\,\,\,\,\,\,\,\,\,\,\,\,\,\,\,\, &&s\bar{\lambda}^a_{\mu} = \lambda^a_{\mu}\,,\nonumber  \\
  s\xi^a_{\mu} &=0\,,\,\,\,\,\,\,\,\,\,\,\,\,\,\,\,\, &&s\lambda^a_{\mu} =0\,,  
\label{hor13}	
\end{alignat}

\noindent and define the following BRST exact term

\begin{eqnarray}
S_{\bar{\lambda}\lambda} &=& s\int d^dx~\bar{\lambda}^{a}_{\mu}(\partial^2h^{a}_{\mu}-\partial_{\mu}\partial_{\nu}A^{a}_{\nu})=\int d^dx~\lambda^a_{\mu}(\partial^2h^{a}_{\mu}-\partial_{\mu}\partial_{\nu}A^{a}_{\nu})\nonumber \\
&-&\int d^dx~\bar{\lambda}^{a}_{\mu}(\partial^2\xi^a_{\mu}+\partial_{\mu}\partial_{\nu}D^{ab}_{\nu}c^b)
\label{hor14}
\end{eqnarray}

\noindent The term (\ref{hor14}) implements the constraint (\ref{hor7}) in a manifest BRST invariant way. What remains now is the localization of the term (\ref{hor8}). Since this term has just the usual non-locality of the Gribov-Zwanziger action in the Landau gauge, given by the inverse of the operator $\EuScript{M}(A-h)$, we can localize it by the introduction of the same set of auxiliary fields employed in the case of the localization of the horizon function in the Landau gauge, see Subsect.~\ref{localizationGZaction}. Thus,  the term  (\ref{hor8}) is replaced by the the following local expression

\begin{eqnarray}
S_H &=& - s\int d^dx~\bar{\omega}^{ac}_{\mu}\EuScript{M}^{ab}(A-h)\varphi^{bc}_{\mu} + \gamma^2g\int d^dx~f^{abc}(A^{a}_{\mu}-h^{a}_{\mu})(\varphi+\bar{\varphi})^{bc}_{\mu} \nonumber \\
&=& \int d^dx~\left(\bar{\varphi}^{ac}_{\mu}\partial_{\nu}D^{ab}_{\nu}\varphi^{bc}_{\mu}-\bar{\omega}^{ac}_{\mu}\partial_{\nu}D^{ab}_{\nu}\omega^{bc}_{\mu}-gf^{adb}(\partial_{\nu}\bar{\omega}^{ac}_{\mu})(D^{de}_{\nu}c^{e})\varphi^{bc}_{\mu}\right) \nonumber \\
&+& \int d^dx~\left(gf^{adb}(\partial_{\nu}\bar{\varphi}^{ac}_{\mu})h^{d}_{\nu}\varphi^{bc}_{\mu}-gf^{adb}(\partial_{\nu}\bar{\omega}^{ac}_{\mu})h^{d}_{\nu}\omega^{bc}_{\mu}-gf^{adb}(\partial_{\nu}\bar{\omega}^{ac}_{\mu})\xi^{d}_{\nu}\varphi^{bc}_{\mu}\right).\nonumber \\
&+& \gamma^2g\int d^dx~f^{abc}(A^{a}_{\mu}-h^{a}_{\mu})(\varphi+\bar{\varphi})^{bc}_{\mu}\,, \nonumber \\
\label{hor15}
\end{eqnarray}

\noindent where

\begin{alignat}{2}
  s\varphi^{ab}_{\mu} &= \omega^{ab}_{\mu}\,,\,\,\,\,\,\,\,\,\,\,\,\,\,\,\,\, &&s\bar{\omega}^{ab}_{\mu} =\bar{\varphi}^{ab}_{\mu} \nonumber  \\
  s\omega^{ab}_{\mu} &=0\,,\,\,\,\,\,\,\,\,\,\,\,\,\,\,\,\, &&s\bar{\varphi}^{ab}_{\mu} = 0\,.  
\label{hor15A}	
\end{alignat}

\noindent Finally,  the action $S_{\mathrm{GZ}}$ given by

\begin{equation}
S_{\mathrm{GZ}}=S_{\mathrm{YM}}+S_{\mathrm{gf}}+S_{\bar{\lambda}\lambda}+S_H
\label{hor16}
\end{equation}

\noindent is local and, on-shell, equivalent to the non-local action (\ref{hor4}). Explicitly, $S_{\mathrm{GZ}}$ is written as

\begin{eqnarray}
S_{\mathrm{GZ}} &=& \frac{1}{4}\int d^dx~F^{a}_{\mu\nu}F^{a}_{\mu\nu} + \int d^dx~b^a\left(\partial_{\mu}A^{a}_{\mu}-\frac{\alpha}{2}b^a\right)+\int d^dx~\bar{c}^a\partial_{\mu}D^{ab}_{\mu}c^b  \nonumber \\
&+& \int d^dx~\lambda^a_{\mu}(\partial^2h^{a}_{\mu}-\partial_{\mu}\partial_{\nu}A^{a}_{\nu}) +\int d^dx~\bar{\lambda}^{a}_{\mu}(\partial^2\xi^a_{\mu}+\partial_{\mu}\partial_{\nu}D^{ab}_{\nu}c^b)\nonumber \\
&+& \int d^dx~\left(\bar{\varphi}^{ac}_{\mu}\partial_{\nu}D^{ab}_{\nu}\varphi^{bc}_{\mu}-\bar{\omega}^{ac}_{\mu}\partial_{\nu}D^{ab}_{\nu}\omega^{bc}_{\mu}-gf^{adb}(\partial_{\nu}\bar{\omega}^{ac}_{\mu})(D^{de}_{\nu}c^{e})\varphi^{bc}_{\mu}\right) \nonumber \\ 
&+&\int d^dx~\left(gf^{adb}(\partial_{\nu}\bar{\varphi}^{ac}_{\mu})h^{d}_{\nu}\varphi^{bc}_{\mu}-gf^{adb}(\partial_{\nu}\bar{\omega}^{ac}_{\mu})h^{d}_{\nu}\omega^{bc}_{\mu}-gf^{adb}(\partial_{\nu}\bar{\omega}^{ac}_{\mu})\xi^{d}_{\nu}\varphi^{bc}_{\mu}\right)\nonumber \\
&+& \gamma^2g\int d^dx~f^{abc}(A^{a}_{\mu}-h^{a}_{\mu})(\varphi+\bar{\varphi})^{bc}_{\mu}-dV\gamma^4(N^2-1)  \,, 
\label{hor17}
\end{eqnarray}

\noindent and we will refer to it as the local GZ action in LCG. We highlight that, in the limit $\gamma\rightarrow 0$, the term (\ref{hor15}) can be trivially integrated out to give a unity.  The remaining action is simply the gauge fixed Yang-Mills action with the addition of the constraint over $h^a_\mu$. This constraint is also easily integrated out, so that the resulting action is simply the usual Yang-Mills action in linear covariant gauges.

\noindent  Let us end this section by noticing that, in the local formulation, the gap equation \eqref{gapequation} takes the following expression 

\begin{equation} 
\frac{ \partial {\cal E}_v}{\partial \gamma^2} = 0  \;, \label{locgap}   
\end{equation}

\noindent where ${\cal E}_v$ denotes the vacuum energy, obtained from 

\begin{equation} 
\mathrm{e}^{-V{\cal E}_v} =  
\int  \left[\EuScript{D}\mu\right] \mathrm{e}^{-S_{\mathrm{GZ}}}\,. \label{ve}     
\end{equation} 

\noindent with $\mu$ being the complete set of fields, {\it i.e.}, the usual ones from the Faddeev-Popov quantization and the auxiliary fields introduced to implement the constraint and to localize the Gribov-Zwanziger action.

\section{BRST soft breaking again} \label{brstbreaking}

As it happens in the case of the GZ action in the Landau and maximal Abelian gauges, expression (\ref{hor17}) is not invariant under the BRST transformations. The only term of the action which is not invariant under BRST transformations is

\begin{equation}
g\gamma^2\int d^dx~f^{abc}(A^{a}_{\mu}-h^{a}_{\mu})(\varphi+\bar{\varphi})^{bc}_{\mu}\,,
\label{brstbreaking1}
\end{equation}

\noindent giving

\begin{equation}
sS_{\mathrm{GZ}}\equiv \Delta_{\gamma^2} = g\gamma^2\int d^dx~f^{abc}\left[-(D^{ad}_{\mu}c^d+\xi^a_{\mu})(\varphi+\bar{\varphi})^{bc}_{\mu}+(A^{a}_{\mu}-h^{a}_{\mu})\omega^{bc}_{\mu}\right]\,.
\label{brstbreaking2}
\end{equation}

\noindent  From eq.(\ref{brstbreaking2}), we see that the BRST breaking is soft, {\it i.e.} it is of dimension two in the quantum fields.  This is precisely the same situation of the Landau and maximal Abelian gauges. The restriction of the domain of integration in the path integral to the Gribov region generates a soft breaking of the BRST symmetry which turns out to be  proportional to the parameter $\gamma^2$. As discussed before, when we take the limit $\gamma\rightarrow 0$, we obtain the usual Faddeev-Popov gauge fixed Yang-Mills action which is BRST invariant. Thus, the breaking of the BRST symmetry is a direct consequence of the restriction of the path integral to the Gribov region $\Omega_{\mathrm{LCG}}$. We also emphasize that, although the gauge condition we are dealing with contains a gauge parameter $\alpha$, the BRST breaking term does not depend on such parameter, due to the fact that the horizon function  takes into account only the transverse component  of the gauge fields, as eq.(\ref{hor4}) shows. 

\section{Gap equation at one-loop order: Explicit analysis} \label{oneloopgap}

As discussed above, in the construction of the effective action which takes into account the presence of infinitesimal Gribov copies a non-perturbative  parameter $\gamma$, {\it i.e.} the Gribov parameter, shows up in the theory. However, this parameter is not free, being determined by the gap equation \eqref{locgap}. In the Landau gauge, the Gribov parameter encodes the restriction of the domain of integration to the Gribov region $\Omega$.  Also, physical quantities like the glueball masses were computed in the Landau gauge \cite{Dudal:2010cd,Dudal:2013wja}, exhibiting an explicit dependence from $\gamma$. It is therefore of primary importance to look at the Gribov parameter in LCG, where both  the longitudinal component of the gauge field and the gauge parameter $\alpha$ are present in the explicit loop computations. We should check out  the possible (in)dependence of $\gamma$ from $\alpha$. Intuitively, from our construction, we would expect that the Gribov parameter would be independent from $\alpha$, as a consequence of the fact that we are imposing a restriction of the domain of integration in the path integral which affects essentially only the transverse sector of the theory. Moreover, the independence from $\alpha$ of the Gribov parameter would also imply that physical quantities like the glueball masses would, as expected,  be $\alpha$-independent. To obtain some computational confirmation of the possible $\alpha$-independence of the Gribov parameter, we provide here the explicit computation of the gap equation at one-loop order. 

\noindent According to eqs.\eqref{locgap} and \eqref{ve},  the one-loop vacuum energy can be  computed by retaining the quadratic part of $S_{\mathrm{GZ}}$ and integrating over the auxiliary fields, being given by 

\begin{equation}
\mathrm{e}^{-V{\cal E}^{(1)}_{v}}=\int \left[\EuScript{D}A\right]\mathrm{e}^{-\int \frac{d^dp}{(2\pi)^d}~\frac{1}{2}A^{a}_{\mu}(p)\tilde{\Delta}^{ab}_{\mu\nu}A^{b}_{\nu}(-p)+dV\gamma^4(N^2-1)}\,,
\label{gapeq1}
\end{equation}

\noindent where 

\begin{equation}
\tilde{\Delta}^{ab}_{\mu\nu}=\delta^{ab}\left[\delta_{\mu\nu}\left(p^{2}+\frac{2Ng^2\gamma^4}{p^2}\right)+p_{\mu}p_{\nu}\left(\left(\frac{1-\alpha}{\alpha}\right)-\frac{2Ng^2\gamma^4}{p^4}\right)\right]\,. 
\label{gapeq2}
\end{equation}

\noindent Performing the functional integral over the gauge fields, we obtain 

\begin{equation}
V{\cal E}^{(1)}_{v} = \frac{1}{2}\mathrm{Tr~ln}\tilde{\Delta}^{ab}_{\mu\nu}-dV\gamma^4(N^2-1)\,.
\label{gapeq3}
\end{equation}

\noindent The remaining step is to compute the functional trace in eq.(\ref{gapeq3}). This is a standard computation, see \cite{Vandersickel:2012tz}. We have 

\begin{equation}
{\cal E}^{(1)}_{v}=\frac{(N^2-1)(d-1)}{2}\int\frac{d^dp}{(2\pi)^d}\mathrm{ln}\left(p^2 + \frac{2Ng^2\gamma^4}{p^2}\right)-d\gamma^4(N^2-1)\,,
\label{gapeq4}
\end{equation}

\noindent We see thus from eq.(\ref{gapeq4}) that the one-loop vacuum energy does not depend on $\alpha$ and the gap equation which determines the Gribov parameter is written as

\begin{equation}
\frac{\partial {\cal E}^{(1)}_{v}}{\partial \gamma^2} = 0\,\,\, \Rightarrow \,\,\, \frac{(d-1)Ng^2}{d}\int\frac{d^dp}{(2\pi)^d}\frac{1}{p^4+2Ng^2\gamma^4} = 1\,.
\label{gapeq5}
\end{equation}

\noindent This equation states that, at one-loop order, the Gribov parameter $\gamma$ is independent from $\alpha$ and, therefore, is the same as in the Landau gauge, which agrees with our expectation. Although being a useful check of our framework, it is important to state that this result has to be extended at higher orders, a non-trivial topic.

\section{Dynamical generation of condensates}\label{condensatessect}

The RGZ action \cite{Dudal:2008sp,Dudal:2011gd} takes into account the existence of dimension two condensates in an effective way already at the level of the starting action as discussed in Ch.~\ref{RGZch}. Here, we expect that, in analogy with the Landau gauge, dimension two condensates will show up in a similar way. In fact, the presence of these dimension two condensates can be established as in Ch.~\ref{RGZch} through a one-loop elementary computation, which shows that the following dimension two condensates 

\begin{equation}
\langle A^{Ta}_{\mu}A^{Ta}_{\mu}\rangle\;,  \qquad \langle \bar{\varphi}^{ab}_{\mu}\varphi^{ab}_{\mu}-\bar{\omega}^{ab}_{\mu}\omega^{ab}_{\mu}\rangle\,.
\label{cond0}
\end{equation} 

\noindent are non-vanishing already at one-loop order, being proportional to the Gribov parameter $\gamma$, in analogy with Landau gauge. In particular, it should be observed that the condensate $\langle A^{Ta}_{\mu}A^{Ta}_{\mu}\rangle$ contains only the transverse component of the gauge field. This is a direct consequence of the fact that the horizon function of the linear covariant gauges, eq.\eqref{hor4}, depends only on the transverse component $A_\mu^{Ta}$. In order to evaluate the condensates  $\langle A^T A^T \rangle$ and $\langle \bar{\varphi}\varphi-\bar{\omega}\omega\rangle$ at one-loop order, one needs the quadratic part of the GZ action in LCG, namely 

\begin{eqnarray}
S^{(2)}_{\mathrm{GZ}} &=& \int d^dx \left[\frac{1}{2}\left(-A^{a}_{\mu}\delta_{\mu\nu}\delta^{ab}\partial^2A^{b}_{\nu}+\left(1-\frac{1}{\alpha}\right)A^{a}_{\mu}\delta^{ab}\partial_{\mu}\partial_{\nu}A^{b}_{\nu}\right) + \bar{c}^{a}\delta^{ab}\partial^2c^{b} \right.\nonumber \\
&+&\left.\bar{\varphi}^{ac}_{\mu}\delta^{ab}\partial^2\varphi^{bc}_{\mu} - \bar{\omega}^{ac}_{\mu}\delta^{ab}\partial^2\omega^{bc}_{\mu}\right]+\gamma^2g\int d^dx~f^{abc}A^{a}_{\nu}\left(\delta_{\mu\nu}-\frac{\partial_{\mu}\partial_{\nu}}{\partial^2}\right)(\varphi + \bar{\varphi})^{bc}_{\mu}\nonumber \\
&-&dV\gamma^4(N^2-1)\,.
\label{cond1}
\end{eqnarray}

\noindent  Further, we  introduce the operators  $\int d^dx A^T A^T $ and $\int d^dx ( \bar{\varphi}\varphi-\bar{\omega}\omega) $ in the action by coupling them to two constant sources $J$ and $m$, and we define the vacuum functional ${\cal E}(m,J)$ defined by 
\begin{equation}
\mathrm{e}^{-V{\cal E}(m,J)}=\int \left[\EuScript{D}\mu\right]\mathrm{e}^{-S^{(2)}_{\mathrm{GZ}}+J\int d^dx\left(\bar{\varphi}^{ac}_{\mu}\varphi^{ac}_{\mu}-\bar{\omega}^{ac}_{\mu}\omega^{ac}_{\mu}\right)-m\int d^dx~A^{a}_{\mu}\left(\delta_{\mu\nu}-\frac{\partial_{\mu}\partial_{\nu}}{\partial^{2}}\right)A^{a}_{\nu}}\,.
\label{cond2}
\end{equation}

\noindent It is apparent to check that the  condensates $\langle A^T A^T \rangle$ and $\langle \bar{\varphi}\varphi-\bar{\omega}\omega\rangle$ are obtained by differentiating ${\cal E}(m,J)$ with respect to the sources $(J,m)$, which are set to zero at the end, {\it i.e.}

\begin{eqnarray}
\langle \bar{\varphi}^{ac}_{\mu}(x)\varphi^{ac}_{\mu}(x)-\bar{\omega}^{ac}_{\mu}(x)\omega^{ac}_{\mu}(x)\rangle &=& - \frac{\partial {\cal E}(m,J)}{\partial J}\Big|_{J=m=0} \nonumber  \\
\langle A^{Ta}_{\mu}(x)A^{Ta}_{\mu}(x)\rangle &=& \frac{\partial {\cal E}(m,J)}{\partial m}\Big|_{J=m=0}\,.
\label{cond3}
\end{eqnarray}

\noindent A direct computation shows that

\begin{equation}
\mathrm{e}^{-V{\cal E}(m,J)}=\mathrm{e}^{-\frac{1}{2}\mathrm{Tr\,ln}\Delta^{ab}_{\mu\nu}+dV\gamma^4(N^2-1)}\,,
\label{cond4}
\end{equation}

\noindent with

\begin{equation}
\Delta^{ab}_{\mu\nu}=\delta^{ab}\left[\delta_{\mu\nu}\left(p^2+\frac{2\gamma^4g^2N}{p^2+J}+2m\right)+p_{\mu}p_{\nu}\left(\left(\frac{1-\alpha}{\alpha}\right)-\frac{2\gamma^4g^2N}{p^2(p^2+J)}-\frac{2m}{p^2}\right)\right]\,. 
\label{cond5}
\end{equation}

\noindent Evaluating the trace, we obtain\footnote{In this computation we are concerned just with the contribution associated to the restriction of the path integral to ${\Omega}_{\mathrm{LCG}}$.}

\begin{equation}
{\cal E}(m,J)=\frac{(d-1)(N^2-1)}{2}\int \frac{d^dp}{(2\pi)^d}~\mathrm{ln}\left(p^2+\frac{2\gamma^4g^2N}{p^2+J}+2m\right)-d\gamma^4(N^2-1)\,. 
\label{cond6}
\end{equation}

\noindent Eq.(\ref{cond3}) and (\ref{cond6}) gives thus

\begin{equation}
\langle \bar{\varphi}^{ac}_{\mu}\varphi^{ac}_{\mu}-\bar{\omega}^{ac}_{\mu}\omega^{ac}_{\mu}\rangle = \gamma^4g^2N(N^2-1)(d-1)\int \frac{d^dp}{(2\pi)^d}\frac{1}{p^2}\frac{1}{(p^4+2g^2\gamma^4N)}
\label{cond7}
\end{equation}

\noindent and

\begin{equation}
\langle A^{Ta}_{\mu}A^{Ta}_{\mu}\rangle = -\gamma^4(N^2-1)(d-1)\int\frac{d^dk}{(2\pi)^d}\frac{1}{k^2}\frac{2g^2N}{(k^4+2g^2\gamma^4N)}\,,
\label{cond8}
\end{equation}

\noindent Eq.(\ref{cond7}) and eq.(\ref{cond8}) show that, already at one-loop order, both condensates  $\langle A^T A^T \rangle$ and $\langle \bar{\varphi}\varphi-\bar{\omega}\omega\rangle$ are non-vanishing and proportional to the Gribov parameter $\gamma$. Notice also that both integrals in eqs.(\ref{cond7}) and (\ref{cond8})  are perfectly convergent in the ultraviolet region by power counting for $d=3,4$. We see thus that, in perfect analogy with the case of the Landau gauge, dimension two condensates are automatically generated by the restriction of the domain of integration to the Gribov region, as encoded in the  parameter $\gamma$. As discussed in Ch.~\ref{RGZch}, the presence of these condensates can be taken into account directly in the starting action giving rise to the refinement of the GZ action. Also, higher order contributions can be systematically evaluated through the calculation of the effective potential for the corresponding dimension two operators by means of the LCO technique, see \cite{Dudal:2011gd} and Ap.~\ref{LCO}.

In the present case, for the refined version of the GZ action which takes into account the presence of the dimension two condensates, we have

 \begin{equation}
S_{\mathrm{RGZ}} = S_{\mathrm{GZ}} + \frac{{\hat m}^2}{2}\int d^dx~(A^{a}_{\mu}-h^{a}_{\mu})(A^{a}_{\mu}-h^{a}_{\mu})-{\hat M}^2\int d^dx~(\bar{\varphi}^{ab}_{\mu}\varphi^{ab}_{\mu}-\bar{\omega}^{ab}_{\mu}\omega^{ab}_{\mu})\,,
\label{propa1}
\end{equation}

\noindent where the parameters $({\hat m},{\hat M})$ can be determined order by order in a self-content way through the evaluation of the corresponding effective potential, as outlined in the case of the Landau gauge  \cite{Dudal:2011gd}.  Let us remark here that the calculation of the vacuum functional ${\cal E}(m,J)$, eq.\eqref{cond2}, done in the previous section shows that, at one-loop order, these parameters turn out to be independent from $\alpha$. The study of the effective potential for the dimension two operators  $\int d^4x A^T A^T $ and $\int d^4x ( \bar{\varphi}\varphi-\bar{\omega}\omega) $ is of utmost importance in order to extend this feature to higher orders. We are now ready to evaluate the tree level gluon propagator in the LCG. This will be the topic of the next section.  

\section{Gluon propagator}\label{propagator}

From the RGZ action, eq.(\ref{propa1}), one can immediately evaluate the  tree level gluon propagator in LCG, given by the following expression 

\begin{equation}
\langle A^{a}_{\mu}(k)A^{b}_{\nu}(-k)\rangle = \delta^{ab}\left[\frac{k^2+{\hat M}^2}{(k^2+{\hat m}^2)(k^2+{\hat M}^2)+2g^2\gamma^4N}\left(\delta_{\mu\nu}-\frac{k_{\mu}k_{\nu}}{k^2}\right)+\frac{\alpha}{k^2}\frac{k_{\mu}k_{\nu}}{k^2}\right]\,.
\label{propa2}
\end{equation}

\noindent A few comments are now in order.  First, the tree-level longitudinal sector is not affected by the restriction to the Gribov region, \textit{i.e.} the longitudinal component of the propagator is the same as the perturbative one. This is an expected result (at tree-level), since the Gribov region ${\Omega}_{\mathrm{LCG}}$ for linear covariant gauges does not impose any restriction to the longitudinal component $A^{La}_\mu$. Second, in the limit $\alpha\rightarrow 0$, the gluon propagator coincides precisely with the well-known result in Landau gauge \cite{Dudal:2008sp}. In particular, this  implies that all features of the gluon propagator in the RGZ framework derived in Landau gauge remains true for  the transverse component of the correlation function \eqref{propa2} at tree-level.

\subsection{Lattice results}

Unlike the Landau, Coulomb and maximal Abelian gauges, the construction of a minimizing functional to define LCG is highly non-trivial. This feature is the source of several complications for a lattice formulation of these gauges. The study of the LCG through lattice numerical simulations represents a big challenge and is gaining special attention in recent years from this community. The first attempt to implement these gauges on the lattice was undertaken by \cite{Giusti:1996kf,Giusti:1999im,Giusti:2000yc,Giusti:2001kr}. More recently, the authors of  \cite{Cucchieri:2008zx,Mendes:2008ux,Cucchieri:2010ku,Cucchieri:2009kk,Cucchieri:2011pp,Cucchieri:2011aa,Bicudo:2015rma} have been able to implement the linear covariant gauges on the lattice by means of a different procedure. 

\noindent With respect to the most recent data \cite{Cucchieri:2008zx,Mendes:2008ux,Cucchieri:2010ku,Cucchieri:2009kk,Cucchieri:2011pp,Cucchieri:2011aa,Bicudo:2015rma} obtained on bigger lattices, our results are in very good qualitative agreement: The tree level transverse gluon propagator does not depend on $\alpha$ and, therefore, behaves like the gluon propagator in the RGZ framework in the Landau gauge. On the other hand, the longitudinal form factor $\mathcal{D}_L(k^2)$ defined by 

\begin{equation} 
\mathcal{D}_L(k^2) =k^2  \frac{\delta^{ab}}{N^2-1} \frac{k_\mu k_\nu}{k^2} \langle A^a_\mu(k) A^b_\nu(-k) \rangle= \alpha   \;. \label{lff}
\end{equation} 
 
\noindent is equal to the gauge parameter $\alpha$, being not affected by the restriction to the Gribov region ${\Omega}_{\mathrm{LCG}}$. These results are in complete agreement with the numerical data of \cite{Cucchieri:2008zx,Mendes:2008ux,Cucchieri:2010ku,Cucchieri:2009kk,Cucchieri:2011pp,Cucchieri:2011aa,Bicudo:2015rma}. Although many properties of the Gribov region ${\Omega}_{\mathrm{LCG}}$ in LCG need to be further established, the qualitative agreement of our results on the gluon propagator with the recent numerical ones are certainly reassuring, providing a good support  for the introduction of the region ${\Omega}_{\mathrm{LCG}}$. Also we point out that, recently, results for the gluon and ghost propagators in linear covariant gauges have been  obtained through the use of the Dyson-Schwinger equations by \cite{Aguilar:2015nqa,Huber:2015ria}.

We are obliged to emphasize, however, although a lot of attention is being devoted to LCG in the recent years from different approaches, the results are not well established as in the Landau gauge. Therefore, the references presented up to now correspond to the status of the art of the gluon propagator from lattice and Schwinger-Dyson equations. The tree-level gluon propagator (\ref{propa2}) is not protected from quantum corrections and up to now, there is no way to establish if $\alpha$-dependent corrections will enter in both transverse and longitudinal sectors. So, it is completely reasonable that loop corrections will destroy the transverse independence from $\alpha$ and also the ``perturbative" like longitudinal part. To understand it, we should go to higher orders in the two-point function computation, a non-trivial task (see the case of Landau gauge, for instance \cite{Gracey:2006dr}). This is far beyond the scope of this thesis. Also, the current lattice and functional results are concentrated mostly for small values of $\alpha$. The role of the $\alpha$-dependence in the gluon two-point function is a very interesting topic that deserves attention from different approaches. 

\noindent An important point to be emphasized is: In functional approaches, the longitudinal part of the gluon propagator is protected to loop corrections, being described by (\ref{lff}) to all orders. This is an immediate consequence of the BRST invariance, an assumption used in these methods to hold at the non-perturbative level. On the other hand, in the RGZ scenario, the BRST symmetry is softly broken and we have no argument to protect the longitudinal sector from loop corrections. Therefore, the result (\ref{lff}) might be just a consequence of the tree-level approximation, rather than a general result.

\section{A first look at the RGZ propagators in the LCG}

An obvious requirement we demand from the RGZ action in the LCG is that as soon as we take the limit $\alpha=0$, all results from Landau gauge should be recovered. A particular and very important check is the full set of propagators of the RGZ action.  The list of all propagators of bosonic fields in the RGZ action in LCG is,

\begin{eqnarray}
\langle A^{a}_{\mu}(-k)A^{b}_{\nu}(k) \rangle &=& \delta^{ab}\left[\frac{k^2+\hat{M}^2}{(k^2+\hat{M}^2)(k^2+\hat{m}^2)+2\gamma^4g^2N}\mathcal{P}_{\mu\nu}+\frac{\alpha}{k^2}\mathcal{L}_{\mu\nu}\right] \\
\langle A^{a}_{\mu}(-k)b^b(k) \rangle &=& -i\delta^{ab}\frac{k_{\mu}}{k^2} \\
\langle A^{a}_{\mu}(-k)\varphi^{bc}_{\nu}(k)\rangle &=& f^{abc}\frac{\gamma^2g}{(k^2+\hat{M}^2)(k^2+\hat{m}^2)+2\gamma^4g^2N}\mathcal{P}_{\mu\nu}\\
\langle A^{a}_{\mu}(-k)\bar{\varphi}^{bc}_{\nu}(k)\rangle &=& f^{abc}\frac{\gamma^2g}{(k^2+\hat{M}^2)(k^2+\hat{m}^2)+2\gamma^4g^2N}\mathcal{P}_{\mu\nu}\\
\langle A^{a}_{\mu}(-k)\lambda^{b}_{\nu}(k) \rangle &=& -\delta^{ab}\frac{\hat{m}^2(k^2+\hat{M}^2)+2\gamma^4g^2N}{k^2\left[(k^2+\hat{M}^2)(k^2+\hat{m}^2)+2\gamma^4g^2N\right]}\mathcal{P}_{\mu\nu} \\
\langle A^{a}_{\mu}(-k)h^{b}_{\nu}(k) \rangle &=& \delta^{ab}\frac{\alpha}{k^2}\mathcal{L}_{\mu\nu} \\
\langle b^a(-k)h^{b}_{\mu}(k)\rangle &=& \delta^{ab}i\frac{k_{\mu}}{k^2} \\
\langle \varphi^{ab}_{\mu}(-k)\varphi^{cd}_{\nu}(k) \rangle &=& f^{abm}f^{mcd}\frac{\gamma^4g^2}{(k^2+\hat{M}^2)\left[(k^2+\hat{M}^2)(k^2+\hat{m}^2)+2\gamma^4g^2N\right]}\mathcal{P}_{\mu\nu}\nonumber\\
\label{prop1}
\end{eqnarray}

\begin{eqnarray}
\langle \varphi^{ab}_{\mu}(-k)\bar{\varphi}^{cd}_{\nu}(k) \rangle &=& f^{abm}f^{mcd}\frac{\gamma^4g^2}{(k^2+\hat{M}^2)\left[(k^2+\hat{M}^2)(k^2+\hat{m}^2)+2\gamma^4g^2N\right]}\mathcal{P}_{\mu\nu}\nonumber\\
&-& \delta^{ac}\delta^{bd}\frac{1}{k^2+\hat{M}^2}\delta_{\mu\nu}\\
\langle \varphi^{ab}_{\mu}(-k)\lambda^{c}_{\nu}(k) \rangle &=& f^{abc}\left[-\frac{1}{k^2(k^2+\hat{M}^2)}\frac{\hat{m}^2\gamma^2g(k^2+\hat{M}^2)+2\gamma^6g^3N}{(k^2+\hat{M}^2)(k^2+\hat{m}^2)+2\gamma^4g^2N}\mathcal{P}_{\mu\nu}\right.\nonumber\\
&+&\left.\frac{2\gamma^2g}{k^2(k^2+\hat{M}^2)}\delta_{\mu\nu}\right] \\
\langle \bar{\varphi}^{ab}_{\mu}(-k)\bar{\varphi}^{cd}_{\nu}(k) \rangle &=& f^{abm}f^{mcd}\frac{\gamma^4g^2}{(k^2+\hat{M}^2)\left[(k^2+\hat{M}^2)(k^2+\hat{m}^2)+2\gamma^4g^2N\right]}\mathcal{P}_{\mu\nu}\nonumber\\
\label{prop1b}
\end{eqnarray}

\begin{eqnarray}
\langle \bar{\varphi}^{ab}_{\mu}(-k)\lambda^{c}_{\nu}(k) \rangle &=& f^{abc}\left[-\frac{1}{k^2(k^2+\hat{M}^2)}\frac{\hat{m}^2\gamma^2g(k^2+\hat{M}^2)+2\gamma^6g^3N}{(k^2+\hat{M}^2)(k^2+\hat{m}^2)+2\gamma^4g^2N}\mathcal{P}_{\mu\nu}\right.\nonumber\\
&+&\left.\frac{2\gamma^2g}{k^2(k^2+\hat{M}^2)}\delta_{\mu\nu}\right] \\
\langle \lambda^{a}_{\mu}(-k)\lambda^{b}_{\nu}(k) \rangle &=&  \delta^{ab}\left\{\left[\frac{\hat{m}^2}{k^4}\frac{\hat{m}^2(k^2+\hat{M}^2)+2\gamma^4g^2N}{(k^2+\hat{M}^2)(k^2+\hat{m}^2)+2\gamma^4g^2N}+\frac{2\gamma^2gN}{k^4(k^2+\hat{M}^2)}\right.\right.\nonumber\\
&\times& \left.\left.\frac{\hat{m}^2\gamma^2g(k^2+\hat{M}^2)+2\gamma^6g^3N}{(k^2+\hat{M}^2)(k^2+\hat{m}^2)+2\gamma^4g^2N}\right]\mathcal{P}_{\mu\nu}-\frac{2\gamma^4g^2N}{k^4(k^2+\hat{M}^2)}\delta_{\mu\nu}\right\}\nonumber\\
\\
\langle h^{a}_{\mu}(-k)h^{b}_{\nu}(k)\rangle &=& \delta^{ab}\frac{\alpha}{k^2}\mathcal{L}_{\mu\nu}\\
\langle b^a(-k)b^b(k) \rangle &=& 0 \label{bbproprgzlcg1}\\
\langle b^a(-k)\varphi^{bc}_{\mu}(k) \rangle &=& 0 \\
\langle b^a(-k)\bar{\varphi}^{bc}_{\mu}(k) \rangle &=& 0 \\
\langle b^a(-k)\lambda^{b}_{\mu}(k)\rangle &=& 0 \\
\langle \varphi^{ab}_{\mu}(-k)h^{c}_{\nu}(k) \rangle &=& 0\\
\langle \bar{\varphi}^{ab}_{\mu}(-k)h^{c}_{\nu}(k) \rangle &=& 0 \\
\langle \lambda^{a}_{\mu}(-k)h^{b}_{\nu}(k)\rangle &=& 0
\label{prop3}
\end{eqnarray}

\noindent Clearly, in this situation we are dealing with more fields, introduced in the localization procedure, and the list of propagators is larger than in the Landau gauge. However, an important difference with Landau gauge shows up. In the standard RGZ for the Landau gauge, the propagator for the Nakanishi-Lautrup field is

\begin{equation}
\langle b^a(-k)b^b(k)\rangle^{\mathrm{RGZ}}_{\mathrm{Landau}} = \delta^{ab}\frac{2\gamma^4g^2N}{k^2(k^2+M^2)}\,,
\label{bb1}
\end{equation}

\noindent which is clearly different from (\ref{bbproprgzlcg1}) at $\alpha=0$. In fact, this incompatibility is crucial as a motivation for the reformulation we will present in the next chapter. As is well-known, in a theory which enjoys BRST symmetry, the Nakanishi-Lautrup two-point function is zero. The reason is that it is nothing but the expectation value of a BRST exact quantity, namely

\begin{equation}
\langle b^a(-k)b^b(k)\rangle = \langle s\left(\bar{c}^a(-k)b^b(k)\right)\rangle = 0\,.
\label{bb2}
\end{equation}

\noindent Eq.(\ref{bb1}) is precisely a signal of BRST breaking (a good check is that we recover a vanishing propagator when $\gamma=0$ in (\ref{bb1})). The origin of these different results stands for the fact that in the Landau gauge the complete field $A_{\mu}$ is transverse on-shell, while in LCG, this is not true. Therefore, the discrimination of the field which enters in the horizon function (the full field or rather its transverse components) is essential to the propagators computations. An unifying framework which solves this ``apparent" inconsistency and reformulates the RGZ action which enjoys a novel BRST symmetry is presented in the next chapter. 

\chapter{A non-perturbative BRST symmetry} \label{nonpBRSTRGZ}

In the last chapters we have presented a concrete way to eliminate a certain class of Gribov copies from the path integral domain. Such elimination can be viewed as an improvement of the Faddeev-Popov gauge fixing procedure, which is not completely efficient in the IR regime. An important consequence of the elimination of gauge copies, as demonstrated for Landau and linear covariant gauges, is the soft breaking of the BRST symmetry. Essentially, the breaking term comes with a mass parameter, the Gribov parameter, which is the agent responsible for making the break explicit in the IR, while it vanishes in the UV. With this, we note the transition of perturbative (UV) to non-perturbative (IR) sectors in pure Yang-Mills theories is realized by a soft breaking of BRST in the (R)GZ context. 

Since the BRST symmetry is an immediate outcome of the Faddeev-Popov procedure, it is rather simple to accept that a modification of such method will produce some effect on this symmetry. However, a natural question is whether a modification of the BRST transformations is possible in such a way that we incorporate the new effects of the elimination of copies and obtain a true symmetry of the (R)GZ action. In this chapter, we discuss some of these points and we present a set of BRST transformations which corresponds to an exact symmetry of the (R)GZ action in the Landau gauge and the nilpotency of the BRST operator is preserved. To introduce it though, we need to reformulate the (R)GZ action using appropriate variables. Before addressing this issue, we will review some important progress obtained in the understanding of the BRST breaking in the (R)GZ context, not only for completeness, but mainly to clarify where complications show up in the standard formulation of the (R)GZ action. This \textit{non-perturbative} BRST symmetry was introduced in \cite{Capri:2015ixa} and was named \textit{non-perturbative} due to the fact that the proposed modifications are non-vanishing at the IR regime, leaving the perturbative regime with the standard BRST symmetry untouched. 

Let us start the logical trip to the non-perturbative BRST symmetry with a discussion on the origin of the BRST breaking and then present the paths pursued to the construction of a new BRST symmetry. We underline that it is not our intention to provide here a review of all attempts to the understanding of the BRST breaking. We refer the reader to \cite{Maggiore:1993wq,Schaden:2014bea,Schaden:2015uua,Serreau:2012cg,Serreau:2013ila,Serreau:2015yna} for a partial list of different perspective on this issue. 

\section{Why the GZ action breaks BRST?}

In Subsect.~\ref{brstbreakingGZ} we presented a qualitative argument why BRST should be broken in the presence of a boundary (horizon) in the configuration space. Inhere, we present this argument more carefully as in \cite{Dudal:2008sp}. Also, for simplicity, we construct the argument in the Landau gauge. Let us prove a very simple statement,

\begin{stmt}
 Given a gauge field $A^{a}_{\mu}\,\in\,\Omega$, then the infinitesimal gauge transformed field $\tilde{A}^{a}_{\mu}=A^{a}_{\mu}-D^{ab}_{\mu}\xi^b$, with $\xi$ the infinitesimal parameter of the gauge transformation, does not belong to $\Omega$.
\label{stmt1}
\end{stmt}

\begin{proof}
The proof is straightforward: Let us assume that $\tilde{A}^{a}_{\mu}\,\in\,\Omega$. Then, since it is transverse by assumption (since we are assuming the Landau gauge condition),

\begin{equation}
\partial_{\mu}\tilde{A}^{a}_{\mu}=\partial_{\mu}A^{a}_{\mu}-\partial_{\mu}D^{ab}_{\mu}(A)\xi^{b}=0\,\,\Rightarrow\,\,-\partial_{\mu}D^{ab}_{\mu}(A)\xi^{b}=0\,,
\label{npbrst1}
\end{equation} 

\noindent which obviously contradicts the assumption that $A^{a}_{\mu}\,\in\,\Omega$ and then $-\partial_{\mu}D^{ab}_{\mu}(A)>0$.    $\blacksquare$
\end{proof}

To supplement Statement~\ref{stmt1}, we also refer to a theorem enunciated and proved by Gribov \cite{Gribov:1977wm,Sobreiro:2005ec},

\begin{stmt}
For a configuration $A^a_{\mu}$ close\footnote{By close to the horizon we mean $A^{a}_{\mu}=C^{a}_{\mu}+a^{a}_{\mu}$, where $C^{a}_{\mu}\,\in\,\partial\Omega$ and $a^{a}_{\mu}$ is an infinitesimal perturbation.} to the horizon $\partial\Omega$ there is a copy $\tilde{A}^{a}_{\mu}$ also close to the horizon, but located on the other side\footnote{We remember that the Gribov region $\Omega$ is bounded in every direction and, therefore, speaking about ``the other side" of the Gribov horizon. } of $A^{a}_{\mu}$.
\label{stmt2}
\end{stmt} 

From these results, we can use the fact that an infinitesimal gauge transformation for the gauge field is formally identical to the BRST transformation, just replacing the infinitesimal gauge parameter by the ghost field. Therefore, following Statements~1 and 2, we automatically conclude that a BRST transformation of a field configuration which lies inside $\Omega$ results in a configuration outside $\Omega$. Since the construction of the GZ action is based on the restriction to the region $\Omega$, the breaking of BRST symmetry is unavoidable. 

\section{The failure of a local modification of the standard BRST}

In the last section, we showed why the BRST symmetry is broken due to the introduction of a boundary $\partial\Omega$ in the configuration space. As discussed in the beginning of this chapter, it is reasonable to conceive a modification of the standard BRST transformations in such a way that we have a symmetry for the GZ action. What we will show in this section is that a \textit{local} modification in the standard formulation of the (R)GZ action is not viable. 

To show it, we need first to characterize the form of the modification:

\begin{itemize}

\item It should be local in the fields;

\item Whenever we take the limit $\gamma\rightarrow 0$, we should recover the standard BRST symmetry.

\end{itemize}

\noindent The first requirement is a working hypothesis, while the second is essential for consistency with the UV behavior. Also, the second requirement implies the modification should contain some $\gamma$ dependence since it must feel the limit $\gamma\rightarrow 0$. However, we have to keep in mind $\gamma$ has mass dimension which constrains the combination of fields we can insert as possible modifications to the BRST transformations. Hence, the possible modification of the BRST operator is

\begin{equation}
s'=s+\tilde{s}_{\gamma}\,,
\label{npbrst2}
\end{equation}

\noindent where $\tilde{s}_{\gamma}$ corresponds to $\gamma$-dependent terms. Considering the dimensionality of the set of fields\footnote{See Tables~\ref{table2} and \ref{tablea1}} $(A^{a}_{\mu},c^a,\bar{c}^a,b^a,\varphi^{ab}_{\mu},\bar{\varphi}^{ab}_{\mu},\omega^{ab}_{\mu},\bar{\omega}^{ab}_{\mu})$, the fact that the BRST operator has ghost number one and does not alter the field dimensions, Lorentz covariance and color group structure, there is no viable local $\tilde{s}_{\gamma}$ modification containing $\gamma$. This implies there is no possible local modification of $s$ in such way that we recover $s$ in the UV limit. This implies we could try to evade the first assumption of \textit{locality} on the modification and provide a non-local modification of the BRST transformations. In the next section we explore this possibility.

\section{Non-local modification of the standard BRST}

In \cite{Sorella:2009vt} and also in \cite{Kondo:2009qz} a different possibility to modify the standard BRST transformations was considered. Since a local modification is forbidden by the argument of previous section, we should not exclude a non-local modification. Also, we have at our disposal the non-local operator\footnote{For concreteness, we will restrict ourselves to the Landau gauge.} $-(\partial D)^{-1}$ which is well-defined due to the restriction of the path integral to the Gribov region $\Omega$. In this thesis we will focus on the presentation of the results presented in \cite{Sorella:2009vt} due to reasons that we will clarify later on. 

Let us consider the Gribov-Zwanziger action in the Landau gauge as presented in (\ref{2.54}). For this action we will demand the auxiliary fields transform as BRST doublets, namely, 

\begin{eqnarray}
s\varphi^{ab}_{\mu}&=&\omega^{ab}_{\mu}\,,\,\,\,\,\,\,\,\,\,\,\,\, s\omega^{ab}_{\mu}=0\nonumber\\
s\bar{\omega}^{ab}_{\mu}&=&\bar{\varphi}^{ab}_{\mu}\,,\,\,\,\,\,\,\,\,\,\,\,\, s\bar{\varphi}^{ab}_{\mu}=0\,.
\label{npbrst3}
\end{eqnarray}

\noindent We note this is a bit different from what is done in (\ref{2.58}) and (\ref{2.59}) where the BRST transformations are defined with respect to a shifted $\omega$ field. However, we will not perform this shift here and still demand (\ref{npbrst3}). Acting with the BRST operator on the Gribov-Zwanziger action (\ref{2.54}), we obtain

\begin{eqnarray}
sS_{\mathrm{GZ}}&=&\int d^dx~gf^{adb}(D^{de}_{\nu}c^e)\left[(\partial_{\nu}\bar{\varphi}^{ac}_{\mu})\varphi^{bc}_{\mu}-(\partial_{\nu}\bar{\omega}^{ac}_{\mu})\omega^{bc}_{\mu}\right]-\gamma^{2}g\int d^dx~\left(f^{abc}(-D^{ad}_{\mu}c^d)(\varphi\right.\nonumber\\
&+&\bar{\varphi})^{bc}_{\mu}+\left.f^{abc}A^{a}_{\mu}\omega^{bc}_{\mu}\right)\,.
\label{npbrst4}
\end{eqnarray} 

\noindent Expression (\ref{npbrst4}) contains terms with and without the Gribov parameter $\gamma$. Let us note, however, the term which does not contain the Gribov parameter is a BRST exact term,

\begin{equation}
\int d^dx~gf^{adb}(D^{de}_{\nu}c^e)\left[(\partial_{\nu}\bar{\varphi}^{ac}_{\mu})\varphi^{bc}_{\mu}-(\partial_{\nu}\bar{\omega}^{ac}_{\mu})\omega^{bc}_{\mu}\right]=-s\int d^dx~gf^{adb}(D^{de}_{\nu}c^e)(\partial_{\nu}\bar{\omega}^{ac}_{\mu})\varphi^{bc}_{\mu}\,,
\label{npbrst5}
\end{equation}

\noindent which implies it is not a genuine breaking. In particular, if we set $\gamma=0$, it is clear that the auxiliary field sector in expression (\ref{2.54}) can be integrated giving an unity, and preserving Yang-Mills gauge fixed action dynamics. The physical content is preserved when $\gamma=0$ due to the fact that the auxiliary fields are introduced as BRST doublets. On the other hand, the breaking term with the Gribov parameter in eq.(\ref{npbrst4}) is not a BRST exact form and thus is a truly breaking term. After these considerations, we can make use of the well-defined inverse of the Faddeev-Popov operator and write the following expressions,

\begin{equation}
\frac{\delta S_{\mathrm{GZ}}}{\delta \bar{c}^a}=-\EuScript{M}^{ab}c^b\,\,\,\,\,\,\Rightarrow\,\,\,\,\,\,c^a(x)=-\int d^dy\left[\EuScript{M}^{-1}\right]^{ab}(x,y)\frac{\delta S_{\mathrm{GZ}}}{\delta \bar{c}^b(y)}
\label{npbrst6}
\end{equation}

\noindent and

\begin{equation}
\frac{\delta S_{\mathrm{GZ}}}{\delta \bar{\omega}^{ac}_{\mu}}=\EuScript{M}^{ab}\omega^{bc}_{\mu}\,\,\,\,\,\,\Rightarrow\,\,\,\,\,\,\omega^{ac}_{\mu}(x)=\int d^dy\left[\EuScript{M}^{-1}\right]^{ab}(x,y)\frac{\delta S_{\mathrm{GZ}}}{\delta \bar{\omega}^{bc}_{\mu}(y)}\,.
\label{npbrst7}
\end{equation}

\noindent We note the fact the Faddeev-Popov operator is well-defined is crucial to write eqs.(\ref{npbrst6}) and (\ref{npbrst7}). After some algebraic gymnastics, we obtain

\begin{equation}
sS_{\mathrm{GZ}}=\int d^dx\left(-(D^{ab}_{\mu}\Lambda^b_{\mu})_x\left[\EuScript{M}^{-1}\right]^{ad}_{xy}\right)\frac{\delta S_{\mathrm{GZ}}}{\delta \bar{c}^{d}(y)}-\gamma^2 g\int d^dx~f^{abc}A^{a}_{\mu}\left[\EuScript{M}^{-1}\right]^{bd}_{xy}\frac{\delta S_{\mathrm{GZ}}}{\delta \bar{\omega}^{dc}_{\mu}(y)}\,,
\label{npbrst8}
\end{equation}

\noindent where we used the shorthand notation

\begin{equation}
\left[\EuScript{M}^{-1}\right]^{ad}_{xy}\left(\ldots\right)=\int d^dy \left[\EuScript{M}^{-1}\right]^{ad}(x,y)\left(\ldots\right)\,.
\label{npbrst9}
\end{equation}

\noindent and 

\begin{equation}
\Lambda^{b}_{\mu}=-gf^{abd}((\partial_{\mu}\bar{\varphi}^{ac}_{\nu})\varphi^{dc}_{\nu}-(\partial_{\mu}\bar{\omega}^{ac}_{\nu})\omega^{dc}_{\nu})-\gamma^2gf^{bac}(\varphi+\bar{\varphi})^{ac}_{\mu}\,.
\label{npbrst10}
\end{equation}

\noindent From eq.(\ref{npbrst8}), we see thus the breaking term can be recast as a contact term, namely, a term related to equations of motion. With this expression at our disposal, is easy to define a set of modified BRST transformations by

\begin{align}
\tilde{s}_{\gamma}A^{a}_{\mu}&=-D^{ab}_{\mu}c^b\,,     &&\tilde{s}_{\gamma}c^a=\frac{g}{2}f^{abc}c^bc^c\,, \nonumber\\
\tilde{s}_{\gamma}\bar{c}^a&=b^{a}+\int d^dy\left((D^{cb}_{\mu}\Lambda^b_{\mu})_y\left[\EuScript{M}^{-1}\right]^{ca}_{yx}\right)\,,     &&\tilde{s}_{\gamma}b^{a}=0\,, \nonumber\\
\tilde{s}_{\gamma}\varphi^{ab}_{\mu}&=\omega^{ab}_{\mu}\,,   &&\tilde{s}_{\gamma}\omega^{ab}_{\mu}=0\,, \nonumber\\
\tilde{s}_{\gamma}\bar{\omega}^{ab}_{\mu}&=\bar{\varphi}^{ab}_{\mu}+\gamma^2 g\int d^dy~f^{deb}A^{d}_{\mu}(y)\left[\EuScript{M}^{-1}\right]^{ea}_{yx}\,,         &&\tilde{s}_{\gamma}\bar{\varphi}^{ab}_{\mu}=0\,.
\label{npbrst11}
\end{align}

\noindent This set of transformations corresponds to a symmetry of the GZ action and the operator $\tilde{s}_{\gamma}$ is a deformation of the the standard BRST $s$ operator where $\gamma$-dependent terms are introduced. We emphasize a crucial fact, however: $\tilde{s}_{\gamma}$ defines a set of non-local transformations which corresponds to a symmetry of the GZ action and reduces to the standard BRST operator $s$ as long as $\gamma=0$, but lacks nilpotency. In summary, $\tilde{s}_{\gamma}$ is a non-local and non-nilpotent BRST symmetry for the GZ action. Although nilpotency is a very important feature of standard BRST, having a non-nilpotent symmetry is still useful for characterization of intrinsically non-perturbative Ward identities. To make this statement more precise, we note that in a BRST soft breaking theory we have, in general,

\begin{equation}
\langle s\left(\ldots\right)\rangle\neq 0\,.
\label{npbrst12}
\end{equation}

\noindent With $\tilde{s}_{\gamma}$, we can study the condition

\begin{equation}
\langle \tilde{s}_{\gamma}\Theta (x)\rangle = \langle 0 |\tilde{s}_{\gamma}\Theta (x)|0 \rangle=0\,.
\label{npbrst13}
\end{equation}

\noindent Let us choose $\Theta (x)=gf^{abc}A^{a}_{\mu}(x)\bar{\omega}^{bc}_{\mu}(x)$. Imposing condition (\ref{npbrst13}), we obtain

\begin{eqnarray}
\langle \tilde{s}_{\gamma}\left(gf^{abc}A^{a}_{\mu}(x)\bar{\omega}^{bc}_{\mu}(x)\right)\rangle &=& -\langle gf^{abc}\left(D^{ad}_{\mu}c^d\right)\bar{\omega}^{bc}_{\mu}\rangle+\langle gf^{abc}A^{a}_{\mu}\bar{\varphi}^{bc}_{\mu}\rangle \nonumber\\
&+& \gamma^2 g^2f^{abc}\int d^dy~f^{dec}\langle A^{a}_{\mu}(x)A^{d}_{\mu}(y)\left[\EuScript{M}^{-1}\right]^{eb}_{yx}\rangle\,.
\label{npbrst14}
\end{eqnarray}

\noindent Before going ahead, we highlight two discrete symmetries of the GZ action,

\begin{eqnarray}
\bar{\varphi}^{bc}_{\mu}& \rightarrow &\varphi^{bc}_{\mu}\,,\,\,\,\,\,\varphi^{bc}_{\mu}\rightarrow \bar{\varphi}^{bc}_{\mu}\nonumber\\
b^{a}& \rightarrow & b^{a} - gf^{abc}\bar{\varphi}^{bd}_{\mu}\varphi^{cd}_{\mu}\,,
\label{npbrst15}
\end{eqnarray}

\noindent and

\begin{equation}
\bar{\omega}^{bc}_{\mu}\rightarrow -\bar{\omega}^{bc}_{\mu}\,,\,\,\,\,\,\,\,\,\,\,\omega^{bc}_{\mu}\rightarrow -\omega^{bc}_{\mu}\,.
\label{npbrst16}
\end{equation}

\noindent Using symmetry (\ref{npbrst16}), we conclude the first term 

\begin{equation}
\langle gf^{abc}\left(D^{ad}_{\mu}c^d\right)\bar{\omega}^{bc}_{\mu}\rangle\,,
\label{npbrst17}
\end{equation}

\noindent automatically vanishes. Now, employing symmetry (\ref{npbrst15}), we recast eq.(\ref{npbrst14}) as 

\begin{eqnarray}
\langle \tilde{s}_{\gamma}\left(gf^{abc}A^{a}_{\mu}(x)\bar{\omega}^{bc}_{\mu}(x)\right)\rangle &=& \frac{1}{2}\langle gf^{abc}A^{a}_{\mu}\bar{\varphi}^{bc}_{\mu}\rangle + \frac{1}{2}\langle gf^{abc}A^{a}_{\mu}\varphi^{bc}_{\mu}\rangle \nonumber\\
&+& \gamma^2 \underbrace{g^2f^{abc}\int d^dy~f^{dec}\langle A^{a}_{\mu}(x)A^{d}_{\mu}(y)\left[\EuScript{M}^{-1}\right]^{eb}_{yx}\rangle}_{-\langle h_L(x)\rangle}\,,
\label{npbrst18}
\end{eqnarray}

\noindent with $h_L(x)$ being such that

\begin{equation}
H_L=\gamma^2\int d^dx~h_L(x)\,,
\label{npbrst19}
\end{equation} 

\noindent where $H_L$ is the horizon function in the Landau gauge. Imposing condition (\ref{2.49}), eq.(\ref{npbrst18}) reduces to

\begin{equation}
\langle \tilde{s}_{\gamma}\left(gf^{abc}A^{a}_{\mu}(x)\bar{\omega}^{bc}_{\mu}(x)\right)\rangle = \frac{1}{2}\langle gf^{abc}A^{a}_{\mu}\bar{\varphi}^{bc}_{\mu}\rangle + \frac{1}{2}\langle gf^{abc}A^{a}_{\mu}\varphi^{bc}_{\mu}\rangle -\gamma^2 d (N^2-1) \nonumber\\
= 0\,,
\label{npbrst20}
\end{equation}

\noindent where eq.(\ref{2.55}), the gap equation in its local form, was imposed. Therefore, we see the gap equation is compatible with the condition (\ref{npbrst13}). 

In a natural way, the modified BRST transformations can be written as

\begin{equation}
\tilde{s}_{\gamma}=s+\tilde{\delta}_{\gamma}\,,
\label{npbrst21}
\end{equation}

\noindent where $\tilde{\delta}_{\gamma}$ denotes the non-local terms introduced in (\ref{npbrst11}) to make $\tilde{s}_{\gamma}$ a symmetry of the Gribov-Zwanziger action. Hence,

\begin{equation}
\langle \tilde{s}_{\gamma}\Theta (x)\rangle=0\,\,\,\Rightarrow\,\,\,\langle s\Theta (x)\rangle=-\langle \tilde{\delta}_{\gamma}\Theta (x)\rangle\,,
\label{npbrst22}
\end{equation}

\noindent which is precisely the statement of correlation functions of (standard) BRST exact quantities may be not zero and proportional to the Gribov parameter. This is a characterization of what we expect from a theory which does not enjoy standard BRST symmetry, namely, we do not access the vacuum of the theory with the standard operator $s$. In this context, the non-vanishing dimension two condensate of the Zwanziger localizing auxiliary fields is naturally computed (and expected), 

\begin{equation}
\langle \tilde{s}_{\gamma}(\bar{\omega}^{ab}_{\mu}(x)\varphi^{ab}_{\mu}(x)) \rangle = 0\,\,\,\Rightarrow\,\,\, \langle (\bar{\varphi}^{ab}_{\mu}(x)\varphi^{ab}_{\mu}(x)-\bar{\omega}^{ab}_{\mu}(x)\omega^{ab}_{\mu}(x)) \rangle= -\langle \tilde{\delta}_{\gamma}(\bar{\omega}^{ab}_{\mu}(x)\varphi^{ab}_{\mu}(x)) \rangle
\label{npbrst23}
\end{equation}

\noindent which automatically implies

\begin{equation}
\langle (\bar{\varphi}^{ab}_{\mu}(x)\varphi^{ab}_{\mu}(x)-\bar{\omega}^{ab}_{\mu}(x)\omega^{ab}_{\mu}(x)) \rangle=-\gamma^2 gf^{deb}\int d^dy~\langle A^{d}_{\mu}(y)\left[\EuScript{M}^{-1}\right]^{ea}_{yx}\varphi^{ab}_{\mu}(x)\rangle\,.
\label{npbrst24}
\end{equation}

\noindent From eq.(\ref{npbrst24}) we can compute the value of the correlation function order by order. At one-loop order, the result was already presented in this thesis, see eq.(\ref{3.5}). 

Since the modified BRST transformations defined by the operator (\ref{npbrst21}) correspond to deformations of the standard BRST transformations with $\gamma$-dependent terms, they carry an intrinsic non-perturbative nature. Albeit an exact symmetry, the modified BRST transformations are not particularly useful for the renormalizability analysis of the GZ action since those are non-local. Nevertheless, as explicitly shown by eqs.(\ref{npbrst22}) and (\ref{npbrst24}) the operator $\tilde{s}_{\gamma}$ is useful to characterize in a practical way non-perturbative Ward identities. A natural question at this point is whether the RGZ action enjoys a similar non-local modified BRST symmetry. The answer to this question is positive and discussed already in $\cite{Sorella:2009vt}$. Since the reasoning is very similar to the GZ case, we do not present more details and simply refer the reader to $\cite{Sorella:2009vt}$.

Remarkably, it is possible to modify the (R)GZ action while keeping its physical content by means of the introduction of extra auxiliary fields to cast the modified BRST transformations (\ref{npbrst11}) in local form. This analysis was carried out in \cite{Dudal:2010hj}. Therefore, a construction of a modification of the (R)GZ action which is physically equivalent to the original one and enjoys a local modified BRST symmetry is allowed. We should remark that albeit local, the modified BRST transformations are still non-nilpotent as in the non-local picture. Also, the construction of a \textit{local} modified BRST transformation in this context is not in contradiction with the ``no-go" theorem of last section due to the fact that the new local transformations act on a larger set of fields than the standard GZ ones. The localization process, as usual, enlarges the amount of fields in the game, a fair price to pay depending on the interests.

We must note, though, the modified BRST symmetry was constructed upon Landau's gauge properties. In particular, the (R)GZ action in the Landau gauge is constructed using the very important properties of the Gribov region $\Omega$, characterized by the Faddeev-Popov operator $\EuScript{M}^{ab}=-\partial_{\mu}D^{ab}_{\mu}$ positivity. This implies the modified BRST symmetry incorporates \textit{particular} features of the Landau gauge as, for instance, the inverse of $\EuScript{M}^{ab}$. This is different from what happens with the standard BRST symmetry which is a feature of Yang-Mills action with a gauge-fixing term, this being arbitrary. This is a signal of the fact that the Gribov problem manifests itself in different ways for different gauges. In this sense, it is very natural to expect these modifications on the BRST transformations capture particular features of the chosen gauge. 

A final remark regarding the (non-)nilpotency of the modified BRST operator. Pictorially, the nilpotency of the standard BRST operator is associated with the fact that, at the perturbative level, we have a systematic way of performing the gauge fixing procedure and thus, we can explore different gauges while physical quantities are not affected. On the other hand, in the (R)GZ setting, the ``proper" gauge fixing is not systematic \textit{i.e.} it depends on particular properties of the gauge choice as discussed before. In this sense, it is not even natural to expect the nilpotency of the aforementioned modified BRST transformations. Although ``expected", the lack of nilpotency prevents us to explore the power of cohomology tools, which have been used in a vast applications in standard gauge theories. 

\section{On the construction of a nilpotent non-perturbative BRST symmetry}

Last sections illustrated how subtle the issue of BRST invariance is in the context of the (R)GZ framework. In particular, it was pointed out that a non-local deformation of the standard BRST transformations should be done in order to find a proper symmetry. Although possible, the resulting transformations are not nilpotent. At least, they can be cast in a local way via auxiliary fields to be introduced. From what was discussed, we have at least two possibilities now: Either we keep ourselves with the non-nilpotent local (or not) modified BRST symmetry or we try to achieve nilpotency through some conceptual reformulation of the (R)GZ action. Here is important to understand what we mean by ``conceptual". As we discussed in the last section, the nilpotency of the standard BRST operator is associated with our ability to gauge fix in a systematic way at the perturbative level. Hence, we can expect that a possible nilpotent modified BRST operator in the (R)GZ context must come from some analogue of what would be a ``systematic" construction of the (R)GZ action in different gauges. 

Although highly counter intuitive, \textit{it is} possible to construct such modified nilpotent operator. For this, we begin by a reformulation of the GZ action in the Landau gauge. Our starting point is the Gribov-Zwanziger action written in its non-local form, namely,

\begin{equation}
\EuScript{Z}=\int_{\Omega}\left[\EuScript{D}A\right]\mathrm{det}\left(\EuScript{M}\right)\delta\left(\partial A\right)\mathrm{e}^{-S_{\mathrm{YM}}}=\int\left[\EuScript{D}A\right]\mathrm{det}\left(\EuScript{M}\right)\delta\left(\partial A\right)\mathrm{e}^{-\left(S_{\mathrm{YM}}+\gamma^4H_L(A)-dV\gamma^4(N^2-1)\right)}\,,
\label{npbrst25}
 \end{equation}

\noindent where

\begin{equation}
H_L(A)=g^2\int d^d x d^dy~f^{abc}A^{b}_{\mu}(x)\left[\EuScript{M}^{-1}(x,y)\right]^{ad}f^{dec}A^{e}_{\mu}(y)\,.
\label{npbrst26}
\end{equation}

\noindent The novelty is the introduction of the transverse field $A^{h}_{\mu}=A^{h,a}_{\mu}T^a$, defined in Ap.~\ref{constructionAh}. This field is obtained through the minimization of the functional $\mathcal{A}^{2}_{\mathrm{min}}$ given by eq.(\ref{ah4}) along the gauge orbit of a given configuration $A^{a}_{\mu}$. It is possible to write a formal series for $A^{h,a}_{\mu}$ (see the derivation in  Ap.~\ref{constructionAh}),

\begin{eqnarray}
A^{h}_{\mu}&=&\left(\delta_{\mu\nu}-\frac{\partial_{\mu}\partial_{\nu}}{\partial^2}\right)\left(A_{\nu}-ig\left[\frac{1}{\partial^2}\partial A,A_\nu\right]+\frac{ig}{2}\left[\frac{1}{\partial^2}\partial A,\partial_{\nu}\frac{1}{\partial^2}\partial A\right]+\mathcal{O}(A^3)\right)\nonumber\\
&=&A_{\mu}-\partial_{\mu}\frac{1}{\partial^2}\partial A+ig\left[A_{\mu},\frac{1}{\partial^2}\partial A\right]-ig\frac{1}{\partial^2}\partial_{\mu}\left[A_{\alpha},\partial_{\alpha}\frac{1}{\partial^2}\partial A\right]+\frac{ig}{2}\frac{1}{\partial^2}\partial_{\mu}\left[\frac{1}{\partial^2}\partial A,\partial A\right]\nonumber\\
&+&\frac{ig}{2}\left[\frac{1}{\partial^2}\partial A,\partial_{\mu}\frac{1}{\partial^2}\partial A\right] + \mathcal{O}(A^3)\,.
\label{npbrst27}
\end{eqnarray}

\noindent From eq.(\ref{npbrst28}) we see $A^{h,a}_{\mu}$ is manifestly transverse, namely, $\partial_{\mu}A^{h,a}_{\mu}=0$. Another fundamental property of $A^{h,a}_{\mu}$ is that it is invariant order by order in $g$ under gauge transformations, see Ap.~\ref{constructionAh}. This property is crucial for the construction. As already discussed in this chapter, infinitesimal gauge transformations are formally equal to BRST transformation of the gauge field $A^{a}_{\mu}$. Therefore, this property automatically implies

\begin{equation}
sA^{h,a}_{\mu}=0\,,
\label{npbrst28}
\end{equation}

\noindent which will be heavily explored later on. Also, from eq.(\ref{npbrst27}) it is clear that disregarding the first term of the formal series which is $A^{a}_{\mu}$, all terms contain at least one factor of $\partial_{\mu}A^{a}_{\mu}$. A consequence of this property is that we are able to rewrite the horizon function as 

\begin{equation}
H(A)=H(A^h)-\int d^dx\, d^dy~R^a(x,y)(\partial_{\mu} A^a_{\mu})_y\,,
\label{npbrst29}
\end{equation}

\noindent with $R^a(x,y)$ being a formal power series of $A^{a}_{\mu}$. We note that we used the fact that we have $\partial_{\mu}A^{a}_{\mu}$ factors at our disposal in all terms but the first one of $A^{h,a}_{\mu}$. This implies the GZ action can be written as 

\begin{equation}
\tilde{S}_{\mathrm{GZ}}=S_{\mathrm{YM}}+\int d^dx\left(b^{h,a}\partial_{\mu}A^{a}_{\mu}+\bar{c}^{a}\partial_{\mu}D^{ab}_{\mu}c^b\right)+\gamma^4H(A^h)\,,
\label{npbrst30}
\end{equation}

\noindent where $b^{h,a}$ is a redefinition of the $b^a$ field with trivial Jacobian, given by

\begin{equation}
b^{h,a}=b^a-\gamma^4\int d^dy~R^a(y,x)\,.
\label{npbrst31}
\end{equation}

\noindent Now, employing the standard localization procedure of the GZ action, we obtain

\begin{eqnarray}
S_{\mathrm{GZ}}&=& S_{\mathrm{YM}}+\int d^dx\left(b^{h,a}\partial_{\mu}A^{a}_{\mu}+\bar{c}^{a}\partial_{\mu}D^{ab}_{\mu}c^{b}\right)\nonumber\\
&-&\int d^dx\left(\bar{\varphi}^{ac}_{\mu}\left[\EuScript{M}(A^h)\right]^{ab}\varphi^{bc}_{\mu}-\bar{\omega}^{ac}_{\mu}\left[\EuScript{M}(A^h)\right]^{ab}\omega^{bc}_{\mu}+g\gamma^2f^{abc}A^{h,a}_{\mu}(\varphi+\bar{\varphi})^{bc}_{\mu}\right)\,, \nonumber\\
\label{npbrst32}
\end{eqnarray}

\noindent with

\begin{equation}
\EuScript{M}^{ab}(A^h)=-\partial_{\mu}D^{ab}_{\mu}(A^h)=-\delta^{ab}\partial^{2}+gf^{abc}A^{h,c}_{\mu}\partial_{\mu}\,\,\,\mathrm{with}\,\,\,\partial_{\mu}A^{h,a}_{\mu}=0\,.
\label{npbrst33}
\end{equation}

Before going ahead, we make some important comments concerning eqs.(\ref{npbrst32}) and (\ref{npbrst33}),

\begin{itemize}

\item The action (\ref{npbrst32}) contains insertions of $A^{h}_{\mu}$ instead of simply $A_{\mu}$ in the terms with the auxiliary Zwanziger's field.

\item Due to the transversality of $A^{h}_{\mu}$, the operator $\EuScript{M}^{ab}(A^h)$ is automatically Hermitian. 

\item The action (\ref{npbrst32}) is still non-local, besides the standard Zwanziger's localization procedure was applied. The reason is that the horizon functional $H_L(A^h)$ displays two types of non-localities: First, the standard inverse of the operator $\EuScript{M}$ non-local structure (which is localized via the standard procedure). Second, the $A^h$ is also non-local. Therefore, even after the localization of the inverse of $\EuScript{M}$ we still have the non-local argument $A^h$. 

\item Using (\ref{npbrst32}) (with the standard vacuum term of the GZ action), the gap equation which is responsible to fix $\gamma$ is written as

\begin{equation}
\langle H_L(A^h) \rangle = dV(N^2-1)\,.
\label{npbrst34}
\end{equation}

This equation displays a very important property: It is manifestly gauge invariant. This is an immediate consequence of the gauge invariance of $A^h$. Although this discussion is a bit misleading at this point due to the fact that we are using the horizon function which was defined to the Landau gauge in principle, we hope it will become clear in the next chapter. 

\end{itemize}

Acting with the standard BRST operator $s$ in the GZ action (\ref{npbrst32}) we obtain

\begin{eqnarray}
sS_{\mathrm{GZ}} &=& \int d^dx\left[-\left(\gamma^4\int d^dy~R^a(y,x)\right)\partial_{\mu}A^{a}_{\mu}-b^{h,a}\partial_{\mu}D^{ab}_{\mu}(A)c^b+b^a\partial_{\mu}D^{ab}_{\mu}(A)c^b\right]\nonumber\\
&+&\int d^dx~gf^{abc}\gamma^2 A^{h,a}_{\mu}\omega^{bc}_{\mu}\,.
\label{npbrst35}
\end{eqnarray}

\noindent Proceeding as in the last section,

\begin{eqnarray}
\frac{\delta S_{\mathrm{GZ}}}{\delta {b}^{h,a}(x)}&=&\partial_{\mu}A^{a}_{\mu}(x)\nonumber\\
\frac{\delta S_{\mathrm{GZ}}}{\delta \bar{c}^{a}(x)}&=&-\EuScript{M}^{ab}(A)c^b(x)\nonumber\\
\frac{\delta S_{\mathrm{GZ}}}{\delta\bar{\omega}^{ac}_{\mu}(x)}&=&\EuScript{M}^{ab}(A^{h})\omega^{bc}_{\mu}(x)\,\,\,\Rightarrow\,\,\, \omega^{ac}_{\mu}(x)=\int d^dy\left[\EuScript{M}^{-1}(A^{h})\right]^{ab}(x,y)\frac{\delta S_{\mathrm{GZ}}}{\delta\bar{\omega}^{bc}_{\mu}(y)}\,,
\label{npbrst36}
\end{eqnarray}

\noindent and plugging into eq.(\ref{npbrst36}),

\begin{eqnarray}
sS_{\mathrm{GZ}} &=& \int d^dx\left[-\left(\gamma^4\int d^dy~sR^a(y,x)\right)\frac{\delta S_{\mathrm{GZ}}}{\delta {b}^{h,a}}-b^{h,a}\frac{\delta S_{\mathrm{GZ}}}{\delta \bar{c}^{a}}+b^a\frac{\delta S_{\mathrm{GZ}}}{\delta \bar{c}^{a}}\right]\nonumber\\
&+&\int d^dx~gf^{abc}\gamma^2 A^{h,a}_{\mu}\left[\EuScript{M}^{-1}(A^{h})\right]^{ac}_{xy}\frac{\delta S_{\mathrm{GZ}}}{\delta\bar{\omega}^{cb}_{\mu}(y)}\nonumber\\
&=& \int d^dx\left[-\left(\gamma^4\int d^dy~sR^a(y,x)\right)\frac{\delta S_{\mathrm{GZ}}}{\delta {b}^{h,a}}+\left(\gamma^4\int d^dy~R^a(y,x)\right)\frac{\delta S_{\mathrm{GZ}}}{\delta \bar{c}^{a}}\right]\nonumber\\
&+&\int d^dx~gf^{abc}\gamma^2 A^{h,a}_{\mu}\left[\EuScript{M}^{-1}(A^{h})\right]^{bd}_{xy}\frac{\delta S_{\mathrm{GZ}}}{\delta\bar{\omega}^{dc}_{\mu}(y)}\,,
\label{npbrst37}
\end{eqnarray}

\noindent we have the breaking written as a contact term. From eq.(\ref{npbrst37}) it is immediate to read off the appropriate modifications for the BRST transformations to construct a symmetry. Therefore,

\begin{align}
s_{\gamma^2}A^{a}_{\mu}&=-D^{ab}_{\mu}c^b\,,     &&s_{\gamma^2}c^a=\frac{g}{2}f^{abc}c^bc^c\,, \nonumber\\
s_{\gamma^2}\bar{c}^a&=b^{h,a}\,,     &&s_{\gamma^2}b^{h,a}=0\,, \nonumber\\
s_{\gamma^2}\varphi^{ab}_{\mu}&=\omega^{ab}_{\mu}\,,   &&s_{\gamma^2}\omega^{ab}_{\mu}=0\,, \nonumber\\
s_{\gamma^2}\bar{\omega}^{ab}_{\mu}&=\bar{\varphi}^{ab}_{\mu}-\gamma^2gf^{cdb}\int d^dy~A^{h,c}_{\mu}(y)\left[\EuScript{M}^{-1}(A^h)\right]^{da}_{yx}\,,         &&s_{\gamma^2}\bar{\varphi}^{ab}_{\mu}=0\,, 
\label{npbrst38}
\end{align}
correspond to a modified BRST symmetry for the GZ action,

\begin{equation}
s_{\gamma^2}S_{\mathrm{GZ}}=0\,.
\label{npbrst39}
\end{equation}

\noindent The transformations defined by $s_{\gamma^2}$ can be decomposed as

\begin{equation}
s_{\gamma^2}=s+\delta_{\gamma^2}\,,
\label{npbrst40}
\end{equation}

\noindent with 

\begin{align}
sA^{a}_{\mu}&=-D^{ab}_{\mu}c^b\,,     &&sc^a=\frac{g}{2}f^{abc}c^bc^c\,, \nonumber\\
s\bar{c}^a&=b^{a}\,,     &&sb^{h,a}=-\gamma^4\int d^dy~sR^a(y,x)\,, \nonumber\\
s\varphi^{ab}_{\mu}&=\omega^{ab}_{\mu}\,,   &&s\omega^{ab}_{\mu}=0\,, \nonumber\\
s\bar{\omega}^{ab}_{\mu}&=\bar{\varphi}^{ab}_{\mu}\,,         &&s\bar{\varphi}^{ab}_{\mu}=0\,, 
\label{npbrst41}
\end{align}

\noindent and

\begin{align}
\delta_{\gamma^2}A^{a}_{\mu}&=0\,,     &&\delta_{\gamma^2}c^a=0\,, \nonumber\\
\delta_{\gamma^2}\bar{c}^a&=-\gamma^4\int d^dy~R^a(y,x)\,,     &&\delta_{\gamma^2}b^{h,a}=\gamma^4\int d^dy~sR^a(y,x)\,, \nonumber\\
\delta_{\gamma^2}\varphi^{ab}_{\mu}&=0\,,   &&\delta_{\gamma^2}\omega^{ab}_{\mu}=0\,, \nonumber\\
\delta_{\gamma^2}\bar{\omega}^{ab}_{\mu}&=-\gamma^2gf^{cdb}\int d^dy~A^{h,c}_{\mu}(y)\left[\EuScript{M}^{-1}(A^h)\right]^{da}_{yx}\,,         &&\delta_{\gamma^2}\bar{\varphi}^{ab}_{\mu}=0\,. 
\label{npbrst42}
\end{align}

\noindent An explicit property of the operator $s_{\gamma^2}$ is that as long as the Gribov parameter is set to zero, we recover the standard BRST operator \textit{i.e.}

\begin{stmt} \label{sgammagoestos}
If the Gribov parameter is set to zero $\gamma\rightarrow 0$, then $s_{\gamma^2}\rightarrow s$.
\end{stmt}

\noindent Statement~\ref{sgammagoestos} is crucial for a consistent treatment of the GZ action. As was presented in this thesis, the limit $\gamma\rightarrow 0$ should bring us back to the full standard Faddeev-Popov theory.

As in the previous section, the deformation of the usual BRST transformation is such that it contains the Gribov parameter. This fact guarantees the smooth reduction to the standard BRST transformations in the perturbative (UV) regime. On the other hand, since these modified BRST transformations bring non-perturbative information due to the nature of $\gamma$. Also, transformations (\ref{npbrst38}) are non-local as those presented in the last section. Nevertheless, the use of the gauge invariant field $A^{h}_{\mu}$ allows a profound difference with respect to the set (\ref{npbrst11}): The operator $s_{\gamma^2}$ is nilpotent,

\begin{equation}
s^2_{\gamma^2}=0\,.
\label{npbrst43}
\end{equation}

\noindent The proof is automatic if we use the fact that $s_{\gamma^2}A^{h}_{\mu}=0$. Another interesting feature which is also easily checked is that the modification $\delta_{\gamma^2}$ is nilpotent alone, namely, 

\begin{equation}
\delta^2_{\gamma^2}=0\,.
\label{npbrst44}
\end{equation}

\noindent Due to the fact that the modified BRST transformations (\ref{npbrst38}) are an exact symmetry of the Gribov-Zwanziger action \textit{and} the operator $s_{\gamma^2}$ is nilpotent, an outstanding property of the standard BRST operator, we call them, from now on, as \textit{non-perturbative BRST symmetry} or \textit{non-perturbative BRST transformations}\footnote{As discussed, they are called ``non-perturbative" due to the presence of $\gamma$.}. As a summary of the properties enjoyed by the non-perturbative BRST operator, we write

\begin{equation}
s^{2}_{\gamma^2}=s^2=\delta^{2}_{\gamma^2}=0\,\,\,\Rightarrow\,\,\, \left\{s,\delta_{\gamma^2}\right\}=0\,,
\label{npbrst45}
\end{equation}

\noindent where $\left\{\cdot,\cdot\right\}$ stands for the anticommutator. 

With the nilpotent non-perturbative BRST operator $s_{\gamma^2}$, a very important check is that

\begin{equation}
s_{\gamma^2}\frac{\partial S_{\mathrm{GZ}}}{\partial\gamma^2}=f^{abc}A^{h,a}_{\mu}\omega^{bc}_{\mu}\neq 0\,\,\,\,\Rightarrow\,\,\, \frac{\partial S_{\mathrm{GZ}}}{\partial\gamma^2}\neq s_{\gamma^2}(\ldots)\,,
\label{npbrst46}
\end{equation}

\noindent which indicates the Gribov parameter $\gamma$ is not akin to a gauge parameter. This conclusion was already achieved through the standard BRST soft breaking. For a more detailed discussion on the physicality of parameters introduced using standard BRST soft breaking, we refer to \cite{Baulieu:2008fy}. Expression (\ref{npbrst46}) is supported by the gauge invariant horizon condition (\ref{npbrst34}), which gives an explicit physical character to $\gamma$. Hence, is expected that $\gamma$ will enter physical quantities as gauge invariant colorless operators. Also, we have at our disposal an exact and nilpotent BRST symmetry and, as a consequence, the cohomology toolbox which enables the classification of quantum extension of classical gauge invariant operators. A final remark concerning the properties of the non-perturbative BRST operator, we write the GZ action in Landau gauge as

\begin{eqnarray}
S_{\mathrm{GZ}}&=& S_{\mathrm{YM}}+s_{\gamma^2}\int d^dx~\bar{c}^{a}\partial_{\mu}A^{a}_{\mu}\nonumber\\
&-&\int d^dx\left(\bar{\varphi}^{ac}_{\mu}\left[\EuScript{M}(A^h)\right]^{ab}\varphi^{bc}_{\mu}-\bar{\omega}^{ac}_{\mu}\left[\EuScript{M}(A^h)\right]^{ab}\omega^{bc}_{\mu}+g\gamma^2f^{abc}A^{h,a}_{\mu}(\varphi+\bar{\varphi})^{bc}_{\mu}\right)\,, \nonumber\\
\label{npbrst47}
\end{eqnarray}
 
\noindent where is explicitly seen that the second line (\ref{npbrst47}) is written in terms of the gauge invariant variable $A^h_{\mu}$, while the first line contains a \textit{non-perturbative} gauge-fixing term written as a non-perturbative BRST exact term. 

From eq.(\ref{npbrst40}), it is very easy to check that

\begin{equation}
\langle s_{\gamma^2}\Theta (x) \rangle=0\,\,\,\,\Rightarrow\,\,\,\langle s\Theta (x) \rangle=-\langle \delta_{\gamma^2}\Theta (x) \rangle\,,
\label{npbrst48}
\end{equation}

\noindent where the non-trivial value of a standard BRST exact correlation function is equal to the non-trivial value of an exact $\delta_{\gamma^2}$ variation, which is necessarily $\gamma$-dependent. This ensures the standard BRST breaking is soft. 

In summary, we have an exact nilpotent BRST operator which defines a symmetry of the GZ action in the Landau gauge and encodes non-perturbative information due to the Gribov parameter. Interesting enough, the construction of this symmetry is based on the introduction of the transverse gauge invariant field $A^{h}_{\mu}$. The structure of the reformulated GZ action in terms of $A^{h}_{\mu}$ as in eq.(\ref{npbrst47}) suggests a ``non-perturbative BRST quantization" procedure which will be elaborated in the next chapter. A final important point is that albeit nilpotent the non-perturbative BRST transformations are non-local as the GZ action itself, when written in terms of $A^h_{\mu}$, a non-local formal power series of $A_{\mu}$. It is highly desirable to cast the reformulated GZ action and the non-perturbative BRST transformations in local form. This will be explored later on in this thesis. 

\chapter{Linear covariant gauges revisited}\label{LCGrevisited}

In Ch.~\ref{ch.5} we presented a first attempt to construct a local action akin to the (R)GZ action to general linear covariant gauges. As we have seen explicitly, this construction is based on the restriction of the path integral to a region $\Omega_{\mathrm{LCG}}$. However, the construction automatically inserts the transverse component of the gauge field in the horizon function. In this case, the longitudinal sector of the gauge field is not vanishing as in Landau gauge and, therefore, including the full gauge field or the transverse field component is a complete different story. As a result, for example, we have checked the propagator for the auxiliary $b$ field vanishes at the tree-level, while in the standard GZ setting for the Landau gauge, this is not the case. It means the deformation of the LCG GZ action to the Landau ones for $\alpha\rightarrow 0$ is not straightforward. In this chapter, inspired by the non-perturbative BRST symmetry introduced in Ch.~\ref{nonpBRSTRGZ}, we promote our proposal for LCG GZ action to a non-perturbative BRST invariant LCG GZ action. We show how this proposal takes into account the removal of zero-modes of the Faddeev-Popov operator in LCG and how the non-perturbative BRST transformations keep the consistency with the non-perturbative BRST invariant GZ action in the Landau gauge. 

\section{The non-perturbative BRST invariant GZ action in LCG}

We define the LCG GZ action in a manifest non-perturbative BRST invariant way by

\begin{eqnarray}
S^{\mathrm{LCG}}_{\mathrm{GZ}}&=& S_{\mathrm{YM}}+s_{\gamma^2}\int d^dx~\bar{c}^{a}\left(\partial_{\mu}A^{a}_{\mu}-\frac{\alpha}{2}b^{h,a}\right)\nonumber\\
&-&\int d^dx\left(\bar{\varphi}^{ac}_{\mu}\left[\EuScript{M}(A^h)\right]^{ab}\varphi^{bc}_{\mu}-\bar{\omega}^{ac}_{\mu}\left[\EuScript{M}(A^h)\right]^{ab}\omega^{bc}_{\mu}+g\gamma^2f^{abc}A^{h,a}_{\mu}(\varphi+\bar{\varphi})^{bc}_{\mu}\right)\nonumber\\
&=&  S_{\mathrm{YM}} + \int d^dx\left(b^{h,a}\left(\partial_{\mu}A^{a}_{\mu}-\frac{\alpha}{2}b^{h,a}\right)+\bar{c}^{a}\partial_{\mu}D^{ab}_{\mu}c^{b}\right)\nonumber\\
&-&\int d^dx\left(\bar{\varphi}^{ac}_{\mu}\left[\EuScript{M}(A^h)\right]^{ab}\varphi^{bc}_{\mu}-\bar{\omega}^{ac}_{\mu}\left[\EuScript{M}(A^h)\right]^{ab}\omega^{bc}_{\mu}+g\gamma^2f^{abc}A^{h,a}_{\mu}(\varphi+\bar{\varphi})^{bc}_{\mu}\right)\,,\nonumber\\
\label{rlcg1}
\end{eqnarray}

\noindent where the gauge condition for linear covariant gauges is defined as

\begin{equation}
\partial_{\mu}A^{a}_{\mu}-\alpha b^{h,a}=0\,,
\label{rlcg2}
\end{equation}

\noindent with $\alpha$ being an arbitrary positive parameter. The action (\ref{rlcg1}) is manifestly invariant due to the nilpotency of the operator $s_{\gamma^2}$. We must note this action is exactly the same as in Landau gauge with the obvious change in the gauge fixing part only. However, since this part is an exact non-perturbative BRST term, moving from Landau to LCG does not spoil the non-perturbative BRST invariance. In this sense we can speak about a ``non-perturbative BRST" quantization. From eq.(\ref{rlcg1}) and (\ref{npbrst47}) it is easy to check

\begin{equation}
S^{\mathrm{LCG}}_{\mathrm{GZ}}\Big|_{\alpha=0}=S^{\mathrm{Landau}}_{\mathrm{GZ}}\,.
\label{rlcg3}
\end{equation}

\noindent From action (\ref{rlcg1}), we can derive the horizon condition/gap equation, which is responsible to fix $\gamma$,

\begin{equation}
\langle H(A^h) \rangle = dV(N^2-1)\,,
\label{rlcg4}
\end{equation}

\noindent with 

\begin{equation}
 H(A^h)=g^2\int d^dx d^dy f^{abc}A^{h,a}_{\mu}(x)\left[\EuScript{M}^{-1}(A^h)\right]^{be}(x,y)f^{dec}A^{h,d}_{\mu}(y)\,,
\label{rlcg5}
\end{equation}

\noindent which is the same expression used in the gap equation (\ref{npbrst34}). Since we are using a general $\alpha$, we drop the subscript $L$ used before. We see thus the gauge invariance of (\ref{rlcg5}) automatically implies independence from $\alpha$ of the Gribov parameter $\gamma$. This is a very important requirement since as discussed in Ch.~\ref{nonpBRSTRGZ} this parameter is physical. 

Another prerequisite action (\ref{rlcg1}) should satisfy is that for $\gamma\rightarrow 0$, we should recover the standard Faddeev-Popov theory. From expression (\ref{rlcg1}) this is evident, because when $\gamma\rightarrow 0$, the resulting part with auxiliary fields can be immediately integrated giving an unity. Therefore, 

\begin{equation}
S^{\mathrm{LCG}}_{\mathrm{GZ}}\Big|_{\gamma=0}=S^{\mathrm{LCG}}_{\mathrm{FP}}=S_{\mathrm{YM}}+\int d^dx\left(b^{h,a}\partial_{\mu}A^{a}_{\mu}-\frac{\alpha}{2}b^{h,a}b^{h,a}+\bar{c}^{a}\partial_{\mu}D^{ab}_{\mu}(A)c^{b}\right)\,.
\label{rlcg6}
\end{equation}

\noindent fulfilling the desired requirement. 

\section{A ``geometrical picture"}

In Ch.~\ref{ch.5}, we proved that infinitesimal Gribov copies in the LCG are associated with zero-modes of the Faddeev-Popov operator, namely

\begin{equation}
\EuScript{M}^{ab}_{\mathrm{LCG}}\xi^{b}=-\partial_{\mu}\xi^{a}+gf^{abc}(\partial_{\mu}A^{c}_{\mu})\xi^{b}+gf^{abc}A^{c}_{\mu}\partial_{\mu}\xi^{b}=0\,,
\label{rlcg7}
\end{equation} 

\noindent and that this operator is not Hermitian. This fact makes the construction of Gribov region out of $\EuScript{M}^{ab}_{\mathrm{LCG}}$ very difficult and as a proposal, we have introduced a subsidiary condition with respect to a \textit{Hermitian} operator $\EuScript{M}^{ab}_{\mathrm{LCG}}(A^T)$, where $A^T$ is the transverse component of the gauge field, to construct the candidate Gribov region $\Omega_{\mathrm{LCG}}$. Theorem~\ref{thma} ensures that demanding the positivity of $\EuScript{M}^{ab}_{\mathrm{LCG}}(A^T)$ automatically implies the removal of zero-modes of $\EuScript{M}^{ab}_{\mathrm{LCG}}$.

The GZ action in LCG (\ref{rlcg1}) compatible with the non-perturbative BRST symmetry imposes the restriction of the path integral domain to the region defined by

\begin{equation}
\EuScript{M}^{ab}(A^h)>0\,,
\label{rlcg8}
\end{equation}

\noindent where we remind $\EuScript{M}^{ab}(A^h)$ \textit{is} Hermitian due to the transversality of $A^h_{\mu}$. Therefore, in the non-perturbative BRST invariant formulation, the Gribov region for LCG is defined as

\begin{definition} \label{omegah}
The Gribov region $\Omega^{h}_{\mathrm{LCG}}$ which is free from regular infinitesimal copies in LCG is given by
\begin{equation}
\Omega^{h}_{\mathrm{LCG}}=\left\{\EuScript{M}^{ab}(A^h)>0\,\,\Big|\,\,\partial_{\mu}A^{a}_{\mu}=\alpha b^a\right\}\,.
\label{rlcg9}
\end{equation}
\end{definition}

\noindent As we did for $\EuScript{M}^{ab}(A^T)$ in Ch.~\ref{ch.5}, we should prove condition (\ref{rlcg8}) implies the removal of zero-modes of the Faddeev-Popov operator $\EuScript{M}^{ab}_{\mathrm{LCG}}$ and as a consequence, Def.~\ref{omegah} is indeed sensible. 

\begin{theorem} \label{thma2}
If $\EuScript{M}^{ab}(A^{h})>0$, then $\EuScript{M}^{ab}_{\mathrm{LCG}}\xi^{b}=0$ is only satisfied\footnote{We are assuming, as before, that $\xi$ is a regular function.} by $\xi^a=0$.
\end{theorem}

\begin{proof}
Let us assume we have a zero-mode $\xi$ which is different from $\xi=0$. Then, by assumption,

\begin{equation}
\EuScript{M}^{ab}_{\mathrm{LCG}}\xi^{b}=0\,\,\,\Rightarrow\,\,\,-\delta^{ab}\partial^2\xi^b+f^{abc}A^{c}_{\mu}\partial_{\mu}\xi^{b}+gf^{abc}(\partial_{\mu}A^{c}_{\mu})\xi^{b}=0\,.
\label{rlcg10}
\end{equation}

\noindent Now, we can write the gauge field $A^{a}_{\mu}$ as

\begin{equation}
A^{a}_{\mu}=A^{h,a}_{\mu}+\tau^{a}_{\mu}\,\,\,\Rightarrow\,\,\,\partial_{\mu}A^{a}_{\mu}=\partial_{\mu}\tau^{a}_{\mu}=\alpha b^a\,
\label{rlcg11}
\end{equation}

\noindent which implies

\begin{eqnarray}
&-&\delta^{ab}\partial^2\xi^b+f^{abc}A^{h,c}_{\mu}\partial_{\mu}\xi^{b}+gf^{abc}(\partial_{\mu}A^{h,c}_{\mu})+gf^{abc}(\partial_{\mu}\tau^{c}_{\mu})+f^{abc}\tau^{c}_{\mu}\partial_{\mu}\xi^{b}=0\nonumber\\
&\Rightarrow& \EuScript{M}^{ab}(A^h)\xi^b=-f^{abc}\partial_{\mu}(\tau^{c}_{\mu}\xi^{b})\nonumber\\
&\Rightarrow& \xi^a= -f^{dbc}\left[\EuScript{M}^{-1}(A^h)\right]^{ad}\partial_{\mu}(\tau^{c}_{\mu}\xi^{b})\,,
\label{rlcg12}
\end{eqnarray}

\noindent where we used the fact $\EuScript{M}^{ab}(A^h)$ is positive and, thus, invertible. Now, using the fact that $\xi$ is a regular zero-mode, we can expand it in a power series like

\begin{equation}
\xi^{a}(x;\alpha)=\sum^{\infty}_{n=0}\alpha^{n}\xi^{a}_{n}(x)\,,
\label{rlcg13}
\end{equation}

\noindent and using eq.(\ref{rlcg11}), we write

\begin{equation}
\tau^{a}=\alpha\frac{1}{\partial^2}\partial_{\mu}b^a\equiv\alpha\psi^a\,,
\label{rlcg14}
\end{equation}

\noindent where expression (\ref{rlcg14}) ensures that for $\alpha=0$, $\tau^{a}=0$. Plugging eqs.(\ref{rlcg13}) and (\ref{rlcg14}) into eq.(\ref{rlcg12}), we obtain

\begin{equation}
\sum^{\infty}_{n=0}\alpha^{n}\xi^{a}_{n}(x)= -f^{dbc}\sum^{\infty}_{n=0}\alpha^{n+1}\left[\EuScript{M}^{-1}(A^h)\right]^{ad}\partial_{\mu}(\psi^{c}\xi^{b}_{n}(x))\,,
\label{rlcg15}
\end{equation}

\noindent which in analogy to the proof of Theorem~\ref{thma} implies

\begin{equation}
\xi^{a}_{n}=0\,,\,\,\,\forall n
\label{rlcg16}
\end{equation}

\noindent which contradicts our assumption that $\xi\neq 0$.    $\blacksquare$
\end{proof}

\noindent Theorem~\ref{thma2} gives support to Def.~\ref{omegah} and makes the restriction to $\Omega^{h}_{\mathrm{LCG}}$ well grounded. As presented is Ch.~\ref{ch.3}, we can construct the no-pole condition for the $\EuScript{M}^{ab}(A^h)$. This construction is simple following the steps described in Ch.~\ref{ch.3} and the fact the no-pole condition can be generalized to all orders as discussed in \cite{Capri:2012wx}. Some details can be found in \cite{Capri:2015ixa}. As a comment about $\Omega^{h}_{\mathrm{LCG}}$ is that it enjoys similar properties as the Gribov region $\Omega$ in the Landau gauge,

\begin{itemize}
\item It is bounded in all directions;

\item It is convenx.
\end{itemize}

\noindent These properties follow directly from the fact that $A^{h}_{\mu}$ is transverse. 

\section{Non-perturbative BRST-invariant RGZ action} \label{npBRSTRGZLCG1}

As discussed in Ch.~\ref{RGZch} and Ch.~\ref{ch.5} it was shown that the GZ action in the Landau and linear covariant gauges suffer from instabilities that give rise to the dynamical generation of condensates. Such effects can be taken into account by the construction of the so-called RGZ action, \cite{Dudal:2008sp,Dudal:2007cw,Gracey:2010cg}. The same issue was analyzed in the maximal Abelian gauge (already under the non-perturbative BRST framework), \cite{Capri:2015pfa} and in the Coulomb gauge \cite{Guimaraes:2015bra}. 

In Ch.~\ref{ch.5}, we constructed the RGZ action in LCG. The construction of this chapter is related to the one of Ch.~\ref{ch.5} in the approximation

\begin{equation}
A^{h,a}_{\mu}\approx A^{a}_{\mu}-\frac{\partial_{\mu}}{\partial^2}\partial_{\nu}A^{a}_{\nu}\equiv A^{Ta}_{\nu}\,,
\label{rlcg17}
\end{equation} 

\noindent where $A^{Ta}_{\mu}$ stands for the transverse component of the gauge field. We see from expression (\ref{ah20}), that the transverse $A^{Ta}_{\mu}$ component is the zeroth order part (in $g$) of $A^{h,a}_{\mu}$. In this perspective, the results obtained in Ch.~\ref{ch.5} might me seen as an approximation of the full construction with $A^{h}$. Clearly, at one-loop order, the computation using action (\ref{hor3}) or the one using expression (\ref{rlcg1}) leads to the same results for existence of the condensates. However, we should keep in mind these formulations are conceptually very different. In the present case, we are concerned with the condensates

\begin{equation}
\langle A^{h,a}_{\mu}A^{h,a}_{\mu}\rangle\,\,\,\,  \mathrm{and}\,\,\,\, \langle \bar{\varphi}^{ab}_{\mu}\varphi^{ab}_{\mu}-\bar{\omega}^{ab}_{\mu}\omega^{ab}_{\mu}\rangle\,,
\label{rlcg18}
\end{equation}

\noindent while in Ch.~\ref{ch.5}, the dimension-two gluon condensate was $\langle A^{Ta}_{\mu}A^{Ta}_{\mu}\rangle$. Explicitly, the value of the condensates (\ref{rlcg18}) can be computed by coupling such composite operators with constant sources $m$ and $J$, namely

\begin{equation}
\mathrm{e}^{-V\mathcal{E}(m,J)}=\int\left[\EuScript{D}\Phi\right]\mathrm{e}^{-(S^{\mathrm{LCG}}_{\mathrm{GZ}}+m\int d^dx~A^{h,a}_{\mu}A^{h,a}_{\mu}-J\int d^dx(\bar{\varphi}^{ab}_{\mu}\varphi^{ab}_{\mu}-\bar{\omega}^{ab}_{\mu}\omega^{ab}_{\mu}))}
\label{rlcg19}
\end{equation}

\noindent and

\begin{eqnarray}
\langle \bar{\varphi}^{ab}_{\mu}\varphi^{ab}_{\mu}-\bar{\omega}^{ab}_{\mu}\omega^{ab}_{\mu} \rangle &=& -\frac{\partial \mathcal{E}(m,J)}{\partial J}\Big|_{m=J=0}\nonumber\\
\langle A^{h,a}_{\mu}A^{h,a}_{\mu}\rangle &=& \frac{\partial\mathcal{E}(m,J)}{\partial m}\Big|_{m=J=0}\,.
\label{rlcg20}
\end{eqnarray}

\noindent At one-loop order, 

\begin{equation}
{\cal E}(m,J)=\frac{(d-1)(N^2-1)}{2}\int \frac{d^dk}{(2\pi)^d}~\mathrm{ln}\left(k^2+\frac{2\gamma^4g^2N}{k^2+J}+2m\right)-d\gamma^4(N^2-1)\,,
\label{rlcg21}
\end{equation}

\noindent which results in

\begin{equation}
\langle \bar{\varphi}^{ac}_{\mu}\varphi^{ac}_{\mu}-\bar{\omega}^{ac}_{\mu}\omega^{ac}_{\mu}\rangle = g^2\gamma^4N(N^2-1)(d-1)\int \frac{d^dk}{(2\pi)^d}\frac{1}{k^2}\frac{1}{(k^4+2g^2\gamma^4N)}
\label{rlcg22}
\end{equation}

\noindent and

\begin{equation}
\langle A^{h,a}_{\mu}A^{h,a}_{\mu}\rangle = -2g^2\gamma^4N(N^2-1)(d-1)\int\frac{d^dk}{(2\pi)^d}\frac{1}{k^2}\frac{1}{(k^4+2g^2\gamma^4N)}\,.
\label{rlcg23}
\end{equation}

From (\ref{rlcg22}) and (\ref{rlcg23}), we see the integrals are perfectly convergent in the UV and depend explicitly on $\gamma$. For $d=3,4$, these integrals are defined in the IR and correspond to well-defined quantities. Nevertheless, in $d=2$, due to the $1/k^2$ factor in the integrals, we have a non-integrable singularity which makes the condensates ill-defined. This IR pathology in $d=2$ is a typical behavior of two-dimensional theories, see \cite{Dudal:2008xd} and references therein. In this way, these results suggest such condensates should be taken into account in $d=3,4$, giving rise to a refinement of the GZ action. In $d=2$, as happens in other gauges, these condensates cannot be safely introduced as they give rise to non-integrable IR singularities. As a consequence, in $d=2$ the GZ theory does not need to be refined.  Therefore, for $d=3,4$, the RGZ action in LCG is written as

\begin{eqnarray}
S^{\mathrm{LCG}}_{\mathrm{RGZ}}&=& S_{\mathrm{YM}} + \int d^dx\left(b^{h,a}\left(\partial_{\mu}A^{a}_{\mu}-\frac{\alpha}{2}b^{h,a}\right)+\bar{c}^{a}\partial_{\mu}D^{ab}_{\mu}c^{b}\right)\nonumber\\
&-&\int d^dx\left(\bar{\varphi}^{ac}_{\mu}\left[\EuScript{M}(A^h)\right]^{ab}\varphi^{bc}_{\mu}-\bar{\omega}^{ac}_{\mu}\left[\EuScript{M}(A^h)\right]^{ab}\omega^{bc}_{\mu}+g\gamma^2f^{abc}A^{h,a}_{\mu}(\varphi+\bar{\varphi})^{bc}_{\mu}\right)\nonumber\\
&+&\frac{m^2}{2}\int d^dx~A^{h,a}_{\mu}A^{h,a}_{\mu}-M^2\int d^dx\left(\bar{\varphi}^{ac}_{\mu}\varphi^{ac}_{\mu}-\bar{\omega}^{ac}_{\mu}\omega^{ac}_{\mu}\right)\,,
\label{rlcg24}
\end{eqnarray}

\noindent while for $d=2$, the action is simply the GZ action, given by eq.(\ref{rlcg1}). Notice that $M^2\geq0$, otherwise the theory would be plagued by a tachyon in the $(\omega,\bar\omega)$-sector.

The action (\ref{rlcg24}) enjoys a non-perturbative nilpotent BRST symmetry, which is precisely the same as (\ref{npbrst38}) with the only modification of

\begin{equation}
s_{\gamma^2}\bar{\omega}^{ab}_{\mu}=\bar{\varphi}^{ab}_{\mu}-\int d^dy~g\gamma^2f^{cdb}A^{h,c}_{\mu}(y)\left(\left[\EuScript{M}(A^h)- \mathds{1} M^2 \right]^{-1}\right)^{da}_{yx}\,,
\label{rlcg25}
\end{equation}

\noindent where $\mathds{1}$ stands for the identity operator. Therefore, the RGZ action in LCG takes into account the presence of dimension-two condensates and is invariant under (\ref{npbrst38}) and (\ref{rlcg25}), a non-perturbative nilpotent BRST symmetry. This construction shows how the RGZ action is compatible with the non-perturbative nilpotent BRST symmetry proposed in Ch.~\ref{nonpBRSTRGZ}.

Besides the IR singularity of the one-loop computation of the condensates, there is an additional problem, of a more fundamental nature, that prohibits the dynamical occurrence of refinement in $d=2$. We recall here that the starting point was to avoid a large class of infinitesimal gauge copies in LCG. This was achieved by requiring that $\EuScript{M}^{ab}(A^h)>0$. For a general classical field $A^h$ we can use Wick's theorem to invert the operator $\EuScript{M}^{ab}(A^h)$. In momentum space, one finds \cite{Capri:2012wx,Capri:2015ixa}

\begin{equation}
\Braket{p | \frac{1}{\EuScript{M}^{ab}(A^h)} |p}=\mathcal{G}^{ab}(A^h, p^2)=\frac{\delta^{ab}}{N^2-1}\mathcal{G}^{cc}(A^h,p^2)=\frac{\delta^{ab}}{N^2-1}\frac{1+\sigma(A^h,p^2)}{p^2}\,.
\label{rlcg26}	
\end{equation}

\noindent At zero momentum, we find consequently \cite{Capri:2012wx}

\begin{equation}
\sigma(A^h,0)=-\frac{g^2}{Vd(N^2-1)}\int \frac{d^dk}{(2\pi)^d}\frac{d^dq}{(2\pi)^d} A_\mu^{h,ab}(-k) \left[(\EuScript{M}(A^h))^{-1}\right]^{bc}_{k-q}A_\mu^{h,ca}(q)\,,
\label{rlcg27}
\end{equation}

\noindent with $A_\mu^{h,ab}(q)\equiv f^{abc}A^{h,c}_{\mu}(q)$, and this leads to the exact identification

\begin{equation}
\sigma(A^h,0)=\frac{H(A^h)}{Vd(N^2-1)}\,.
\label{rlcg28}
\end{equation}

\noindent At the level of expectation values, we can rewrite eq.~\eqref{rlcg26} as

\begin{equation}
\mathcal{G}^h(p^2)= \braket{ \mathcal{G}^{aa}(A^h, p^2)}_{\mathrm{conn}}=\frac{1}{p^2(1-\braket{\sigma(A^h,p^2)}_{\mathrm{1PI}})},
\label{rlcg29}
\end{equation}

\noindent so that we must impose at the level of the path integral

\begin{equation}
\sigma(0)\equiv\braket{\sigma(A^h,0)}_{\mathrm{1PI}} < 1
\label{rlcg30}
\end{equation}

\noindent to ensure a positive operator\footnote{To avoid confusion, we emphasize that the quantity $\sigma(k)$ introduced in eq.~\eqref{rlcg31} is not referring to the (inverse) Faddeev-Popov ghost propagator for general $\alpha$. The connection with the ghost self-energy is only valid for the Landau gauge $\alpha=0$.} $\EuScript{M}(A^h)$.

\begin{figure}[t]\label{ghostAh}
    \begin{center}
        \scalebox{0.4}{\includegraphics{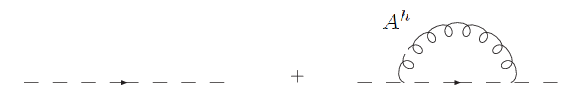}}
        \caption{The leading order correction to $\EuScript{M}^{-1}(A^h)$. The wiggled line represents a $\braket{A^h A^h}$ propagator, the broken line represents the tree level approximation to $\EuScript{M}^{-1}(A^h)$, $\frac{1}{p^2}$ in momentum space.}
    \end{center}
\end{figure}

We will now show that in the presence of the extra mass scale $M^2$, it is impossible to comply with the necessary condition \eqref{rlcg30} in $d=2$. It is sufficient to work at leading order, as the problem will already reveal itself at this order. Since this corresponds to working at order $g^2$ with two factors of $g$ already coming from the term $gf^{abc}A^{h,c}_\mu\partial_\mu$ in the operator $\EuScript{M}^{ab}(A^h)$, we may cut off the expansion of $A^h$ at order $g^0$, \textit{i.e.}~use the approximation \eqref{rlcg17}. Doing so, we find at leading order (see also Fig.~\ref{ghostAh})

\begin{equation}
\sigma(k) = g^2N \frac{k_{\mu} k_{\nu}}{k^2} \int \frac{d^2
q}{(2\pi)^2} \frac{1}{(k-q)^2} \frac{q^2 + M^2}{q^4 + (M^2+m^2)q^2 +
\lambda^4}\left(\delta_{\mu\nu}-\frac{q_\mu q_\nu}{q^2}\right)\,.
\label{rlcg31}
\end{equation}

\noindent We set here $\lambda^4=2g^2N \gamma^4+m^2M^2$.  The quantity $\frac{q^2 + M^2}{q^4 + (M^2+m^2)q^2 +\lambda^4}$ is the transverse piece of the would-be RGZ gluon propagator in $d=2$ (for the RGZ gluon propagator, see eq.(\ref{3.24})).

\noindent The above integral $\sigma(k)$ can be evaluated exactly using polar coordinates. Choosing the $q_x$-axis along
$\vec{k}$, we get

\begin{eqnarray}
\sigma(k) &=& \frac{g^2N}{4\pi^2}\int_0^{\infty}q d
q\frac{q^2+M^2}{q^4+(M^2+m^2)q^2+\lambda^4}\int_{0}^{2\pi} d\phi
\frac{1}{k^2+q^2-2qk\cos\phi}(1-\cos^2\phi)\nonumber\\
&=&\frac{g^2N}{4\pi}\left(\frac{1}{k^2}\int_0^{k}\frac{q(q^2+M^2)}{q^4+(M^2+m^2)q^2+\lambda^4}d
q+\int_k^{\infty}\frac{q^2+M^2}{q(q^4+(M^2+m^2)q^2+\lambda^4)}d
q\right)\,.\nonumber\\
\label{rlcg32}
\end{eqnarray}

\noindent where we employed $\vec{k}\cdot\vec{q}=kq\cos\phi$ and

\begin{equation}
\int_{0}^{2\pi}d\phi\frac{1-\cos^2\phi}{k^2+q^2-2qk\cos\phi}=  \frac{\pi}{q^ 2}\theta(q^2-k^2)+\frac{\pi}{k^ 2}\theta(k^2-q^2)\,.
\label{rlcg33}
\end{equation}

From the integrals appearing in \eqref{rlcg32}, we can extract the leading small $k^2$ behavior to be

\begin{equation}
\left.\sigma(k)\right|_{k^2\approx 0}\approx -\frac{g^2N}{8\pi}\frac{M^2}{\lambda^4}\ln(k^2)\,.
\label{rlcg34}		
\end{equation}

\noindent Since $M^2\geq 0$, we unequivocally find that $\sigma(k^2)$ will become (much) larger than $1$ if the momentum gets too small for $M^2>0$, that is it would become impossible to fulfill condition \eqref{rlcg30} and thus to ensure the positivity of $\EuScript{M}(A^h)$.

We are thus forced to conclude that $M^2=0$. Notice however that we are not able to prove that $m^2=0$. Indeed, if $M^2=0$, we are already back to the scaling case irrespective of the value for $m^2$. Scaling implies a vanishing of the transverse gluon form factor at zero momentum, which is in general sufficient to eliminate IR problems in the ghost form factor $\sigma(k^2)$, see \cite{Cucchieri:2012cb} for a general discussion.  Only an explicit discussion of the effective potential of the condensate related to $m^2$ (that is, $\braket{A^h A^h}$) will reveal whether it can be introduced or not into the theory. Though, this will not affect the conclusion that in $d=2$, a masssive/decoupling behavior is excluded.

\section{Gluon propagator} \label{GpropLCGAh}

As discussed in Sect.~\ref{npBRSTRGZLCG1}, the restriction of the path integral to a suitable region which is free of a large set of Gribov copies and is intimately related to the introduction of the Gribov parameter $\gamma$ generates dynamically dimension-two condensates. This generation is consistent in $d=3,4$, while in $d=2$ is absent. In this way, the gluon propagator is further affected by the introduction of such operators.

In $d=3,4$, the tree-level gluon two-point function is\footnote{In this expression we should keep in mind the meaning of indices and dimensions for different choices of $d$}

\begin{equation}
\langle A^{a}_{\mu}(k)A^{b}_{\nu}(-k)\rangle_{d=3,4} = \delta^{ab}\left[\frac{k^2+{M}^2}{(k^2+{m}^2)(k^2+{M}^2)+2g^2\gamma^4N}\left(\delta_{\mu\nu}-\frac{k_{\mu}k_{\nu}}{k^2}\right)+\frac{\alpha}{k^2}\frac{k_{\mu}k_{\nu}}{k^2}\right]\,.
\label{rlcg35}
\end{equation}

\noindent and in $d=2$,

\begin{equation}
\langle A^{a}_{\mu}(k)A^{b}_{\nu}(-k)\rangle_{D=2} = \delta^{ab}\left[\frac{k^2}{k^4+2g^2\gamma^4N}\left(\delta_{\mu\nu}-\frac{k_{\mu}k_{\nu}}{k^2}\right)+\frac{\alpha}{k^2}\frac{k_{\mu}k_{\nu}}{k^2}\right]\,.
\label{rlcg36}
\end{equation}

As is clear from (\ref{rlcg35}) and (\ref{rlcg36}), the longitudinal part of the tree level gluon propagator is not affected by non-perturbative effects, \textit{i.e.}, it has the same form as in the standard Faddeev-Popov quantization scheme. It is ensured by the non-perturbative BRST symmetry to hold to all orders and, therefore, is not a peculiarity of the tree-level approximation. At the current stage we can already provide a path integral proof of this fact and in the next chapter we shall discuss it from a more rigorous point of view.

For the path integral proof, we add a source term $\int \d^dx J^a b^{h,a}$ to the action  $S^{\mathrm{LCG}}_{\mathrm{RGZ}}$ to write (suppressing color indices)

\begin{equation}
\braket{b^{h,a}(x) b^{h,b}(y)}=\frac{\delta^2}{\delta J^{a}(x) \delta J^b(y)}\int [\EuScript{D}\Phi][\EuScript{D} b^h]e^{-S^{\mathrm{LCG}}_{\mathrm{RGZ}}}\Big|_{J=0}\,.
\label{rlcg37}
\end{equation}

As the $b^h$-field appears at most quadratically, we find exactly

\begin{equation}
\int [\EuScript{D}\Phi]e^{-S}= \int [\EuScript{D}\Phi]e^{-\int d^dx\left(\frac{1}{2\alpha} (\partial_{\mu} A^a_{\mu})^2+\frac{1}{\alpha}J^a\partial_{\mu} A^a_{\mu}+\frac{J^aJ^a}{2\alpha}+\textrm{rest}\right)}\,.
\label{rlcg38}		
\end{equation}

\noindent or, using \eqref{rlcg37}

\begin{equation}
\braket{b^{h,a}(x) b^{h,b}(y)}=\frac{1}{\alpha^2}\braket{\partial_{\mu} A^{a}_{\mu}(x) \partial_{\nu} A^{b}_{\nu}(y)}-\delta^{ab}\frac{\delta(x-y)}{\alpha}\,.
\label{rlcg39}	
\end{equation}

\noindent On the other hand, 

\begin{equation}
\braket{b^{h,a}(x) b^{h,b}(y)}=\braket{s_{\gamma^2}(\bar c^{a}(x) b^{h,b}(y) )}=0
\label{rlcg40}
\end{equation}

\noindent because of the non-perturbative BRST symmetry generated by $s_{\gamma^2}$. Combining \eqref{rlcg39} and \eqref{rlcg40}, we obtain

\begin{equation}
0=\frac{1}{\alpha^2}\braket{\partial_{\mu} A^{a}_{\mu}(x) \partial_{\nu} A^{b}_{\nu}(y)}-\delta^{ab}\frac{\delta(x-y)}{\alpha}\,,
\label{rlcg41}	
\end{equation}

\noindent which becomes in momentum space

\begin{equation}
\langle A^{a}_{\mu}(k)A^{b}_{\nu}(-k)\rangle =\delta^{ab}\left[\mathcal{D}_T(k^2)\left(\delta_{\mu\nu}-\frac{k_\mu k_\nu}{k^2}\right)+\alpha\frac{k_\mu k_\nu}{k^4}\right]\,,
\label{rlcg42}
\end{equation}

\noindent where the non-trivial information is encoded in the transverse form factor $D_T(k^2)$. For the transverse component of the gluon propagator, we see that a \textit{decoupling}-like behavior for $d=3,4$ is apparent, \textit{i.e.}, it has a non-vanishing form factor for zero momentum, while in $d=2$, a \textit{scaling}-like behavior is observed.

This result deserves some comments: 

\begin{itemize}

\item We derived the behavior of the gluon propagator for different values of $d$ in LCG. We cannot characterize how general this property is because the propagator is gauge dependent. Nevertheless, it was known and confirmed by lattice simulations the same behavior for the gluon propagator in Landau gauge, see \cite{Dudal:2008xd}. Of course, this is a particular case of LCG and, therefore, our results should reproduce Landau's features at least for $\alpha=0$. It turns out that the property seems to hold for arbitrary $\alpha$.

\item This scaling-like Vs. decoupling-like behavior in different dimensions was also observed in the MAG and in the Coulomb gauge, see \cite{Guimaraes:2015bra,Capri:2015pfa}. As we explicitly checked, this difference arises from typical IR singularities in $d=2$ and, therefore, we can conjecture the refinement will be potentially problematic in $d=2$ and these different behaviors will be present in different dimensions. 

\item Up to now, there is no lattice results for the gluon propagator in $d=2,3$ for LCG. We see then, that the non-perturbative BRST invariant construction could be partially test by confronting it against lattice simulations data regarding this quantity.

\end{itemize}

The gluon propagator presented here was computed at tree-level. The transverse component, irrespective of $d$, is equal to the tree-level gluon propagator in the (R)GZ set up (which is transverse). This is ensured by the independence from $\alpha$ of the Gribov parameter which appears in the propagator. On the other hand, we should emphasize that loop corrections can, in principle, add $\alpha$-dependent terms to the transverse sector of this propagator, displaying different results than Landau gauge ones. For this fact, we have no results to report yet and this is a future perspective. On the other hand, we are able to make some ``practical" comments: \textit{(i)} The assumption we made about the ``regularity" of zero-modes to prove Theorems~\ref{thma} and \ref{thma2} is exact as long as $\alpha << 1$, where all zero-modes are regular, but we have no \textit{a priori} strong argument to prevent completely the existence of ``pathological" zero-modes. Therefore, we cannot discard the possibility that as long as $\alpha$ increases, these pathological zero-modes play a relevant role and our elimination procedure does not take into account all infinitesimal Gribov copies and therefore not enough. Nevertheless, for many practical reasons, mainly from the numerical point of view, we restrict our analysis to values of alpha considered ``small" and then, we expect our procedure to be efficient. Up to now, the most recent lattice/functional results are roughly limited to $\alpha \leq 1$, but in principle, nothing (but practicality) forbids the choice \textit{e.g.} $\alpha = 10^10$. The current lattice and Schwinger-Dyson results indicate the $\alpha$ dependence of the transverse sector of the gluon two-point function should not be too strong, at least for values of $\alpha \leq 1$. For this reason, we are tempted to believe the loop corrections from the RGZ framework will introduce a not so strong $\alpha$-dependence on the transverse part. Just an explicit analysis will reveal if this is the case or not. Let us also point out that our reformulation of the (R)GZ action has introduced an important conceptual change in the variables used. The consistency of this formulation with gauge independence of physical quantities is a strong consistency check of the formulation, since in principle, for a BRST soft breaking theory as the (R)GZ, it is not trivial how to prove this. We will come back to this point in a more rigorous language in the next section; \textit{(ii)} The non-perturbative nilpotent BRST operator starts to show its power: Thanks to this symmetry, we were able to prove the exact form of the longitudinal part of the gluon propagator. In particular, in the standard Faddeev-Popov quantization, the perturbative (or standard) BRST invariance allows us to prove that the longitudinal part of the propagator is precisely the tree-level ones, \textit{i.e.} it does not receive loop corrections. Inhere, the same conclusion is achieved through the non-perturbative BRST symmetry, proving that unlike in the initial proposal of Ch.\ref{ch.5}, which there was no \textit{a priori} reason to declare the non-renormalization of the longitudinal sector. 

\section{A short look at the ghost propagator} \label{ghost-0}

Having worked out the expression of the gluon propagator in $d=4$, eq.\eqref{rlcg42}, we can have a short preliminary look at the ghost propagator. In this thesis, we limit ourselves to the one-loop order and a more complete and detailed analysis is beyond the our purposes.

For the one-loop ghost propagator in $d=4$, we have

\begin{equation}
\frac{1}{N^2-1} \sum_{ab} \delta^{ab}\langle {\bar c}^a(k) c^b(-k) \rangle_{1-loop} = \frac{1}{k^2} \frac{1}{1 - \omega(k^2)}  \,, 
\label{ghost-1}
\end{equation}

\noindent where

\begin{equation}
\omega(k^2) = \frac{Ng^2}{k^2(N^2-1)} \int \frac{d^4q}{(2\pi)^4} \frac{k_\mu (k-q)_\nu}{(k-q)^2} \langle A^{a}_{\mu}(q)A^{a}_{\nu}(-q)\rangle \,.
\label{ghost-2}
\end{equation}

\noindent From expression \eqref{rlcg42}, we get

\begin{equation}
\omega(k^2) = \omega^T(k^2) + \omega^L(k^2)  \,, 
\label{ghost-3}
\end{equation}

\noindent where $\omega^T(k^2) $ stands for the contribution corresponding to the transverse component of the gluon propagator, {\it i.e.}

\begin{equation}
\omega^T(k^2) = {Ng^2}\frac{k_\mu k_\nu}{k^2}  \int \frac{d^4q}{(2\pi)^4} \frac{1}{(k-q)^2} \frac{q^2+{M}^2}{(q^2+{m}^2)(q^2+{M}^2)+2g^2\gamma^4N}\left(\delta_{\mu\nu}-\frac{q_{\mu}q_{\nu}}{q^2}\right) \,, 
\label{ghost-4}
\end{equation}

\noindent while $\omega^L(k^2)$ is the contribution stemming from the longitudinal component, namely

\begin{equation}
\omega^L(k^2) = \alpha \frac{Ng^2}{k^2} \int \frac{d^4q}{(2\pi)^4} \frac{k_\mu (k-q)_\nu}{(k-q)^2} \frac{q_\mu q_\nu}{q^4}   \,. 
\label{ghost-5}
\end{equation}

Employing dimensional regularization in the $\overline{\mathrm{MS}}$ scheme, expression \eqref{ghost-5} yields

\begin{equation}
\omega^L(k^2) = \alpha \frac{Ng^2}{64 \pi^2} \log{\frac{k^2}{{\bar \mu}^2} }\,. 
\label{ghost-6}
\end{equation}

\noindent This result for $\omega^L(k^2)$ obviously coincides with the standard perturbative result at one loop. It is worth emphasizing that the result \eqref{ghost-6} is a consequence of the non-trivial fact that the longitudinal component of the gluon propagator is left unmodified by the addition of the horizon function $H(A^h)$. The presence of terms of the type of eq.\eqref{ghost-6} seems therefore unavoidable when evaluating the ghost form factor for non-vanishing values of the gauge parameter $\alpha$. When passing from the $1PI$ Green function to the connected one, such terms should lead to a ghost form factor which is severely suppressed in the infrared region $k^2\approx 0$ with respect to the case of the Landau gauge, {\it i.e.} $\alpha=0$, as discussed recently within the framework of the Dyson-Schwinger equations \cite{Aguilar:2015nqa,Huber:2015ria}.

\chapter{Non-perturbative BRST-invariant RGZ action: A local formulation} \label{locnonpBRST}

The construction of a would-be (R)GZ action for LCG forced us to introduce the non-local gauge invariant variable $A^h$. Fortunately, this reformulation enables us to define a nilpotent exact BRST symmetry. As a byproduct, we were able to cast the gap equation which fixes the Gribov parameter $\gamma$ in a self-consistent way in a gauge invariant form, a non-trivial but highly desirable feature, since it enters expectation values of physical operators. Also, the exact non-perturbative BRST symmetry allowed us to prove \textit{e.g.} the exact form of the longitudinal form of the gluon propagator. 

Although the construction of a non-perturbative exact \textit{and} nilpotent BRST symmetry corresponds to a strong conceptual advance in the (R)GZ setting, it is based on a non-local action and BRST transformations so far. This sort of non-locality forbids the use of the most powerful tools in QFT, an in particular, the new set of BRST transformations cannot be used for a renormalizability analysis. Hence, a local formulation of the (R)GZ action in LCG as well as of the non-perturbative BRST transformations would be a big technical step enabling us to use the full local QFT machinery. 

We remind the reader that the standard (R)GZ action also displays a non-locality due to the horizon function, see eq.(\ref{2.46}). As discussed in Sect.~\ref{localizationGZaction}, this non-local expression can be traded by a local one upon the introduction of Zwanziger's auxiliary fields. In the reformulation of the (R)GZ action with $A^{h}$, a novel non-locality arises from the very non-local form of $A^h$ itself as explained in Ap.~\ref{constructionAh}. Already at Ch.~\ref{ch.5}, this problem was faced in the approximation $A^h\approx A^T$. There, we were able to present a full local (R)GZ action via the introduction of extra auxiliary fields. 

In this chapter we present a complete localization of both (R)GZ action in LCG and non-perturbative BRST transformations. The procedure, again, is based on the introduction of auxiliary fields. At the end of the procedure, we obtain a full local framework allowing the use of local QFT principles and theorems. Our main goal is to present the localization procedure step by step always showing how to reobtain the non-local formulation. 

\section{Local non-perturbative BRST-invariant GZ action in LCG}

The non-perturbative BRST-invariant GZ action in LCG is given by

\begin{eqnarray}
S^{\mathrm{LCG}}_{\mathrm{GZ}}&=& S_{\mathrm{YM}}+s_{\gamma^2}\int d^dx~\bar{c}^{a}\left(\partial_{\mu}A^{a}_{\mu}-\frac{\alpha}{2}b^{h,a}\right)\nonumber\\
&-&\int d^dx\left(\bar{\varphi}^{ac}_{\mu}\left[\EuScript{M}(A^h)\right]^{ab}\varphi^{bc}_{\mu}-\bar{\omega}^{ac}_{\mu}\left[\EuScript{M}(A^h)\right]^{ab}\omega^{bc}_{\mu}+g\gamma^2f^{abc}A^{h,a}_{\mu}(\varphi+\bar{\varphi})^{bc}_{\mu}\right)\nonumber\\
&=&  S_{\mathrm{YM}} + \int d^dx\left(b^{h,a}\left(\partial_{\mu}A^{a}_{\mu}-\frac{\alpha}{2}b^{h,a}\right)+\bar{c}^{a}\partial_{\mu}D^{ab}_{\mu}c^{b}\right)\nonumber\\
&-&\int d^dx\left(\bar{\varphi}^{ac}_{\mu}\left[\EuScript{M}(A^h)\right]^{ab}\varphi^{bc}_{\mu}-\bar{\omega}^{ac}_{\mu}\left[\EuScript{M}(A^h)\right]^{ab}\omega^{bc}_{\mu}+g\gamma^2f^{abc}A^{h,a}_{\mu}(\varphi+\bar{\varphi})^{bc}_{\mu}\right)\,,\nonumber\\
\label{lrlcg1}
\end{eqnarray} 

\noindent where the set of Zwanziger's auxiliary fields were already introduced. We recall these fields were originally introduced to localize the standard horizon function $H_{L}(A)$ in the Landau gauge. In the present case, we have a horizon function with two sorts of non-localities,

\begin{equation}
H(A^h)=g^2\int d^dx d^dy~f^{abc}A^{h,b}_{\mu}(x)\left[\EuScript{M}^{-1}(A^{h})\right]^{ad}(x,y)f^{dec}A^{h,e}_{\mu}(y)\,,
\label{lrlcg2}
\end{equation}

\noindent namely the inverse of $\EuScript{M}(A^{h})$ which is similar to standard non-locality of the horizon function in the Landau gauge $H_{L}(A)$ and the non-local form of $A^{h}$ itself. The set of auxiliary fields $(\bar{\varphi},\varphi,\bar{\omega},\omega)$ is responsible to localize the first type of non-locality, as explicitly showed in (\ref{lrlcg1}). Still, the resulting action (\ref{lrlcg1}) is non-local due to the presence of $A^h$, which is explicitly written as a formal non-local power series, see eq.(\ref{ah19}). 

The localization procedure we propose for the $A^h$ field relies on the introduction of a Stueckelberg-like field $\xi^a$. This field appears as

\begin{equation}
h=\mathrm{e}^{ig\xi^a T^a}\equiv \mathrm{e}^{ig\xi}\,,
\label{lrlcg3}
\end{equation}

\noindent and we rewrite the $A^h$ field as\footnote{When color indices are suppressed, we are employing matrix variables.}

\begin{equation}
A^h_{\mu}=h^{\dagger}A_{\mu}h+\frac{i}{g}h^{\dagger}\partial_{\mu}h\,.
\label{lrlcg4}
\end{equation}

\noindent Clearly, to recover the non-local power series expression for $A^{h}$ from eq.(\ref{lrlcg5}) we still have to impose the transversality condition,

\begin{equation}
\partial_{\mu}A^h_{\mu}=0\,.
\label{lrlcg5}
\end{equation}

\noindent Before introducing constraint (\ref{rlcg5}), let us discuss some properties of eq.(\ref{lrlcg4}). First, the defining expression (\ref{lrlcg4}) is \textit{local} albeit non-polynomial on $\xi$. Second, under a gauge transformation with group element $V$,

\begin{equation}
A_{\mu}\,\,\longrightarrow\,\, A'_{\mu}=V^{\dagger}A_{\mu}V+\frac{i}{g}V^{\dagger}\partial_{\mu}V\,,\,\,\,\,h\,\,\longrightarrow\,\, h'=V^{\dagger}h\,,\,\,\,\,h^{\dagger}\,\,\longrightarrow\,\, h'^{\dagger}=h^{\dagger}V
\label{lrlcg6}
\end{equation} 

\noindent it is very simple to check

\begin{equation}
A^{h}_{\mu},\,\longrightarrow\,\, A'^{h}_{\mu} = A^{h}_{\mu}\,,
\label{lrlcg7}
\end{equation}

\noindent which establishes the gauge invariance of $A^{h}_{\mu}$. To obtain the expression of $A^{h}_{\mu}$ in terms of $\xi$, we refer to Ap.~\ref{constructionAh}, where this computation was explicitly carried out where instead of a field $\xi$, a gauge parameter $\phi$ is used. This establishes the promotion of a gauge parameter to a dynamical field of the theory. Moreover, we should demand transversality of the now local field $A^{h}_{\mu}$. This is enforced by the introduction of a Lagrangian multiplier $\tau^a$ in the following way,

\begin{equation}
S_{\tau}=\int d^dx~\tau^a\partial_{\mu}A^{h,a}_{\mu}\,.
\label{lrlcg8}
\end{equation}

\noindent Finally, the \textit{local} GZ action in LCG, $S^{\mathrm{loc}}_{\mathrm{LCG}}$ is

\begin{eqnarray}
S^{\mathrm{loc}}_{\mathrm{LCG}} &=&  S_{\mathrm{YM}} + \int d^dx\left(b^{h,a}\left(\partial_{\mu}A^{a}_{\mu}-\frac{\alpha}{2}b^{h,a}\right)+\bar{c}^{a}\partial_{\mu}D^{ab}_{\mu}c^{b}\right)-\int d^dx\left(\bar{\varphi}^{ac}_{\mu}\left[\EuScript{M}(A^h)\right]^{ab}\varphi^{bc}_{\mu}\right.\nonumber\\
&-&\left.\bar{\omega}^{ac}_{\mu}\left[\EuScript{M}(A^h)\right]^{ab}\omega^{bc}_{\mu}+g\gamma^2f^{abc}A^{h,a}_{\mu}(\varphi+\bar{\varphi})^{bc}_{\mu}\right)+\int d^dx~\tau^a\partial_{\mu}A^{h,a}_{\mu}\,,
\label{lrlcg9}
\end{eqnarray} 

\noindent where $A^{h,a}_{\mu}$ is given by eq.(\ref{lrlcg4}) supplied by eq.(\ref{lrlcg3}). Expression (\ref{lrlcg9}) is a \textit{local} action which reduces to the non-local expression (\ref{lrlcg1}) upon elimination of the auxiliary fields. It is worth underlining that albeit local, the local GZ action in LCG (\ref{lrlcg9}) is a non-polynomial expression. With expression (\ref{lrlcg9}) at our disposal, we now proceed to the localization of the non-perturbative BRST transformations. 

\section{Localization of the non-perturbative BRST transformations}

In Ch.~\ref{nonpBRSTRGZ} we introduced the non-perturbative BRST transformations which correspond to an exact symmetry of the GZ action \textit{and} enjoy nilpotency. As argued in the same chapter, this sort of symmetry is viable if these transformations are non-local and the explicit construction shows this is the case. To make use of the powerful BRST set up, it is important to cast these transformations in local way. As we also pointed out, the non-nilpotent set of BRST transformation which also forms a symmetry of the GZ action in the Landau gauge was localized in \cite{Dudal:2010hj}. In this section, we localize the nilpotent non-perturbative BRST transformations (\ref{npbrst38}). Before doing that, we set the standard/perturbative BRST transformations for the auxiliary field introduced in previous section. 

The Stueckelberg field $\xi$ enters in the GZ action inside expression (\ref{lrlcg3}). For $h$, the BRST transformation is given by

\begin{equation}
sh=-igch\,.
\label{lrlcg10}
\end{equation}

\noindent Expanding $h$ in power series, we can write

\begin{eqnarray}
s\left(\mathrm{e}^{ig\xi}\right)&=&-igc\left(\mathrm{e}^{ig\xi}\right)\Rightarrow s\left(1+ig\xi-\frac{g^2}{2}\xi\xi-i\frac{g^3}{3!}\xi\xi\xi + \ldots\right)=-igc\left(1+ig\xi\phantom{\frac{1}{2}}\right.\nonumber\\
&-&\left.\frac{g^2}{2}\xi\xi +\ldots\right)\Rightarrow igs\xi - \frac{g^{2}}{2}((s\xi)\xi+\xi(s\xi))-i\frac{g^3}{3!}((s\xi)\xi+\xi(s\xi)\xi\nonumber\\
&+&\xi\xi(s\xi))+\ldots=-igc+g^2c\xi+i\frac{g^3}{2}c\xi\xi + \ldots\,\,\Rightarrow\nonumber\\
s\xi&=&-c-igc\xi+\frac{g^2}{2}c\xi\xi-i\frac{g}{2}\left[\left(-c-igc\xi-i\frac{g}{2}((s\xi)\xi+\xi(s\xi))\right)\xi\right.\nonumber\\
&+&\left.\xi\left(-c-igc\xi-i\frac{g}{2}((s\xi)\xi+\xi(s\xi))\right)\right]-\frac{g^2}{3!}(c\xi\xi+\xi c\xi + \xi\xi c)+\ldots\,.
\label{lrlcg11}
\end{eqnarray}

\noindent We can solve eq.(\ref{lrlcg11}) iteratively to obtain an explicit expression for $s\xi$,

\begin{equation}
s\xi^a=-c^a+\frac{g}{2}f^{abc}c^b \xi^c-\frac{g^2}{12}f^{acd}f^{cmb}c^m \xi^b \xi^d + \mathcal{O}(g^3)\,.
\label{lrlcg12}
\end{equation}

\noindent As discussed in Ap.~\ref{constructionAh}, the field $A^{h}_{\mu}$ should be BRST invariant. Written in terms of $h$, eq.(\ref{lrlcg4}), we can easily check it,

\begin{eqnarray}
sA^{h}_{\mu} &=& (sh^{\dagger})A_{\mu}h+h^{\dagger}(sA_{\mu})h+h^{\dagger}A_{\mu}(sh)+\frac{i}{g}(sh^{\dagger})\partial_{\mu}h+\frac{i}{g}h^{\dagger}\partial_{\mu}(sh) \nonumber \\
&=& (igh^{\dagger}c)A_{\mu}h-h^{\dagger}(\partial_{\mu}c)h+igh^{\dagger}[A_{\mu},c]h-ig h^{\dagger}A_{\mu}ch-h^{\dagger}c\partial_{\mu}h+h^{\dagger}\partial_{\mu}(ch)\nonumber\\
&=& igh^{\dagger}c A_{\mu}h-h^{\dagger}(\partial_{\mu}c)h + igh^{\dagger}A_{\mu}ch-igh^{\dagger}c A_{\mu}h-igh^{\dagger}A_{\mu}ch-h^{\dagger}c\partial_{\mu}h\nonumber\\
&+&h^{\dagger}(\partial_{\mu}c)h+h^{\dagger}c\partial_{\mu}h=0\,.
\label{lrlcg13}
\end{eqnarray}

\noindent Finally, the auxiliary field $\tau^a$, introduced to enforce the transversality of $A^{h}_{\mu}$, is coupled to the BRST invariant $\partial_{\mu}A^{h}_{\mu}$ term. Therefore, we demand $\tau^a$ to be BRST invariant to guarantee $S_{\tau}$ is BRST invariant,

\begin{equation}
s\tau^{a}=0\,.
\label{lrlcg14}
\end{equation}

\noindent In summary, the standard/perturbative BRST transformations are

\begin{align}
sA^{a}_{\mu}&=-D^{ab}_{\mu}c^{b}\,,     && s\varphi^{ab}_{\mu}=\omega^{ab}_{\mu}\,,&&&      sh =-igch\,,\nonumber\\
sc^a&=\frac{g}{2}f^{abc}c^{b}c^{c}\,,     && s\omega^{ab}_{\mu}=0\,,&&&     sA^{h,a}_{\mu}=0\,,\nonumber\\
s\bar{c}^{a}&=b^{a}\,,   && s\bar{\omega}^{ab}_{\mu}=\bar{\varphi}^{ab}_{\mu}\,,&&&     s\tau^{a} =0\,.\nonumber\\
sb^a&=0\,,         && s\bar{\varphi}^{ab}_{\mu}=0\,. 
\label{lrlcg15}
\end{align}

\noindent We see thus the standard BRST breaking of (\ref{lrlcg9}) is the same as in the non-local formulation (\ref{lrlcg1}). From the non-local transformations (\ref{npbrst38}), we see this breaking is healed by the non-local modification of the BRST transformation of $\bar{\omega}^{ab}_{\mu}$. To localize this structure, we must introduce auxiliary fields whose equations of motion reproduce the non-local modification. To accomplish this task, the following trick is useful: First we rewrite the starting point, the horizon function as

\begin{equation}
\mathrm{e}^{-\gamma^4 H(A^{h})}=\mathrm{e}^{-\frac{\gamma^4}{2} H(A^{h})}\mathrm{e}^{-\frac{\gamma^4}{2} H(A^{h})}\,.
\label{lrlcg16}
\end{equation}

\noindent The second step is to use Zwanziger's localization procedure to localize each factor of eq.(\ref{lrlcg16}) separately,

\begin{equation}
\mathrm{e}^{-\frac{\gamma^4}{2} H(A^{h})}=\int \left[\EuScript{D}\varphi\right]\left[\EuScript{D}\bar{\varphi}\right]\left[\EuScript{D}\omega\right]\left[\EuScript{D}\bar{\omega}\right]\mathrm{e}^{-\int d^dx\left(-\bar{\varphi}^{ac}_{\mu}\EuScript{M}^{ab}(A^h)\varphi^{bc}_{\mu}+\bar{\omega}^{ac}_{\mu}\EuScript{M}^{ab}(A^h)\omega^{bc}_{\mu}+g\frac{\gamma^2}{\sqrt{2}}f^{abc}A^{h,a}_{\mu}(\varphi+\bar{\varphi})^{bc}_{\mu}\right)}
\label{lrlcg17}
\end{equation}

\noindent and

\begin{equation}
\mathrm{e}^{-\frac{\gamma^4}{2} H(A^{h})}=\int \left[\EuScript{D}\beta\right]\left[\EuScript{D}\bar{\beta}\right]\left[\EuScript{D}\zeta\right]\left[\EuScript{D}\bar{\zeta}\right]\mathrm{e}^{-\int d^dx\left(-\bar{\beta}^{ac}_{\mu}\EuScript{M}^{ab}(A^h)\beta^{bc}_{\mu}+\bar{\zeta}^{ac}_{\mu}\EuScript{M}^{ab}(A^h)\zeta^{bc}_{\mu}-g\frac{\gamma^2}{\sqrt{2}}f^{abc}A^{h,a}_{\mu}(\beta+\bar{\beta})^{bc}_{\mu}\right)}
\label{lrlcg18}
\end{equation}

\noindent where we explore the symmetry $\gamma^2\,\leftrightarrow\, -\gamma^2$ in the localization procedure. The set $(\varphi,\bar{\varphi},\omega,\bar{\omega})$ are the standard localizing Zwanziger's fields, $(\beta,\bar{\beta})$ are commuting fields and $(\zeta,\bar{\zeta})$ anti-commuting ones. 

We see from eq.(\ref{lrlcg18}) that the equation of motion of $\beta$ gives

\begin{equation}
-\EuScript{M}^{ab}(A^h)\bar{\beta}^{bc}_{\mu}-g\frac{\gamma^2}{\sqrt{2}}f^{bac}(A^h)^{b}_{\mu}=0\,\,\,\Rightarrow\,\,\,\bar{\beta}^{ab}_{\mu}=-g\frac{\gamma^2}{\sqrt{2}}f^{kcb}(A^h)^{k}_{\mu}\left[\EuScript{M}^{-1}(A^h)\right]^{ac}\,,
\label{lrlcg19}
\end{equation}

\noindent which is precisely the non-local structure present in the non-perturbative BRST transformation for $\bar{\omega}$. In eq.(\ref{lrlcg19}), we used the fact that $\EuScript{M}(A^h)$ is positive and, thus, invertible. Immediately, we see that the field $\bar{\beta}$ allows the localization of the non-perturbative BRST transformations and the non-local expression is recovered by the integration of $(\beta,\bar{\beta},\zeta,\bar{\zeta})$. 

The local GZ action in LCG written with the auxiliary fields introduced in eq.(\ref{lrlcg18}) is written as

\begin{eqnarray}
S^{\mathrm{loc}}_{{\mathrm{LCG}}}&=& S_{\mathrm{YM}}+\int d^dx\left(b^{h,a}\partial_{\mu}A^{a}_{\mu}-\frac{\alpha}{2}b^{h,a}b^{h,a}+\bar{c}^{a}\partial_{\mu}D^{ab}_{\mu}(A)c^b\right)+\int d^dx~\tau^a\partial_{\mu}A^{h,a}_{\mu}\nonumber\\
&-&\int d^dx\left(\bar{\varphi}^{ac}_{\mu}\EuScript{M}^{ab}(A^h)\varphi^{bc}_{\mu}-\bar{\omega}^{ac}_{\mu}\EuScript{M}^{ab}(A^h)\omega^{bc}_{\mu}-g\frac{\gamma^2}{\sqrt{2}}f^{abc}A^{h,a}_{\mu}(\varphi+\bar{\varphi})^{bc}_{\mu}\right)\nonumber\\
&-&\int d^dx\left(\bar{\beta}^{ac}_{\mu}\EuScript{M}^{ab}(A^h)\beta^{bc}_{\mu}-\bar{\zeta}^{ac}_{\mu}\EuScript{M}^{ab}(A^h)\zeta^{bc}_{\mu}+g\frac{\gamma^2}{\sqrt{2}}f^{abc}A^{h,a}_{\mu}(\beta+\bar{\beta})^{bc}_{\mu}\right)\,. \nonumber\\
\label{lrlcg20}
\end{eqnarray}

\noindent Action (\ref{lrlcg20}) is invariant under the non-perturbative \textit{and} nilpotent BRST transformations, generated by the operator $s_{l}$,

\begin{align}
s_{l}A^{a}_{\mu}&=-D^{ab}_{\mu}c^{b}\,,     && s_{l}\varphi^{ab}_{\mu}=\omega^{ab}_{\mu}\,,&&&      s_{l}h =-igch\,,&&&& s_{l}\beta^{ab}_{\mu}=\omega^{ab}_{\mu}\,,\nonumber\\
s_{l}c^a&=\frac{g}{2}f^{abc}c^{b}c^{c}\,,     && s_{l}\omega^{ab}_{\mu}=0\,,&&&     s_{l}A^{h,a}_{\mu}=0\,,&&&& s_{l}\bar{\zeta}^{ab}_{\mu}=0\,,\nonumber\\
s_{l}\bar{c}^{a}&=b^{h,a}\,,   && s_{l}\bar{\omega}^{ab}_{\mu}=\bar{\varphi}^{ab}_{\mu}+\bar{\beta}^{ab}_{\mu}\,,&&&     s_{l}\tau^{a} =0\,,&&&&s_{l}\zeta^{ab}_{\mu}=0\,,\nonumber\\
s_{l}b^{h,a}&=0\,,         && s_{l}\bar{\varphi}^{ab}_{\mu}=0\,,&&&s_{l}\bar{\beta}^{ab}_{\mu}=0\,.
\label{lrlcg21}
\end{align}

\noindent The structure of the BRST transformations for $(\beta,\bar{\beta},\zeta,\bar{\zeta})$ naturally arises from the nilpotency requirement of $s_{l}$ and invariance of (\ref{lrlcg20}). Hence, $s_l$ defines a \textit{local, non-perturbative and nilpotent} BRST symmetry of action (\ref{lrlcg20}). The non-local GZ action, written in terms of the horizon function, is manifestly invariant under  $\gamma^2\,\leftrightarrow\, -\gamma^2$. In local form, this symmetry is manifest through a discrete transformation of the auxiliary localizing fields, namely

\begin{align}
\varphi^{ab}_{\mu}&\rightarrow -\beta^{ab}_{\mu}\,, &&\bar{\varphi}^{ab}_{\mu}\rightarrow -\bar{\beta}^{ab}_{\mu}\,,\nonumber\\
\beta^{ab}_{\mu}&\rightarrow -\varphi^{ab}_{\mu}\,, &&\bar{\beta}^{ab}_{\mu}\rightarrow -\bar{\varphi}^{ab}_{\mu}\,, \nonumber\\
\omega^{ab}_{\mu}&\rightarrow -\zeta^{ab}_{\mu}\,, &&\bar{\omega}^{ab}_{\mu}\rightarrow -\bar{\zeta}^{ab}_{\mu}\,, \nonumber\\
\zeta^{ab}_{\mu}&\rightarrow -\omega^{ab}_{\mu}\,, &&\bar{\zeta}^{ab}_{\mu}\rightarrow -\bar{\omega}^{ab}_{\mu}\,.
\label{lrlcg22}
\end{align}

Conveniently, we can introduce the following variables

\begin{eqnarray}
\kappa^{ab}_{\mu}&=&\frac{1}{\sqrt{2}}\left(\varphi^{ab}_{\mu}+\beta^{ab}_{\mu}\right)\,,\nonumber\\
\lambda^{ab}_{\mu}&=& \frac{1}{\sqrt{2}}\left(\varphi^{ab}_{\mu}-\beta^{ab}_{\mu}\right)\,,
\label{lrlcg23}
\end{eqnarray}

\noindent which satisfy

\begin{eqnarray}
s_l\kappa^{ab}_{\mu}&=&\sqrt{2}\omega^{ab}_{\mu}\,,\nonumber\\
s_l\lambda^{ab}_{\mu}&=&0\,.
\label{lrlcg24}
\end{eqnarray}

\noindent From relations (\ref{lrlcg23}), we can write

\begin{eqnarray}
\varphi^{ab}_{\mu}&=&\frac{1}{\sqrt{2}}\left(\kappa^{ab}_{\mu}+\lambda^{ab}_{\mu}\right)\,,\nonumber\\
\beta^{ab}_{\mu}&=& \frac{1}{\sqrt{2}}\left(\kappa^{ab}_{\mu}-\lambda^{ab}_{\mu}\right)\,,
\label{lrlcg25}
\end{eqnarray}

\noindent and plugging eq.(\ref{lrlcg25}) into (\ref{lrlcg20}), we obtain

\begin{eqnarray}
S^{\mathrm{loc}}_{{\mathrm{LCG}}}&=& S_{\mathrm{YM}}+\int d^dx\left(b^{h,a}\partial_{\mu}A^{a}_{\mu}-\frac{\alpha}{2}b^{h,a}b^{h,a}+\bar{c}^{a}\partial_{\mu}D^{ab}_{\mu}(A)c^b\right)+\int d^dx~\tau^a\partial_{\mu}A^{h,a}_{\mu}\nonumber\\
&-&\int d^dx\left(\bar{\lambda}^{ac}_{\mu}\EuScript{M}^{ab}(A^h)\lambda^{bc}_{\mu}-\bar{\zeta}^{ac}_{\mu}\EuScript{M}^{ab}(A^h)\zeta^{bc}_{\mu}-g\gamma^2f^{abc}A^{h,a}_{\mu}(\lambda+\bar{\lambda})^{bc}_{\mu}\right)\nonumber\\
&-&\int d^dx\left(\bar{\kappa}^{ac}_{\mu}\EuScript{M}^{ab}(A^h)\kappa^{bc}_{\mu}-\bar{\omega}^{ac}_{\mu}\EuScript{M}^{ab}(A^h)\omega^{bc}_{\mu}\right)\,. 
\label{lrlcg26}
\end{eqnarray}

\noindent In terms of variables (\ref{lrlcg23}), the complete set of local non-perturbative BRST transformations is

\begin{align}
s_{l}A^{a}_{\mu}&=-D^{ab}_{\mu}c^{b}\,,     && s_{l}\kappa^{ab}_{\mu}=\sqrt{2}\omega^{ab}_{\mu}\,,&&&      s_{l}h =-igch\,,&&&& s_{l}\zeta^{ab}_{\mu}=0\,,\nonumber\\
s_{l}c^a&=\frac{g}{2}f^{abc}c^{b}c^{c}\,,     && s_{l}\omega^{ab}_{\mu}=0\,,&&&     s_{l}A^{h,a}_{\mu}=0\,,&&&& s_{l}\bar{\lambda}^{ab}_{\mu}=0\,,\nonumber\\
s_{l}\bar{c}^{a}&=b^{h,a}\,,   && s_{l}\bar{\omega}^{ab}_{\mu}=\sqrt{2}\bar{\kappa}^{ab}_{\mu}\,,&&&     s_{l}\tau^{a} =0\,,&&&&s_{l}\lambda^{ab}_{\mu}=0\,.\nonumber\\
s_{l}b^{h,a}&=0\,,         && s_{l}\bar{\kappa}^{ab}_{\mu}=0\,,&&&s_{l}\bar{\zeta}^{ab}_{\mu}=0\,,
\label{lrlcg27}
\end{align}

\noindent From transformations (\ref{lrlcg27}), we identify immediately $(A^{h},\bar{\lambda},\lambda,\bar{\zeta},\zeta,\tau)$ as BRST singlets. Endowed with (\ref{lrlcg27}), we can give a simple algebraic proof of the physicality of $\gamma$,

\begin{equation}
s_{l}\frac{\partial S^{\mathrm{loc}}_{\mathrm{LCG}}}{\partial\gamma^2}=s_{l}\int d^dx~gf^{abc}A^{h,a}_{\mu}(\lambda+\bar{\lambda})^{bc}_{\mu}=0\,,\,\,\,\mathrm{and}\,\,\,\frac{\partial S^{\mathrm{loc}}_{\mathrm{LCG}}}{\partial\gamma^2}\neq s_l\left(\ldots\right)\,.
\label{lrlcg28}
\end{equation}

\noindent Therefore, $\gamma$ couples to a closed BRST form, establishing its physical nature. This is in agreement with the gauge invariant gap equation (\ref{rlcg4}) which establishes the independence from $\alpha$ of $\gamma$. On the other hand, the parameter $\alpha$ is introduced in (\ref{lrlcg26}) via a BRST exact term, and due to the invariance of (\ref{lrlcg26}) under (\ref{lrlcg27}), we ensure it cannot enter physical quantities and, in particular, to (\ref{lrlcg28}).  

\section{A comment on the Stueckelberg field}

Action (\ref{lrlcg20}) is a local action which implements the restriction of the path integral domain to $\EuScript{M}(A^h)>0$ in a non-perturbative BRST invariant way. With this at our disposal we may write down its correspondent Wards identities following the Quantum Action Principle and draw possible conclusions. In particular, the existence of a non-perturbative exact BRST symmetry allows the construction of a ``non-perturbative" Slavnov-Taylor identity which \textit{e.g.} should be responsible to control the gauge parameter dependence. However, before moving to this direction, we point out a feature which appears at the level of the computation of tree-level propagators associated with (\ref{lrlcg20}). To avoid an abrupt interruption in the text, we collect the propagators in Ap.~\ref{proplrgzlcg}. As an important particular case, the Stueckelberg field has a propagator which behaves like

\begin{equation}
\langle \xi^{a}(p)\xi^{b}(-p)\rangle = \alpha\frac{\delta^{ab}}{p^4}\,,
\label{lrlcg29}
\end{equation}

\noindent which might generate IR singularities for explicit computations. Expression (\ref{lrlcg29}) can be regularized in a BRST-invariant fashion though. To show this, we begin by the observation that, from eq.(\ref{lrlcg12}), we can write the following identity,

\begin{equation}
s\left(\frac{\xi^a\xi^a}{2}\right)=-c^a\xi^a\,.
\label{lrlcg30}
\end{equation}
This is proved as follows: Expanding the exponential in Taylor series of $s_l\left(\mathrm{e}^{ig\xi}\right)$, one gets

\begin{equation}
s_{l} \left( 1 + ig \xi - \frac{g^2}{2} \xi \xi - i \frac{g^3}{3!} \xi \xi \xi + \cdot  \cdot  \right)  = -igc \left(   1 + ig \xi - \frac{g^2}{2} \xi \xi - i \frac{g^3}{3!} \xi \xi \xi + \cdot  \cdot \right)  \;. \label{lrlcg30.1}
\end{equation}
Multiplying both sides of eq.\eqref{lrlcg30.1} by $\xi$, yields

\begin{equation}
\xi \; s_{l} \left( 1 + ig \xi - \frac{g^2}{2} \xi \xi - i \frac{g^3}{3!} \xi \xi \xi + \cdot  \cdot  \right)  = -ig \xi \;c \left(   1 + ig \xi - \frac{g^2}{2} \xi \xi - i \frac{g^3}{3!} \xi \xi \xi + \cdot  \cdot \right)  \;. \label{lrlcg30.2}
\end{equation}
Equating order by order in $g$ the expression \eqref{lrlcg30.2} immediately provides eq.\eqref{lrlcg30} at leading order. Now, we can straightforwardly introduce the following non-perturbative BRST exact term

\begin{equation}
S_{\mathrm{R}}=s_l\int d^dx~\left(\frac{1}{2}\rho\xi^a\xi^a\right)=\int d^dx~\left(\frac{1}{2}\tilde{m}^{4}\xi^a\xi^a+\rho c^a\xi^a\right)\,,
\label{lrlcg31}
\end{equation}

\noindent with 

\begin{equation}
s_l\rho=\tilde{m}^4\,\,\,\,\,\,\mathrm{and}\,\,\,\,\,\, s_l\tilde{m}^4=0\,.
\label{lrlcg32}
\end{equation}

\noindent With the introduction of the manifest non-perturbative BRST invariant term (\ref{lrlcg31}), the Stueckelberg field propagator is written as 

\begin{equation}
\langle \xi^{a}(p)\xi^{b}(-p)\rangle = \alpha\frac{\delta^{ab}}{p^4+\alpha \tilde{m}^{4}}\,,
\label{lrlcg33}
\end{equation}

\noindent which, thanks to the introduction of (\ref{lrlcg32}) is regularized in the IR. To recover the ``correct" physical limit of the original theory, after explicit computations we should take $\tilde{m}^4\rightarrow 0$. Although for the purposes of this thesis we will not perform any explicit computation under this framework, it is somehow reassuring we are able to implement a BRST-invariant regularization for this field. 

\section{Non-perturbative Ward identities}

With a fully local and BRST invariant set up at our disposal, an immediate natural task is the analysis of the Ward identities and which sort of restrictions they impose to our quantum action. In particular, the theory has an exact Slavnov-Taylor identity associated with the non-perturbative BRST invariance. 

In the literature, a trick that provided an explicit control over the gauge parameter $\alpha$-dependence is to introduce a BRST transformation of the gauge parameter itself,

\begin{equation}
s_l\alpha=\chi\,,\,\,\,\,\,\, s_l\chi=0\,,
\label{lrlcg34}
\end{equation}

\noindent with $\chi$ an anti-commuting parameter. We refer to \cite{Piguet:1984js} for further details on this contruction. The set of the already introduced local BRST transformations plus those introuced in (\ref{lrlcg34}) form the so-called \textit{extended BRST transformations}. As a final ``preparation step" to introduce the Ward identities, we couple each non-linear BRST transformation to an external source, a standard procedure (see Ap.~\ref{appendixA}),

\begin{equation}
S_{\mathrm{ext}}=\int d^dx~\left(\Omega^{a}_{\mu}s_l A^{a}_{\mu}+L^{a}s_l c^a+K^{a}s_l\xi^{a}\right)\,,
\label{lrlcg35}
\end{equation}

\noindent with 

\begin{equation}
s_l\Omega^{a}_{\mu}=s_l L^a = s_l \xi^a=0\,.
\label{lrlcg36}
\end{equation}

\noindent This automatically implies

\begin{equation}
s_l S_{\mathrm{ext}}=0\,.
\label{lrlcg37}
\end{equation}

\noindent Hence, the full action we start with, taking into account the extended BRST transformations, the IR regulator for the Stueckelberg field and external sources coupled to non-linear BRST transformations is given by

\begin{eqnarray}
\Sigma &=& S_{\mathrm{YM}}+\int d^dx~\left(b^{h,a}\partial_{\mu}A^{a}_{\mu}-\frac{\alpha}{2}b^{h,a}b^{h,a}+\bar{c}^{a}\partial_{\mu}D^{ab}_{\mu}(A)c^b+\frac{\chi}{2}\bar{c}^{a}b^{h,a}\right)\nonumber\\
&+&\int d^dx\left(\tau^a\partial_{\mu}A^{h,a}_{\mu}-\bar{\varphi}^{ac}_{\mu}\EuScript{M}^{ab}(A^h)\varphi^{bc}_{\mu}-\bar{\omega}^{ac}_{\mu}\EuScript{M}^{ab}(A^h)\omega^{bc}_{\mu}-g\frac{\gamma^2}{\sqrt{2}}f^{abc}A^{h,a}_{\mu}(\varphi+\bar{\varphi})^{bc}_{\mu}\right)\nonumber\\
&-&\int d^dx\left(\bar{\beta}^{ac}_{\mu}\EuScript{M}^{ab}(A^h)\beta^{bc}_{\mu}-\bar{\zeta}^{ac}_{\mu}\EuScript{M}^{ab}(A^h)\zeta^{bc}_{\mu}+g\frac{\gamma^2}{\sqrt{2}}f^{abc}A^{h,a}_{\mu}(\beta+\bar{\beta})^{bc}_{\mu}\right)+S_{\mathrm{R}}+S_{\mathrm{ext}}\,. \nonumber\\
\label{lrlcg38}
\end{eqnarray}

\noindent By construction, action $\Sigma$ is invariant under $s_l$. The Slavnov-Taylor identity associated with this symmetry is

\begin{eqnarray}
\mathcal{S}(\Sigma)&=&\int d^dx\left[\frac{\delta\Sigma}{\delta\Omega^{a}_{\mu}}\frac{\delta\Sigma}{\delta A^{a}_{\mu}}+\frac{\delta\Sigma}{\delta L^{a}}\frac{\delta\Sigma}{\delta c^a}+\frac{\delta\Sigma}{\delta K^a}\frac{\delta\Sigma}{\delta \xi^a}+b^a\frac{\delta\Sigma}{\delta\bar{c}^a}+\omega^{ab}_{\mu}\frac{\delta\Sigma}{\delta\varphi^{ab}_{\mu}}+\omega^{ab}_{\mu}\frac{\delta\Sigma}{\delta\beta^{ab}_{\mu}}\right.\nonumber\\
&+&\left.(\bar{\varphi}^{ab}_{\mu}+\bar{\beta}^{ab}_{\mu})\frac{\delta\Sigma}{\delta\bar{\omega}^{ab}_{\mu}}\right]+\tilde{m}^4\frac{\partial\Sigma}{\partial\rho}+\chi\frac{\partial\Sigma}{\partial\alpha}=0\,.
\label{lrlcg39}
\end{eqnarray}

\noindent It is also simple to read off two more Ward identities from \eqref{lrlcg38}: The equation of motion for $b$ (\textit{i.e.} the gauge fixing conditions) and the antighost equation. They are, respectively,

\begin{equation}
\frac{\delta\Sigma}{\delta b^{h,a}}=\partial_{\mu}A^{a}_{\mu}-\alpha b^{h,a}+\frac{1}{2}\chi\bar{c}^a\,,
\label{lrlcg40}
\end{equation}
and

\begin{equation}
\frac{\delta\Sigma}{\delta\bar{c}^{a}}+\partial_{\mu}\frac{\delta\Sigma}{\delta\Omega^{a}_{\mu}}=-\frac{1}{2}\chi b^{h,a}\,.
\label{lrlcg41}
\end{equation}
Employing the \textit{Quantum Action Principle} we promote the aforementioned functional equations to the classical action $\Sigma$ to symmetries of the quantum action (generating functional of 1PI diagrams) $\Gamma$, \textit{i.e.}

\begin{eqnarray}
\mathcal{S}(\Gamma)&=&\int d^dx\left[\frac{\delta\Gamma}{\delta\Omega^{a}_{\mu}}\frac{\delta\Gamma}{\delta A^{a}_{\mu}}+\frac{\delta\Gamma}{\delta L^{a}}\frac{\delta\Gamma}{\delta c^a}+\frac{\delta\Gamma}{\delta K^a}\frac{\delta\Gamma}{\delta \xi^a}+b^a\frac{\delta\Gamma}{\delta\bar{c}^a}+\omega^{ab}_{\mu}\frac{\delta\Gamma}{\delta\varphi^{ab}_{\mu}}+\omega^{ab}_{\mu}\frac{\delta\Gamma}{\delta\beta^{ab}_{\mu}}\right.\nonumber\\
&+&\left.(\bar{\varphi}^{ab}_{\mu}+\bar{\beta}^{ab}_{\mu})\frac{\delta\Gamma}{\delta\bar{\omega}^{ab}_{\mu}}\right]+\tilde{m}^4\frac{\partial\Gamma}{\partial\rho}+\chi\frac{\partial\Gamma}{\partial\alpha}=0\,,
\label{lrlcg42}
\end{eqnarray}

\begin{equation}
\frac{\delta\Gamma}{\delta b^{h,a}}=\partial_{\mu}A^{a}_{\mu}-\alpha b^{h,a}+\frac{1}{2}\chi\bar{c}^a\,,
\label{lrlcg43}
\end{equation}
and

\begin{equation}
\frac{\delta\Gamma}{\delta\bar{c}^{a}}+\partial_{\mu}\frac{\delta\Gamma}{\delta\Omega^{a}_{\mu}}=-\frac{1}{2}\chi b^{h,a}\,.
\label{lrlcg44}
\end{equation}
Having the Ward Identities \eqref{lrlcg42},\eqref{lrlcg43} and \eqref{lrlcg44} written in terms of $\Gamma$, we can Legendre transform and write them with respect to the connected diagrams $\EuScript{W}$ (see Ap.~\ref{EA}),

\begin{equation}
\EuScript{W}[J,\EuScript{J},\mu]=\Gamma[\Phi,\EuScript{J},\mu]+\sum_{i}\int d^dx~J^{(\Phi)}_{i}\Phi_i\,,
\label{lrlcg45}
\end{equation}
where 

\begin{eqnarray}
\Phi&\equiv&\left\{A,b,\bar{c},c,\xi,\tau,\varphi,\bar{\varphi},\omega,\bar{\omega},\beta,\bar{\beta},\zeta,\bar{\zeta}\right\}\nonumber\\
\EuScript{J}&\equiv&\left\{\Omega,L,K\right\}\nonumber\\
\mu&\equiv&\left\{\rho,\tilde{m}^4,\alpha,\chi\right\}\,,
\label{lrlcg46}
\end{eqnarray}
with $J^{(\Phi)}$ being the source coupled to the field $\Phi$ at the path integral. Taking care of the (anti-)commuting nature of fields and sources, we can write

\begin{eqnarray}
\frac{\delta\Gamma}{\delta\Phi^{(c)}_{i}}&=&-J^{(\Phi^{(c)})}_i\nonumber\\
\frac{\delta\Gamma}{\delta\Phi^{(a)}_{i}}&=&J^{(\Phi^{(a)})}_i\,,
\label{lrlcg47}
\end{eqnarray}
where the superscripts $(a)$ and $(c)$ stand for anti-commuting and commuting, respectively. For the sources $\EuScript{J}$ and the parameters $\mu$ satisfy,

\begin{equation}
\frac{\delta\Gamma}{\delta\EuScript{J}}=\frac{\delta\EuScript{W}}{\delta\EuScript{J}}\,\,\,\,\,\,\,\,\,\,\mathrm{and}\,\,\,\,\,\,\,\,\,\,\frac{\partial\Gamma}{\partial\mu}=\frac{\partial\EuScript{W}}{\partial\mu}\,.
\label{lrlcg48}
\end{equation}
To avoid confusion we remind the reader that $\Phi$ stands for $\delta\EuScript{W}/\delta J^{(\Phi)}$, namely, it denotes the ``classical field" introduced in Ap.~\ref{EA}. We are ready to write the Ward identities with respect to $\EuScript{W}$,

\begin{itemize}

\item Slavnov-Taylor identity

\begin{eqnarray}
&\phantom{=}&\int d^dx\left[-J^{(A),a}_{\mu}(x)\frac{\delta\EuScript{W}}{\delta\Omega^{a}_{\mu}(x)}+J^{(c),a}(x)\frac{\delta\EuScript{W}}{\delta L^{a}(x)}-J^{(\xi),a}(x)\frac{\delta\EuScript{W}}{\delta K^a(x)}\right.\nonumber\\
&+&\left.J^{(\bar{c}),a}(x)\frac{\delta\EuScript{W}}{\delta J^{(b),a}(x)}-J^{(\varphi),ab}_{\mu}(x)\frac{\delta\EuScript{W}}{\delta J^{(\omega),ab}_{\mu}(x)}-J^{(\beta),ab}_{\mu}(x)\frac{\delta\EuScript{W}}{\delta J^{(\omega),ab}_{\mu}(x)}\right.\nonumber\\
&+&\left.\left(\frac{\delta\EuScript{W}}{\delta J^{(\bar{\varphi}),ab}_{\mu}(x)}+\frac{\delta\EuScript{W}}{\delta J^{(\bar{\beta}),ab}_{\mu}(x)}\right)J^{(\bar{\omega}),ab}_{\mu}(x)\right]+\tilde{m}^4\frac{\partial\EuScript{W}}{\partial\rho}+\chi\frac{\partial\EuScript{W}}{\partial\alpha}=0\,,
\label{lrlcg49}
\end{eqnarray}

\item Gauge-fixing condition

\begin{equation}
-J^{(b),a}(x)=\partial_{\mu}\frac{\delta\EuScript{W}}{\delta J^{(A),a}_{\mu}(x)}-\alpha\frac{\delta\EuScript{W}}{\delta J^{(b),a}(x)}+\frac{1}{2}\chi\frac{\delta\EuScript{W}}{\delta J^{(\bar{c}),a}(x)}\,,
\label{lrlcg50}
\end{equation}

\item Anti-ghost equation of motion

\begin{equation}
J^{(\bar{c}),a}(x)+\partial_{\mu}\frac{\delta\EuScript{W}}{\delta\Omega^{a}_{\mu}(x)}=-\frac{1}{2}\chi\frac{\delta\EuScript{W}}{\delta J^{(b),a}(x)}\,.
\label{lrlcg51}
\end{equation}

\end{itemize}
These Ward identities give us a powerful set up to prove the non-renormalization of the longitudinal sector of the gluon propagator and the gauge independence of physical operators. We expose the details of these derivations in the following subsections.

\subsection{Longitudinal part of the gluon propagator remains unchanged}

To prove this fact (already proved at the path integral level in Ch.~\ref{LCGrevisited}), we start by acting with the operator

\begin{equation}
\frac{\delta}{\delta J^{(A),a}_{\mu}(z)}\frac{\delta}{\delta J^{(\bar{c}),b}(y)}
\label{lrlcg52}
\end{equation}
on \eqref{lrlcg49}, the (non-perturbative) Slavnov-Taylor identity, and setting all sources and the parameters $\tilde{m}^{4}$ and $\chi$ to zero yields

\begin{equation}
\frac{\delta^{2}\EuScript{W}}{\delta J^{(\bar{c}),b}(y)\delta\Omega^{a}_{\mu}(z)}-\frac{\delta^2\EuScript{W}}{\delta J^{(A),a}_{\mu}(z)\delta J^{(b),b}(y)}=0\,.
\label{lrlcg53}
\end{equation}
Applying $\partial^{z}_{\mu}$ to eq.(\ref{lrlcg54}), we immediately get

\begin{equation}
\frac{\delta}{\delta J^{(\bar{c}),b}(y)}\partial^{z}_{\mu}\frac{\delta\EuScript{W}}{\delta\Omega^{a}_{\mu}(z)}-\partial^{z}_{\mu}\frac{\delta^2\EuScript{W}}{\delta J^{(A),a}_{\mu}(z)\delta J^{(b),b}(y)}=0\,.
\label{lrlcg54}
\end{equation}
At this point, it is clear the first term of (\ref{lrlcg54}) can be properly determined through the anti-ghost equation of motion (taking $\chi=0$). So, from (\ref{lrlcg51}) we get

\begin{equation}
J^{(\bar{c}),a}(z)+\partial^{z}_{\mu}\frac{\delta\EuScript{W}}{\delta\Omega^{a}_{\mu}(z)}=0\,,
\label{lrlcg55}
\end{equation}
which is plugged into eq.(\ref{lrlcg54}) giving rise to

\begin{equation}
\delta^{ab}\delta (y-z)+\partial^{z}_{\mu}\frac{\delta^{2}\EuScript{W}}{\delta J^{(A),a}_{\mu}(z)\delta J^{(b),b}(y)}=0\,\,\Rightarrow\,\,\delta^{ab}\delta (y-z)+\partial^{z}_{\mu}\langle A^{a}_{\mu}(z)b^{b}(y)\rangle_{c}=0\,.
\label{lrlcg56}
\end{equation}
Taking the Fourier transform of eq.(\ref{lrlcg56}) results

\begin{equation}
\delta^{ab}-ip_{\mu}\langle A^{a}_{\mu}b^{b}\rangle_{c}(p)=0 \,\, \Rightarrow \langle A^{a}_{\mu}b^{b}\rangle_{c}(p)=-i\frac{p_{\mu}}{p^2}\delta^{ab}\,,
\label{lrlcg57}
\end{equation}
where we explored Lorentz invariance to obtain the mixed $\langle Ab\rangle$ two-point function. Finally, we act with $\delta/\delta J^{(A),b}_{\nu}(y)$ on the Ward identity associated with the gauge-fixing condition eq.(\ref{lrlcg50}), with $\chi=0$. The result is

\begin{equation}
\partial^{z}_{\mu}\frac{\delta^{2}\EuScript{W}}{\delta J^{(A),b}_{\nu}(y)\delta J^{(A),a}_{\mu}(z)}-\alpha\frac{\delta^2\EuScript{W}}{\delta J^{(A),b}_{\nu}(y)\delta J^{(b),a}(z)}=0
\label{lrlcg58}
\end{equation}
which is automatically translated to

\begin{equation}
\partial^{z}_{\mu}\langle A^{b}_{\nu}(y)A^{a}_{\mu}(z)\rangle_c-\alpha\langle A^{b}_{\nu}(y)b^{a}(z)\rangle_c=0\,.
\label{lrlcg59}
\end{equation}
By taking the Fourier transform of eq.(\ref{lrlcg59}) we obtain the longitudinal part of the gluon propagator,

\begin{equation}
p_{\mu}\langle A^{a}_{\mu}A^{b}_{\nu}\rangle_c(p)=\alpha\frac{p_{\nu}}{p^2}\delta^{ab}\,.
\label{lrlcg60}
\end{equation}
It is clear, from eq.(\ref{lrlcg60}) that the general form of the gluon propagator is

\begin{equation}
\langle A^{a}_{\mu}A^{b}_{\nu}\rangle_c(p)=\delta^{ab}\left[\left(\delta_{\mu\nu}-\frac{p_{\mu}p_{\nu}}{p^2}\right)\mathcal{D}_T(p^2)+\frac{p_{\mu}p_{\nu}}{p^2}\mathcal{D}_{L}(p^2)\right]\,,
\label{lrlcg61}
\end{equation}
with

\begin{equation}
\mathcal{D}_{L}(p^2)=\frac{\alpha}{p^2}\,.
\label{lrlcg62}
\end{equation}
Once again we reinforce the non-triviality of the non-renormalization of the longitudinal sector in a theory which breaks (standard) BRST softly. In standard perturbative Yang-Mills, it is easy to prove the relation (\ref{lrlcg62}) to all orders by using the standard BRST symmetry. In the GZ setting, this was possible thanks to the non-perturbative BRST symmetry. Without this novel symmetry, this would be a highly non-trivial question. It is remarkable that this novel symmetry, although a deformation of the perturbative BRST preserves its strength to establish such result. 

\subsection{Physical operators are independent of gauge parameter}

One of the main worries one might have in a standard-BRST breaking theory is gauge dependence of what would-be physical observables. As is widely known in standard gauge theories, it is precisely the BRST symmetry enjoyed by them which controls gauge parameter dependence of physical operators \cite{Piguet:1984js}. In fact, some claims that standard-BRST soft breaking theories are inconsistent with gauge independence were put forward in \cite{Lavrov:2011wb}. In this subsection, we show that the non-perturbative Slavnov-Taylor identity allows an elegant algebraic proof for gauge independence for correlation functions of operators $\EuScript{O}$ which belong to the non-perturbative BRST cohomology and have vanishing ghost number \textit{i.e.}

\begin{equation}
s_l\EuScript{O}=0\,,\,\,\,\,\,\,\,\,\, \EuScript{O}\neq s_l\tilde{\EuScript{O}}\,,
\label{lrlcg63}
\end{equation}
for any local operator $\tilde{\EuScript{O}}$. In order to compute the $n$-point function $\langle \EuScript{O}(x_1)\ldots\EuScript{O}(x_n)\rangle$, we introduce in the starting action the term

\begin{equation}
\int d^dx~J^{(\EuScript{O})}(x)\EuScript{O}(x)\,,
\label{lrlcg64}
\end{equation}
where we choose the external sources $J^{(\EuScript{O})}$ to be BRST invariant. Since $\EuScript{O}$ is also BRST invariant, the introduction of (\ref{lrlcg64}) keeps the non-perturbative Slavnov-Taylor identity \eqref{lrlcg49} untouched. To compute the correlator $\langle \EuScript{O}(x_1)\ldots\EuScript{O}(x_n)\rangle$, we proceed in the standard way,

\begin{equation}
\langle \EuScript{O}(x_1)\ldots\EuScript{O}(x_n)\rangle_c=\frac{\delta}{\delta J^{(\EuScript{O})}(x_1)}\ldots\frac{\delta}{\delta J^{(\EuScript{O})}(x_n)}\EuScript{W}\Big|_{J=\EuScript{J}=\tilde{m}=\rho=\chi=0}\,.
\label{lrlcg65}
\end{equation}
The proof that \eqref{lrlcg65} in independent of $\alpha$ follows immediately from the Slavnov-Taylor identity \eqref{lrlcg49}. First, we act with the operator

\begin{equation}
\frac{\delta}{\delta J^{(\EuScript{O})}(x_1)}\ldots\frac{\delta}{\delta J^{(\EuScript{O})}(x_n)}
\label{lrlcg66}
\end{equation}
and set sources and $\tilde{m}$ to zero. After we apply $\partial/\partial\chi$ on \eqref{lrlcg49} on the resulting expression, which becomes

\begin{equation}
\frac{\partial}{\partial\alpha}\frac{\delta^{n}\EuScript{W}}{\delta J^{(\EuScript{O})}(x_1)\ldots \delta J^{(\EuScript{O})}(x_n)}\Big|_{J=\EuScript{J}=\tilde{m}=\rho=0}-\chi\frac{\partial^2}{\partial\chi\partial\alpha}\frac{\delta^{n}\EuScript{W}}{\delta J^{(\EuScript{O})}(x_1)\ldots \delta J^{(\EuScript{O})}(x_n)}\Big|_{J=\EuScript{J}=\tilde{m}=\rho=0}=0\,.
\label{lrlcg67}
\end{equation}
Finally, setting $\chi=0$ and using eq.(\ref{lrlcg65}) we obtain,

\begin{equation}
\frac{\partial}{\partial\alpha}\langle \EuScript{O}(x_1)\ldots\EuScript{O}(x_n)\rangle_c=0\,,
\label{lrlcg68}
\end{equation}
which is nothing but the all-order proof of the independence of $n$-point function of operators which belong to the non-perturbative BRST cohomology of the gauge parameter. As in the previous subsection, we emphasize how important is the presence of an exact and nilpotent (non-perturbative) BRST symmetry for this proof. In fact, this is an immediate consequence of its existence. 

\section{Refinement in local fashion}

In Ch.~\ref{LCGrevisited} we have proposed a refinement of the GZ action in LCG which is consistent with the proposed non-perturbative BRST symmetry\footnote{At this level, the proposal was based on the non-local BRST transformations (\ref{npbrst38})}. In this section, we translate the refinement of the GZ action in LCG within the complete local setting introduced in this chapter. 

From our experience on the construction of the refinement of the GZ action, the following local composite operators are candidates to be taken into account,

\begin{equation}
A^{h,a}_{\mu}A^{h,a}_{\mu}\,,\,\,\,\,\,\bar{\omega}^{ab}_{\mu}\omega^{ab}_{\mu}\,,\,\,\,\,\,\bar{\varphi}^{ab}_{\mu}\varphi^{ab}_{\mu}\,,\,\,\,\,\,\bar{\beta}^{ab}_{\mu}\beta^{ab}_{\mu}\,\,\,\,\,\mathrm{and}\,\,\,\,\,\bar{\zeta}^{ab}_{\mu}\zeta^{ab}_{\mu}\,.
\label{lrlcg69}
\end{equation} 
In principle, the form of the refinement term to be added to the GZ action (\ref{lrlcg20}) is written as

\begin{equation}
\tilde{S}_{\mathrm{cond}}=\int d^dx\left[\frac{m^2}{2}A^{h,a}_{\mu}A^{h,a}_{\mu}+M^2_{1}\bar{\omega}^{ab}_{\mu}\omega^{ab}_{\mu}+M^{2}_{2}\bar{\varphi}^{ab}_{\mu}\varphi^{ab}_{\mu}+M^2_{3}\bar{\beta}^{ab}_{\mu}\beta^{ab}_{\mu}+M^{2}_{4}\bar{\zeta}^{ab}_{\mu}\zeta^{ab}_{\mu}\right]\,.
\label{lrlcg70}
\end{equation}

\noindent Nevertheless, we have to impose some contraints over (\ref{lrlcg70}). First, we demand BRST invariance of (\ref{lrlcg70}). This imposes

\begin{equation}
M^2_1=-M^2_2=-M^2_3\equiv M^2\,.
\label{lrlcg71}
\end{equation}
Now, requiring invariance of \eqref{lrlcg70} with respect to the discrete transformations (\ref{lrlcg22}), we constraint

\begin{equation}
M^{2}_4=M^2\,.
\label{lrlcg72}
\end{equation}
Finally, the refinement term is reduced to 

\begin{equation}
S_{\mathrm{cond}}=\int d^dx\left[\frac{m^2}{2}A^{h,a}_{\mu}A^{h,a}_{\mu}+M^2\left(\bar{\omega}^{ab}_{\mu}\omega^{ab}_{\mu}-\bar{\varphi}^{ab}_{\mu}\varphi^{ab}_{\mu}-\bar{\beta}^{ab}_{\mu}\beta^{ab}_{\mu}+\bar{\zeta}^{ab}_{\mu}\zeta^{ab}_{\mu}\right)\right]\,.
\label{lrlcg73}
\end{equation}
As previously discussed, the mass parameters introduced in the refinement are not free. They are fixed dynamically through the minimization of the effective potential when these dimension two operators are taken into account. The computation of these parameters, although crucial, is not part of the scope of this thesis and is object of ongoing investigations.

A very important remark at this level is that the mass parameters $m^2$ and $M^2$ are coupled to dimension two operators which belong to the cohomology of the non-perturbative BRST. Since these mass parameters are not coupled to BRST-exact terms, they are not akin to gauge parameters \textit{i.e.} they are physical parameters. Also, we note that although $\bar{\zeta}\zeta$ is BRST invariant by its own, the discrete symmetry (\ref{lrlcg22}) connects its mass parameter with the mass parameter of the other auxiliary fields.

\noindent \textbf{Remark:} After the submission of this thesis, the paper \cite{Capri:2016aqq} was published. It contains the main results here presented.

\chapter{En route to non-linear gauges: Curci-Ferrari gauges}\label{CFGaugeCh}

The entire machinery introduced and developed in the last chapters was simply to extend in a consistent fashion the RGZ setting to LCG. This problem is notably non-trivial and many subtleties show up. From the technical issue of dealing with a non Hermitian Faddeev-Popov operator to the presence of a gauge parameter, the RGZ scenarion in LCG led to the introduction of many important technical and conceptual novelties with respect to the standard formulation in the Landau gauge. 

One might very well insist on the idea of extending further the formalism to even more complicated gauges and see what kind of new effects to the formalism these more elaborated choices can bring. It turns out, however, that the formalism introduced to deal with LCG is more powerful than one might expect at first glance. In this chapter, we will argue why what was introduced so far is enough to construct the RGZ action in a class of non-linear gauges known as Curci-Ferrari gauges. We emphasize that, so far, neither lattice nor functional methods results are available for this class of gauges in such a way that the results here presented, once confronted with future results from these different approaches, might represent a very non-trivial check of our proposal. 

\section{Establishing the Gribov problem in Curci-Ferrari gauges} \label{conv}

In \cite{Baulieu:1981sb,Delduc:1989uc} an one-parameter family of renormalizable non-linear gauges was introduced. Quite often, these gauges are called Curci-Ferrari gauges because the gauge-fixing Lagrangian is exactly the same introduced in \cite{Curci:1976bt,Curci:1976kh} by Curci and Ferrari. There, however, a mass term for the gluons is introduced to discuss massive Yang-Mills theories. Inhere, we will deal with the massless case.  

\subsection{Conventions and standard BRST quantization}

The gauge fixed Yang-Mills action in Curci-Ferrari gauges in $d$ Euclidean dimensions is given by

\begin{eqnarray}
S_{\mathrm{FP}}&=& S_{\mathrm{YM}}+s\int d^dx~\bar{c}^{a}\left[\partial_{\mu}A^{a}_{\mu}-\frac{\alpha}{2}\left(b^a-\frac{g}{2}f^{abc}\bar{c}^{b}c^{c}\right)\right]\nonumber\\
&=&S_{\mathrm{YM}}+\int d^dx\left[b^{a}\partial_{\mu}A^{a}_{\mu}+\bar{c}^{a}\partial_{\mu}D^{ab}_{\mu}(A)c^{b}-\frac{\alpha}{2}b^{a}b^{a}+\frac{\alpha}{2}gf^{abc}b^{a}\bar{c}^{b}c^{c}\right.\nonumber\\
&+&\left.\frac{\alpha}{8}g^{2}f^{abc}f^{cde}\bar{c}^{a}\bar{c}^{b}c^{d}c^{e}\right]\,,
\label{csbrst1}
\end{eqnarray}
with $\alpha$ a non-negative gauge parameter. This action is manifestly invariant under the standard BRST transformations,

\begin{eqnarray}
sA^{a}_{\mu}&=&-D^{ab}_{\mu}c^{b}\nonumber\\
sc^{a}&=&\frac{g}{2}f^{abc}c^{b}c^{c}\nonumber\\
s\bar{c}^{a}&=&b^{a}\nonumber\\
sb^{a}&=&0\,,
\label{csbrst2}
\end{eqnarray}
and is renormalizable to all orders in perturbation theory \cite{Delduc:1989uc}. It is worth mentioning that action \eqref{csbrst1} contains an interaction term between Faddeev-Popov ghosts and the auxiliary field $b$ and a quartic interaction of ghosts. The presence of such terms is responsible to drive different dynamical effects with respect to linear covariant gauges, for instance, as we shall see. In particular, the equation of motion for the auxiliary field $b$ and also for the anti-ghost $\bar{c}$ do not correspond to Ward identities in this case due to the non-linear character of this gauge. In linear gauges, these equations do correspond to Ward identities which are pivotal to the renormalizability proof. 

On the other hand, action \eqref{csbrst1} enjoys another global symmetry besides BRST which will generate a Ward identity that plays the analogous role of the anti-ghost equation. This symmetry is known as $SL(2,\mathbb{R})$ symmetry and together with the Slavnov-Taylor identity guarantee the all order proof of renormalizability of such gauge \cite{Delduc:1989uc}. The $SL(2,\mathbb{R})$ symmetry is defined by the following set of transformations:

\begin{eqnarray}
\delta\bar{c}^{a}&=&c^a\nonumber\\
\delta b^{a}&=&\frac{g}{2}f^{abc}c^{b}c^{c}\nonumber\\
\delta A^{a}_{\mu}&=&\delta c^{a} = 0\,.
\label{csbrst3}
\end{eqnarray}

\noindent An useful property is that the $SL(2,\mathbb{R})$ operator $\delta$ commutes with the BRST operator $s$ \textit{i.e.} $[s,\delta]=0$.

\subsection{Construction of a copies equation}

As discussed in Ch.~\ref{ch.2}, given a gauge condition $F[A]=0$, one way to characterize the existence of Gribov copies by performing a gauge transformation over $F[A]=0$ and looking for solutions of the resulting equation - the copies equation. Nevertheless, in the case of Curci-Ferrari gauges, is not clear how to read off an equation as $F[A]=0$, with $F$ a functional of the gauge field, from action \eqref{csbrst1}. To see this, let us compute the equation of motion for $b$, which in linear gauges gives the gauge-fixing condition,

\begin{equation}
\frac{\delta S_{\mathrm{FP}}}{\delta b^a}=\partial_{\mu}A^{a}_{\mu}-\alpha b^a + \frac{\alpha}{2}g f^{abc}b^{a}\bar{c}^{b}c^c\,.
\label{cce1}
\end{equation}
It is clear, due to the presence of $(\bar{c},c)$ in \eqref{cce1} that it cannot be written as $F[A]=0$. As a comparison, for instance, in LCG, the last term of \eqref{cce1} is not present and at the level of gauge-fixing, the field $b$ is nothing but a fixed function. 

On the other hand, it is possible to cast the Curci-Ferrari gauges in a similar fashion of linear covariant gauges by a convenient shift on the $b$ field. Therefore, at the level of the path integral, we can perform the shift

\begin{equation}
b^a\,\,\longrightarrow\,\, b'^a=b^a-\frac{g}{2}f^{abc}\bar{c}^bc^c\,,
\label{cce2}
\end{equation}
which entails a trivial Jacobian. The Yang-Mills action in Curci-Ferrari gauges is then written as

\begin{equation}
S_{\mathrm{FP}} = S_{\mathrm{YM}}+s\int d^4x~\bar{c}^{a}\left(\partial_{\mu}A^{a}_{\mu}-\frac{\alpha}{2}b'^a\right)\,,
\label{cce3}
\end{equation}
which is formally the same as in LCG. However, the difference between these gauges arises from the fact that the action of BRST transformations also change. Nevertheless, we can still exploit the similarity between these gauges at the formal level and keep in mind the different roles played by $b$ and $b'$. So, as a ``gauge-fixing equation", we write Curci-Ferrari gauges as

\begin{equation}
\partial_{\mu}A^{a}_{\mu}=\alpha b'^a\,.
\label{cce4}
\end{equation}
We can treat \eqref{cce4} as our desired $F[A]=0$ equation. Since it is formally identical to the gauge-fixing equation for LCG, we can immediately conclude that their solutions are formally the same. As a consequence, the framework contructed in Ch.~\ref{LCGrevisited} and Ch.~\ref{locnonpBRST} to deal with the Gribov problem in LCG, can be trivially imported to the Curci-Ferrari case. This is precisely the subject of next section.

The shifted BRST transformations are expressed as

\begin{eqnarray}
sA^{a}_{\mu}&=&-D^{ab}_{\mu}c^{b}\nonumber\\
sc^{a}&=&\frac{g}{2}f^{abc}c^{b}c^{c}\nonumber\\
s\bar{c}^{a}&=&b'^a+\frac{g}{2}f^{abc}\bar{c}^{b}c^{c}\nonumber\\
sb'^a&=&-\frac{g}{2}f^{abc}b'^bc^c+\frac{g^{2}}{8}f^{abc}f^{cde}\bar{c}^{b}c^{d}c^{e}\,.
\label{cce5}
\end{eqnarray}
Explicitly, the Faddeev-Popov action in terms of $b'$ is given by

\begin{equation}
S_{\mathrm{FP}}=S_{\mathrm{YM}}+\int d^4x\left[b'^a\partial_{\mu}A^{a}_{\mu}+\frac{1}{2}\bar{c}^{a}(\partial_{\mu}D^{ab}_{\mu}+D^{ab}_{\mu}\partial_{\mu})c^{b}-\frac{\alpha}{2}b'^{a}b'^{a}+\frac{\alpha}{8}g^{2}f^{abc}f^{cde}\bar{c}^{e}\bar{c}^{a}c^{b}c^{d}\right]\,,
\label{cce6}
\end{equation}

\noindent and the equation of motion of $b'^a$ enforces the gauge condition (\ref{cce4}),

\begin{equation}
\frac{\delta S_{\mathrm{FP}}}{\delta b'^a}=\partial_{\mu}A^{a}_{\mu}-\alpha b'^a\,,
\label{cce7}
\end{equation}
The $SL(2,\mathbb{R})$ symmetry takes the simpler form

\begin{eqnarray}
\delta\bar{c}^{a}&=&c^{a}\nonumber\\
\delta b'^a &=& 0\nonumber\\
\delta A^{a}_{\mu}&=&\delta c^{a}=0\,.
\label{cce8}
\end{eqnarray}
We see that the shift over the $b$ field simplifies the structure of the action (there are no $b'$-ghosts vertices) and the $SL(2,\mathbb{R})$ transformations form. However, it introduces a much more involved form for the BRST transformations. Therefore, we should be able to explore when using $b'$ instead of $b$ (and vice-versa) is more convenient.

\section{Construction of the GZ action}

In the last section we have established a connection between the manifestation of the Gribov problem in Curci-Ferrari gauges with LCG. The latter was object of study of Ch.~\ref{LCGrevisited} and Ch.~\ref{locnonpBRST}. In particular, since the copies equation for Curci-Ferrari and LCG are formally identical, the removal of Gribov copies from the Curci-Ferrari path integral follows exactly the same route as in LCG. As a byproduct, the resulting GZ action in Curci-Ferrari gauges enjoys non-perturbative BRST invariance. From a different perspective, we can establish from the beginning a non-perturbative BRST quantization as already proposed in \cite{Capri:2015ixa,Capri:2015nzw}.  Following this prescription, we employ the non-perturbative BRST quantization to Curci-Ferrari gauges,

\begin{eqnarray}
S^{\mathrm{CF}}_{\mathrm{GZ}}&=& S_{\mathrm{YM}}+s_{\gamma^2}\int d^dx~\bar{c}^{a}\left[\partial_{\mu}A^{a}_{\mu}-\frac{\alpha}{2}\left(b^{h,a}-\frac{g}{2}f^{abc}\bar{c}^{b}c^{c}\right)\right]\nonumber\\
&+&\int d^dx\left(\bar{\varphi}^{ac}_{\mu}\left[\EuScript{M}(A^h)\right]^{ab}\varphi^{bc}_{\mu}-\bar{\omega}^{ac}_{\mu}\left[\EuScript{M}(A^h)\right]^{ab}\omega^{bc}_{\mu}+g\gamma^2f^{abc}A^{h,a}_{\mu}(\varphi+\bar{\varphi})^{bc}_{\mu}\right)\nonumber\\
&=& S_{\mathrm{YM}}+\int d^dx\left[b^{h,a}\partial_{\mu}A^{a}_{\mu}+\bar{c}^{a}\partial_{\mu}D^{ab}_{\mu}c^{b}-\frac{\alpha}{2}b^{h,a}b^{h,a}+\frac{\alpha}{2}gf^{abc}b^{h,a}\bar{c}^{b}c^{c}\right.\nonumber\\
&+&\left.\frac{\alpha}{8}g^{2}f^{abc}f^{cde}\bar{c}^{a}\bar{c}^{b}c^{d}c^{e}\right] 
+\int d^dx\left(\bar{\varphi}^{ac}_{\mu}\left[\EuScript{M}(A^h)\right]^{ab}\varphi^{bc}_{\mu}-\bar{\omega}^{ac}_{\mu}\left[\EuScript{M}(A^h)\right]^{ab}\omega^{bc}_{\mu}\right.\nonumber\\
&+&\left.g\gamma^2f^{abc}A^{h,a}_{\mu}(\varphi+\bar{\varphi})^{bc}_{\mu}\right)\,.
\label{gz1}
\end{eqnarray}
with $s_{\gamma^2}$ the non-perturbative and nilpotent BRST operator, see eq.\eqref{npbrst38}.

As discussed in the context of LCG, the proposed non-perturbative BRST quantization gives rise to a non-local action. From eq.\eqref{npbrst38}, even the non-perturbative BRST transformations are non-local. It is of uttermost interest to cast all the framework in local fashion so that all the powerful machinery of local quantum field theories are at our disposal. As discussed in Ch.~\ref{locnonpBRST} the localization of this set up in LCG is possible. The extension to Curci-Ferrari gauges is straightforward and we will report the explicit local form in Sect.~7. However, before turning to this issue, we expose some features of the tree-level gluon propagator and for this purpose, it is not necessary to go through all the localization procedure. 

For completeness, we present the form of the GZ action in Curci-Ferrari gauges in terms of the shifted field $b'^{h}$ \eqref{cce2},

\begin{eqnarray}
S^{\mathrm{CF}}_{\mathrm{GZ}}&=&S_{\mathrm{YM}}+\int d^dx\left[b'^{h,a}\partial_{\mu}A^{a}_{\mu}+\frac{1}{2}\bar{c}^{a}(\partial_{\mu}D^{ab}_{\mu}+D^{ab}_{\mu}\partial_{\mu})c^{b}-\frac{\alpha}{2}b'^{h,a}b'^{h,a}+\frac{\alpha}{8}g^{2}f^{abc}f^{cde}\bar{c}^{e}\bar{c}^{a}c^{b}c^{d}\right]\nonumber\\
&+&\int d^dx\left(\bar{\varphi}^{ac}_{\mu}\left[\EuScript{M}(A^h)\right]^{ab}\varphi^{bc}_{\mu}-\bar{\omega}^{ac}_{\mu}\left[\EuScript{M}(A^h)\right]^{ab}\omega^{bc}_{\mu}+g\gamma^2f^{abc}A^{h,a}_{\mu}(\varphi+\bar{\varphi})^{bc}_{\mu}\right)\,,
\label{gz3}
\end{eqnarray}
which is invariant under the non-perturbative set of BRST transformations,

\begin{align}
s_{\gamma^2}A^{a}_{\mu}&=-D^{ab}_{\mu}c^b\,,     &&s_{\gamma^2}c^a=\frac{g}{2}f^{abc}c^bc^c\,, \nonumber\\
s_{\gamma^2}\bar{c}^a&=b'^{h,a}+\frac{g}{2}f^{abc}\bar{c}^bc^c\,,     &&s_{\gamma^2}b'^{h,a}=-\frac{g}{2}f^{abc}b'^bc^c+\frac{g^2}{8}f^{abc}f^{cde}\bar{c}^bc^dc^e\,, \nonumber\\
s_{\gamma^2}\varphi^{ab}_{\mu}&=\omega^{ab}_{\mu}\,,   &&s_{\gamma^2}\omega^{ab}_{\mu}=0\,, \nonumber\\
s_{\gamma^2}\bar{\omega}^{ab}_{\mu}&=\bar{\varphi}^{ab}_{\mu}-\gamma^2gf^{cdb}\int d^dy~A^{h,c}_{\mu}(y)\left[\EuScript{M}^{-1}(A^h)\right]^{da}_{yx}\,,         &&s_{\gamma^2}\bar{\varphi}^{ab}_{\mu}=0\,. 
\label{gz2}
\end{align}
As in LCG, the gap equation which fixes the Gribov parameter is manifestly gauge invariant, namely

\begin{equation}
\frac{\partial\mathcal{E}_0}{\partial\gamma^2}=0\,\,\,\Rightarrow\,\,\,\langle gf^{abc}A^{h,a}_{\mu}(\varphi+\bar{\varphi})^{bc}_{\mu}\rangle=2d\gamma^2(N^2-1)\,. 
\label{gz23}
\end{equation}
The integration over $b'^h$ can be performed and the resulting action is

\begin{eqnarray}
S^{\mathrm{CF}}_{\mathrm{GZ}}&=&S_{\mathrm{YM}}+\int d^dx\left[\frac{(\partial_{\mu}A^{a}_{\mu})^{2}}{2\alpha}+\frac{1}{2}\bar{c}^{a}(\partial_{\mu}D^{ab}_{\mu}+D^{ab}_{\mu}\partial_{\mu})c^{b}+\frac{\alpha}{8}g^{2}f^{abc}f^{cde}\bar{c}^{e}\bar{c}^{a}c^{b}c^{d}\right]\nonumber\\
&+&\int d^dx\left(\bar{\varphi}^{ac}_{\mu}\left[\EuScript{M}(A^h)\right]^{ab}\varphi^{bc}_{\mu}-\bar{\omega}^{ac}_{\mu}\left[\EuScript{M}(A^h)\right]^{ab}\omega^{bc}_{\mu}+g\gamma^2f^{abc}A^{h,a}_{\mu}(\varphi+\bar{\varphi})^{bc}_{\mu}\right)\,.\nonumber\\
\label{gz5}
\end{eqnarray}
From action \eqref{gz3} - or \eqref{gz5} - the tree-level gluon propagator computation is trivial and yields

\begin{equation}
\langle A^{a}_{\mu}(p)A^{b}_{\nu}(-p)\rangle=\delta^{ab}\left[\frac{p^2}{p^4+2g^2N\gamma^4}\left(\delta_{\mu\nu}-\frac{p_{\mu}p_{\nu}}{p^2}\right)+\frac{\alpha}{p^2}\frac{p_{\mu}p_{\nu}}{p^2}\right]\,.
\label{gz6}
\end{equation}
The transverse part receives effects from the restriction of the path integral domain due to the presence of $\gamma$, while the longitudinal part is equal to the perturbative result. We emphasize this is a tree-level computation. The transverse part has the Gribov-type behavior. It is IR suppressed and its form factor goes to zero at zero-momentum. Also, this propagator violates positivity and as such, no physical particle interpretation can be attached to the gluon field. However, as presented in Ch.~\ref{RGZch}, the GZ action suffers from IR instabilities and dimension-two condensates are formed. In the next section we take into account these effects and discuss their consequences to the gluon propagator. 

\section{Dynamical generation of condensates}

\subsection{Refinement of the Gribov-Zwanziger action}

In the Landau gauge, it was noted that the Gribov-Zwanziger action suffers from IR instabilities, \cite{Dudal:2008sp}. In particular, already at the one-loop level it is possible to compute a non-vanishing value for dimension-two condensates. Those are proportional to the Gribov parameter $\gamma$ putting in evidence that the non-trivial background of the Gribov horizon contributes to the formation of dimension-two condensates. In \cite{Capri:2015pja,Capri:2015nzw}, these results were extended to linear covariant gauges in the non-perturbative BRST framework. Also, analogous results were obtained for maximal Abelian and Coulomb gauges, \cite{Capri:2015pfa,Guimaraes:2015bra}. 

For Curci-Ferrari gauges, we can proceed in full analogy to the linear covariant gauges. In particular, both condensates considered in \cite{Dudal:2008sp}, namely,

\begin{equation}
\langle A^{h,a}_{\mu}(x)A^{h,a}_{\mu}(x)\rangle\,\,\,\,\,\mathrm{and}\,\,\,\,\,\langle \bar{\varphi}^{ab}_{\mu}(x)\varphi^{ab}_{\mu}(x)-\bar{\omega}^{ab}_{\mu}(x)\omega^{ab}_{\mu}(x)\rangle\,,
\label{dgc5.0}
\end{equation}
are dynamically generated. This fact is easily proved by the introduction of the aforementioned dimension-two operators coupled to constant sources into the Gribov-Zwanziger action. Therefore, let us consider the generating functional $\mathcal{E}(m,J)$ defined as 

\begin{equation}
\mathrm{e}^{-V\mathcal{E}(m,J)}=\int\left[\EuScript{D}\Phi\right]\mathrm{e}^{-(S^{\mathrm{CF}}_{\mathrm{GZ}}+m\int d^dx~A^{h,a}_{\mu}A^{h,a}_{\mu}-J\int d^dx(\bar{\varphi}^{ab}_{\mu}\varphi^{ab}_{\mu}-\bar{\omega}^{ab}_{\mu}\omega^{ab}_{\mu}))}
\label{dgc6.0}
\end{equation}
with $m$ and $J$ being constant sources. Hence

\begin{eqnarray}
\langle \bar{\varphi}^{ab}_{\mu}\varphi^{ab}_{\mu}-\bar{\omega}^{ab}_{\mu}\omega^{ab}_{\mu} \rangle &=& -\frac{\partial \mathcal{E}(m,J)}{\partial J}\Big|_{m=J=0}\nonumber\\
\langle A^{h,a}_{\mu}A^{h,a}_{\mu}\rangle &=& \frac{\partial\mathcal{E}(m,J)}{\partial m}\Big|_{m=J=0}\,.
\label{dgc7.0}
\end{eqnarray}
At one-loop order, employing dimensional regularization,

\begin{equation}
{\cal E}(m,J)=\frac{(d-1)(N^2-1)}{2}\int \frac{d^dp}{(2\pi)^d}~\mathrm{ln}\left(p^2+\frac{2\gamma^4g^2N}{p^2+J}+2m\right)-d\gamma^4(N^2-1)\,,
\label{dgc8.0}
\end{equation}
which results in

\begin{equation}
\langle \bar{\varphi}^{ac}_{\mu}\varphi^{ac}_{\mu}-\bar{\omega}^{ac}_{\mu}\omega^{ac}_{\mu}\rangle = g^2\gamma^4N(N^2-1)(d-1)\int \frac{d^dp}{(2\pi)^d}\frac{1}{p^2}\frac{1}{(p^4+2g^2\gamma^4N)}
\label{dgc9.0}
\end{equation}
and

\begin{equation}
\langle A^{h,a}_{\mu}A^{h,a}_{\mu}\rangle = -2g^2\gamma^4N(N^2-1)(d-1)\int\frac{d^dp}{(2\pi)^d}\frac{1}{p^2}\frac{1}{(p^4+2g^2\gamma^4N)}\,.
\label{dgc10.0}
\end{equation}
From eq.\eqref{dgc9.0} and \eqref{dgc10.0}, we see immediately the presence of the Gribov parameter as a prefactor. This implies the non-triviality of the value of such condensates due to the restriction of the path integral domain to the Gribov region, encoded in $\gamma$. Also, as discussed in \cite{Capri:2015nzw,Dudal:2008xd}, the integrals appearing in \eqref{dgc9.0} and \eqref{dgc10.0} are perfectly convergent for $d=3,4$, while for $d=2$ those develop an IR singularity. This behavior suggests the inclusion of \eqref{dgc5.0} to the Gribov-Zwanziger action for $d=3,4$, while keeping the action untouched for $d=2$. The absence of refinement of the Gribov-Zwanziger action in $d=2$ can be made more precise, see \cite{Capri:2015nzw,Dudal:2008xd}. Essentially, in $d=2$ it turns out to be impossible to stay within the Gribov region by introducing dimension two condensates \cite{Capri:2015nzw,Dudal:2008xd}. The same argument is easily extended to Curci-Ferrari gauges.

Taking into account these considerations, for the Refined Gribov-Zwanziger action in $d=3,4$ we obtain

\begin{eqnarray}
S^{\mathrm{CF}}_{\mathrm{RGZ}}&=& S_{\mathrm{YM}}+\int d^dx\left[b^{h,a}\partial_{\mu}A^{a}_{\mu}+\bar{c}^{a}\partial_{\mu}D^{ab}_{\mu}c^{b}-\frac{\alpha}{2}b^{h,a}b^{h,a}+\frac{\alpha}{2}gf^{abc}b^{h,a}\bar{c}^{b}c^{c}\right.\nonumber\\
&+&\left.\frac{\alpha}{8}g^{2}f^{abc}f^{cde}\bar{c}^{a}\bar{c}^{b}c^{d}c^{e}\right] 
+\int d^dx\left(\bar{\varphi}^{ac}_{\mu}\left[\EuScript{M}(A^h)\right]^{ab}\varphi^{bc}_{\mu}-\bar{\omega}^{ac}_{\mu}\left[\EuScript{M}(A^h)\right]^{ab}\omega^{bc}_{\mu}\right.\nonumber\\
&+&\left.g\gamma^2f^{abc}A^{h,a}_{\mu}(\varphi+\bar{\varphi})^{bc}_{\mu}\right)+\frac{m^2}{2}\int d^dx\,A^{h,a}_{\mu}A^{h,a}_{\mu}-M^2\int d^dx\,\left(\bar{\varphi}^{ab}_{\mu}\varphi^{ab}_{\mu}-\bar{\omega}^{ab}_{\mu}\omega^{ab}_{\mu}\right)\,.\nonumber\\
\label{dgc11.0}
\end{eqnarray}
while in $d=2$ the Gribov-Zwanziger action is left unmodified and expression \eqref{gz1} is preserved.

\subsection{A remark on the gluon-ghost condensate}

In the last decade, much effort has been undertaken to understand QCD vacuum and, in particular, pure Yang-Mills vacuum. Particular attention was devoted to the dynamical formation of condensates which could introduce non-perturbative effects related to chiral symmetry breaking (in the specific case of QCD) and color confinement. Also, dimension-two gluon condensates were on the mainstream of analytical and numerical approaches to confinement due to the possibility of giving rise to a possible mechanism for dynamical mass generation. On the other hand, the dimension-two gluon condensate $\langle A^{a}_{\mu}A^{a}_{\mu}\rangle$ is not gauge invariant for a generic choice of a covariant renormalizable gauge and a direct physical interpretation is unclear. Moreover, a genuine gauge invariant expression  is provided by $\langle A^h_\mu A^h_\mu\rangle$. Albeit gauge invariant, this quantity  is highly non-local, with the notable exception of the  Landau gauge, where $\mathcal{A}^{2}_{\mathrm{min}}$ reduces to the simple expression $A^{a}_{\mu}A^{a}_{\mu}$. This is a very special feature of Landau gauge. On the other hand, the existence of other dimension-two condensates is also possible. A particular example is the ghost condensate $\langle \bar{c}^ac^a\rangle$. Though, Yang-Mills theories quantized in Landau gauge displays an additional  Ward identity, the anti-ghost equation of motion, which forbids the existence of $\langle \bar{c}^ac^a\rangle$. The same Ward identity holds for linear covariant gauges. Therefore, in these cases, just the gluon condensate is allowed. However, in the Curci-Ferrari gauges, the anti-ghost equation is not a Ward identity anymore and there is no a priori reason to exclude the condensate $\langle \bar{c}^ac^a\rangle$. Hence, we can introduce the general term

\begin{equation}
\tilde{S}_{\mathrm{cond}}=\int d^dx\left(\kappa_1 A^{a}_{\mu}A^{a}_{\mu}+\kappa_2\bar{c}^ac^a\right)\,,
\label{dgc2.0}
\end{equation}
and demand invariance under BRST and the $SL(2,\mathbb{R})$ symmetry. The latter does not impose any constraint on the coefficients $\kappa_1$ and $\kappa_2$. BRST, however, does\footnote{There is no difference in making use of the standard BRST or the non-perturbative one, due to the fact that for $(A,c,\bar{c})$ these transformations are identical.}:

\begin{equation}
s\tilde{S}_{\mathrm{cond}} = \int d^dx\left(2\kappa_1(\partial_{\mu}A^{a}_{\mu})c^a+k_2b'ac^a\right)\approx 0\,\,\,\,\Rightarrow\,\,\,\, \kappa_2=-2\alpha\kappa_1\,,
\label{dgc3.0}
\end{equation}
where the symbol $\approx$ denotes that we have used the equations of motion. Therefore, modulo a prefactor, the (on-shell) BRST invariant operator is

\begin{equation}
\EuScript{O}=\frac{1}{2}A^{a}_{\mu}A^{a}_{\mu}-\alpha\bar{c}^ac^a\,.
\label{dgc4.0}
\end{equation}
Some remarks concerning expression \eqref{dgc4.0} are in order: $(i)$ The limit $\alpha\rightarrow 0$ corresponds to the Landau gauge. In this case, the operator \eqref{dgc4.0} reduces to the dimension-two gluon operator $A^{a}_{\mu}A^{a}_{\mu}$ and no ghost condensate is included. $(ii)$ As is well-known, the presence of the quartic interaction term of Faddeev-Popov ghosts is  responsible for (eventually) generating a non-vanishing ghost condensate $\langle \bar{c}^ac^a\rangle$. 

Evidences for the existence of the condensate \eqref{dgc4.0} were presented in \cite{Kondo:2001nq,Dudal:2003gu}. In \cite{Kondo:2001nq} the modification of the OPE for the gluon and ghost due to the dimension two-condensate \eqref{dgc4.0} was pointed out,  while in \cite{Dudal:2003gu} an effective potential analysis was carried out. Unfortunately, the lack of lattice simulations results for Curci-Ferrari gauges limits ourselves to have a more conclusive statement concerning the relevance of the condensate \eqref{dgc4.0}. 

Nevertheless, within the new non-perturbative BRST framework, we introduced directly the gauge invariant quantity $\langle A^h_{\mu}A^h_{\mu}\rangle$ in the refinement of the Gribov-Zwanziger action. This condensate, as the gluon-ghost condensate (\ref{dgc4.0}), reduces to $\langle A^a_{\mu}A^{a}_{\mu}\rangle$ in the Landau gauge. In this sense, the introduction of both condensates seems to be redundant. Moreover, as will be discussed in Sect.~7, we have a local set up for $\langle A^h_{\mu}A^h_{\mu}\rangle$, evading the main difficulties that earlier studies had to deal with this operator. In summary, $\langle A^h_{\mu}A^h_{\mu}\rangle$ should be responsible to carry all physical information of \eqref{dgc4.0}. A very attractive feature is that the gauge invariance of $\langle A^h_{\mu}A^h_{\mu}\rangle$ together with the non-perturbative BRST symmetry gives to us full control of the independence from $\alpha$ of correlation funtions of gauge invariant operators. Therefore, the inclusion of \eqref{dgc4.0} seems to be superfluous, due to the use of the operator $A^h_\mu A^h_\mu$. 

We remark that the formation of different ghost condensates was also studied in Curci-Ferrari gauges, see \cite{Dudal:2002ye,Lemes:2002jv}. In principle, we should take them into account as well. However, in this work we are concerned with the behavior of the gluon propagator and, for this purpose, the inclusion of these extra condensates is irrelevant. Moreover, these condensates affect the ghost propagator and, again, it  would be  desirable to have access to lattice simulations for such propagator in order to estimate the relevance played by these novel condensates.

\section{Gluon propagator}

In the last section we discussed non-trivial dynamical effects generated in Curci-Ferrari gauges. As happens in the Gribov-Zwanziger theory in the gauges already studied in the literature, the presence of the Gribov horizon contributes to the formation of dimension-two condensates. The Refined Gribov-Zwanziger action in Curci-Ferrari gauges is given by \eqref{dgc11.0}, where such condensates are taken into account from the beginning through the presence of the dynamical parameters $(M^2, m^2)$. Hence, we can easily compute the gluon propagator out of \eqref{dgc11.0}, namely

\begin{equation}
\langle A^{a}_{\mu}(p)A^{b}_{\nu}(-p)\rangle_{d=3,4}=\delta^{ab}\left[\frac{p^2+M^2}{(p^2+m^2)(p^2+M^2)+2g^2\gamma^4N}\left(\delta_{\mu\nu}-\frac{p_{\mu}p_{\nu}}{p^2}\right)+\frac{\alpha}{p^2}\frac{p_{\mu}p_{\nu}}{p^2}\right]\,,
\label{cfgp1.0}
\end{equation}
while in $d=2$, we use the Gribov-Zwanziger action \eqref{gz1},

\begin{equation}
\langle A^{a}_{\mu}(p)A^{b}_{\nu}(-p)\rangle_{d=2}=\delta^{ab}\left[\frac{p^2}{p^4+2g^2\gamma^4N}\left(\delta_{\mu\nu}-\frac{p_{\mu}p_{\nu}}{p^2}\right)+\frac{\alpha}{p^2}\frac{p_{\mu}p_{\nu}}{p^2}\right]\,.
\label{cfgp2.0}
\end{equation}
Several remarks are in order.  For $d=3,4$,

\begin{itemize}

\item The form factor of the transverse part of the propagator is IR suppressed, positivity violating and attains a finite non-vanishing value at zero-momentum, a property which follows from the inclusion of the dimension two condensate of the auxiliary fields $\langle \bar{\varphi}\varphi-\bar{\omega}\omega\rangle$. Also, at tree-level, this form factor is independent from  $\alpha$. Hence, the transverse component of the gluon propagator displays the so-called decoupling/massive behavior.

\item The limit $\alpha\rightarrow 0$ brings us back to the gluon propagator for the Refined Gribov-Zwanziger action in the Landau gauge.

\item In the linear covariant gauges, the longitudinal part of the gluon propagator does not receive non-perturbative corrections. It remains as in perturbation theory, which is known to be just the tree-level result without quantum corrections. However, in Curci-Ferrari gauges, non-linearity jeopardizes this property as follows, for example, from the existence of the interaction vertex $b$-$c$-${\bar c}$. Therefore, inhere we expect that loop corrections will affect the longitudinal sector, although an explicit verification is far beyond the scope of this work. 

\end{itemize}

\noindent In the case of $d=2$,

\begin{itemize}

\item Since in $d=2$ the Gribov-Zwanziger action does not suffer from refinement, the gluon propagator is of Gribov-type \textit{i.e.} the transverse part is IR suppressed, positivity violating and vanishes at zero-momentum. This characterizes the so-called  scaling behavior.

\item As in $d=3,4$, the Landau propagator is easily obtained for $\alpha\rightarrow 0$, giving the scaling Gribov gluon propagator in $d=2$.

\end{itemize}

\noindent From these comments we can conclude that for $d=3,4$, the transverse gluon propagator displays a decoupling/massive behavior while in $d=2$, it is of scaling type. This is precisely the same behavior obtained in the Landau gauge and reported by very large lattice simulations. As pointed out in \cite{Capri:2015nzw,Capri:2015pfa,Guimaraes:2015bra}, this feature is more general than a particular property of  the Landau gauge, being also present in the linear covariant, maximal Abelian and Coulomb gauges. Inhere, we provide evidence that this property should also hold in Curci-Ferrari gauges. The novelty here with respect to the gauges already studied is the non-triviality of the longitudinal part which, due to the very non-linear character of the Curci-Ferrari gauges,  might very well acquire corrections from higher loops. 

\section{Local Refined Gribov-Zwanziger action in Curci-Ferrari gauges}

In this section we present a localization procedure to cast the action \eqref{gz1} and transformations \eqref{gz2} in a suitable local fashion. This puts the (Refined) Gribov-Zwanziger action in Curci-Ferrari gauges within the  well-developed realm of local quantum field theory. Before starting the description of the procedure, we emphasize the already mentioned feature that the original formulation of Gribov-Zwanziger action in the Landau gauge relies on the introduction of a non-local horizon function, displaying thus a non-local character. As shown previously, this non-locality can be handled through the introduction of suitable auxiliary fields which provide a local and renormalizable framework. 

Nevertheless, as soon as we introduce the gauge invariant field $A^h$, we introduce a new source of non-locality, see eq.\eqref{ah19}. Hence, even after the introduction of the auxiliary fields introduced in the standard construction, the resulting action is still non-local due to the explicit presence of $A^h$. 

The localization of the transverse gauge invariant field $A^h$ is performed by the introduction of a Stueckelberg-type field $\xi^a$ in the form

\begin{equation}
h=\mathrm{e}^{ig\xi^aT^a}\,.
\label{lcf1.0}
\end{equation}
With \eqref{lcf1.0} we rewrite the $A^h$ field as

\begin{equation}
A^h_{\mu}=h^{\dagger}A_{\mu}h+\frac{i}{g}h^{\dagger}\partial_{\mu}h\,,
\label{lcf2.0}
\end{equation}
where a matrix notation is being employed. Expression \eqref{lcf2.0} is local albeit non-polynomial. For a $SU(N)$ element $v$, $A^h$ is left invariant under the gauge transformations

\begin{equation}
A'_{\mu}=v^{\dagger}A_{\mu}v+\frac{i}{g}v^{\dagger}\partial_{\mu}v\,,\,\,\,\,\, h'=v^{\dagger}h\,\,\,\,\, \mathrm{and}\,\,\,\,\, h'^{\dagger}=h^{\dagger}v\,,
\label{lcf3.0}
\end{equation}
{\it i.e.}
\begin{equation} 
 ({A^{h}_\mu})'  \leftarrow A^h_{\mu}  \;. \label{gah}
\end{equation}
Although gauge invariance of $A^h$ is guaranteed by \eqref{lcf3.0}, we still have to impose the transversality condition of $A^h$. This is done by means of a Lagrange multiplier $\tau^a$ which enforces this constraint, namely, we introduce the following term

\begin{equation}
S_{\tau}=\int d^dx\,\tau^a\partial_{\mu}A^{h,a}_{\mu}\,.
\label{lcf4.0}
\end{equation}
Solving the transversality condition $\partial A^h=0$ for $\xi$, we obtain the non-local expression \eqref{ah19} for $A^h$, see, for example, Ap.~D. Then, the Gribov-Zwanziger action in Curci-Ferrari gauges can be expressed in local form as follows,

\begin{eqnarray}
S^{\mathrm{loc}}_{\mathrm{CF}} &=&  S_{\mathrm{YM}} + \int d^dx\left(b^{h,a}\partial_{\mu}A^{a}_{\mu}-\frac{\alpha}{2}b^{h,a}b^{h,a}+\bar{c}^{a}\partial_{\mu}D^{ab}_{\mu}c^{b}+\frac{\alpha}{2}gf^{abc}b^{a}\bar{c}^{b}c^{c}\right.\nonumber\\
&+&\left.\frac{\alpha}{8}g^{2}f^{abc}f^{cde}\bar{c}^{a}\bar{c}^{b}c^{d}c^{e}\right)-\int d^dx\left(\bar{\varphi}^{ac}_{\mu}\left[\EuScript{M}(A^h)\right]^{ab}\varphi^{bc}_{\mu}-\bar{\omega}^{ac}_{\mu}\left[\EuScript{M}(A^h)\right]^{ab}\omega^{bc}_{\mu}\right.\nonumber\\
&+&\left.g\gamma^2f^{abc}A^{h,a}_{\mu}(\varphi+\bar{\varphi})^{bc}_{\mu}\right)+\int d^dx~\tau^a\partial_{\mu}A^{h,a}_{\mu}\,,\nonumber\\
\label{lcf5.0}
\end{eqnarray} 
with $A^h$ given by \eqref{lcf3.0}. 

The non-perturbative BRST transformations, which correspond to a symmetry of \eqref{lcf5.0}, are also non-local. As shown in Ch.~\ref{locnonpBRST}, the localization of these transformations is achieved through the introduction of extra auxiliary fields. Before doing this, we note that the standard BRST transformations for $\tau$ and $\xi$ (written implicitly in terms of $h$) are

\begin{equation}
sh=-igch\,\,\,\,\,\,\,\,\mathrm{and}\,\,\,\,\,\,\,\,s\tau^a=0\,.
\label{lcf6.0}
\end{equation}
Proceeding to the localization of the non-perturbative BRST transformations, we make use of the following trick: We rewrite the horizon function $H(A^h)$ in the path integral  as

\begin{equation}
\mathrm{e}^{-\gamma^4 H(A^{h})}=\mathrm{e}^{-\frac{\gamma^4}{2} H(A^{h})}\mathrm{e}^{-\frac{\gamma^4}{2} H(A^{h})}\,.
\label{lcf7.0}
\end{equation}
Now, employing the same localization procedure used in the standard Gribov-Zwanziger framework, we obtain

\begin{equation}
\mathrm{e}^{-\frac{\gamma^4}{2} H(A^{h})}=\int \left[\EuScript{D}\varphi\right]\left[\EuScript{D}\bar{\varphi}\right]\left[\EuScript{D}\omega\right]\left[\EuScript{D}\bar{\omega}\right]\mathrm{e}^{-\int d^dx\left(-\bar{\varphi}^{ac}_{\mu}\EuScript{M}^{ab}(A^h)\varphi^{bc}_{\mu}+\bar{\omega}^{ac}_{\mu}\EuScript{M}^{ab}(A^h)\omega^{bc}_{\mu}+g\frac{\gamma^2}{\sqrt{2}}f^{abc}A^{h,a}_{\mu}(\varphi+\bar{\varphi})^{bc}_{\mu}\right)}\,,
\label{lcf8.0}
\end{equation}
and

\begin{equation}
\mathrm{e}^{-\frac{\gamma^4}{2} H(A^{h})}=\int \left[\EuScript{D}\beta\right]\left[\EuScript{D}\bar{\beta}\right]\left[\EuScript{D}\zeta\right]\left[\EuScript{D}\bar{\zeta}\right]\mathrm{e}^{-\int d^dx\left(-\bar{\beta}^{ac}_{\mu}\EuScript{M}^{ab}(A^h)\beta^{bc}_{\mu}+\bar{\zeta}^{ac}_{\mu}\EuScript{M}^{ab}(A^h)\zeta^{bc}_{\mu}-g\frac{\gamma^2}{\sqrt{2}}f^{abc}A^{h,a}_{\mu}(\beta+\bar{\beta})^{bc}_{\mu}\right)}\,.
\label{lcf9.0}
\end{equation}
In \eqref{lcf8.0}, the fields $(\varphi,\bar{\varphi},\omega,\bar{\omega})$ are Zwanziger's localizing fields, $(\beta,\bar{\beta})$ are commuting ones while $(\zeta,\bar{\zeta})$ are anti-commuting and play the same role as Zwanziger's fields. The resulting Gribov-Zwanziger action is given by

\begin{eqnarray}
S^{\mathrm{loc}}_{{\mathrm{CF}}}&=& S_{\mathrm{YM}}+\int d^dx\left(b^{h,a}\partial_{\mu}A^{a}_{\mu}-\frac{\alpha}{2}b^{h,a}b^{h,a}+\bar{c}^{a}\partial_{\mu}D^{ab}_{\mu}c^{b}+\frac{\alpha}{2}gf^{abc}b^{a}\bar{c}^{b}c^{c}\right.\nonumber\\
&+&\left.\frac{\alpha}{8}g^{2}f^{abc}f^{cde}\bar{c}^{a}\bar{c}^{b}c^{d}c^{e}\right)-\int d^dx\left(\bar{\varphi}^{ac}_{\mu}\EuScript{M}^{ab}(A^h)\varphi^{bc}_{\mu}-\bar{\omega}^{ac}_{\mu}\EuScript{M}^{ab}(A^h)\omega^{bc}_{\mu}\right.\nonumber\\
&-&\left.g\frac{\gamma^2}{\sqrt{2}}f^{abc}A^{h,a}_{\mu}(\varphi+\bar{\varphi})^{bc}_{\mu}\right)-\int d^dx\left(\bar{\beta}^{ac}_{\mu}\EuScript{M}^{ab}(A^h)\beta^{bc}_{\mu}-\bar{\zeta}^{ac}_{\mu}\EuScript{M}^{ab}(A^h)\zeta^{bc}_{\mu}\right.\nonumber\\
&+&\left.g\frac{\gamma^2}{\sqrt{2}}f^{abc}A^{h,a}_{\mu}(\beta+\bar{\beta})^{bc}_{\mu}\right)+\int d^dx~\tau^a\partial_{\mu}A^{h,a}_{\mu}\,. \nonumber\\
\label{lcf10.0}
\end{eqnarray}
The local Gribov-Zwanziger action written as \eqref{lcf10.0} is invariant under the following local non-perturbative BRST transformations,

\begin{align}
s_{l}A^{a}_{\mu}&=-D^{ab}_{\mu}c^{b}\,,     && s_{l}\varphi^{ab}_{\mu}=\omega^{ab}_{\mu}\,,&&&      s_{l}h =-igch\,,&&&& s_{l}\beta^{ab}_{\mu}=\omega^{ab}_{\mu}\,,\nonumber\\
s_{l}c^a&=\frac{g}{2}f^{abc}c^{b}c^{c}\,,     && s_{l}\omega^{ab}_{\mu}=0\,,&&&     s_{l}A^{h,a}_{\mu}=0\,,&&&& s_{l}\bar{\zeta}^{ab}_{\mu}=0\,,\nonumber\\
s_{l}\bar{c}^{a}&=b^{h,a}\,,   && s_{l}\bar{\omega}^{ab}_{\mu}=\bar{\varphi}^{ab}_{\mu}+\bar{\beta}^{ab}_{\mu}\,,&&&     s_{l}\tau^{a} =0\,,&&&&s_{l}\zeta^{ab}_{\mu}=0\,.\nonumber\\
s_{l}b^{h,a}&=0\,,         && s_{l}\bar{\varphi}^{ab}_{\mu}=0\,,&&&s_{l}\bar{\beta}^{ab}_{\mu}=0\,,
\label{lcf11.0}
\end{align}
It is an immediate check that $s_l$ is nilpotent, $s^2_l=0$. Integration over $(\beta,\bar{\beta},\zeta,\bar{\zeta})$ gives back the non-local BRST transformations \eqref{gz2}.

In local fashion, the refinement of the Gribov-Zwanziger action is obtained by the introduction of the following term to \eqref{lcf10.0},

\begin{equation}
S_{\mathrm{cond}}=\int d^dx\left[\frac{m^2}{2}A^{h,a}_{\mu}A^{h,a}_{\mu}+M^2\left(\bar{\omega}^{ab}_{\mu}\omega^{ab}_{\mu}-\bar{\varphi}^{ab}_{\mu}\varphi^{ab}_{\mu}-\bar{\beta}^{ab}_{\mu}\beta^{ab}_{\mu}+\bar{\zeta}^{ab}_{\mu}\zeta^{ab}_{\mu}\right)\right]\,.
\label{lcf12.0}
\end{equation}
The resulting Refined Gribov-Zwanziger action, written in local form and invariant under \eqref{lcf11.0} is

\begin{eqnarray}
S^{\mathrm{RGZ}}_{\mathrm{CF}}&=&S_{\mathrm{YM}}+\int d^dx\left(b^{h,a}\partial_{\mu}A^{a}_{\mu}-\frac{\alpha}{2}b^{h,a}b^{h,a}+\bar{c}^{a}\partial_{\mu}D^{ab}_{\mu}c^{b}+\frac{\alpha}{2}gf^{abc}b^{a}\bar{c}^{b}c^{c}\right.\nonumber\\
&+&\left.\frac{\alpha}{8}g^{2}f^{abc}f^{cde}\bar{c}^{a}\bar{c}^{b}c^{d}c^{e}\right)-\int d^dx\left(\bar{\varphi}^{ac}_{\mu}\EuScript{M}^{ab}(A^h)\varphi^{bc}_{\mu}-\bar{\omega}^{ac}_{\mu}\EuScript{M}^{ab}(A^h)\omega^{bc}_{\mu}\right.\nonumber\\
&-&\left.g\frac{\gamma^2}{\sqrt{2}}f^{abc}A^{h,a}_{\mu}(\varphi+\bar{\varphi})^{bc}_{\mu}\right)-\int d^dx\left(\bar{\beta}^{ac}_{\mu}\EuScript{M}^{ab}(A^h)\beta^{bc}_{\mu}-\bar{\zeta}^{ac}_{\mu}\EuScript{M}^{ab}(A^h)\zeta^{bc}_{\mu}\right.\nonumber\\
&+&\left.g\frac{\gamma^2}{\sqrt{2}}f^{abc}A^{h,a}_{\mu}(\beta+\bar{\beta})^{bc}_{\mu}\right)+\int d^dx~\tau^a\partial_{\mu}A^{h,a}_{\mu}+\int d^dx\left[\frac{m^2}{2}A^{h,a}_{\mu}A^{h,a}_{\mu}
\right.\nonumber\\
&+&\left.M^2\left(\bar{\omega}^{ab}_{\mu}\omega^{ab}_{\mu}-\bar{\varphi}^{ab}_{\mu}\varphi^{ab}_{\mu}-\bar{\beta}^{ab}_{\mu}\beta^{ab}_{\mu}+\bar{\zeta}^{ab}_{\mu}\zeta^{ab}_{\mu}\right)\right]\,. 
\label{lcf14.0}
\end{eqnarray}
The Refined Gribov-Zwanziger action \eqref{lcf14.0} is an effective action which takes into account the presence of Gribov copies in the standard Faddeev-Popov procedure in Curci-Ferrari gauges. Moreover, this action also incorporates further non-perturbative dynamics effects as the formation of dimension-two condensates. All this setting is written in local fashion and enjoys non-perturbative BRST symmetry \eqref{lcf11.0} which ensures gauge parameter independence of correlation functions of gauge invariant composite operators, see Ch.~\ref{locnonpBRST} for a purely  algebraic proof of this statement.

\noindent \textbf{Remark:} After the submission of this thesis, the paper \cite{Pereira:2016fpn} was published. It contains the main results here presented.

\chapter{Conclusions} \label{Conc}

In this thesis, a particular approach to the infrared regime of Yang-Mills theories, the Refined Gribov-Zwanziger scenaio, was analyzed. Originally constructed in the Landau gauge, this program aims at providing a consistent quantization of Yang-Mills theories beyond perturbation theory. Although very fruitful in the Landau gauge, with very nice agreement with lattice and functional methods results, the issue of extending this framework to different gauges was not addressed in detail as in the Landau gauge. Partially, this is due to the lack of data concerning different gauges (with few exceptions) from different approaches being thus an obstacle to compare the results and partially due to technical complications as described in this thesis. 

Recently, different groups on different approaches to the IR issue in Yang-Mills theories started investigating different gauge choices than Landau gauge, a fact that plays an important role to establish an interplay between the communities. Also, from the Gribov-Zwanziger scenario point of view, a lot of non-trivial results enabled the extension of the formalism to different gauges. The prominent example of such results is the construction of the non-perturbative and nilpotent BRST symmetry for the Gribov-Zwanziger action. This, not only established a better understanding of the scenario in the Landau gauge itself, but also provided a more systematic line of attack for the formalism. Hence, the construction of the Refined Gribov-Zwanziger action in linear covariant gauges with the power of this new symmetry seems to provide a consistent picture. Gauge independence of physical observables and agreement with the most recent lattice data are certainly attractive features of the construction. 

Also, the construction of the non-perturbative BRST symmetry naturally suggested a way to take into account the Gribov problem in a class of non-linear gauges. At this level, no comparison with lattice or functional methods is possible, since these gauges were not implemented by these approaches so far. However, it is remarkable how simple the extesion to these gauges was once we had the non-perturbative BRST symmetry at our disposal. So far, we can only provide consistency checks of the formalism.

The developments here presented open a rich window of perspectives. An urgent issue to be studied is the renormalizability property of Refined Gribov-Zwanziger action in linear covariant gauges as well as in Curci-Ferrari gauges. Having a local formulation for these actions, the analysis is possible, albeit non-trivial. Also, the novel BRST symmetry here introduced encodes a non-trivial information about the geometry of functional space of gauge fields. It is not only interesting by its own, but also very pertinent to understand the geometrical interpretation of the non-perturbative BRST operator. Also, a detailed investigation of the coupling with matter fields in a non-perturbative BRST invariant way is desired. 

In the non-perturbative BRST invariant setting, a physical meaning was given to the dimension-two condensates dynamically generated. The explicit computation of these condensates through effective potential methods is a natural path to pursue. The interplay among this computation with lattice's fitting may provide well grounded conclusions about their values. 

Finally, the extension to supersymmetric theories as well as finite temperature computations are natural topics to be investigated, mainly because they were worked out (partially) in the standard (Refined) Gribov-Zwanziger formalism in the Landau gauge. In summary, a substantial ammount of applications is possible and this should fortify the power of the here presented framework.

\appendix

\chapter{Conventions} \label{appendixA}

\section{Yang-Mills theory - Setting the stage}

The central object of this thesis is pure Yang-Mills theory \textit{i.e.} we do not include fermionic matter. The arena where computations are performed is $d$-dimensional Euclidean space (with $d=2,3,4$) and the gauge group is $SU(N)$. Gauge group generators will be denoted as $T^a$, with $a\in\left\{1,\ldots,N^2-1\right\}$ being the \textit{color} index. They obey the algebra

\begin{equation}
\left[T^a,T^b\right]=if^{abc}T^c\,,
\label{a0}
\end{equation}

 \noindent with $T$ being Hermitian generators, namely $T=T^\dagger$ and $f^{abc}$ the totally antisymmetric (real) structure constants. For our purposes, the generators $T$ belong to the adjoint representation of $SU(N)$. The structure constants satisfy very useful relations, given by

\begin{equation}
f^{abc}f^{cde}+f^{adc}f^{ceb}+f^{aec}f^{cbd} = 0\,\,\,\,\mathrm{and}\,\,\,\, f^{adc}f^{bdc}=N\delta^{ab}\,.
\label{a0.1}
\end{equation}

\noindent The Killing metric is  

\begin{equation}
\mathrm{Tr}\left(T^aT^b\right)=\frac{1}{2}\delta^{ab}\,.
\label{a1}
\end{equation}

\noindent Generically, an element $U$ of $SU(N)$ will be written as

\begin{equation}
U=\mathrm{exp}\left(-ig\xi^a T^a\right)\,,
\label{a2}
\end{equation}

\noindent with $g$ and $\xi$ being real variables and, by definition, $UU^\dagger = 1$. 

\noindent After this set of definitions we can introduce the algebra-valued \textit{gauge} field or, simply, the \textit{gluon} field $A_{\mu}=A^{a}_{\mu}T^{a}$ and the Yang-Mills action\footnote{Greek indices as $\alpha,\beta,\mu,\nu,\ldots$ will denote $d$-dimensional spacetime indices and since we are in Euclidean space, we do not bother with the distinction among up or down indices.}

\begin{equation}
S_{\mathrm{YM}}=\frac{1}{2}\int d^dx~\mathrm{Tr}\left(F_{\mu\nu}F_{\mu\nu}\right)=\frac{1}{4}\int d^dx~F^{a}_{\mu\nu}F^{a}_{\mu\nu}\,,
\label{a3}
\end{equation}

\noindent with $F_{\mu\nu}=F^{a}_{\mu\nu}T^{a}=\partial_{\mu}A_{\nu}-\partial_{\nu}A_{\mu}-ig\left[A_{\mu},A_{\nu}\right]$ the \textit{field strength}. In terms of algebra components, the field strength is written as

\begin{equation}
F^{a}_{\mu\nu}=\partial_{\mu}A^{a}_{\nu}-\partial_{\nu}A^{a}_{\mu}+gf^{abc}A^{b}_{\mu}A^{c}_{\nu}\,.
\label{a4}
\end{equation}

\noindent The \textit{gauge transformation} which leaves action (\ref{a3}) invariant is defined by

\begin{equation}
A'_{\mu}=UA_{\mu}U^{\dagger}+\frac{i}{g}U\partial_{\mu}U^{\dagger}\,,
\label{a5}
\end{equation}

\noindent which induces a gauge covariant transformation of the field strength, \textit{i.e.}

\begin{equation}
F'_{\mu\nu}=UF_{\mu\nu}U^{\dagger}\,.
\label{a6}
\end{equation}

\noindent The invariance of (\ref{a3}) under (\ref{a6}) is trivial using the cyclic property of the trace with $UU^{\dagger}=1$. We define an \textit{infinitesimal} gauge transformation by taking the element $U$ of eq.(\ref{a5}) close to the identity. This corresponds to take the real parameter $\xi^{a}$ as infinitesimal, which implies

\begin{equation}
U=1-ig\xi^{a}T^{a}+\mathcal{O}(\xi^2)\approx 1-ig\xi^{a}T^{a}\,.
\label{a7}
\end{equation}

\noindent Plugging eq.(\ref{a7}) into eq.(\ref{a5}), the gauge field transforms under an infinitesimal gauge transformation as

\begin{equation}
A'^{a}_{\mu}=A^{a}_{\mu}-D^{ab}_{\mu}\xi^{b}\,,
\label{a8}
\end{equation}

\noindent whereby $D^{ab}_{\mu}$ is the covariant derivative in the adjoint representation,

\begin{equation}
D^{ab}_{\mu}=\delta^{ab}\partial_{\mu}-gf^{abc}A^{c}_{\mu}\,.
\label{a9}
\end{equation}

\noindent A very important decomposition which is heavily used in the text is the transverse/longitudinal split of the gauge field, namely

\begin{equation}
A^{a}_{\mu}=\underbrace{\left(\delta_{\mu\nu}-\frac{\partial_{\mu}\partial_{\nu}}{\partial^2}\right)}_{\mathcal{P}_{\mu\nu}}A^{a}_{\mu}+\underbrace{\frac{\partial_{\mu}\partial_{\nu}}{\partial^2}}_{\mathcal{L}_{\mu\nu}}A^{a}_{\mu}\,,
\label{a9.1}
\end{equation}

\noindent with $\mathcal{P}_{\mu\nu}$ and $\mathcal{L}_{\mu\nu}$ the \textit{transverse and longitudinal projectors}, respectively. Therefore, we define the transverse component of the gauge field as

\begin{equation}
A^{Ta}_{\mu}\equiv \mathcal{P}_{\mu\nu}A^{a}_{\nu}\,,
\label{a9.2}
\end{equation}

\noindent and the longitudinal sector,

\begin{equation}
A^{La}_{\mu}\equiv \mathcal{L}_{\mu\nu}A^{a}_{\nu}\,.
\label{a9.3}
\end{equation}

\noindent At the level of classical Yang-Mills theory, these are the most important (and widely used throughout this thesis) definitions and conventions. Now, we move to some important conventions at the quantum level. 

\section{``Quantum" conventions}

The quantization of Yang-Mills theory relies on the introduction of a constraint on the gauge field $A^{a}_{\mu}$. As discussed in Subsect.~\ref{FPtrick}, this is achieved, in the path integral quantization, through the Faddeev-Popov method at the perturbative level. Effectively, to implement the \textit{gauge fixing}, we introduce Faddeev-Popov ghosts fields $\left(\bar{c}^{a},c^{b}\right)$, with $\bar{c}$ the anti-ghost and $c$ the ghost field and a Lagrange multiplier $b^{a}$, also known as the auxiliary Lautrup-Nakanishi field, responsible to enforce the gauge fixing condition. The gauge-fixed action is given by

\begin{eqnarray}
S_{\mathrm{FP}}&=&S_{\mathrm{YM}}+s\int d^dx~\bar{c}^{a}F^a\nonumber\\
&=&S_{\mathrm{YM}}+\int d^dx~\left(b^{a}F^a+\bar{c}^{a}\frac{\delta\Delta^{a}}{\delta A^{b}_{\mu}}D^{bc}_{\mu}c^{c}\right)\,,
\label{a10}
\end{eqnarray}

\noindent with $s$ the BRST operator and $F^{a}=0$, the gauge condition. At this level, gauge symmetry is manifestly broken due to the introduction of a gauge fixing constraint\footnote{This is true just at the perturbative level.}, but BRST enters the game. In this thesis, the conventions we use for the BRST transformations of the set of fields $(A,\bar{c},c,b)$ are

\begin{eqnarray}
sA^{a}_{\mu}&=&-D^{ab}_{\mu}c^{b}\nonumber\\
sc^{a}&=&\frac{g}{2}f^{abc}c^{b}c^{c}\nonumber\\
s\bar{c}^{a}&=&b^{a}\nonumber \\
sb^{a}&=& 0\,.
\label{a11}
\end{eqnarray}

\noindent The BRST operator $s$ is nilpotent \textit{i.e.} $s^2=0$ and carries ghost number 1. To account for the non-linearity of the BRST transformation of $A$ and $c$, we introduce BRST invariant external sources $\Omega^{a}_{\mu}$ and $L^{a}$ coupled to the non-linear transformations. These are introduced by the addition of the following term to (\ref{a10}),

\begin{equation}
S_{\mathrm{ext}}=s\int d^dx\left(-\Omega^a_\mu A^a_\mu+L^ac^a\right)=\int d^dx\left(-\Omega^a_\mu D^{ab}_\mu c^b+\frac{g}{2}f^{abc}L^ac^bc^c\right)\,.
\label{a12}
\end{equation}

\noindent The quantum numbers of fields and sources are summarized in table~\ref{tablea1}.

\begin{table}[t]
\centering
\begin{tabular}{|c|c|c|c|c|c|c|c|c|}
\hline
Fields/Sources & $A$ & $F$ & $b$ & $c$ & $\bar{c}$ & $\Omega$ & $L$ & $g$ \\ \hline
Dimension & $(d-2)/2$ & $\kappa$ & $d-\kappa$ & $(d-4)/2$ & $d-\kappa$ & $d-1$ & $(3d-4)/2$ & $(4-d)/2$\\
Ghost number & 0 & 0 & 0 & 1 & $-1$ & $-1$ & $-2$ & 0 \\ \hline
\end{tabular}
\caption{Quantum numbers of the fields. The BRST operator has ghost number $1$ and is chosen to be of dimension $0$.}
\label{tablea1}
\end{table}

\chapter{Effective Action and Symmetries}\label{EA}

In classical physics, the most important object is the action. With this, it is possible to derive the equations of motion using Hamilton's principle and once we have the equations of motion and boundary conditions, we are able to describe the complete motion of the system. Curiously, in quantum field theory, the classical action also plays a fundamental role, but this time, the object which fundamentally describes the system is the partition function, or simply, generating functional. It is defined by

\begin{equation}
\EuScript{Z} = \EuScript{N}\int \left[\EuScript{D}\phi\right]\mathrm{e}^{-\frac{1}{\hbar}(S(\phi) + \int d^dx\,J\phi)}\;,
\label{b1}
\end{equation}

\noindent where $S(\phi)$ denotes the classical action defined for the system to be described, $J$ are external sources coupled to the fields and $\EuScript{N}$ is a normalization constant. As we mentioned, the classical action still plays a very important role even in quantum field theory. If we take the functional derivative of $Z$ with respect to the source $J(x)$, we obtain

\begin{equation}
{\frac{\delta \EuScript{Z}}{\delta J(x)}}\Bigg|_{J=0} = -\EuScript{N}\frac{1}{\hbar}\int \left[\EuScript{D}\phi\right]\phi(x)\mathrm{e}^{-\frac{1}{\hbar}S(\phi)} \equiv -\hbar\langle \phi(x) \rangle\;,
\label{b2}
\end{equation}

\noindent which is nothing else than the one-point function for the field $\phi$. Clearly, we could take the functional derivative with respect to an arbitrary number of sources, let us say $n$, and obtain the $n$-point function for the field $\phi$, which is given by

\begin{eqnarray}
\frac{\delta^{n}\EuScript{Z}}{\delta J(x_{1})\ldots\delta J(x_{n})}\Bigg|_{J=0} &=&\EuScript{N}\left(\frac{-1}{\hbar}\right)^{n}\int \left[\EuScript{D}\phi\right]\phi(x_{1})\ldots\phi(x_{n})\mathrm{e}^{-\frac{1}{\hbar}S(\phi)} \nonumber \\
&\equiv& (-\hbar)^{n}\langle \phi(x_{1})\ldots \phi(x_{n}) \rangle \;.
\label{b3}
\end{eqnarray}

\noindent Since it generates all correlation functions, its name of generating functional is justified. Finally, we can introduce a power series to define the generating functional as follows

\begin{equation}
\EuScript{Z}[J] = \sum^{\infty}_{n=0}\frac{(-1/\hbar)^n}{n!}\int d^dx_{1} \ldots \int d^dx_{n}J(x_{1})\ldots J(x_{n})\langle \phi (x_{1})\ldots \phi(x_{n}) \rangle\;.
\label{b4}
\end{equation}

\section{The $\EuScript{W}[J]$ Generating Functional}

As reviewed in the last section, we can construct an object which generates all correlation functions for a system. Conveniently, these correlation functions are represented through Feynman diagrams. However, the generating functional $\EuScript{Z}$ generates diagrams that can be divided as ``products" of other diagrams. These diagrams are called \textit{disconnected} ones. We say they are built by \textit{connected} diagrams. Hence, the knowledge of disconnected and connected diagrams is not bigger than the knowledge of connected diagrams, since we can simply take ``products" of them and build disconnected ones. Since $\EuScript{Z}[J]$ generates all diagrams (connected and disconnected), we may look for a quantity $\EuScript{W}[J]$ which might be seen as more fundamental than $Z[J]$, that generates exclusively connected diagrams. It is possible to show, \cite{zinnjustin} that this quantity is defined by

\begin{equation}
\EuScript{Z}[J] = \mathrm{e}^{-\frac{1}{\hbar}\EuScript{W}[J]}\;.
\label{b5}
\end{equation}

\noindent It is conventional to express the $\EuScript{W}[J]$ generating functional as $\EuScript{W}[J] = -\hbar \mathrm{ln}\EuScript{Z}[J]$. Since we are using the normalization $\EuScript{Z}[0] = 1$, we have that $\EuScript{W}[0] = 0$. In analogy with eq.\eqref{b4}, we can write the following power series for $\EuScript{W}[J]$,

\begin{equation}
\EuScript{W}[J] = \sum^{\infty}_{n=1}\frac{(-1/\hbar)^{n-1}}{n!}\int d^dx_{1} \ldots \int d^dx_{n}~J(x_1)\ldots J(x_n)\langle \phi(x_1)\ldots\phi(x_n) \rangle_{c}\;,
\label{b6}
\end{equation}

\noindent where the subscipt $c$ denotes the fact that these are connected correlation functions. 

\section{The $\Gamma$ Generating Functional}

In the previous section, we introduced a generating functional which is, in a certain sense, more fundamental than the $\EuScript{Z}[J]$ generating functional. The reason is that disconnected diagrams are not accounted since they can be generated by the union of connected diagrams. We could go ahead with this kind of analysis and look for more fundamental diagrams which are responsible to generate all the possibles in a theory. These diagrams are called \textit{one particle irreducible} (1PI) and to understand their origin we can imagine the following: take a connected diagram and make a cut on the internal lines. If after this, the resultant diagrams are still genuine diagrams of the theory, than we say that the initial diagram is \textit{reducible}. In the end of this procedure, we will have diagrams that cannot be reduced further and we call them as the 1PI diagrams. Therefore, if it is possible to define a generating functional which is responsible to generate just 1PI diagrams, then it is the most fundamental quantity that we have to build correlation functions. In fact, such generating functional exists and is denoted as $\Gamma$. This functional is of great importance for quantum field theory and it is usually called as \textit{effective action} or \textit{quantum action}. 

In order to define the effective action, we must introduce the \textit{classical field} $\varphi$, defined by

\begin{equation}
\varphi_{J}(x) = \frac{\delta \EuScript{W}[J]}{\delta J(x)}\;.
\label{b7}
\end{equation}

\noindent This quantity is essentially the expectation value of $\phi(x)$ in the presence of a source $J$. The classical field is a function of the source $J$ and as such, we assume it is invertible and that we can write $J$ as a function of $\varphi$. The effective or quantum action is defined by the following Legendre transform

\begin{equation}
\Gamma(\varphi) = \EuScript{W}[J(\varphi)] - \int d^dx\,J(x)\varphi(x)\;.
\label{b8}
\end{equation}

\noindent It is a simple exercise to show that 

\begin{equation}
\frac{\delta \Gamma[\varphi]}{\delta \varphi(x)} = - J(x)\;,
\label{b9}
\end{equation}

\noindent and if we set $J=0$, we obtain the following equation

\begin{equation}
\frac{\delta \Gamma[\varphi]}{\delta \varphi(x)}\Bigg|_{J=0} = 0\;,
\label{b10}
\end{equation}

\noindent which solutions are the expectation values of the fields $\phi$. This is the ``quantum version" of the classical equations of motions, where we replace $\Gamma$ by $S$. A very important and well-known feature about the effective action is that it can be expanded in a formal power series of $\hbar$, \textit{i.e.}

\begin{equation}
\Gamma[\varphi] = \sum^{\infty}_{n=0}\hbar^n\Gamma^{(n)}(\varphi)\;,
\label{b11}
\end{equation}

\noindent and for $n=0$, the effective action coincides with the classical action $S[\varphi]$. The power of $\hbar$ is associated with the number of loops of a diagram. In this sense, the classical action contributes with diagrams which do not have loops, \textit{i.e.}, at the tree level and quantum corrections are introduced with loops. This is why we call $\Gamma$ as the quantum action. 

\section{Composite Operators}

In many cases, we are interested in correlation functions involving composite field operators, \textit{i.e.}, which are local polynomials of field operators. Inhere, this is very important, since BRST transformations are non-linear in some fields and we are also interested on the introduction of dimension-two operators in some circumstances. In order to deal with these objects, we must introduce a new term in the classical action. To do so, let us call the composite operators as $Q$. For each composite operator, we introduce a source $\rho$, and add in the classical action, the following term

\begin{equation}
S_\mathrm{sources} = \int d^dx\,\rho Q\;.
\label{b12}
\end{equation}

\noindent Thus, the generating functional $\EuScript{Z}$ becomes

\begin{equation}
\EuScript{Z}[J,\rho] = \int \left[\EuScript{D}\phi\right]\mathrm{e}^{-\frac{1}{\hbar}(S(\phi) + \int d^dx\,\rho Q + \int d^dx\,J\phi)}\;,
\label{b13}
\end{equation}

\noindent and we can calculate the correlation functions as

\begin{equation}
\frac{\delta^{(n+m)}\EuScript{Z}[J,\rho]}{\delta J(x_{1})\ldots\delta J(x_{n})\delta\rho(y_1)\ldots\delta\rho(y_{m})}\Bigg|_{J,\rho = 0} = \left(-\frac{1}{\hbar}\right)^{n+m}\langle \phi(x_1)\ldots\phi(x_n)Q(y_1)\ldots Q(y_m)\rangle\;.
\label{b14}
\end{equation}

\noindent We can generalize the definition of $\EuScript{W}[J]$ and $\Gamma$ directly as

\begin{equation}
\EuScript{Z}[J,\rho] = \mathrm{e}^{-\frac{1}{\hbar}\EuScript{W}[J,\rho]}
\label{b15}
\end{equation}

\noindent and

\begin{equation}
\Gamma[\varphi,\rho] = \EuScript{W}[J(\varphi),\rho] - \int d^dx\, J(x)\varphi(x)\;.
\label{b16}
\end{equation}

\noindent Finally, we introduce the following useful notation:

\begin{equation}
\frac{\delta \EuScript{Z}[J,\rho]}{\delta \rho(x)}\Bigg|_{\rho=0} \equiv Q(x)\cdot \EuScript{Z}[J], \;\;\; \frac{\delta \EuScript{W}[J,\rho]}{\delta \rho(x)}\Bigg|_{\rho=0} \equiv Q(x)\cdot \EuScript{W}[J] \;\; \mathrm{and} \;\; \frac{\delta \Gamma[\varphi,\rho]}{\delta \rho(x)}\Bigg|_{\rho=0} \equiv Q(x)\cdot \Gamma[\varphi]\;.
\label{b17}
\end{equation}

\section{Symmetries}

In the context of classical field theory, symmetries are very important not only to simplify the analysis of a system, but even to define important conserved quantities via Noether's theorem. A very special question to address is if symmetries in the classical world survive quantization. In quantum field theory, symmetries are represented by Ward identities, which are relations between correlation functions. 

To avoid confusion, we will introduce an index $i$ on the fields, to say explicitly that we are dealing with a set of fields $\left\{\phi_i\right\}$, and not necessarily with just one field $\phi$. Let us suppose that the classical action which describes this system is known and we denote it by $S(\phi_i)$. A transformation $\delta$ defined by

\begin{equation}
\delta\phi_i(x) = i\epsilon^{a}R^{a}_{i}(x)
\label{b18}
\end{equation}

\noindent is called a symmetry if it leaves the action invariant, \textit{i.e.}

\begin{equation}
\delta S[\phi_i] = 0\;,
\label{b19}
\end{equation}

\noindent where $R^{a}_{i}(x)$ are formal power series of fields and its derivatives and $\epsilon^a$ are infinitesimal constant parameters. We assume that the transformation defined by eq.\eqref{b18} belongs to a representation of a Lie group, whose generators $G_a$ satisfy

\begin{equation}
[G^a,G^b]=if^{abc}G^c\;,
\label{b20}
\end{equation}

\noindent with $f^{abc}$ being the totally antisymmetric symbol. We demand that

\begin{equation}
\int d^dy~\left(R^b_j(y)\frac{\delta R^{a}_{i}(x)}{\delta \phi_j(y)} - R^{a}_{j}(y)\frac{\delta R^b_i(x)}{\delta \phi_j(y)}\right) = if^{abc}R^{c}_{i}(x)\;.
\label{b21}
\end{equation}

\noindent Given a generic functional of the fields $F[\phi]$, we can write its variation under $\delta$ transformations as

\begin{equation}
\delta F[\phi] = -i\epsilon^{a}\mathcal{W}^{a}F\;,
\label{b22}
\end{equation}

\noindent with

\begin{equation}
\mathcal{W}^{a} = -\int d^dx~R^{a}_{i}(x)\frac{\delta}{\delta \phi_i(x)}\;,
\label{b23}
\end{equation}

\noindent and using eq.(\ref{b21}), we obtain

\begin{equation}
[\mathcal{W}^{a},\mathcal{W}^{b}] = if^{abc}\mathcal{W}^{c}\;.
\label{b24}
\end{equation}

\subsection{Linear Transformations}

Let us consider the particular example of transformations which are linear in the fields, \textit{i.e.} 

\begin{equation}
R^{a}_{i}(x) = {\Lambda^{a}_i}^{j}\phi_j(x)\;,
\label{b25}
\end{equation}

\noindent where $\Lambda^{a}$ are matrices which satisfy the algebra (\ref{b20}). We can express the invariance of the classical action under this transformation as

\begin{equation}
\mathcal{W}^{a}S = - \int d^dy\,{\Lambda^{a}_i}^{j}\phi_j\frac{\delta S}{\delta \phi_i} = 0\;.
\label{b26}
\end{equation}

\noindent Using eq.(\ref{b9}), it is possible to show that

\begin{equation}
\mathcal{W}^{a}\Gamma^{(0)}[\varphi] = \int d^dx\,J_i{\Lambda^{a}_{i}}^{j}\frac{\delta \EuScript{W}^{(0)}[J]}{\delta J_{j}} = 0\;.
\label{b27}
\end{equation}

\noindent Equations (\ref{b26}) and (\ref{b27}) are called \textit{Ward Identities}.

\subsection{Nonlinear Transformations}

We assume a general transformation, $R^{a}_{i}(x)$ which is not linear on the fields anymore. As explained before, we must couple these non-linearities to sources $\rho$. So, the classical action is modified by the introduction of sources terms in the form

\begin{equation}
S'[\phi] = S[\phi] + i\epsilon^{a}\int d^dx\,R^{a}_{i}\rho^i \equiv \Gamma^{(0)}\;.
\label{b28}
\end{equation} 

\noindent It happens, however, that the action defined by eq.(\ref{b28}) is not invariant under the transformations given by eq.(\ref{b18}). It is invariant if we consider the infinitesimal parameters $\epsilon$ to be Grassmann numbers and if we consider that they transform as well by

\begin{equation}
\delta\epsilon^{a} = -\frac{1}{2}f^{abc}\epsilon^{b}\epsilon^{c}\;.
\label{b29}
\end{equation}

\noindent In this way, the operator $\delta$ which defines the transformation is nilpotent and the we can write the symmetry in a functional form

\begin{equation}
\mathcal{S}(\Gamma^{(0)}) = \int d^dx~\frac{\delta\Gamma^{(0)}}{\delta\rho^{i}}\frac{\delta\Gamma^{(0)}}{\delta\phi^{i}} - \frac{1}{2}f^{abc}\epsilon^{b}\epsilon^{c}\frac{\partial\Gamma^{(0)}}{\partial\epsilon^{a}} = 0\;.
\label{b30}
\end{equation}

\noindent This is precisely the BRST transformations. Eq.(\ref{b30}) is known as \textit{Slavnov-Taylor Identity}.

\section{Quantum action principle - QAP}

A nice way to introduce the QAP \cite{Schwinger:1951xk,Lowenstein:1971vf,Lowenstein:1971jk,Lam:1972mb,Clark:1976ym} is dividing the discussion as we done in the last sections into linear and nonlinear symmetries. 

\subsection{Linear Symmetries}

For linear symmetries, it is not necessary to introduce external sources, since there is not any non-linearity to couple with. To simplify, we will denote this symmetry as

\begin{equation}
\delta\phi_i=P_i(x)\;,
\label{b31}
\end{equation}

\noindent and if it is a symmetry, than we know that

\begin{equation}
\delta S = \int d^dx\,P_i\frac{\delta S}{\delta\phi_{i}} = 0\;.
\label{b32}
\end{equation}

\noindent The QAP states that, at the quantum level, this Ward identity is 

\begin{equation}
\int d^dx\,P_i\frac{\delta\Gamma}{\delta\phi_i} = \Delta \cdot \Gamma\;,
\label{b33}
\end{equation}

\noindent where in the right hand side, there is an insertion as introduced in eq.(\ref{b17}). This insertion is arbitrary, but has to satisfy some minimal conditions: (i) $\Delta$ must be an integrated polynominal of fields, sources and their derivatives; (ii) By consistency, the dimension of the field combination is boundend by $(d - d_i + d_P)$, where $d$ is the dimension of spacetime, $d_i$ is the dimension of the field $\phi_i$ and $d_P$ is the dimension of the symmetry $P$; (iii) Again, by consistency, $\Delta$ must have the same quantum numbers of $\mathcal{W}$, which defines the Ward identity.

\subsection{Nonlinear Symmetries}

The QAP for nonlinear symmetries is of the same form for linear ones, but now we have to introduce external sources $\rho$ to couple with non-linearities. Therefore, the QAP is written as

\begin{equation}
\int d^dx\,\frac{\delta \Gamma}{\delta\rho^a_i}\frac{\delta \Gamma}{\delta \phi_i} = \Delta^a \cdot \Gamma \;.
\label{b34}
\end{equation}

\chapter{Local composite operator formalism} \label{LCO}

In this appendix we present a short overview of the so-called \textit{local composite operator technique} (LCO technique), a method to consistently obtain an effective potential for local composite operators where basic requirements such as renormalizability and consistency with the renormalization group at all orders in perturbation theory are obeyed. This method was introduced and developed in \cite{Verschelde:1995jj} and largely used \textit{e.g.} for the computation of the effective potential associated with dimension 2 condensates in Yang-Mills theories. In the following section, we show how to introduce dimension 2 composite operators using the LCO framework in such a way to be consistent with with QAP and BRST symmetry.

\section{Introduction of dimension two LCO}

Let us consider a LCO\footnote{The index $A$ is a multi-index representation.} $\mathcal{O}^{A}$ of dimension two with ghost number $z$ and an action $\Sigma$ invariant under BRST transformations\footnote{We will present the method with particular attention to the case of Yang-Mills theories.} (\ref{a11}). Following the QAP we must couple $\mathcal{O}^{A}$ to a source $J^{A}$ with ghost number $-z$ to introduce it in the action $\Sigma$. In order to ensure $J^A$ does not enter the non-trivial part of the BRST cohomology we introduce it as BRST doublet, namely

\begin{equation}
s\lambda^{A}=J^A\,\,\,\,\,\mathrm{and}\,\,\,\,\,sJ^{A}=0\,,
\label{c1}
\end{equation}

\noindent with $\lambda^A$ (with ghost number $-z-1$) introduces to form the BRST doublet. The LCO action, which is the original action $\Sigma$ taking into account the introduction of $\mathcal{O}^{A}$ is defined as

\begin{eqnarray}
\Sigma_{\mathrm{LCO}}&=&\Sigma+s\int d^4x\left(\lambda^{A}\mathcal{O}^{A}-\beta\frac{\zeta}{2}\lambda^AJ^A\right)\nonumber\\
&=&\Sigma+\int d^4x\left(J^A\mathcal{O}^{A}+(-1)^{-z-1}\lambda^A s\mathcal{O}^{A}-\beta(z)\frac{\zeta}{2}J^{A}J^{A}\right)\,,
\label{c2}
\end{eqnarray}

\noindent where the quadratic term in the sources is allowed by power counting as long as $z=0$ and $\zeta$ is in principle a free parameter known as LCO parameter. The parameter $\beta$ is such that $\beta(0)=1$ and $\beta(z\neq 0)=0$. Novel divergences $\propto J^2$ arise from the behavior of the correlator

\begin{equation}
\lim_{x\rightarrow y}\langle\mathcal{O}^{A}(x)\mathcal{O}^{A}(y)\rangle\,.
\label{c3}
\end{equation}

\noindent For this reason, the term $\zeta J^2/2$ is responsible to absorb $\delta\zeta J^2$ counterterms. From eq.(\ref{c2}), we compute immediately the equation of motion for $\lambda^A$,

\begin{equation}
\frac{\delta\Sigma_{\mathrm{LCO}}}{\delta\lambda^{A}}= (-1)^{-z-1}s\mathcal{O}^{A}\,,
\label{c4}
\end{equation}

\noindent which is equal to a BRST variation. If $\mathcal{O}^{A}$ is BRST invariant even at the \textit{on-shell} level, this equation of motion corresponds to a Ward identity which enables us to control the introduction of the LCO at the level of renormalizability. Although we have introduced the operator $\mathcal{O}^{A}$ in a consistent way with the QAP and BRST symmetry, the presence of quadratic term in the source $J^A$ spoils the standard interpretations adopted in quantum field theory. This ``problem" can be solved though by the introduction of the so-called Hubbard-Stratonovich fields. We discuss it in the next section.

\section{Hubbard-Stratonovich fields} \label{HSFields}

Since our interest is to study the effective potential associated with $\mathcal{O}^{A}$, we set $\lambda^{A}$ to zero. The functional $W[J]$ is defined as

\begin{equation}
\mathrm{e}^{-W[J]}\equiv\int\left[\EuScript{D}\Phi\right]\exp\left[-\Sigma-\int d^4x\left(J^A\mathcal{O}^A-\frac{\zeta}{2}J^AJ^A\right)\right]\,,
\label{c5}
\end{equation}

\noindent which implies

\begin{equation}
\frac{\delta W[J]}{\delta J^{A}}\Big|_{J^{A}=0}=-\langle \mathcal{O}^{A}\rangle\,.
\label{c6}
\end{equation}

\noindent Now we can write the unity in the following way

\begin{equation}
1=\EuScript{N}\int\left[\EuScript{D}\sigma\right]\exp\left[-\frac{1}{2\zeta}\int d^4x\left(\frac{\sigma^A}{g}+\mathcal{O}^{A}-\zeta J^A\right)\left(\frac{\sigma^A}{g}+\mathcal{O}^{A}-\zeta J^A\right)\right]\,,
\label{c7}
\end{equation}

\noindent with $\EuScript{N}$ is a normalization factor. An easy algebraic manipulation leads to 

\begin{eqnarray}
&-&\frac{1}{2\zeta}\left(\frac{\sigma^A}{g}+\mathcal{O}^{A}-\zeta J^A\right)\left(\frac{\sigma^A}{g}+\mathcal{O}^{A}-\zeta J^A\right)
=-\frac{\sigma^{A}\sigma^{A}}{2\zeta g^2}-\frac{\sigma^{A}\mathcal{O}^{A}}{\zeta g}\nonumber\\
&+&\frac{\sigma^{A}J^{A}}{g}-\frac{1}{2\zeta}\mathcal{O}^{A}\mathcal{O}^{A}-\frac{\zeta}{2}J^{A}J^{A}+\mathcal{O}^{A}J^{A}\,.
\label{c8}
\end{eqnarray}

\noindent Inserting (\ref{c7}) in (\ref{c5}) we obtain the following partition function (we omit the normalization factor),

\begin{equation}
\mathrm{e}^{-W[J]}=\int \left[\EuScript{D}\Phi\right]\left[\EuScript{D}\sigma\right]\exp\left[-\Sigma-\int d^4x\left(\frac{\sigma^{A}\sigma^{A}}{2\zeta g^2}+\frac{\sigma^{A}\mathcal{O}^{A}}{\zeta g}+\frac{1}{2\zeta}\mathcal{O}^{A}\mathcal{O}^{A}-\frac{\sigma^{A}J^{A}}{g}\right)\right]\,.
\label{c9}
\end{equation}

\noindent From eq.(\ref{c9}) is immediate 

\begin{equation}
\langle \mathcal{O}^{A}\rangle=-\frac{1}{g}\langle \sigma^{A}\rangle\,,
\label{c10}
\end{equation}

\noindent and also from eq.(\ref{c9}) without the $J^2$ term we can proceed with a standard analysis of the effective potential. 

\section{Condensation}

From eq.(\ref{c9}) we define the action $S_{\sigma}$ associated with the Hubbard-Stratonovich fields,

\begin{equation}
S_{\sigma}=\int d^4x\left(\frac{\sigma^{A}\sigma^{A}}{2\zeta g^2}+\frac{\sigma^{A}\mathcal{O}^{A}}{\zeta g}+\frac{1}{2\zeta}\mathcal{O}^{A}\mathcal{O}^{A}\right)\,.
\label{c11}
\end{equation}

\noindent To define the quantum action, we first introduce the classical or background field $\sigma^{\mathrm{cl}}_{J}$ given by

\begin{equation}
\sigma^{\mathrm{cl},A}=\frac{\delta W[J]}{\delta J^A(x)}=-\langle \sigma^{A}(x)\rangle_{J}\,.
\label{c12}
\end{equation} 

\noindent Considering we can write $J^A$ as a function of $\sigma^{\mathrm{cl}}$, we define the quantum action as 

\begin{equation}
\Gamma[\sigma^{\mathrm{cl}}]=W[\sigma^{\mathrm{cl}}]-\int d^4x~J^{A}_{\sigma^{\mathrm{cl}}}(x)\sigma^{\mathrm{cl},A}(x)\,,
\label{c13}
\end{equation}

\noindent where as previously discussed, 

\begin{equation}
\frac{\delta\Gamma[\sigma^{\mathrm{cl}}]}{\delta \sigma^{\mathrm{cl},A}(x)}=-J^{A}_{\sigma^{\mathrm{cl}}}(x)\,.
\label{c14}
\end{equation}

\noindent The effective potential is defined for a constant background configuration $\sigma^{A}_{*}$. The effective potential computed for $\sigma^{A}_{*}$ is defined as

\begin{equation}
\Gamma[\sigma_{*}]=-\int d^4x~V_{\mathrm{eff}}(\sigma_{*})\,\,\,\,\Rightarrow\,\,\,\,V_{\mathrm{eff}}(\sigma_{*})=-\frac{1}{V}\Gamma[\sigma_{*}]\,,
\label{c15}
\end{equation}

\noindent with $V$ the spacetime volume. Therefore,

\begin{equation}
\frac{d}{d\sigma^{A}_{*}}V_{\mathrm{eff}}=J^{A}\,,
\label{c16}
\end{equation}

\noindent where $V$ is absorbed somehow in $J$. We see then if $J^{A}=0$ and a trivial solution $\sigma^{A}_{*}\neq 0$ to the equation

\begin{equation}
\frac{d}{d\sigma^{A}_{*}}V_{\mathrm{eff}}=0\,,
\label{c17}
\end{equation}

\noindent we have a non-trivial extrema for the effective potential which meand $\mathcal{O}^{A}$ condenses. For a minimization of the effective potential, this condensate is said to be energetically favored and might be take into account in the theory. 

A last remark concerning the parameter $\zeta$. In principle, this parameter is free. However, if we demand the effective action $\Gamma$ obeys a homogeneous renormalization group equation, then we can fix $\zeta$ as a function of $g$ as

\begin{equation}
\zeta = \frac{\zeta_0}{g^2}+\zeta_1+g^2\zeta_2+g^4\zeta_3+\ldots\,.
\label{c18}
\end{equation}

\chapter{The construction of $A^h_{\mu}$} \label{constructionAh}

In this Appendix, we provide a derivation for the expression of the gauge invariant $A^h_\mu$ field which was crucial to formulate the non-perturbative BRST symmetry in Ch.~\ref{nonpBRSTRGZ}. We begin with the definition of the functional $f_{A}[u]$ given by

\begin{equation}
f_{A}[u]\equiv \mathrm{Tr\,}\int d^4x\,A^{u}_{\mu}A^{u}_{\mu}= \mathrm{Tr\,}\int d^4x\,\left(u^{\dagger}A_{\mu}u+\frac{i}{g}u^{\dagger}\partial_{\mu}u\right)\left(u^{\dagger}A_{\mu}u+\frac{i}{g}u^{\dagger}\partial_{\mu}u\right)\,.
\label{ah1}
\end{equation}

\noindent For a given gauge field configuration $A_{\mu}$, $f_{A}[u]$ is functional over its gauge orbit. A minimum $f_{A}[h]$ is obtained when

\begin{eqnarray}
\delta f_{A}[u]\big|_{u=h} &=& 0\nonumber\\
\delta^2 f_{A}[u]\big|_{u=h} &>& 0\,,
\label{ah2}
\end{eqnarray}

\noindent and a minimum is \textit{absolute} if

\begin{equation}
f_{A}[h] \leq f_{A}[u]\,,\,\,\, \forall u\,\,\in\,\, \mathcal{U}\,,
\label{ah3}
\end{equation}

\noindent where $\mathcal{U}$ is the space of local gauge transformations. With the absolute minimum $f_{A}[h]$ at our disposal, is possible to define a gauge invariant quantity through

\begin{equation}
\mathcal{A}^{2}_{\mathrm{min}}=\underset{\left\{u\right\}}{\mathrm{min}}\,\mathrm{Tr\,}\int d^4x\,A^{u}_{\mu}A^{u}_{\mu}=f_{A}[h]\,.
\label{ah4}
\end{equation}

\noindent Searching for absolute minimum is a tremendous task. However, we should start at least by demanding conditions (\ref{ah2}). To achieve this, we can perform an infinitesimal expansion around $h$ (which satisfies conditions (\ref{ah2})). We define

\begin{equation}
v=h\mathrm{e}^{ig\omega}\equiv h\mathrm{e}^{ig\omega^{a}T^a}\,,
\label{ah5}
\end{equation} 

\noindent with $T^a$ the $SU(N)$ generators as defined in Ap.~\ref{appendixA} and $\omega^a$ a small parameter. Due to this assumption, we retain terms up to $\omega^2$, which is enough for our interests. By definition,

\begin{eqnarray}
A^{v}_{\mu}&=& v^{\dagger}A_{\mu}v\frac{i}{g}v^{\dagger}\partial_{\mu}v\nonumber\\
&=&\mathrm{e}^{-ig\omega}h^{\dagger}A_{\mu}h\mathrm{e}^{ig\omega}+\frac{i}{g}\mathrm{e}^{-ig\omega}h^{\dagger}(\partial_{\mu}h)\mathrm{e}^{ig\omega}+\frac{i}{g}\mathrm{e}^{-ig\omega}\partial_{\mu}\mathrm{e}^{ig\omega}\nonumber\\
&=&\mathrm{e}^{-ig\omega}A^{h}_{\mu}\mathrm{e}^{ig\omega}+\frac{i}{g}\mathrm{e}^{-ig\omega}\partial_{\mu}\mathrm{e}^{ig\omega}\,,
\label{ah6}
\end{eqnarray}

\noindent where we have used the definition of $A^h_{\mu}$ and $h^{\dagger}h=1$. Expanding eq.(\ref{ah6}) up to quadratic order in $\omega$, we obtain

\begin{eqnarray}
A^{v}_{\mu}&=&\left(1-ig\omega-\frac{g^2}{2}\omega^2+\mathcal{O}(\omega^3)\right)A^{h}_{\mu}\left(1+ig\omega-\frac{g^2}{2}\omega^2+\mathcal{O}(\omega^3)\right)\nonumber\\
&+&\frac{i}{g}\left(1-ig\omega-\frac{g^2}{2}\omega^2+\mathcal{O}(\omega^3)\right)\partial_{\mu}\left(1+ig\omega-\frac{g^2}{2}\omega^2+\mathcal{O}(\omega^3)\right)\nonumber\\
&=& A^h_{\mu}+igA^{h}_{\mu}\omega-\frac{g^{2}}{2}A^{h}_{\mu}\omega^2-ig\omega A^{h}_{\mu}+g^2\omega A^{h}_{\mu}\omega-\frac{g^2}{2}\omega^2 A^{h}_{\mu}+\frac{i}{g}\left(ig\partial_{\mu}\omega\phantom{\frac{1}{2}}\right.\nonumber\\
&-&\left.\frac{g^2}{2}(\partial_{\mu}\omega)\omega-\frac{g^2}{2}\omega\partial_{\mu}\omega+g^2\omega\partial_{\mu}\omega\right)+\mathcal{O}(\omega^3)\,.
\label{ah7}
\end{eqnarray}

\noindent After few simple manipulations, we can rewrite eq.(\ref{ah7}) as

\begin{equation}
A^{v}_{\mu}=A^{h}_{\mu}-\partial_{\mu}\omega+\frac{ig}{2}[\omega,\partial_{\mu}\omega]+ig[A^{h}_{\mu},\omega]+\frac{g^2}{2}[[\omega,A^{h}_{\mu}],\omega]+\mathcal{O}(\omega^3)\,.
\label{ah8}
\end{equation}

\noindent Now, we explicitly compute $f_{A}[v]$,

\begin{eqnarray}
f_{A}[v]&=&\mathrm{Tr\,}\int d^4x\,A^{v}_{\mu}A^{v}_{\mu}\nonumber\\
&=& \mathrm{Tr\,}\int d^4x\,\left(A^{h}_{\mu}-\partial_{\mu}\omega+\frac{ig}{2}[\omega,\partial_{\mu}\omega]+ig[A^{h}_{\mu},\omega]+\frac{g^2}{2}[[\omega,A^{h}_{\mu}],\omega]+\mathcal{O}(\omega^3)\right)\nonumber\\
&\times&\left(A^{h}_{\mu}-\partial_{\mu}\omega+\frac{ig}{2}[\omega,\partial_{\mu}\omega]+ig[A^{h}_{\mu},\omega]+\frac{g^2}{2}[[\omega,A^{h}_{\mu}],\omega]+\mathcal{O}(\omega^3)\right)
\label{ah9}
\end{eqnarray}

\noindent which implies

\begin{eqnarray}
f_{A}[v]&=&  \mathrm{Tr\,}\int d^4x\,\left[A^{h}_{\mu}A^{h}_{\mu}-A^{h}_{\mu}\partial_{\mu}\omega+\frac{ig}{2}A^{h}_{\mu}[\omega,\partial_{\mu}\omega]+igA^{h}_{\mu}[A^{h}_{\mu},\omega]+\frac{g^{2}}{2}A^{h}_{\mu}[[\omega,A^{h}_{\mu}],\omega]\right.\nonumber\\
&-&\left.(\partial_{\mu}\omega)A^{h}_{\mu}+(\partial_{\mu}\omega)(\partial_{\mu}\omega)-ig(\partial_{\mu}\omega)[A^{h}_{\mu},\omega]+\frac{ig}{2}[\omega,\partial_{\mu}\omega]A^{h}_{\mu}+ig[A^{h}_{\mu},\omega]A^{h}_{\mu}\right.\nonumber\\
&-&\left.ig[A^{h}_{\mu},\omega]\partial_{\mu}\omega-g^{2}[A^{h}_{\mu},\omega][A^{h}_{\mu},\omega]+\frac{g^{2}}{2}[[\omega,A^{h}_{\mu}],\omega]A^{h}_{\mu}+\mathcal{O}(\omega^3)\right]\,.
\label{ah10}
\end{eqnarray}

\noindent After few algebraic steps, we obtain

\begin{eqnarray}
f_{A}[v]&=&f_{A}[h]+2\,\mathrm{Tr\,}\int d^4x\,\omega(\partial_{\mu}A^{h}_{\mu})+\mathrm{Tr\,}\int d^4x\,\left(2g^2\omega A^{h}_{\mu}\omega A^{h}_{\mu}-2g^2 A^{h}_{\mu}A^{h}_{\mu}\omega^2\right)\nonumber\\
&-& g^2\mathrm{Tr\,}\int d^4x\,(A^{h}_{\mu}\omega-\omega A^{h}_{\mu})(A^{h}_{\mu}\omega-\omega A^{h}_{\mu})+\mathrm{Tr\,}\int d^4x\,(\partial_{\mu}\omega)(\partial_{\mu}\omega)\nonumber\\
&-&\mathrm{Tr\,}\int d^4x\,ig(\partial_{\mu}\omega)[A^{h}_{\mu},\omega]+\mathcal{O}(\omega^3)\nonumber\\
&=& f_{A}[h]+2\,\mathrm{Tr\,}\int d^4x\,\omega(\partial_{\mu}A^{h}_{\mu})+\mathrm{Tr\,}\int d^4x\,(\partial_{\mu}\omega)\underbrace{(\partial_{\mu}\omega-ig[A^{h}_{\mu},\omega])}_{D_{\mu}(A^h)\omega}+\mathcal{O}(\omega^3)\nonumber\\
&=& f_{A}[h]+2\,\mathrm{Tr\,}\int d^4x\,\omega(\partial_{\mu}A^{h}_{\mu})-\mathrm{Tr\,}\int d^4x\,\omega\partial_{\mu}D_{\mu}(A^h)\omega+\mathcal{O}(\omega^3)\,.
\label{ah11}
\end{eqnarray}

\noindent From eq.(\ref{ah11}) we automatically satisfy condition (\ref{ah2}) with

\begin{eqnarray}
\partial_{\mu}A^{h}_{\mu}&=& 0 \nonumber\\
-\partial_{\mu}D_{\mu}(A^h)&>& 0\,.
\label{ah12}
\end{eqnarray}

\noindent Using the transversality condition $\partial_{\mu}A^{h}_{\mu} = 0$, we can solve $h=h(A)$ as a power series in $A_{\mu}$. As a result, we write $A^h_{\mu}=A^h_{\mu}(A)$ which is very useful. We start from the definition of $A^h_{\mu}$,

\begin{equation}
A^h_{\mu}=h^{\dagger}A_{\mu}h+\frac{i}{g}h^{\dagger}\partial_{\mu}h\,,
\label{ah13}
\end{equation}

\noindent and we write 

\begin{equation}
h=\mathrm{e}^{ig\phi^aT^a}\equiv\mathrm{e}^{ig\phi}=1+ig\phi-\frac{g^2}{2}\phi^2+\mathcal{O}(\phi^3)\,.
\label{ah14}
\end{equation}

\noindent Plugging eq.(\ref{ah14}) in eq.(\ref{ah13}), we obtain

\begin{eqnarray}
A^{h}_{\mu}&=&\left(1-ig\phi-\frac{g^2}{2}\phi^2+\mathcal{O}(\phi^3)\right)A_{\mu}\left(1+ig\phi-\frac{g^2}{2}\phi^2+\mathcal{O}(\phi^3)\right)\nonumber\\
&+&\frac{i}{g}\left(1-ig\phi-\frac{g^2}{2}\phi^2+\mathcal{O}(\phi^3)\right)\partial_{\mu}\left(1+ig\phi-\frac{g^2}{2}\phi^2+\mathcal{O}(\phi^3)\right)\nonumber\\
&=& A_{\mu}+igA_{\mu}\phi-\frac{g^{2}}{2}A_{\mu}\phi^2-ig\phi A_{\mu}+g^2\phi A_{\mu} \phi-\frac{g^2}{2}\phi^2 A_{\mu}-\partial_{\mu}\phi \nonumber\\
&-&\frac{ig}{2}(\phi\partial_{\mu}\phi+(\partial_{\mu}\phi)\phi)+ig\phi\partial_{\mu}\phi + \mathcal{O}(\phi^3)
\label{ah15}
\end{eqnarray}

\noindent The result is

\begin{equation}
A^{h}_{\mu}=A_{\mu}+ig[A_{\mu},\phi]-\frac{g^2}{2}A_{\mu}\phi^2+g^2\phi A_{\mu}\phi-\frac{g^2}{2}\phi^2 A_{\mu}-\partial_{\mu}\phi+\frac{ig}{2}[\phi,\partial_{\mu}\phi]+\mathcal{O}(\phi^3)\,.
\label{ah16}
\end{equation}

\noindent Imposing the transversality of $A^{h}_{\mu}$ on (\ref{ah16}), we write

\begin{eqnarray}
\partial^2\phi&=&\partial_{\mu}A_{\mu}+ig\partial_{\mu}[A_{\mu},\phi]-\frac{g^2}{2}\partial_{\mu}(A_{\mu}\phi^2)+g^2\partial_{\mu}(\phi A_{\mu}\phi)-\frac{g^2}{2}\partial_{\mu}(\phi^2 A_{\mu})+\frac{ig}{2}\partial_{\mu}[\phi,\partial_{\mu}\phi]\nonumber\\
&+&\mathcal{O}(\phi^3)\,,
\label{ah17}
\end{eqnarray}

\noindent which can be solved iteratively. For concreteness we will retain terms up to $A^2$,

\begin{equation}
\phi=\frac{1}{\partial^2}\partial A + \frac{ig}{2}\frac{1}{\partial^2}\left[\partial A,\frac{1}{\partial^2}\partial A\right]+ig\frac{1}{\partial^2}\left[A_{\alpha},\frac{\partial_{\alpha}}{\partial^2}\partial A\right]+\mathcal{O}(A^3)\,.
\label{ah18}
\end{equation}

\noindent Substituting eq.(\ref{ah18}) in (\ref{ah16}), we obtain an explicit expression for $A^h_{\mu}$ as a power series of $A_{\mu}$,

\begin{eqnarray}
A^{h}_{\mu}&=&A_{\mu}-\partial_{\mu}\frac{1}{\partial^2}\partial A+ig\left[A_{\mu},\frac{1}{\partial^2}\partial A\right]-ig\frac{1}{\partial^2}\partial_{\mu}\left[A_{\alpha},\partial_{\alpha}\frac{1}{\partial^2}\partial A\right]+\frac{ig}{2}\frac{1}{\partial^2}\partial_{\mu}\left[\frac{1}{\partial^2}\partial A,\partial A\right]\nonumber\\
&+&\frac{ig}{2}\left[\frac{1}{\partial^2}\partial A,\partial_{\mu}\frac{1}{\partial^2}\partial A\right] + \mathcal{O}(A^3)\,.
\label{ah19}
\end{eqnarray}

\noindent We can also rewrite eq.(\ref{ah19}) using the transverse projector $\mathcal{P}_{\mu\nu}$ which makes the transversality condition manifest,

\begin{equation}
A^{h}_{\mu}=\underbrace{\left(\delta_{\mu\nu}-\frac{\partial_{\mu}\partial_{\nu}}{\partial^2}\right)}_{\mathcal{P}_{\mu\nu}}\underbrace{\left(A_{\nu}-ig\left[\frac{1}{\partial^2}\partial A,A_\nu\right]+\frac{ig}{2}\left[\frac{1}{\partial^2}\partial A,\partial_{\nu}\frac{1}{\partial^2}\partial A\right]+\mathcal{O}(A^3)\right)}_{\phi_\nu}\,.
\label{ah20}
\end{equation}

\noindent It is possible to prove $A^{h}_{\mu}$ is gauge invariant order by order in $g$. As a first evidence, we use expression (\ref{ah20}) to prove it up to first order in $g$. The infinitesimal gauge transformation is

\begin{equation}
\delta A_{\mu}=-\partial_{\mu}\omega+ig[A_{\mu},\omega]\,,
\label{ah21}
\end{equation}

\noindent and we apply it to $\phi_{\nu}$,

\begin{eqnarray}
\delta\phi_\nu &=& -\partial_\nu\omega+ig[A_\nu,\omega]-ig\left[\frac{1}{\partial^2}(-\partial^2\omega),A_\nu\right]-ig\left[\frac{1}{\partial^2}(\partial A),(-\partial_{\nu}\omega)\right]\nonumber\\
&+&\frac{ig}{2}\left[\frac{1}{\partial^2}(-\partial^2\omega),\partial_{\nu}\frac{1}{\partial^2}\partial A\right]+\frac{ig}{2}\left[\frac{1}{\partial^2}\partial A,\partial_{\nu}\frac{1}{\partial^2}(-\partial^2\omega)\right]+\mathcal{O}(g^2)\nonumber\\
&=&-\partial_{\nu}\omega+ig[A_{\nu},\omega]+ig[\omega,A_{\nu}]+ig\left[\frac{1}{\partial^2}\partial A,\partial_{\nu}\omega\right]-ig\left[\omega,\partial_{\nu}\frac{1}{\partial^2}\partial A\right]\nonumber\\
&-&\frac{ig}{2}\left[\frac{1}{\partial^2}\partial A,\partial_{\nu}\omega\right]+\mathcal{O}(g^2)\nonumber\\
&=& -\partial_{\nu}\omega + \frac{ig}{2}\left(\left[\frac{1}{\partial^2}\partial A,\partial_{\nu}\omega\right]+\left[\partial_{\nu}\frac{1}{\partial^2}\partial A,\omega\right]\right)+\mathcal{O}(g^2)\nonumber\\
&=&-\partial_\nu\left(\omega-\frac{ig}{2}\left[\frac{1}{\partial^2}\partial A,\omega\right]\right)+\mathcal{O}(g^2)
\label{ah22}
\end{eqnarray}

\noindent Therefore

\begin{equation}
\delta A^{h}_{\mu}=\mathcal{P}_{\mu\nu}\delta\phi_\nu=\mathcal{O}(g^2)\,,
\label{ah23}
\end{equation}

\noindent which establishes the gauge invariance up to first order in $g$.

\chapter{Propagators of the local RGZ action in LCG}\label{proplrgzlcg}

In this appendix we collect the tree-level propagators for the fields appearing in the local RGZ action in LCG presented in Ch.~\ref{locnonpBRST}. The list of propagators is

\begin{eqnarray}
\langle A^a_{\mu}(p)A^b_{\nu}(-p)\rangle &=& \frac{p^2+M^2}{p^4+(m^2+M^2)p^2+m^2M^2+2Ng^2\gamma^4}\delta^{ab}\mathcal{P}_{\mu\nu}
+\frac{\alpha}{p^4}p_{\mu}p_{\nu}
\nonumber\\
\\
\langle A_{\mu}^a(p)b^b(-p)\rangle&=& -i\frac{p^2}{p^4+\alpha \tilde{m}^4}\delta^{ab}p_{\mu}
\\
\langle A_{\mu}^a(p)\varphi_{\nu}^{bc}(-p)\rangle=\langle A_{\mu}^a(p)\bar{\varphi}_{\nu}^{bc}(-p)\rangle
&=&\frac{g\gamma^2f^{abc}}{p^4+p^2(m^2+M^2)+m^2M^2+2Ng^2\gamma^4}\mathcal{P}_{\mu\nu}\\
\langle A_{\mu}^a(p)\xi^b(-p)\rangle&=& i\frac{\alpha\delta^{ab}}{p^4+\alpha \tilde{m}^4}p_{\mu}
\\
\langle A_{\mu}^a(p)\tau^b(-p)\rangle&=&-i\frac{\alpha \tilde{m}^4}{p^2(p^4+\alpha \tilde{m}^4)}p_{\mu}\delta^{ab}
\\
\langle b^a(p)b^b(-p)\rangle &=&-\frac{\tilde{m}^4}{p^4+\alpha \tilde{m}^4}\delta^{ab}
\\
\langle b^a(p)\xi^b(-p)\rangle&=&-\frac{p^2\delta^{ab}}{p^4+\alpha \tilde{m}^4}
\\
\langle b^a(p)\tau^b(-p)\rangle&=&\frac{\tilde{m}^4}{p^2}\delta^{ab}
\end{eqnarray}

\begin{eqnarray}
\langle\bar{c}^a(p)A^b_{\mu}(-p)\rangle &=&-i\frac{\rho\,\alpha}{p^2(p^4+\alpha \tilde{m}^4)}\delta^{ab}p_{\mu}\\
\langle\bar{c}^a(p)b^b(-p)\rangle &=&-\frac{\rho}{p^4+\alpha \tilde{m}^4}\delta^{ab}\\
\langle\bar{c}^a(p)\tau^b(-p)\rangle &=&\frac{\rho}{p^4+\alpha \tilde{m}^4}\delta^{ab}\\
\langle\bar{c}^a(p)\xi^b(-p)\rangle &=&\frac{\rho\,\alpha}{p^2(p^4+\alpha \tilde{m}^4)}\delta^{ab}
\end{eqnarray}
\begin{eqnarray}
\langle \bar\varphi_{\mu}^{ab}(p)\bar\varphi_{\nu}^{cd}(-p)\rangle
&=&\frac{g^2\gamma^4f^{abm}f^{mcd}}{(p^2+M^2)[p^4+p^2(m^2+M^2)+m^2M^2+2Ng^2\gamma^4]}\mathcal{P}_{\mu\nu}
\nonumber\\
&=& \langle \varphi_{\mu}^{ab}(p)\varphi_{\nu}^{cd}(-p)\rangle\\
&&\hspace{-6cm}\langle \bar\varphi_{\mu}^{ab}(p)\varphi_{\nu}^{cd}(-p)\rangle
=\frac{g^2\gamma^4f^{abm}f^{mcd}}{(p^2+M^2)[p^4+p^2(m^2+M^2)+m^2M^2+2Ng^2\gamma^4]}\mathcal{P}_{\mu\nu}
-\frac{\delta^{ac}\delta^{bd}}{p^2+M^2}\delta_{\mu\nu}\nonumber\\
\\
\langle \varphi_{\mu}^a(p)\tau^b(-p)\rangle=\langle \bar{\varphi}_{\mu}^a(p)\tau^b(-p)\rangle
&=&-i\frac{g\gamma^2}{p^2(p^2+M^2)}p_{\mu}f^{abc}
\\
\langle\xi^a(p)\xi^b(-p)\rangle&=&\frac{\alpha\delta^{ab}}{p^4+\alpha \tilde{m}^4}
\\
\langle\xi^a(p)\tau^b(-p)\rangle &=&\frac{p^2}{p^4+\alpha \tilde{m}^4}\delta^{ab}
\\
\langle\tau^a(p)\tau^b(-p)\rangle&=&-\left\{\frac{m^2(p^4-\alpha \tilde{m}^4)+\tilde{m}^4p^2}{p^2(p^4+\alpha \tilde{m}^4)}
+\frac{2Ng^2\gamma^4}{p^2(p^2+M^2)}\right\}\delta^{ab}  
\end{eqnarray}


\end{document}